%% file: main.tex
\documentclass{lmcs} %%% last changed 2014-08-20
\pdfoutput=1

\usepackage{lastpage}
\lmcsdoi{19}{4}{23}
\lmcsheading{}{\pageref{LastPage}}{}{}%
{Feb.~01,~2023}{Dec.~14,~2023}{}

\keywords{linear logic, cut elimination, cost models, explicit substitutions}

\usepackage{hyperref}
\theoremstyle{plain} %\crefname{satz}{Satz}{S\"atze}
\def\eg{{\em e.g.}}

\newcommand{\macrospath}{./macros}

\input{\macrospath/ben-macros-preamble}
\setboolean{LMCSstyle}{true}
\input{\macrospath/ben-macros-paper}

\input{\macrospath/ben-macros-diagrams}
\input{\macrospath/ben-macros-nets}

\setboolean{techr}{true}
\input{local-macros}
\setboolean{withproofs}{true}

\begin{document}

\title[Exponentials as Substitutions]{Exponentials as Substitutions,\texorpdfstring{\\}{} and the Cost of Cut Elimination in Linear Logic}
\author[B. Accattoli]{Beniamino Accattoli\lmcsorcid{0000-0003-4944-9944}}

% affiliation 1 (automatically numbered a)
\address{Inria \& LIX, \'Ecole Polytechnique, France}	%optional
% write emails for all authors having that affiliation
\email{beniamino.accattoli@inria.fr}  %optional

%% etc.

%% required for running head on odd and even pages, use suitable
%% abbreviations in case of long titles and many authors:

%%%%%%%%%%%%%%%%%%%%%%%%%%%%%%%%%%%%%%%%%%%%%%%%%%%%%%%%%%%%%%%%%%%%%%%%%%%

%% the abstract has to PRECEDE the command \maketitle:
%% be sure not to issue the \maketitle command twice!

\begin{abstract}
 \input{00-Abstract}

\end{abstract}

\maketitle

\input{01-Introduction}
\input{02-Contributions}
\input{03-Sub-Term_Property_and_Size_Explosion}
\input{04-The_Linear_Substitution_Calculus}
\input{05-Towards_the_IMELL_calculus}
\input{06-The_IMELL_substitution_calculus}
\input{07-Basic_Properties}
\input{08-Local_Termination}
\input{09-Some_Technicalities}
\input{10-Untyped_Confluence}
\input{11-The_Good_Strategy}

\input{12-Untyped_PSN}
\input{13-Typed_Strong_Normalization}

\input{14-Conclusions}

\medskip

\section*{Acknowledgments}
\noindent  To Delia Kesner, for asking the question that triggered this work. To Giulio Guerrieri, Claudio Sacerdoti Coen, and Olivier Laurent for feedback. 

  %% the following bibliography is gererated manually for the sake of brevity
  %% only; please use a separate .bib file in your submission

\bibliographystyle{alphaurl}
\bibliography{biblio.bib}

\cameratech{}{
\onecolumn
\pagebreak
\appendix
\input{APP-Typed_Strong_Normalization}

}
\end{document}

%% file: macros/ben-macros-preamble.tex
\usepackage{ifthen}
 \usepackage{xspace}
 
\newboolean{talk}
\setboolean{talk}{false}
\newboolean{paper}
\setboolean{paper}{false}

\newboolean{LMCSstyle}
\setboolean{LMCSstyle}{false}
\newboolean{IEEEstyle}
\setboolean{IEEEstyle}{false}
\newboolean{lipicsstyle}
\setboolean{lipicsstyle}{false}
\newboolean{eptcsstyle}
\setboolean{eptcsstyle}{false}
\newboolean{entcsstyle}
\setboolean{entcsstyle}{false}
\newboolean{sigplanstyle}
\setboolean{sigplanstyle}{false}
\newboolean{easychairstyle}
\setboolean{easychairstyle}{false}
\newboolean{scrartclstyle}
\setboolean{scrartclstyle}{false}
\newboolean{lncsstyle}
\setboolean{lncsstyle}{false}
\newboolean{acmartstyle}
\setboolean{acmartstyle}{false}
\newboolean{needstheorems}
\setboolean{needstheorems}{false}
\newboolean{withimages}
\setboolean{withimages}{false}
\newboolean{withproofs}
\setboolean{withproofs}{true}
\newboolean{complexity}
\setboolean{complexity}{false}

\newboolean{techr}
\setboolean{techr}{false}

\newboolean{french}
\setboolean{french}{false}

\newboolean{colored}
\setboolean{colored}{false}

%% file: macros/ben-macros-paper.tex
\ifthenelse{\boolean{talk}}{}{\usepackage[utf8]{inputenc}}
\ifthenelse{\boolean{french}}{\usepackage[french]{babel}}{
  \ifthenelse{\boolean{lipicsstyle}}{}{\usepackage[english]{babel}}}
\ifthenelse{\boolean{acmartstyle}}{}{
\usepackage{amssymb}	
}
\usepackage{amsmath}
\usepackage{amsfonts}
\usepackage{graphicx}
\usepackage{bussproofs}
\usepackage{cmll}
\usepackage{stmaryrd}
 \usepackage{multirow}
\ifthenelse{\boolean{talk}}{\usepackage{color}}{
 	\ifthenelse{\boolean{lipicsstyle}}{\usepackage{color}}{
   		\ifthenelse{\boolean{acmartstyle}}{\usepackage{color}}{
			\ifthenelse{\boolean{LMCSstyle}}{}{
				\usepackage[usenames]{xcolor}
			}
		}
	}
}
 \usepackage{fancybox}
 \usepackage{hhline}
\usepackage{nameref}
\ifthenelse{\boolean{LMCSstyle}}{}{
	\usepackage[basic]{complexity}
	}
\usepackage{appendix}
\usepackage[normalem]{ulem}
\usepackage{proof}
\usepackage{mathtools}
\ifthenelse{\boolean{lipicsstyle}}{}{\usepackage{tabls}}

\ifthenelse{\boolean{IEEEstyle}}{
 	\usepackage[bookmarks={false}]{hyperref}
	}{
	\ifthenelse{\boolean{acmartstyle}}{}{\usepackage{hyperref}}
	}

\ifthenelse{\boolean{IEEEstyle}}{
\setboolean{needstheorems}{true}}{}

\ifthenelse{\boolean{eptcsstyle}}{
\setboolean{needstheorems}{true}}{}

\ifthenelse{\boolean{scrartclstyle}}{
\setboolean{needstheorems}{true}}{}

\ifthenelse{\boolean{sigplanstyle}}{
\setboolean{needstheorems}{true}}{}

\ifthenelse{\boolean{easychairstyle}}{
\setboolean{needstheorems}{true}}{}

\ifthenelse{\boolean{lipicsstyle}}{
    }{}

\ifthenelse{\boolean{needstheorems}}{
\input{\macrospath/ben-macros-thm-environments}}{}

\newcommand{\myproof}[1]{
\ifthenelse{\boolean{withproofs}}{#1}{}
}

\newcommand{\la}[1]{\lambda #1.}

\newcommand{\tm}{t}
\newcommand{\tmtwo}{s}
\newcommand{\tmthree}{u}
\newcommand{\tmfour}{r}
\newcommand{\tmfive}{p}
\newcommand{\tmsix}{q}

\newcommand{\var}{x}
\newcommand{\vartwo}{y}
\newcommand{\varthree}{z}

\newcommand{\mvar}{m}
\newcommand{\mvartwo}{n}
\newcommand{\mvarthree}{o}

\newcommand{\dom}[1]{\symfont{dom}(#1)}

\newcommand{\rootRew}[1]{\mapsto_{#1}}
\newcommand{\Rew}[1]{\rightarrow_{#1}}
\newcommand{\Rewn}[1]{\rightarrow_{#1}^*}

\renewcommand{\to}{\Rew{}}
\newcommand{\tob}{\Rew{\beta}}

\newcommand{\toone}{\Rew{1}}
\newcommand{\totwo}{\Rew{2}}

\newcommand{\lRew}[1]{\; \mbox{}_{#1}{\leftarrow}}

\newcommand{\towh}{\Rew{wh}}

\newcommand{\sn}[1]{SN_{#1}}

\newcommand{\symfont}[1]{\mathtt{#1}}
\newcommand{\ssym}{\symfont{s}}

\newcommand{\lssym}{\symfont{ls}}
\newcommand{\ls}{\lssym}
\newcommand{\gcsym}{\symfont{gc}}
\newcommand{\gc}{\gcsym}
\newcommand{\db}{\symfont{dB}}

\newcommand{\esym}{\symfont{e}}
\newcommand{\wsym}{\symfont{w}}

\newcommand{\msym}{\symfont{m}}

\newcommand{\csym}{\symfont{c}}

\newcommand{\val}{v}
\newcommand{\valtwo}{\val'}
\newcommand{\valthree}{\val''}

\newcommand{\ctxholep}[1]{\langle #1\rangle}
\newcommand{\ctxhole}{\ctxholep{\cdot}}

\makeatletter
\newsavebox{\@brx}
\newcommand{\llangle}[1][]{\savebox{\@brx}{\(\m@th{#1\langle}\)}%
  \mathopen{\copy\@brx\kern-0.7\wd\@brx\usebox{\@brx}}}
\newcommand{\rrangle}[1][]{\savebox{\@brx}{\(\m@th{#1\rangle}\)}%
  \mathclose{\copy\@brx\kern-0.7\wd\@brx\usebox{\@brx}}}
\makeatother

\newcommand{\ctxholefp}[1]{\llangle #1\rrangle}
\newcommand{\ctxholef}{\ctxholefp{\cdot}}

\newcommand{\ctx}{C}
\newcommand{\ctxtwo}{D}
\newcommand{\ctxthree}{E}
\newcommand{\ctxfour}{F}
\newcommand{\ctxp}[1]{\ctx\ctxholep{#1}}

\newcommand{\ctxtwop}[1]{\ctxtwo\ctxholep{#1}}
\newcommand{\ctxthreep}[1]{\ctxthree\ctxholep{#1}}
\newcommand{\ctxfourp}[1]{\ctxfour\ctxholep{#1}}

\newcommand{\ctxfp}[1]{\ctx\ctxholefp{#1}}
\newcommand{\ctxtwofp}[1]{\ctxtwo\ctxholefp{#1}}
\newcommand{\ctxthreefp}[1]{\ctxthree\ctxholefp{#1}}
\newcommand{\ctxfourfp}[1]{\ctxfour\ctxholefp{#1}}

\newcommand{\hctx}{H}
\newcommand{\hctxp}[1]{\hctx\ctxholep{#1}}

\newcommand{\nbvctxtwo}[1]{\nbvctxtwo{#1}}

\newcommand{\defeq}{:=}

\newcommand{\grameq}{::=}

\newcommand{\esub}[2]{[#1{\shortleftarrow}#2]}
\newcommand{\isub}[2]{\{#1{\shortleftarrow}#2\}}

\ifthenelse{\boolean{complexity}}{ }{
  
  }

\newcommand{\rtodb}{\rootRew{\db}}

\newcommand{\tos}{\Rew{\ssym}}
\newcommand{\todbp}[1]{\Rew{\db#1}}
\newcommand{\todb}{\todbp{}}%{\rightsquigarrow}

\newcommand{\rtols}{\rootRew{\lssym}}
\newcommand{\rtogc}{\rootRew{\gcsym}}
\newcommand{\tols}{\Rew{\lssym}}
\newcommand{\togc}{\Rew{\gcsym}}

\newcommand{\tolh}{\Rew{\symfont{lh}}}

\newcommand{\llbrace}{\{ \kern -0.27em \vert}
\newcommand{\rrbrace}{\vert \kern -0.27em \}}

\newcommand{\streq}{\equiv}
\renewcommand{\l}{\lambda}
\newcommand{\ie}{i.e.\xspace}
\ifthenelse{\boolean{LMCSstyle}}{	
\renewcommand{\eg}{{e.g.}\xspace}
}{
	\newcommand{\eg}{{e.g.}\xspace}
}
\newcommand{\ih}{{\textit{i.h.}}\xspace}
\newcommand{\fv}[1]{\symfont{fv}(#1)}
\newcommand{\dfv}[1]{\symfont{dv}(#1)}

\newcommand{\mfv}[1]{\symfont{mfv}(#1)}
\newcommand{\efv}[1]{\symfont{efv}(#1)}
\ifthenelse{\boolean{entcsstyle}}{}{
	}

\newcommand{\red}[1]{{\color{red} {#1}}}

\definecolor{azure}{rgb}{0.0, 0.5, 1.0}
\definecolor{dcyan}{rgb}{0.0, 0.72, 0.92}
\definecolor{deepskyblue}{rgb}{0.0, 0.75, 1.0}
\definecolor{lightseagreen}{rgb}{0, 0.65, 0.65}

\newcommand{\lblue}[1]{{\color{lightseagreen} {#1}}}

\definecolor{dgreen}{rgb}{0.0, 0.5, 0.0}

\newcommand{\ignore}[1]{}

\newcommand{\colspace}{@{\hspace{.5cm}}}
\newcommand{\myinput}[1]{\ifthenelse{\boolean{withimages}}{\input{#1}}{}}
\newcommand{\inputproof}[1]{\ifthenelse{\boolean{withproofs}}{\input{#1}}{}}

\newcommand{\reflemma}[1]{Lemma~\ref{l:#1}}
\newcommand{\reflemmap}[2]{Lemma~\ref{l:#1}.\ref{p:#1-#2}}

\newcommand{\reflemmaeq}[1]{{L.\ref{l:#1}}}
\newcommand{\reflemmaeqp}[2]{{L.\ref{l:#1}.\ref{p:#1-#2}}}
\newcommand{\reflemmapeq}[2]{\reflemmaeqp{#1}{#2}}

\newcommand{\refthm}[1]{Theorem~\ref{thm:#1}}

\newcommand{\reftm}[1]{Theorem~\ref{tm:#1}}

\newcommand{\refprop}[1]{Proposition~\ref{prop:#1}}
\newcommand{\refpropp}[2]{Prop.~\ref{prop:#1}.\ref{p:#1-#2}} %Giulio
\newcommand{\refsect}[1]{Sect.~\ref{sect:#1}}

\newcommand{\refapp}[1]{Appendix~\ref{sect:#1}}

\ifthenelse{\boolean{easychairstyle}}{}{
  \ifthenelse{\boolean{lncsstyle}}{}{	
	\renewcommand{\refeq}[1]{(\ref{eq:#1})}
  }
}
\newcommand{\reffig}[1]{Fig.~\ref{fig:#1}}
\newcommand{\refcoro}[1]{Corollary~\ref{coro:#1}}

\newcommand{\refdef}[1]{Definition~\ref{def:#1}}

\newcommand{\refrem}[1]{Remark~\ref{rem:#1}}

\newcommand{\mellies}{{Melli{\`e}s}\xspace}
\newcommand{\levy}{{L{\'e}vy}\xspace}

\newcommand{\IMELL}{{IMELL}\xspace}
\newcommand{\imell}{\IMELL}

\newcommand{\form}{A}
\newcommand{\formtwo}{B}
\newcommand{\formthree}{C}

\newcommand{\aform}{X}

\newcommand{\multiForm}{\Gamma}
\newcommand{\multiFormtwo}{\Delta}
\newcommand{\multiFormthree}{\Pi}
\newcommand{\multiform}{\multiForm}
\newcommand{\multiformtwo}{\multiFormtwo}
\newcommand{\multiformthree}{\multiFormthree}

\newcommand{\exch}{\symfont{exch}}

\newcommand{\ax}{\symfont{ax}}
\newcommand{\cut}{\symfont{cut}}

\newcommand{\tens}{\otimes}
\newcommand{\lolli}{\multimap}
\newcommand{\substr}{\varogreaterthan}

\newcommand{\bang}{{\mathsf{!}}}
\newcommand{\whyn}{{\mathsf{?}}}
\newcommand{\der}{{\mathsf{d}}}
\newcommand{\contr}{{\mathsf{c}}}
\newcommand{\weak}{{\mathsf{w}}}

\newcommand{\rightRuleSym}{r}
\newcommand{\leftRuleSym}{l}

\newcommand{\tensRightRule}{\tens_\rightRuleSym}
\newcommand{\tensLeftRule}{\tens_\leftRuleSym}

\newcommand{\lolliRightRule}{\lolli_\rightRuleSym}
\newcommand{\lolliLeftRule}{\lolli_\leftRuleSym}

\newcommand{\bangRightRule}{\bang_\rightRuleSym}
\newcommand{\bangLeftRule}{\bang_\leftRuleSym}

\newcommand{\contrRule}{\contr}
\newcommand{\weakRule}{\weak}

\newcommand{\set}[1]{\{#1\}}

\newcommand{\nat}{\mathbb{N}}

\newcommand{\settone}{\mathcal{S}}

\newcommand{\setecone}{\mathcal{E}}

\newcommand{\net}{P}

\newcommand{\tolsc}{\Rew{LSC}}
\newcommand{\rtom}{\mapsto_\msym}

\newcommand{\tom}{\Rew{\msym}}

\newcommand{\tow}{\Rew{\wsym}}

\newcommand{\rtow}{\rootRew{\wsym}}

\newcommand{\betav}{\beta_\val}

\newcommand{\snsubst}{\sn{\to}}

\newcommand{\redc}[1]{\llbracket#1\rrbracket}

\newcommand{\size}[1]{|#1|}
\newcommand{\sizep}[2]{|#1|_{#2}}

\ifthenelse{\boolean{acmartstyle}}{
	
	}{
	
}

\newcommand{\pair}[2]{(#1,#2)}

\newcommand{\eqstruct}{\equiv}

\newcommand{\withproofs}[1]{\ifthenelse{\boolean{withproofs}}{#1}{}}

\newcommand{\withoutproofs}[1]{\ifthenelse{\boolean{withproofs}}{}{#1}}

\ifthenelse{\boolean{entcsstyle}}{}{
  \ifthenelse{\boolean{lncsstyle}}{}{	
    
  }
 }

\newcounter{numberone}

\newcommand{\bigo}{{\mathcal{O}}}

\ifthenelse{\boolean{acmartstyle}}{}{
	
}

\newcommand{\tderiv}{\pi}
\newcommand{\tderivtwo}{\rho}

\newcommand{\Id}{\symfont{I}}

\newcommand{\hastype}{\,{:}\,}

\makeatletter
\newcommand{\exder}{%
  \def\exderW[##1]{\triangleright_{##1}\ }%
  \def\exderWO{\triangleright\ }%
  \@ifnextchar[\exderW\exderWO%
  }
\makeatother

\newcommand\Deribbase[5]{{#3}\ {\pmb\vdash}_{#2}^{#1} {#4}\  {:}\  {#5}}

\makeatletter

\newcommand{\Deribase}[1]{%
  \def\DeribW[##1]{\Deribbase{##1}{#1}}%
  \def\DeribWO{\Deribbase{}{#1}}%
  \@ifnextchar[\DeribW\DeribWO%
  }

  \newcommand{\Deri}{%
  \def\DeriW_##1{\Deribase{##1}}%
  \def\DeriWO{\Deribase{}}%
  \@ifnextchar_\DeriW\DeriWO%
  }
\makeatother

\newcommand{\app}{\symfont{app}}

\makeatletter
\newcommand{\appresult}{%
  \def\appresultW<##1>{\app_\result^{##1}}%
  \def\appresultWO{\app_\result}%
  \@ifnextchar<\appresultW\appresultWO%
  }
\makeatother

\newcommand\mydots{\hbox to .6em{.\hss.}}

\newcommand{\pof}{\;\triangleright}

\newcommand{\bdersym}{\bang_{\der}}
\renewcommand{\bdersym}{\bang}
\newcommand{\bder}{\bdersym}

\newcommand{\bw}{\bwsym}

\newcommand{\axe}{\axesym}

\newcommand{\tobder}{\Rew{\bdersym}}
\newcommand{\tobw}{\Rew{\bwsym}}

\newcommand{\ems}{\esym^{\mssym}}

\newcommand{\ess}{\esym^{\sssym}}

\newcommand{\toems}{\Rew{\ems}}

\newcommand{\rtoess}{\rootRew{\ess}}

\newcommand{\rtobder}{\rootRew{\bdersym}}

\newcommand{\varmeas}[2]{\sizep{#1}{#2}}

\newcommand{\llterms}{\symfont{Terms}}
\newcommand{\lltermsp}[1]{\llterms_{#1}}

\newcommand{\elctxs}{\symfont{ElCtxs}}
\newcommand{\lltermsform}{\lltermsp{\vdash\form}}

\newcommand{\elctxsform}{\elctxs_{\form}}
\newcommand{\vars}{\symfont{Vars}}
\renewcommand{\vars}{\mathcal{V}}
\newcommand{\elvars}{\symfont{ElVars}}

\newcommand{\elctx}{E}

\newcommand{\elctxfp}[1]{\elctx\ctxholefp{#1}}
\newcommand{\elctxtwo}{\elctx'}

\newcommand{\elctxtwofp}[1]{\elctxtwo\ctxholefp{#1}}

\newcommand{\gctx}{G}
\newcommand{\gctxp}[1]{\gctx\ctxholep{#1}}

\newcommand{\gctxtwo}{\gctx'}

\newcommand{\potmul}[2]{\symfont{p}_{#2}(#1)}
\newcommand{\ltmeas}[1]{\symfont{ms}(#1)}

\newcommand{\Gsym}{\symfont{G}}
\newcommand{\tog}{\togp{}}
\newcommand{\togp}[1]{\Rew{\Gsym #1}}

\newcommand{\Bsym}{\symfont{B}}
\newcommand{\tobad}{\tobadp{}}
\newcommand{\tobadp}[1]{\Rew{\Bsym #1}}

\newcommand{\bctx}{B}
\newcommand{\bctxtwo}{\bctx'}

\newcommand{\bctxp}[1]{\bctx\ctxholep{#1}}

\newcommand{\bctxtwop}[1]{\bctxtwo\ctxholep{#1}}

\renewcommand{\outer}{\prec}

\newcommand{\cctx}{\mathbb{C}}
\newcommand{\cctxtwo}{\cctx'}

\newcommand{\cctxp}[1]{\cctx\ctxholep{#1}}

\newcommand{\mctx}{M}
\newcommand{\mctxp}[1]{\mctx\ctxholep{#1}}
\newcommand{\mctxfp}[1]{\mctx\ctxholefp{#1}}
\newcommand{\mctxtwo}{\mctx'}
\newcommand{\mctxtwop}[1]{\mctxtwo\ctxholep{#1}}
\newcommand{\mctxtwofp}[1]{\mctxtwo\ctxholefp{#1}}
\newcommand{\mctxthree}{\mctx''}
\newcommand{\mctxthreep}[1]{\mctxthree\ctxholep{#1}}

\newcommand{\vctx}{V}
\newcommand{\vctxp}[1]{\vctx\ctxholep{#1}}
\newcommand{\vctxfp}[1]{\vctx\ctxholefp{#1}}

\newcommand{\vmctx}{V_{\msym}}

\newcommand{\rcuteq}{\sim_\cut}
\newcommand{\cuteq}{\equiv_\cut}

\newcommand{\letin}[3]{{\sf let}\ #1=#2\ {\sf in}\ #3}

\newcommand{\vgctx}{V_{G}}

\newcommand{\lambdamutcalc}{\overline\lambda\mu\tilde{\mu}}

\newcommand{\cuta}[2]{\redbrackets{#1{\shortrightarrow}#2}}
\newcommand{\cutsub}[2]{\set{#1{\shortrightarrow}#2}}

\newcommand{\redbrackets}[1]{\red[#1\red]}

\newcommand{\bluebrackets}[1]{\lblue[#1\lblue]}

\newcommand{\para}[3]{\bluebrackets{#1{\parr}#2,#3}}
\newcommand{\suba}[3]{\bluebrackets{#1{\substr}#2,#3}}

\newcommand{\dera}[2]{\bluebrackets{#1{\whyn}#2}}

\newcommand{\axmsym}{\ax\msym}
\newcommand{\axesym}{\ax\esym}
\newcommand{\axmone}{\axmsym_{1}}
\newcommand{\axmtwo}{\axmsym_{2}}

\newcommand{\rtoaxmone}{\rootRew{\axmone}}
\newcommand{\rtoaxmtwo}{\rootRew{\axmtwo}}
\newcommand{\rtotens}{\rootRew{\tens}}
\newcommand{\rtololli}{\rootRew{\lolli}}

\newcommand{\toaxmone}{\Rew{\axmone}}
\newcommand{\toaxmtwo}{\Rew{\axmtwo}}
\newcommand{\totens}{\Rew{\tens}}
\newcommand{\tololli}{\Rew{\lolli}}

\newcommand{\axeone}{\axe_{1}}
\newcommand{\axetwo}{\axe_{2}}

\newcommand{\rtoaxeone}{\rootRew\axeone}
\newcommand{\toaxeone}{\Rew\axeone}
\newcommand{\rtoaxetwo}{\rootRew\axetwo}
\newcommand{\toaxetwo}{\Rew\axetwo}

\newcommand{\rtobang}{\rootRew{\bang}}
\newcommand{\tobang}{\Rew{\bang}}

\newcommand{\mssym}{\symfont{ms}}
\newcommand{\sssym}{\symfont{ss}}
\newcommand{\toms}{\Rew{\mssym}}
\newcommand{\toss}{\Rew{\sssym}}

\newcommand{\toess}{\Rew{\ess}}
\newcommand{\tomsnw}{\Rew{\mssym_{\neg\wsym}}}
\newcommand{\msnlolli}{\mssym_{\neg\lolli}}
\newcommand{\tomsnlolli}{\Rew{\msnlolli}}
\newcommand{\ssnlolli}{\sssym_{\neg\lolli}}
\newcommand{\tossnlolli}{\Rew{\ssnlolli}}

\newcommand{\lctx}{L}
\newcommand{\lctxtwo}{\lctx'}
\newcommand{\lctxthree}{\lctx''}

\newcommand{\lctxp}[1]{\lctx\ctxholep{#1}}

\newcommand{\lctxtwop}[1]{\lctxtwo\ctxholep{#1}}
\newcommand{\lctxthreep}[1]{\lctxthree\ctxholep{#1}}

\newcommand{\mvars}{\vars_{\msym}}
\newcommand{\evars}{\vars_{\esym}}

\newcommand{\evar}{e}
\newcommand{\evartwo}{f}
\newcommand{\evarthree}{g}

\newcommand{\exval}{\val_{\esym}}
\newcommand{\exvaltwo}{\exval'}

\newcommand{\mval}{\val_{\msym}}
\newcommand{\mvaltwo}{\mval'}

\newcommand{\cameratech}[2]{\ifthenelse{\boolean{techr}}{#2}{#1}}

\definecolor{LightGray}{gray}{.80}
\definecolor{DarkGray}{gray}{.60}

%% file: macros/ben-macros-diagrams.tex
\usepackage{tikz}
\usetikzlibrary{calc}
\usetikzlibrary{matrix,arrows}
\usetikzlibrary{decorations.shapes}
\usetikzlibrary{decorations.text}
\usetikzlibrary{decorations.pathmorphing}
\usetikzlibrary{decorations.markings}
\usetikzlibrary{positioning}

\tikzset{
node distance=1.3cm, auto,
every node/.style={font=\scriptsize },
ocenter/.style={baseline={([yshift=-.5ex, xshift=-.5ex]current bounding box)}},  
labelBeginAbove/.style={postaction={decorate,decoration={markings,mark=at position 0 with {\node[inner sep= 0.6pt, above=1pt]{\tiny #1};}} } },
labelBeginBelow/.style={postaction={decorate,decoration={markings,mark=at position 0 with {\node[inner sep= 0.6pt, below=1pt]{\tiny #1};}}}},
labelEndAbove/.style={postaction={decorate,decoration={markings,mark=at position 1 with {\node[inner sep= 0.6pt, above=2pt]{\tiny #1};}}}},
labelEndBelow/.style={postaction={decorate,decoration={markings,mark=at position 1 with {\node[inner sep= 0.6pt, below=2pt]{\tiny #1};}}}},
labelEndRight/.style={postaction={decorate,decoration={markings,mark=at position 1 with {\node[inner sep= 0.6pt, right=2pt]{\tiny #1};}}}},
labelEndLeft/.style={postaction={decorate,decoration={markings,mark=at position 1 with {\node[inner sep= 0.6pt, left=2pt]{\tiny #1};}}}}
}

\newcommand{\med}[2]{
($(#1)!.5!(#2)$)
}

%% file: macros/ben-macros-nets.tex
\usepackage{tikz}
\usetikzlibrary{calc}
\usetikzlibrary{backgrounds}
\usetikzlibrary{decorations.shapes}
\usetikzlibrary{decorations.text}
\usetikzlibrary{decorations.pathmorphing}
\usetikzlibrary{decorations.markings}
\usetikzlibrary{matrix}
\usetikzlibrary{decorations.pathreplacing}
\usetikzlibrary{arrows}
\usetikzlibrary{positioning}
\usepackage[nofancy]{tikz-inet}

\setboolean{withimages}{true}

\tikzset{%
  ocenter/.style={baseline={([yshift=-.5ex, xshift=-.5ex]current bounding box)}},  
  nospace/.style={inner sep= 0pt},
  etic/.style={inner sep= 0.5pt, fill=white, anchor= center},
  bigetic/.style={inner sep= 0.8pt, fill=white, anchor= center},
  letic/.style={inner sep= 0.7pt, fill=white, anchor= center,circle,draw=black},
  port/.style={inner sep= 0.3pt, anchor= center,circle,draw=black,fill=black, minimum size = 0.3pt},
  eport/.style={inner sep= 0.8pt, anchor= center,circle,draw=cyan,fill=white, minimum size = 0.8pt},
  mport/.style={inner sep= 0.8pt, anchor= center, circle,draw=white,fill=brown, minimum size = 0.8pt, solid, 
line width=0.10ex},
  eportn/.style={inner sep= 0.8pt, anchor= center,circle,draw=blue,fill=white, minimum size = 0.8pt},
  mportn/.style={inner sep= 0.8pt, anchor= center, circle,draw=red,fill=red, minimum size = 0.8pt, solid, 
line width=0.10ex},
  nopol/.style={->, shorten <=0.5pt, shorten >=0.5pt, draw=gray, line width=0.18ex},
  nopolrev/.style={<-, shorten <=1pt, shorten >=1pt, draw=gray, line width=0.18ex},
  nopolgen/.style={preaction={decorate},decoration={markings,mark=at position .5 with {\draw [shorten >=0pt, shorten <=0pt,draw=gray,-](0pt,3pt) -- (0pt,-3pt);}},->, shorten <=0.5pt, shorten >=0.5pt, draw=gray, line width=0.18ex},
  nopolrevgen/.style={preaction={decorate},decoration={markings,mark=at position .5 with {\draw [shorten >=0pt, shorten <=0pt,draw=gray,-](0pt,3pt) -- (0pt,-3pt);}},<-, shorten <=1pt, shorten >=1pt, draw=gray, line width=0.18ex},
  outEdge/.style={<-, shorten >=0.5pt, shorten <=0.5pt, draw=blue, line width=0.18ex},
  inpEdge/.style={->, densely dotted, shorten >=0.5pt, shorten <=0.5pt, draw=red, line width=0.18ex, overlay},
  inpEdgeGen/.style={preaction={decorate},decoration={markings,mark=at position .5 with {\draw [shorten >=0pt, shorten <=0pt,draw=red,-](0pt,3pt) -- (0pt,-3pt);}},->, densely dotted, shorten >=0.5pt, shorten <=0.5pt, draw=red, line width=0.18ex, overlay},
    eprinc/.style={postaction={decorate},decoration={markings,mark=at position .4 with {\node [inner sep= 0.5pt, anchor= center,circle,draw=cyan, fill=cyan, minimum size = 0.8pt, solid, line width=0.10ex]{};}}},
    mprinc/.style={postaction={decorate},decoration={markings,mark=at position .4 with {\node [inner sep= 0.5pt, anchor= center,circle,draw=brown, fill=brown, minimum size = 0.8pt, solid, line width=0.10ex]{};}}},
    mprincpar/.style={postaction={decorate},decoration={markings,mark=at position .6 with {\node [inner sep= 0.5pt, anchor= center,circle,draw=brown, fill=brown, minimum size = 0.8pt, solid, line width=0.10ex]{};}}},
    pure/.style={->, shorten <=1pt, shorten >=1pt, draw=brown, line width=0.18ex},
    pureRev/.style={<-, shorten <=1pt, shorten >=1pt, draw=brown, line width=0.18ex},
    exppure/.style={->, shorten <=1pt, shorten >=1pt, draw=cyan, line width=0.18ex, densely dotted},
    exppureRev/.style={<-, shorten <=1pt, shorten >=1pt, draw=cyan, line width=0.18ex, densely dotted},    genexppure/.style={preaction={decorate},decoration={markings,mark=at position .5 with {\draw [shorten >=0pt, shorten <=0pt,draw=cyan,-](0pt,3pt) -- (0pt,-3pt);}},->, shorten <=1pt, shorten >=1pt, draw=cyan, line width=0.18ex, densely dotted},
    enaryedge/.style={preaction={decorate},decoration={markings,mark=at position .5 with {\draw [shorten >=0pt, shorten <=0pt,draw=cyan,-](0pt,3pt) -- (0pt,-3pt);}}},
    eprincn/.style={postaction={decorate},decoration={markings,mark=at position .4 with {\node [inner sep= 0.7pt, anchor= center,circle,draw=blue, fill=blue, minimum size = 0.8pt, solid, line width=0.10ex]{};}}},
    mprincn/.style={postaction={decorate},decoration={markings,mark=at position .6 with {\node [inner sep= 0.7pt, anchor= center,circle,draw=red, fill=white, minimum size = 0.8pt, solid, line width=0.10ex]{};}}},
    mprincnpar/.style={postaction={decorate},decoration={markings,mark=at position .4 with {\node [inner sep= 0.7pt, anchor= center,circle,draw=red, fill=white, minimum size = 0.8pt, solid, line width=0.10ex]{};}}},
    puren/.style={->, shorten <=1pt, shorten >=1pt, draw=red, line width=0.18ex},
    purenRev/.style={<-, shorten <=1pt, shorten >=1pt, draw=red, line width=0.18ex},
    exppuren/.style={->, shorten <=1pt, shorten >=1pt, draw=blue, line width=0.18ex, densely dotted},
    exppurenRev/.style={<-, shorten <=1pt, shorten >=1pt, draw=blue, line width=0.18ex, densely dotted},    genexppuren/.style={preaction={decorate},decoration={markings,mark=at position .5 with {\draw [shorten >=0pt, shorten <=0pt,draw=blue,-](0pt,3pt) -- (0pt,-3pt);}},->, shorten <=1pt, shorten >=1pt, draw=blue, line width=0.18ex, densely dotted},
    enaryedge/.style={preaction={decorate},decoration={markings,mark=at position .5 with {\draw [shorten >=0pt, shorten <=0pt,draw=blue,-](0pt,3pt) -- (0pt,-3pt);}}},
  jboxline/.style={draw= gray,rounded corners, line width=0.20ex, overlay},
  exboxline/.style={draw= gray,line width=0.20ex, overlay},
  noboxline/.style={draw= white,rounded corners, line width=0ex},
  net/.style={draw=gray,inner sep=2pt,thick,ellipse, anchor=center, font=\scriptsize},
  inductiveTr/.style={draw=black!50, minimum size=0.9cm},
  inductiveTrSmall/.style={draw=black!50, minimum size=0.6cm},
  every label/.style={label distance = 1pt, font=\scriptsize, inner sep= 1pt},  
  every node/.style={font=\tiny }%\small
}
\newcommand{\altax}{8pt}

\newcommand{\stlarb}{25pt}

\newcommand{\stalt}{22pt}
\newcommand{\stlar}{15pt}
\newcommand{\hstalt}{\stalt/2}
\newcommand{\hstlar}{\stlar/2}

\newcommand{\ilar}{12pt}

\newcommand{\pendingdist}{10pt}

\input{\macrospath/graphical-macros/generic-links}



\newcommand{\gdots}[3]{
\node at ($(#1)!.5!(#2)$) [#3]{\tiny $\ldots$};
}

%% file: macros/graphical-macros/generic-links.tex
\newcommand{\lBinEdgesAbove}[5]{
\draw[#4, in=150, out=-90] (#1) to (#3);
\draw[#5, in=30, out=-90] (#2) to (#3);
}

\newcommand{\stboxlw}{6pt}
\newcommand{\stboxh}{20pt}

\newcommand{\stboxrwaux}{35pt}
\newcommand{\sepbox}{2pt}

\newcommand{\boxnodes}[4]{
\node at (#1.center)[left = #2,nospace](#1so){};
\node at (#1.center)[right = #3,nospace](#1se){};
\node at (#1se.center)[above=#4,nospace](#1ne){};
\node at (#1so.center)[above=#4,nospace](#1no){};
}

\newcommand{\boxline}[2]{
\draw[#2](#1.center) -- (#1se.center) -- (#1ne.center) -- (#1no.center) -- (#1so.center)--(#1.center);}

\newcommand{\abox}[5]{
\boxnodes{#1}{#3}{#4}{#5}
\boxline{#1}{#2line}}

\newcommand{\boxauxnodesright}[3]{
\node at (#1so.center)[right = #2, nospace](#1auxk){};
\node at (#1auxk.center)[right= #3, nospace](#1aux1){};
}

\newcommand{\boxauxnodeslab}[4]{
\node at (#1auxk.center)[below = #4,etic](#1auxktarget){\tiny #3};
\node at (#1aux1.center)[below = #4,etic](#1aux1target){\tiny #2};

\draw[nopol] (#1auxk) to (#1auxktarget);
\draw[nopol] (#1aux1) to (#1aux1target);
}

%% file: local-macros.tex
\usepackage{appendix}
\capturecounter{thm}
\capturecounter{figure}

\usepackage{microtype}
\usepackage{version}

\renewcommand{\l}{\lambda}

\renewcommand{\net}{G}

\renewcommand{\mvar}{m}
\renewcommand{\mvartwo}{n}
\renewcommand{\mvarthree}{o}

\newcommand{\proper}{proper\xspace}

\renewcommand{\ctxtwo}{\ctx'}
\renewcommand{\ctxthree}{\ctx''}
\renewcommand{\ctxfour}{\ctx'''}

\renewcommand{\mctx}{M}

\renewcommand{\aform}{X_{\msym}}

\renewcommand{\refprop}[1]{Prop.~\ref{prop:#1}}

\renewcommand{\tobw}{\tow}
\renewcommand{\bw}{\wsym}

\renewcommand{\refsect}[1]{Section~\ref{sect:#1}}

\newcommand{\lhs}{\symfont{lhs}}
\newcommand{\rtolhs}{\rootRew{\lhs}}
\newcommand{\tolhs}{\Rew{\lhs}}
\renewcommand{\tolh}{\Rew{\symfont{lh}}}

%% file: 00-Abstract.tex
% !TeX spellcheck = en_US
% !TEX root = main.tex
%%%%%%%%%%%%%%%%%%%%%%
%%%%%%%%%%%%%%%%%%%%%%
%%%%%%%%%%%%%%%%%%%%%%
 \noindent This paper introduces the \emph{exponential substitution calculus} (ESC), a new presentation of cut elimination for IMELL based on proof terms and building on the idea that exponentials can be seen as explicit substitutions. The idea in itself is not new, but here it is pushed to a new level, inspired by Accattoli and Kesner's linear substitution calculus (LSC). 

One of the key properties of the LSC is that it naturally models the sub-term property of abstract machines, which is the key ingredient for the study of reasonable time cost models for the $\lambda$-calculus. The new ESC is then used to design a cut elimination strategy with the sub-term property, providing the first polynomial  cost model for cut elimination with unconstrained exponentials. 

For the ESC, we also prove untyped confluence and typed strong normalization, showing that it is an alternative to proof nets for an advanced study of cut elimination.

%% file: 01-Introduction.tex
% !TeX spellcheck = en_US
% !TEX root = main.tex
%%%%%%%%%%%%%%%%%%%%%%
%%%%%%%%%%%%%%%%%%%%%%
%%%%%%%%%%%%%%%%%%%%%%
\section{Introduction}
\label{sect:intro}
Two key aspects of linear logic are  its  resource-awareness and that it  models the evaluation of $\l$-terms via cut elimination, even of untyped $\l$-terms, if recursive formulas are allowed. One would then expect that, given a $\l$-term $\tm$ represented as a linear proof $\pi_{\tm}$, the length of cut elimination in $\pi_{\tm}$ could provide estimates about the time complexity of $\tm$. Surprisingly, this is not (yet) the case. 

\paragraph{Linear Logic and Complexity Classes} Linear logic is often used in \emph{implicit computational complexity}, a field that aims at characterizing complexity classes with no explicit references to machine models. In this line of work, classes are characterized by seeing program execution as cut elimination in fragments of, and variations on, linear logic, usually obtained by constraining the exponential connectives in some way. Some representative papers are \cite{DBLP:journals/tcs/GirardSS92,DBLP:journals/iandc/Girard98,DBLP:journals/iandc/DanosJ03,DBLP:journals/tcs/Lafont04,DBLP:conf/ictcs/MairsonT03}. %The simplest instance is  \emph{elementary linear logic} \cite{DBLP:journals/iandc/Girard98,DBLP:journals/iandc/DanosJ03}, which is obtained from linear logic by removing the exponential principles of \emph{dereliction} and \emph{digging}, and that characterizes the elementary time class (cut elimination takes at most elementary time and every elementary time Turing machine can be represented). 

A first limitation of these results is that the bound is given on \emph{whole} fragments of linear logic, without saying how to compute the cost of a fixed proof, which might be much lower than the bound for the fragment. A second limitation is that, as soon as one steps out of the fragment that characterizes the class, nothing is known. %And this does not even require stepping out of the class: for instance, there are many proofs for which cut elimination has elementary cost and that do not belong to elementary linear logic because they use dereliction or digging. 
%At present, for a proof of MELL---which is the smallest fragment where the exponentials are  \emph{unconstrained}---there is no abstract way to extimate the cost of cut elimination. 

\paragraph{The Cut Elimination Clock} The underlying problem is that there are no known complexity measures for cut elimination in linear logic. The use of linear logic for implicit complexity is so implicit that, not only it avoids machine models, it also never states what is the underlying logical unit for time, that is, what are the \emph{ticks} of the cut elimination \emph{clock}. 

Formally, what is missing is a \emph{polynomial time cost model} for cut elimination. That is, a cut elimination strategy of which \emph{the number of steps} is a bound, up to a polynomial overhead, to the time complexity of implementing cut elimination according to that strategy. But be careful, it is not the number of steps that has to be polynomial: it would not be possible, linear logic cut elimination can take way more than polynomially many steps (in the size of the initial proof). It is the cost of implementing the sequence of steps (on a random access machine, or any other reasonable framework) that must be polynomial---ideally linear---in their number (which can be whatever) and in the size of the proof. 
Then, the number of steps taken by the strategy on a proof $\pi$ provides a reliable measure for the time cost of $\pi$. Roughly, the steps of such a strategy become the ticks of a cut elimination clock.

One might wonder why not any strategy would do. The point is that cut elimination strategies often suffer of \emph{size explosion}: one can build families of proofs $\{\pi_n\}_{n\in\nat}$ on which the strategy iterates duplications in malicious ways, leading to a growth of the proof size which is exponential in the number $k_n$ of cut elimination steps taken by the strategy on $\pi_n$. Therefore, $k_n$ cannot be taken as the time cost for implementing the cut elimination sequence (not even up to a polynomial), if the sequence has to be implemented as it is. Intuitively, taking such strategies as clocks would correspond to having irregular ticks of non-uniform length, some of which take a very long time. In \refsect{sub-term-size-explosion}, we shall show that cut elimination \emph{by levels},  the  strategy of reference in linear logic (especially for implicit computational complexity), suffers from size explosion.

The aim of this paper is to provide a clock of which the ticks are regular enough to serve as a reliable time measure for unconstrained exponentials. That is, we shall provide a new strategy of which the number of steps is a polynomial time cost model. In our pursuit, we are inspired by developments in the study of cost models for the $\l$-calculus and in the modern theory of $\l$-calculi with explicit substitutions. 

\paragraph{The Linear Substitution Calculus} In the last decade, the study of \emph{reasonable} time cost models, that is, of polynomial cost models that \emph{additionally} are equivalent to the time cost model of Turing machines, has advanced considerably in the sister field of $\l$-calculus, starting with Accattoli and Dal Lago's result about the number of steps of the leftmost strategy \cite{DBLP:journals/corr/AccattoliL16}. The advances have been enabled and developed over Accattoli and Kesner's \emph{linear substitution calculus} (shortened to LSC)  \cite{DBLP:journals/entcs/Milner07,DBLP:conf/rta/Accattoli12,DBLP:conf/popl/AccattoliBKL14}, which is a neat and compact refinement of the $\l$-calculus that  probably can be considered the answer to the quest for the canonical $\l$-calculus with explicit substitutions (shortened to ESs). The LSC refines ESs with ideas from both proof nets, namely using \emph{contextual} rewriting rules on terms---also called  \emph{at a distance}---to avoid commuting constructors (taken from Accattoli and Kesner  \cite{DBLP:conf/csl/AccattoliK10}), and the $\pi$-calculus/bigraphs, in modeling duplication as replication in $\pi$ (following Milner \cite{DBLP:journals/entcs/Milner07}). 

\paragraph{The Sub-Term Property} The relevance of the LSC for cost models is due to its natural modeling of strategies having the \emph{sub-term property} often found in abstract machines: 
\begin{center}
\emph{All sub-terms duplicated  or erased along an evaluation sequence from $\tm$ are sub-terms of $\tm$}.
\end{center}
When a strategy $\Rew{\symfont{st}}$ has the sub-term property, the cost of implementing a sequence $\tm \Rew{\symfont{st}}^{k} \tmtwo$  is polynomial, and in general \emph{linear}, in $k$ and in the size $\size\tm$ of $\tm$. Therefore, the sub-term property implies that the number of steps is a polynomial cost model for time. Intuitively, it states that the steps of the strategy are the \emph{regular ticks} of a reliable clock. It represents, for the study of cost models, what the sub-formula property is for proof-search, or the cut-elimination theorem for sequent calculi.

Typically, important strategies of the $\l$-calculus such as weak head, head, or leftmost-outermost reduction do not have the sub-term property (actually, we show in \refsect{sub-term-size-explosion} that no strategy in the $\l$-calculus has the property), while their analogous in the LSC does.  One of the key points is that the sub-term property requires \emph{micro-steps} evaluation, that is, performing one variable replacement at a time (as in the LSC), rather than small-step evaluation, that is, using rules resting on meta-level substitution such as $\beta$. The sub-term property and the related degeneracy of size explosion are discussed in detail in \refsect{sub-term-size-explosion}.

\paragraph{Linear Logic and the Sub-Term Property} In linear logic, the problem is not just that there are no known cut elimination strategies providing polynomial cost models. What is worse, is that there are no known micro-step strategies with the sub-term property (as we discuss below), despite the clear linear logical flavor of the property. Therefore, the study of cost models is, at present, simply \emph{out of reach}.

How can one recover the sub-term property in linear logic? This is the challenge addressed by this paper. There are two natural possible routes. One is working inside linear logic as it is usually presented, and try to recover the sub-term property. This would inevitably mean working with proof nets, as cut elimination in the sequent calculus requires to deal with too many commutative cases. 
Another one is to develop an alternative, \emph{commutation-free} presentation of linear logic akin to the LSC, and use it to design a strategy with the sub-term property. We choose the second option, for three reasons. 
\begin{enumerate}
\item \emph{Terms are more easily manageable}: graphical syntaxes are too hard to manage for the long and delicate proofs typical of the study of cost models. 
\item \emph{Better rewriting}: the rewriting theory of the LSC is better behaved than the one of proof nets. 
\item \emph{Novelty}: we aim at a fresh look at linear logic, importing ideas from a sister field.
\end{enumerate}

\paragraph{This Paper} We provide three main contributions:
\begin{enumerate}
\item \emph{Design}: we introduce the \emph{exponential substitution calculus} (ESC), which is a new presentation of linear logic cut elimination based over the theory of ESs and the LSC and free from commutative cases. It is a design study, done in a principled way. It is also a \emph{stress test}, as linear logic has many more constructors and rewriting rules than the LSC. We shall see that, in order to guarantee some expected properties of ESs, one is forced to turn to \emph{intuitionistic} linear logic.

\item \emph{Sub-term strategy}: we define a new cut elimination strategy for the ESC and prove that it has the sub-term property. Therefore, we solve the issue mentioned at the beginning of the paper, providing the first polynomial time cost model for unconstrained exponentials.

\item \emph{Foundations}: we prove various key properties of the ESC, among which untyped confluence and typed strong normalization, thus providing solid foundations for our presentation. The non-trivial proofs of these properties are developed from scratch, using elegant proof techniques developed for the LSC.
\end{enumerate}
The next section provides overviews of these contributions.

\paragraph{Related Work} To our knowledge, there are no works in the literature studying cost models for linear logic. Many term and process calculi for or related to linear logic have been proposed, for instance \cite{DBLP:conf/lics/LincolnM92,DBLP:journals/tcs/Abramsky93,DBLP:conf/tlca/BentonBPH93,DBLP:conf/mfps/Wadler93,DBLP:journals/tcs/BellinS94,DBLP:conf/lics/BentonW96,DBLP:journals/sLogica/RoccaR97,DBLP:journals/tcs/MaraistOTW99,DBLP:conf/rta/Simpson05,DBLP:conf/rta/OhtaH06,DBLP:conf/concur/CairesP10,DBLP:conf/icfp/Wadler12,DBLP:conf/popl/CurienFM16,DBLP:conf/ppdp/EhrhardG16}. There are also proposals for $\l$-calculi with ESs with linear features such as  \cite{DBLP:journals/igpl/GhaniPR00,DBLP:journals/iandc/KesnerL07,DBLP:journals/logcom/FernandezS14}. None of these calculi employs rewriting rules at a distance as we do here. Mazza uses a natural deduction linear calculus at a distance \cite{Mazza:Hab}, and Kesner develops a $\l$-calculus with ESs reflecting proof nets cut elimination \cite{DBLP:journals/pacmpl/Kesner22} but only for the fragment representing the $\l$-calculus. Various authors consider terms for proof \emph{structures} \cite{DBLP:journals/tcs/Abramsky93,DBLP:conf/ppdp/FernandezM99,DBLP:journals/entcs/MackieS08,DBLP:conf/csl/Ehrhard14,DBLP:journals/lmcs/ChouquetA21} which then need correctness criteria, not required here.

The sub-term property is a folklore property first called as such by Accattoli and Dal Lago  \cite{DBLP:conf/rta/AccattoliL12}, who also show a surprising link with the standardization theorem in \cite{DBLP:journals/corr/AccattoliL16}.

\paragraph{Journal Version and Proofs} This paper is the journal version of the LICS 2022 conference paper with the same title. It adds explanations, in particular \refsect{sub-term-size-explosion} is new, and most proofs (together with the intermediary lemmas) that were omitted from the conference version. A few proofs are particularly long and tedious so they are still omitted or partially omitted, but can be found on Arxiv in the technical report \cite{DBLP:journals/corr/abs-lics} associated to the conference paper, which is accessible as version one ("v1", see the bibliography entry for the link) of the present paper on Arxiv (which is instead "v3").

%% file: 02-Contributions.tex
% !TeX spellcheck = en_US
% !TEX root = main.tex
%%%%%%%%%%%%%%%%%%%%%%%%
%%%%%%%%%%%%%%%%%%%%%%%%
\section{Overview of the Contributions}
%%%%%%%%%%%%%%%%%%%%%%%%
%%%%%%%%%%%%%%%%%%%%%%%%
\subsection{Contribution 1: The Exponential Substitution Calculus}
By allowing duplication and erasure of
sub-proofs, the exponentials can be seen as a substitution device: it is the core of the simulation of the $\l$-calculus into linear logic due to Girard \cite{DBLP:journals/tcs/Girard87}. 
Historically, ESs became popular with a calculus by Abadi et al. \cite{DBLP:journals/jfp/AbadiCCL91}, later shown defective at the rewriting level by \mellies \cite{DBLP:conf/tlca/Mellie95}. Linear logic and the exponentials were thus considered as a more solid formalism to borrow from. Di Cosmo and Kesner were the first ones to do so \cite{DBLP:conf/lics/CosmoK97}, seeing ESs as exponential cuts.

Along the years, however, the theory of ESs made progresses of its own, and here   we reverse the transfer. The idea is enforcing and pushing to the extreme the slogan:
\begin{center}
\begin{tabular}{ccccc}
Exponentials &=& Explicit Substitutions
\end{tabular}
\end{center}
To continue, we need to recall the very basics of ESs. 

\paragraph{Basics of Explicit Substitutions} The idea is to extend the $\l$-calculus with a new term constructor $\tm\esub\var\tmtwo$
(which is just a compact notation for $\letin\var\tmtwo\tm$, but with no fixed order of evaluation between $\tmtwo$ and $\tm$) denoting a delayed or \emph{explicit} substitution, and decomposing the $\beta$-rule:
\begin{center}$
\begin{array}{ccccc}
(\la\var\tm)\tmtwo &\tob &\tm\isub\var\tmtwo
\end{array}
$\end{center}
where $\tm\isub\var\tmtwo$ is meta-level substitution, into two rules ($\symfont{e}$xplicit $\beta$ and $\symfont{s}$ubstitution):
\begin{center}$
\begin{array}{ccccc}
(\la\var\tm)\tmtwo &\Rew{\symfont{e}\beta}& \tm\esub\var\tmtwo &\tos& \tm\isub\var\tmtwo
\end{array}
$\end{center}
Now,  rule $\tos$ is usually further decomposed into various \emph{micro-step} rules. There are many possible sets of such substitution rules. What they all have in common, even the defective ones mentioned above, is that, when considered \emph{separately} from rule $\Rew{\symfont{B}}$, they are strongly normalizing even without types.

\input{figure-classical-diverging}
\paragraph{Design Principles} Our \emph{exponentials as substitutions} design of the ESC is based on the following principles.

\emph{Principle 1: proofs as typing derivations}. Calculi with ESs are studied as \emph{untyped} calculi. Then we want the ESC to be an untyped calculus with good properties by itself, and see proofs of linear logic as typing derivations for it.

\emph{Principle 2: ESs terminates/intuitionism}. As we recalled above, the rules that manage ESs are always strongly normalizing without types. It is possible to define also \emph{untyped} linear logic proofs, as done \eg in de Carvalho et al. \cite{DBLP:journals/tcs/CarvalhoPF11}, obtaining a notion of untyped proof net. In the case of classical linear logic \emph{à la} Girard (with involutive negation), the exponentials of such proof nets can diverge by themselves, as in \reffig{classical-diverging}. In such a setting, then, exponentials are \emph{not} as ESs. We then switch to \emph{intuitionistic} linear logic, and prove that therein untyped exponentials are strongly normalizing (\refthm{local-termination}, p. \pageref{thm:local-termination}). To our knowledge, this is a new result and the first time that such a discrepancy between the classical and intuitionistic case is pointed out. The proof of the theorem requires long calculations but it is otherwise simple. Termination is obtained via a measure defined by induction over terms. Essentially, the result states that in the intuitionistic setting \emph{exponentials are as ESs}. It is one of the main contributions of the paper.

\emph{Principle 3: micro and small steps together}. The sub-term property requires a micro-step operational semantics. But studying micro-step ESs constantly requires to refer to meta-level substitution, which is small-step. For linear logic, Girard originally considered micro-step cut elimination \cite{DBLP:journals/tcs/Girard87}, but soon afterwards Regnier introduced a small-step variant, to study the $\l$-calculus \cite{Reg:Thesis:92}, and both are frequently used in the literature. In the ESC, we want small and micro steps to \emph{co-exist}, so as to provide a comprehensive framework. %For instance, we shall prove strong normalization at the small-step level and then transfer it to the micro one, but for confluence we shall go the other way around.

%\emph{Principle 4: the sequent calculus as is}. At the logical level, we want a standard sequent calculus which should not be tweaked in any way, only decorated with terms. On the one hand, to show that the LSC technology is flexible enough. On the other hand, for the results to be as solid as possible. In particular, there is the ambition of showing that our framework is an alternative to proof nets, usually considered \emph{the} tool for studying cut elimination in linear logic.

\paragraph{Left Splitting} The design of the ESC systematically exploits a simple fact about IMELL proofs: they can be seen as ending on a sequence of left rules following a right rule or an axiom. Representing a proof as a term $\tm$, one can thus uniquely split it as $\tm = \lctxp\val$, that is, a sub-proof ending on a right rule or an axiom, what we shall call a \emph{value} $\val$, and a possibly empty sequence of left rules $\lctx$ from the last sequent of $\val$ to the last sequent of $\tm$. Such a \emph{left splitting} is a basic property of sequent proofs that is not related to polarity or focussing---it is in fact simpler---and yet plays a crucial role in our rewriting rules at a distance for IMELL.

\paragraph{Additives} We do not consider the additive connectives, and so we deal with IMELL. The reason is that they are both trivial and challenging, depending on the approach. If one adopts \emph{additive slices} \cite{DBLP:journals/tcs/Girard87,LTdF04}, then our results lift smoothly. Without slices, instead, it is unclear how to combine additives and LSC-style micro steps while retaining all the good properties of the ESC.

\paragraph{Syntactical Variants} A number of alternative syntaxes could have been adopted, typically Curien \& Herbelin $\lambdamutcalc$ calculi \cite{DBLP:conf/icfp/CurienH00,DBLP:conf/popl/CurienFM16}, Benton's linear-non-linear approach \cite{DBLP:conf/csl/Benton94}, Pfenning \& Caires' processes \cite{DBLP:conf/concur/CairesP10}, and, of course, proof nets or natural deduction. We prefer avoiding them to show that their additional ingredients (the distinguished conclusion on the left of sequents \cite{DBLP:conf/icfp/CurienH00,DBLP:conf/popl/CurienFM16}, the linear-non-linear separation \cite{DBLP:conf/csl/Benton94}, and the graphical syntax) are not required for studying cost models for linear logic, while we avoided processes because their evaluation does not compute cut-free proofs, since they do not evaluate under prefixes. Finally, we prefer the sequent calculus to natural deduction because it is more commonly used for presenting linear logic, and because in this way we provide a novel use of the LSC technology, given that the LSC is based on natural deduction. Our study can however be adapted to any of these settings. %Perhaps also to classical linear logic.

\subsection{Contribution 2: the Sub-Term Strategy}
The proof nets strategy of reference is the notion of \emph{least level} (or \emph{by levels}) strategy $\Rew{ll}$ (introduced---we believe---by Girard \cite{DBLP:journals/iandc/Girard98} and studied for instance by de Carvalho et al. \cite{DBLP:journals/tcs/CarvalhoPF11}), which reduces cuts of minimal level, where the level is the number of $\bang$-boxes surrounding a cut. 

Unfortunately, $\Rew{ll}$ is not a good candidate for a linear logic clock. Consider the following local confluence diagram made out of micro steps at level $0$ (the used rewriting rules are in \reffig{pn-rules}), which is an important sub-strategy $\Rew{l_0}$ of $\Rew{ll}$:
\begin{figure}
\input{figure-pn-rules}
\caption{The two rewriting rules that are used in the discussions on the sub-term property and size explosion for proof nets.}
\label{fig:pn-rules}
\end{figure}
\begin{center}
\input{figure-least-level}
\end{center} 
The right-down path shows that $\Rew{l_0}$, and thus $\Rew{ll}$, lacks the sub-term property, as the box duplicated by the second step is not a sub-proof of the initial one. Moreover, iterating such a pattern gives size explosion for $\Rew{l_0}$, as shall we show in \refsect{sub-term-size-explosion}. The down-left path has the sub-term property but it shows that different $\Rew{l_0}$ paths can have different lengths---that is, $\Rew{l_0}$ is not \emph{diamond}---forbidding to take its number of steps as a measure, because the number of its steps to normal form is an ambiguously defined quantity.

\paragraph{The Linear Head Strategy} There actually is a micro-step strategy in the linear logic literature which is both diamond and with the sub-term property. It is Mascari \& Pedicini's and Danos \& Regnier's \emph{linear head strategy} $\tolh$ \cite{DBLP:journals/tcs/MascariP94,Danos04headlinear} (which is a sub-strategy of level 0 cut elimination) but it is defined only on the fragment representing the $\l$-calculus and it does not compute cut-free proofs. 

The linear head strategy is naturally modelled by the LSC, used by Accattoli and co-authors for proving properties of $\tolh$, design variants and extensions \cite{DBLP:conf/rta/Accattoli12,DBLP:conf/popl/AccattoliBKL14,DBLP:conf/icfp/AccattoliBM14,DBLP:conf/aplas/AccattoliBM15}, as well as to prove that it provides a reasonable time cost model \cite{DBLP:conf/rta/AccattoliL12}. 

\paragraph{The Good Strategy} The ESC strategy introduced here is simply called \emph{the good strategy} $\tog$, because it avoids \emph{bad} steps breaking the sub-term property. It is the generalization of the linear head strategy to the whole of IMELL and extended as to compute cut-free proofs. 
The good strategy is micro-step, diamond, and it has the sub-term property, thus providing the first polynomial cost model for IMELL\footnote{Is it a \emph{reasonable} cost model? Roughly, yes. For time reasonability, the subtle part is polynomiality of the cost model, that is, the polynomial simulation of the ESC on Turing machines/RAM. The complementary part---here missing---is simulating Turing machines in the untyped ESC within polynomial overhead. It is true, but for very minor reasons it does not follow from results in the literature and it has to be (tediously) reproved from scratch. Moreover, the missing part is not relevant/cannot hold for the \emph{typed} ESC, as IMELL is not Turing-complete (IMELL cut-elimination being strongly normalizing, it cannot model diverging computations).}. Its design is based on the \emph{creation of cuts}, a notion due to the intuitionistic arrow $\lolli$ and invisible in classical linear logic. 
 
\subsection{Overview of Contribution 3: Foundations}
We provide a foundational study of the ESC, ensuring that the new cut elimination rules are well behaved, but also to promote it as an alternative to the use of proof nets for IMELL. We prove three important properties:
\begin{enumerate}
\item \emph{Untyped confluence}: the \emph{untyped} case is possibly divergent and thus more difficult than the typed case, as one cannot exploit termination in the proof technique. We prove confluence using the Hindley-Rosen method, and exploit termination nonetheless by building on the strong normalization of the exponentials, following an approach pushed forward by Accattoli in the study of the LSC \cite{DBLP:conf/rta/Accattoli12}. 
\item \emph{Untyped PSN}: preservation of untyped strong normalization (shortened to PSN) is a property typical of calculi with ESs, stating that if a term is SN for the small-step rules (in the untyped setting where terms can also be divergent) then it is SN for the micro-step rules. Essentially, it states that the decomposition of  small steps in micro steps does not introduce degeneracies. The proof is based on a technique due to Kesner \cite{DBLP:journals/corr/abs-0905-2539}. This theorem is notoriously technical for $\l$-calculi with ESs. We adapt a simple proof by Accattoli and Kesner \cite{DBLP:conf/csl/AccattoliK10}, where the simplicity is enabled by adopting rules at a distance. Our adaptation is even simpler than the proof in \cite{DBLP:conf/csl/AccattoliK10}. 
\item \emph{Typed SN}: we prove strong normalization (SN) of the small-step rules in the typed case, using the reducibility method. This theorem too is notoriously technical to prove for fragments of linear logic including the exponentials. We here provide what is probably the simplest and cleanest proof of SN in the literature, improving the already simplified approach of Accattoli \cite{DBLP:conf/rta/Accattoli13} based on rules at a distance. We then use PSN to transport SN to micro steps.
\end{enumerate}

Confluence and strong normalization are not usually studied for sequent calculi with traditional  cut elimination, that is, with both principal and commutative cases, because of the following two facts:
\begin{itemize}
\item \emph{No SN}: cut elimination is \emph{not} SN, because cut commutes with itself, leading to non-termination:
\[
\input{figure-silly-diverging-sequents}
\]
\item \emph{No confluence}: cut elimination is not confluent, not even in IMLL (!), because commutations affect the result, as shown by this counter-example, courtesy of Olivier Laurent:
\[
\input{figure-bad-traditional-cut-elim}
\]
\end{itemize}
Retrieving confluence and SN requires cut elimination \emph{modulo commutations}, or via proof nets, or adding some rigidity to proofs such as focalization. In stark contrast, the cut elimination \emph{at a distance} of the ESC, being free from commutations, achieves them without rewriting modulo nor modifying the deductive system.

Our study thus lies a new foundation for the study of cut elimination. It also sums up ten years of research on the rewriting of the LSC, at the same time generalizing the developed techniques to the considerably more general setting of IMELL.

%% file: figure-classical-diverging.tex
% !TEX root = main.tex
\begin{figure*}
$\begin{array}{cccccccc}
\pi  \defeq &
\begin{tikzpicture}[ocenter]
\node at (0,0) [etic](axRightConclusion){};
\node at (axRightConclusion.center) [etic, left = 1.4*\stlar](der){\scriptsize $\der$};
\node at \med{axRightConclusion}{der} [etic, above = \hstalt ](axSym){\scriptsize $\ax$};
\node at (axSym) [etic, below = 1.4*\stalt](contr){\scriptsize $\csym$};
\draw[nopol, out=0, in=45](axSym)to(contr);
\draw[nopol, out=180, in=90](axSym)to(der);
\draw[nopol, out = -90, in=135](der)to(contr);

\node at (contr) [nospace, below = .7*\stalt](bang){};
\abox{bang}{exbox}{18pt}{18pt}{56pt}
\node at (bang.center)[etic] (bangsym){$!$};
\draw[nopol](contr)to(bangsym);

\node at (axRightConclusion) [etic, below left = .7*\stalt and 2.6*\stlar](axRightConclusion2){};
\node at (axRightConclusion2.center) [etic, left = 1.4*\stlar](der2){\scriptsize $\der$};
\node at \med{axRightConclusion2}{der2} [etic, above = \hstalt ](axSym2){\scriptsize $\ax$};
\node at (axSym2) [etic, below = 1.4*\stalt](contr2){\scriptsize $\csym$};
\draw[nopol, out=0, in=45](axSym2)to(contr2);
\draw[nopol, out=180, in=90](axSym2)to(der2);
\draw[nopol, out = -90, in=135](der2)to(contr2);

\node at \med{bangsym}{contr2} [etic, below = \hstalt ](cut){\scriptsize $\cut$};
\draw[nopol, out=-90, in=0](bangsym)to(cut);
\draw[nopol, out = -90, in=180](contr2)to(cut);
\end{tikzpicture}
&
\to
&
\begin{tikzpicture}[ocenter]
\node at (0,0) [etic](axRightConclusion){};
\node at (axRightConclusion.center) [etic, left = 1.4*\stlar](der){\scriptsize $\der$};
\node at \med{axRightConclusion}{der} [etic, above = \hstalt ](axSym){\scriptsize $\ax$};
\node at (axSym) [etic, below = 1.4*\stalt](contr){\scriptsize $\csym$};
\draw[nopol, out=0, in=45](axSym)to(contr);
\draw[nopol, out=180, in=90](axSym)to(der);
\draw[nopol, out = -90, in=135](der)to(contr);

\node at (contr) [nospace, below = .7*\stalt](bang){};
\abox{bang}{exbox}{18pt}{18pt}{56pt}
\node at (bang.center)[etic] (bangsym){$!$};
\draw[nopol](contr)to(bangsym);

\node at (axRightConclusion) [etic, below left = 2*\stalt and 2.6*\stlar](axRightConclusion2){};
\node at (axRightConclusion2.center) [etic, left = 1.4*\stlar](der2){\scriptsize $\der$};
\node at \med{axRightConclusion2}{der2} [etic, above = \hstalt ](axSym2){\scriptsize $\ax$};
\draw[nopol, out=180, in=90](axSym2)to(der2);

\node at \med{bangsym}{axRightConclusion2} [etic, below = .5*\hstalt ](cut){\scriptsize $\cut$};
\draw[nopol, out=-90, in=0](bangsym)to(cut);
\draw[nopol, out = 0, in=180, looseness=1.5](axSym2)to(cut);

\node at (axRightConclusion) [etic, below right = \hstalt and 2.8*\stlar](axRightConclusion3){};
\node at (axRightConclusion3.center) [etic, left = 1.4*\stlar](der3){\scriptsize $\der$};
\node at \med{axRightConclusion3}{der3} [etic, above = \hstalt ](axSym3){\scriptsize $\ax$};
\node at (axSym3) [etic, below = 1.4*\stalt](contr3){\scriptsize $\csym$};
\draw[nopol, out=0, in=45](axSym3)to(contr3);
\draw[nopol, out=180, in=90](axSym3)to(der3);
\draw[nopol, out = -90, in=135](der3)to(contr3);

\node at (contr3) [nospace, below = .7*\stalt](bang3){};
\abox{bang3}{exbox}{18pt}{18pt}{56pt}
\node at (bang3.center)[etic] (bangsym3){$!$};
\draw[nopol](contr3)to(bangsym3);

\node at \med{bangsym3}{der2} [etic, below = 1.4*\hstalt ](cut2){\scriptsize $\cut$};
\draw[nopol, out=-60, in=180](der2)to(cut2);
\draw[nopol, out=-120, in=0](bangsym3)to(cut2);
\end{tikzpicture}
&
\to
&
\begin{tikzpicture}[ocenter]
\node at (0,0) [etic](axRightConclusion){};
\node at (axRightConclusion.center) [etic, left = 1.4*\stlar](der){\scriptsize $\der$};
\node at \med{axRightConclusion}{der} [etic, above = \hstalt ](axSym){\scriptsize $\ax$};
\node at (axSym) [etic, below = 1.4*\stalt](contr){\scriptsize $\csym$};
\draw[nopol, out=0, in=45](axSym)to(contr);
\draw[nopol, out=180, in=90](axSym)to(der);
\draw[nopol, out = -90, in=135](der)to(contr);

\node at (contr) [nospace, below = .7*\stalt](bang){};
\abox{bang}{exbox}{18pt}{18pt}{56pt}
\node at (bang.center)[etic] (bangsym){$!$};
\draw[nopol](contr)to(bangsym);

\node at (axRightConclusion) [etic, below left = 2*\stalt and 2.6*\stlar](axRightConclusion2){};
\node at (axRightConclusion2.center) [nospace, left = 1.4*\stlar](der2){};
\node at \med{axRightConclusion2}{der2} [etic, above = \hstalt ](axSym2){\scriptsize $\ax$};
\node at (der2.center) [nospace, above = 1 pt](der2ghost){};

\node at \med{bangsym}{axRightConclusion2} [etic, below = .5*\hstalt ](cut){\scriptsize $\cut$};
\draw[nopol, out=-90, in=0](bangsym)to(cut);
\draw[nopol, out = 0, in=180, looseness=1.5](axSym2)to(cut);

\node at (axRightConclusion) [etic, below right = \stalt and 2.6*\stlar](axRightConclusion3){};
\node at (axRightConclusion3.center) [etic, left = 1.4*\stlar](der3){\scriptsize $\der$};
\node at \med{axRightConclusion3}{der3} [etic, above = \hstalt ](axSym3){\scriptsize $\ax$};
\node at (axSym3) [etic, below = 1.4*\stalt](contr3){\scriptsize $\csym$};
\draw[nopol, out=0, in=45](axSym3)to(contr3);
\draw[nopol, out=180, in=90](axSym3)to(der3);
\draw[nopol, out = -90, in=135](der3)to(contr3);

\node at \med{contr3}{der2} [etic, below = 1.4*\hstalt ](cut2){\scriptsize $\cut$};
\draw[draw=gray, line width=0.18ex,out=180, in=90](axSym2)to(der2);
\draw[->, draw=gray, line width=0.18ex, out=-90, in=180](der2ghost)to(cut2);
\draw[nopol, out=-120, in=0](contr3)to(cut2);
\end{tikzpicture}
&\to\pi
\end{array}$
\caption{Untyped proof net of classical LL reducing to itself in 3 micro steps.}
\label{fig:classical-diverging}
\end{figure*}

%% file: figure-pn-rules.tex
% !TEX root = main.tex
\begin{center}
\hspace*{-15pt}
\begin{tabular}{c|c}
\textsc{Promotion / contraction} & \textsc{Promotion / auxiliary port}
\\
\input{figure-pn-contraction-redrule}\hspace*{-8pt}
&
\input{figure-pn-boxbox-redrule}\hspace*{-10pt}
\end{tabular}
\end{center}

%% file: figure-pn-contraction-redrule.tex
% !TEX root = main.tex
\begin{tabular}{ccc}
\scalebox{1.1}{
\begin{tikzpicture}[ocenter]
 %contraction
\node at (0,0) [etic](parlink){$\contr$};
\node at (parlink.center) [above right=\hstalt and \hstlar, etic](parauxr){\hspace*{5pt}\tiny$\whyn\form^\bot$}; % inserted space so the letters are not cut off
\node at (parlink.center) [above left=\hstalt and \hstlar, etic](parauxl){\tiny$\whyn\form^\bot$};
\lBinEdgesAbove{parauxl}{parauxr}{parlink}{nopol}{nopol}
\node at (parlink.center) [below left=1.3*\altax and \stlar, etic](cutlink){\tiny$\cut$};
\node at (cutlink.center) [above left=1.3*\altax and \stlar, etic](poslink){};
\node at (parlink.center) [nospace, below right=6pt and -1pt](label){\tiny$\whyn\form^\bot$};
\node at (poslink.center) [nospace, below right=3pt and 3pt](label){\tiny$\bang\form$};
\boxnodes{poslink}{0.7*\stboxrwaux}{2*\stboxlw/3}{0.6*\stboxh}
\boxline{poslink}{exboxline}
\node at (poslink.center)[etic](possym){$!$};
%\boxlabel{poslink}{{ $\net$}}
\boxauxnodesright{poslink}{\stboxlw/2}{3*\stlar/4}
\boxauxnodeslab{poslink}{$\whyn\formtwo_1$}{$\whyn\formtwo_k$\hspace*{5pt}}{12pt}
\gdots{poslinkaux1}{poslinkauxk}{below =.5pt}
\draw[nopol, in=0, out=-90] (parlink) to (cutlink);
\draw[nopol, in=180, out=-90] (possym) to (cutlink);
\end{tikzpicture}}
&
\hspace*{-5pt}$\rightarrow_{\contr}$\hspace*{-5pt}
&
%&
%\boxinit
\scalebox{1.1}{
\begin{tikzpicture}[ocenter]
 %contraction
  \node [etic](contrpaxl){\hspace*{5pt}\tiny$\whyn\form^\bot$}; 
 \node at (contrpaxl.center) [nospace, above=15pt](spacingup){};
 \node at (contrpaxl.center) [left= \stlar, etic](contrpaxr){\tiny$\whyn\form^\bot$};
%posL
\node at (contrpaxr.center) [left= \stlarb, etic](pos){};
%\posboxlinkmfake{pospal}{pos}{}
\boxnodes{pos}{0.6*\stboxrwaux}{2*\stboxlw/3}{0.6*\stboxh}
\boxline{pos}{exboxline}
\node at (pos.center)[etic](possym){$!$};
\boxauxnodesright{pos}{\stboxlw/2}{3*\stlar/4}
%\boxlabel{pos}{{ $\net$}}
\gdots{posaux1}{posauxk}{below=1pt}
\node at \med{contrpaxr}{pos}[below=0.8*\hstalt, etic](cutlink){\tiny$\cut$};
\draw[nopol, in=0, out=-90] (contrpaxr) to (cutlink);
\draw[nopol, in=180, out=-90]  (possym) to (cutlink);
%\lcutnpdistfixp{contrpaxr}{pospal}{cut}{7pt};
%bangR
\node at (pos.center) [left= 2.2*\stlar, etic](posr){};
%\posboxlinkmfake{posrpal}{posr}{}
\boxnodes{posr}{0.6*\stboxrwaux}{2*\stboxlw/3}{0.6*\stboxh}
\boxline{posr}{exboxline}
\node at (posr.center)[etic](posrsym){$!$};
\boxauxnodesright{posr}{\stboxlw/2}{3*\stlar/4}
%\boxlabel{posr}{{ $\net$}}
\gdots{posraux1}{posrauxk}{below=1pt}
\node at (pos.center) [nospace, below right=2pt and 3pt](label){\tiny$\bang\form$};
\node at (posr.center) [nospace, below right=2pt and 3pt](label){\tiny$\bang\form$};
%contractions
\node at \med{posraux1}{posaux1}[etic, below right=\stalt and 2pt] (contrl){$\contr$};
\node at (contrl.center)[below =\hstalt, etic](contrlpal){\tiny$\whyn\formtwo_1$};

\draw[nopol, out =-75, in =135](posraux1)to(contrl);
\draw[nopol, out =-105, in =45](posaux1)to(contrl);
\draw[nopol](contrl)to(contrlpal);

\node at \med{posrauxk}{posauxk}[etic, below left=\stalt and 2pt](contrr){$\contr$};
\node at (contrr.center)[below =\hstalt, etic] (contrrpal){\tiny$\whyn\formtwo_k$};

\draw[nopol, out =-75, in =135](posrauxk)to(contrr);
\draw[nopol, out =-105, in =45](posauxk)to(contrr);
\draw[nopol](contrr)to(contrrpal);
\gdots{contrlpal}{contrrpal}{above =1pt}
% \node at (contrlpal.center) [nospace, left=3pt](label){};
% \node at (contrrpal.center) [nospace, right=3pt](label){\tiny$\whyn\formtwo_k$};
%CUT
\node at (cutlink.center)[etic, below=.5*\hstalt](cut2){\tiny$\cut$};
%\lcutnpdistfixp{contrpaxl}{bangrpal}{cut2}{16pt};
\draw[nopol, out=-90, in=0](contrpaxl) to(cut2);
\draw[nopol, in=180, out=-75](posrsym)to (cut2);
\end{tikzpicture}}
\end{tabular}

%% file: figure-pn-boxbox-redrule.tex
% !TEX root = main.tex
\begin{tabular}{ccc}
%\boxinit
\scalebox{1.1}{
\begin{tikzpicture}[ocenter]%Hole
%bang
\node at (0,0) [ etic](pospal){};

%box
\abox{pospal}{exbox}{0.7*\stboxrwaux}{\stboxlw-2pt}{\stboxh-3pt}
\node at (pospal) [ etic](possym){$!$};
\node at (possym) [below=1.3*\pendingdist, etic](ghostsource){\tiny $!\form$};

\draw[nopolrev](ghostsource)to(possym);

%net g
\node at (pospal.center)[above left=0.4*\stalt and 0.8*\stlar, net](netg){$\net$};

%auxiliary ports
\node at (ghostsource-|netg)[etic](delta_node){\tiny$\whyn\Delta$};
\draw[nopolgen](netg)to(delta_node);

\node at (delta_node.center)[left= 1.3*\stlar, etic](cutlink){\tiny $\cut$};

\draw[nopol, in=0, out=-150](netg)to(cutlink);
\draw[nopol, in =90, out=0](netg)to(possym);

%pos r
\node at (pospal-|cutlink) [left= 0.9*\ilar, etic](posr){};
\boxnodes{posr}{\stlar}{\stboxlw-1pt}{\stboxh-9pt}
\boxline{posr}{exboxline}
\node at (posr.center)[etic](posrsym){$!$};
%\boxlabel{posr}{{$\nettwo$}}

\node at (posr.center) [left=1.7*\hstlar-3pt, nospace](aux){};
\node at (cutlink-|aux) [etic](auxlabel){\tiny $\whyn\Gamma$};
\draw[nopolgen](aux)to(auxlabel);

\draw[nopolrev, in =-90, out=180](cutlink)to(posrsym);

\node at (posr.center) [nospace, below right=3pt and 3pt](label){\tiny $\bang\formtwo$};

\end{tikzpicture}}
&

\hspace*{-5pt}$\Rew{\boxbox}$\hspace*{-5pt}
&

%\boxinit
\scalebox{1.1}{
\begin{tikzpicture}[ocenter]%Hole
%bang
\node at (0,0) [ etic](pospal){};

\node at (pospal.center) [nospace, above=33pt](spacingup){};

%box
\node at (pospal) [ etic](possym){$!$};
\node at (possym) [below=1.3*\pendingdist, etic](ghostsource){\tiny $!\form$};
\draw[nopolrev](ghostsource)to(possym);

%net g
\node at (pospal.center)[above left=0.5*\stalt and 0.8*\stlar, net](netg){$\net$};
\draw[nopol, in =90, out=0](netg)to(possym);

%auxiliary ports
\node at (ghostsource-|netg)[etic](delta_node){\tiny$\whyn\Delta$};

\draw[nopolgen](netg)to(delta_node);

%cut
\node at (netg.center)[below left=\hstalt/2 and 1.3*\stlar, etic](cutlink){\tiny $\cut$};
\draw[nopol, in=0, out=210](netg)to(cutlink);

%pos r
\node at (cutlink) [above left=1.3*\altax and 1.1*\ilar, etic](posr){};
\boxnodes{posr}{\stlar}{\stboxlw-1pt}{\stboxh-9pt}
\boxline{posr}{exboxline}
\node at (posr.center)[etic](posrsym){$!$};
%\boxlabel{posr}{{$\nettwo$}}

\node at (posr.center) [left=1.7*\hstlar-3pt, nospace](aux){};
\node at (delta_node-|aux) [etic](auxlabel){\tiny $\whyn\Gamma$};
\draw[nopolgen](aux)to(auxlabel);

\draw[nopol, out=-90, in=180](posrsym)to(cutlink);

\node at (posr.center) [nospace, below right=3pt and 3pt](label){\tiny $\bang\formtwo$};

%\node at (posraux1.center)[ nospace, below left=2pt and 5pt](dummy2){};
\node at (pospal.center)[ nospace, above=1pt](dummy3){};
\abox{pospal}{exbox}{\stboxrwaux+1.9*\stlar}{\stboxlw-2pt}{1.5*\stboxh}

\node at (pospal) [ etic](possym){$!$};

\end{tikzpicture}}

\end{tabular}

%% file: figure-least-level.tex
% !TEX root = main.tex
\begin{tikzpicture}[ocenter]
\node at (0,0)[etic](origin){
	\begin{tikzpicture}[ocenter]
	\node at (0,0) [etic](contr){\tiny$\csym$};
	\node at (contr) [etic, above left = \hstalt and \hstlar](ghost1){};
	\node at (contr) [etic, above right = \hstalt and \hstlar](ghost2){};
	\draw[nopol, out = -90, in = 135](ghost1)to(contr);
	\draw[nopol, out = -90, in = 45](ghost2)to(contr);

	% left cut
	\node at (contr)[below right= .7*\hstalt and .8*\stlar, etic](cutlink){\tiny $\cut$};
	\draw[nopol, in=180, out=-90](contr)to(cutlink);

	% box 2
	\node at (contr-|cutlink) [right= .8*\stlar, etic](posr){};
	\abox{posr}{exbox}{\stboxlw-2pt}{0.2*\stboxrwaux}{\stboxh-5pt}
	\node at (posr.center)[etic](posrsym){$!$};
	\node at (posr) [above = .7*\hstalt,etic](label2){1};
	\draw[nopolrev, in =-90, out=0](cutlink)to(posrsym);

	% right cut
	\node at (posr.center)[right=.4*\stlar, nospace](aux2){};
	\node at (aux2)[below right= .7*\hstalt and .8*\stlar, etic](cutlink2){\tiny $\cut$};
	\draw[nopol, in=180, out=-90](aux2)to(cutlink2);

	% box 2
	\node at (posr-|cutlink2) [right= .8*\stlar, etic](pos3){};
	\abox{pos3}{exbox}{\stboxlw-2pt}{0.2*\stboxrwaux}{\stboxh-5pt}
	\node at (pos3.center)[etic](pos3sym){$!$};
	\node at (pos3) [above = .7*\hstalt,etic](label3){2};
	\draw[nopolrev, in =-90, out=0](cutlink2)to(pos3sym);
	\end{tikzpicture}
};

\node at (origin.center) [etic, right = 300pt](origin-right){
	\begin{tikzpicture}[ocenter]%Hole
	%box 1
	\node at (0,0) [etic](contr){\tiny$\csym$};
	\node at (contr) [etic, above left = \hstalt and \hstlar](ghost1){};
	\node at (contr) [etic, above right = \hstalt and \hstlar](ghost2){};
	\draw[nopol, out = -90, in = 135](ghost1)to(contr);
	\draw[nopol, out = -90, in = 45](ghost2)to(contr);

	% left cut
	\node at (contr.center)[right=0.4*\stlar, nospace](aux1){};
	\node at (aux1)[below right= .7*\hstalt and .5*\stlar, etic](cutlink){\tiny $\cut$};
	\draw[nopol, in=180, out=-90](contr)to(cutlink);

	% box 2
	\node at (contr-|cutlink) [right= .9*\stlar, etic](posr-r){};
	\abox{posr-r}{exbox}{\stboxlw-2pt}{.7*\stboxrwaux}{1.3*\stboxh}
	\node at (posr-r.center)[etic](posrsym){$!$};
	\node at (posr-r) [above = .7*\hstalt,etic](label2){1};
	\draw[nopolrev, in =-90, out=0](cutlink)to(posrsym);

	% box 3
	\node at (posr-r) [above right= .5*\hstalt and .7*\stlar, etic](pos3-r){};
	\abox{pos3-r}{exbox}{\stboxlw-2pt}{0.2*\stboxrwaux}{\stboxh-5pt}
	\node at (pos3-r.center)[etic](pos3sym){$!$};
	\node at (pos3-r) [above = .7*\hstalt,etic](label3){2};
	\end{tikzpicture}
};

%%%%%%%%%
%%%% ORIGIN-DOWN
\node at (origin.center) [etic, below = 40pt](origin-down){
	\begin{tikzpicture}[ocenter]
	% box 2
	\node at (0,0) [etic](posr1){};
	\abox{posr1}{exbox}{\stboxlw-2pt}{0.2*\stboxrwaux}{\stboxh-5pt};
	\node at (posr1.center)[etic](posrsym){$!$};
	\node at (posr1) [above = .7*\hstalt,etic](label2){1};

	% box 2 bis
	\node at (posr1.center) [right= 1.3*\stlar, etic](posr2){};
	\abox{posr2}{exbox}{\stboxlw-2pt}{0.2*\stboxrwaux}{\stboxh-5pt}
	\node at (posr2.center)[etic](posrsym2){$!$};
	\node at (posr2) [above = .7*\hstalt,etic](label22){1};

	%contr
	\node at (posr1.center)[right=.4*\stlar, nospace](aux3){};
	\node at (posr2.center)[right=.4*\stlar, nospace](aux22){};
	\node at \med{aux3}{aux22} [etic, below=\hstalt](contr){\tiny$\csym$};
	\draw[nopol, out = -90, in = 135](aux3)to(contr);
	\draw[nopol, out = -90, in = 45](aux22)to(contr);

	% left cut
	\node at (contr)[below right= .7*\hstalt and .8*\stlar, etic](cutlink){\tiny $\cut$};
	\draw[nopol, in=180, out=-90](contr)to(cutlink);

	% box 2
	\node at (contr) [right= 1.6*\stlar, etic](pos3){};
	\abox{pos3}{exbox}{\stboxlw-2pt}{0.2*\stboxrwaux}{\stboxh-5pt}
	\node at (pos3.center)[etic](pos3sym){$!$};
	\node at (pos3) [above = .7*\hstalt,etic](label3){2};
	\draw[nopolrev, in =-90, out=0](cutlink)to(pos3sym);
	\end{tikzpicture}
};

%%%%%%%%%
%%%% TARGET
\node at (origin-right|-origin-down) [etic](target){
	\begin{tikzpicture}[ocenter]%Hole
	% box 2
	\node at (0,0) [right= .7*\stlar, etic](posr){};
	\abox{posr}{exbox}{\stboxlw-2pt}{.7*\stboxrwaux}{1.3*\stboxh}
	\node at (posr.center)[etic](posrsym){$!$};
	\node at (posr) [above = .7*\hstalt,etic](label2){1};
	% box 3
	\node at (posr) [above right= .5*\hstalt and .7*\stlar, etic](pos3){};
	\abox{pos3}{exbox}{\stboxlw-2pt}{0.2*\stboxrwaux}{\stboxh-5pt}
	\node at (pos3.center)[etic](pos3sym){$!$};
	\node at (pos3) [above = .7*\hstalt,etic](label3){2};

	% box 2
	\node at (posr) [right= 2.5*\stlar, etic](posr2){};
	\abox{posr2}{exbox}{\stboxlw-2pt}{.7*\stboxrwaux}{1.3*\stboxh}
	\node at (posr2.center)[etic](posrsym2){$!$};
	\node at (posr2) [above = .7*\hstalt,etic](label22){1};
	% box 3
	\node at (posr2) [above right= .5*\hstalt and .7*\stlar, etic](pos32){};
	\abox{pos32}{exbox}{\stboxlw-2pt}{0.2*\stboxrwaux}{\stboxh-5pt}
	\node at (pos32.center)[etic](pos3sym){$!$};
	\node at (pos32) [above = .7*\hstalt,etic](label32){2};
	\end{tikzpicture}
};

\node at \med{origin-down.center}{target.center}(ghost){};

%%%%%%%%%
%%%% ORIGIN-DOWN-RIGHT
\node at \med{origin-down.east}{ghost.center}[right=10pt, anchor = center](origin-down-right){ 
	\begin{tikzpicture}[ocenter]%Hole
	% box 2
	\node at (0,0) [right= .7*\stlar, etic](posr){};
	\abox{posr}{exbox}{\stboxlw-2pt}{0.2*\stboxrwaux}{\stboxh-5pt}
	\node at (posr.center)[etic](posrsym){$!$};
	\node at (posr) [above = .7*\hstalt,etic](label2){1};

	% box 2 bis
	\node at (posr.center) [right= 1.3*\stlar, etic](posr2){};
	\abox{posr2}{exbox}{\stboxlw-2pt}{0.2*\stboxrwaux}{\stboxh-5pt}
	\node at (posr2.center)[etic](posrsym2){$!$};
	\node at (posr2) [above = .7*\hstalt,etic](label22){1};

	%contr
	\node at (posr.center)[right=.4*\stlar, nospace](aux3){};
	\node at (posr2.center)[right=.4*\stlar, nospace](aux22){};
	\node at \med{aux3}{aux22} [etic, below=\hstalt](contr){};
	%\draw[nopol, out = -90, in = 135](aux3)to(contr);
	%\draw[nopol, out = -90, in = 45](aux22)to(contr);

	% below cut
	\node at (contr)[below right= .7*\hstalt and .3*\stlar, etic](cutlink){\tiny $\cut$};
	\draw[nopol, in=170, out=-90](aux3)to(cutlink);
	% above cut
	\node at (cutlink)[right= 1.5*\stlar, etic](cutlink2){\tiny $\cut$};
	\draw[nopol, in=170, out=-90](aux22)to(cutlink2);

	% box 2
	\node at (contr) [above right= .5*\hstalt and 1.5*\stlar, etic](pos3){};
	\abox{pos3}{exbox}{\stboxlw-2pt}{0.2*\stboxrwaux}{\stboxh-5pt}
	\node at (pos3.center)[etic](pos3sym){$!$};
	\node at (pos3) [above = .7*\hstalt,etic](label3){2};
	\draw[nopolrev, in =-90, out=0](cutlink)to(pos3sym);

	\node at (pos3.center) [right= 1.3*\stlar, etic](pos4){};
	\abox{pos4}{exbox}{\stboxlw-2pt}{0.2*\stboxrwaux}{\stboxh-5pt}
	\node at (pos4.center)[etic](pos3sym){$!$};
	\node at (pos4) [above = .7*\hstalt,etic](label3){2};
	\draw[nopolrev, in =-90, out=0](cutlink2)to(pos3sym);
	\end{tikzpicture}
};

%%%%%%%%%
%%%% TARGET-LEFT
\node at \med{ghost.center}{target.west}[left=10pt, anchor = center](target-left){ 
	\begin{tikzpicture}[ocenter]%Hole
	\node at (0,0) [right= .7*\stlar, etic](posr){};
	\abox{posr}{exbox}{\stboxlw-2pt}{0.2*\stboxrwaux}{\stboxh-5pt}
	\node at (posr.center)[etic](posrsym){$!$};
	\node at (posr) [above = .7*\hstalt,etic](label2){1};

	\node at (posr) [right= 1.3*\stlar, etic](posr2){};
	\abox{posr2}{exbox}{\stboxlw-2pt}{.7*\stboxrwaux}{1.3*\stboxh}
	\node at (posr2.center)[etic](posrsym2){$!$};
	\node at (posr2) [above = .7*\hstalt,etic](label22){1};
	% box 3
	\node at (posr2) [above right= .5*\hstalt and .7*\stlar, etic](pos32){};
	\abox{pos32}{exbox}{\stboxlw-2pt}{0.2*\stboxrwaux}{\stboxh-5pt}
	\node at (pos32.center)[etic](pos3sym){$!$};
	\node at (pos32) [above = .7*\hstalt,etic](label32){2};

	\node at (posr2.center) [right= 2.5*\stlar, etic](pos4){};
	\abox{pos4}{exbox}{\stboxlw-2pt}{0.2*\stboxrwaux}{\stboxh-5pt}
	\node at (pos4.center)[etic](pos3sym){$!$};
	\node at (pos4) [above = .7*\hstalt,etic](label3){2};

	\node at (posr.center)[right=.4*\stlar, nospace](aux3){};
	\node at \med{aux3}{pos4}[below= \hstalt, etic](cutlink){\tiny $\cut$};
	\draw[nopol, in=180, out=-90](aux3)to(cutlink);
	\draw[nopolrev, in =-90, out=0](cutlink)to(pos3sym);
	\end{tikzpicture}
};

\draw[->, shorten <=3pt, shorten >=3pt](origin) to node[above] {\scriptsize $\boxbox $} (origin-right);
\draw[->, shorten <=3pt, shorten >=3pt](origin) to node[left] {\scriptsize $\csym $} (origin-down);

\draw[->, dotted, shorten <=3pt, shorten >=3pt](origin-down) to node[above] {\scriptsize $\csym $} (origin-down-right);
\draw[->, dotted, shorten <=3pt, shorten >=3pt](origin-down-right) to node[above] {\scriptsize $\boxbox $} (target-left);
\draw[->, dotted, shorten <=3pt, shorten >=3pt](target-left) to node[above] {\scriptsize $\boxbox $} (target);
\draw[->, dotted, shorten <=3pt, shorten >=3pt](origin-right) to node[right] {\scriptsize $\csym $} (target);
\end{tikzpicture}
%\end{array}$

%% file: figure-silly-diverging-sequents.tex
% !TEX root = main.tex
\begin{tabular}{c\colspace c\colspace c}
	\AxiomC{$\Gamma \vdash   \formtwo$}
		\AxiomC{$\Pi \vdash   \form$}
		\AxiomC{$\Delta,  \form,  \formtwo \vdash \formthree$}

		\RightLabel{cut}
		\BinaryInfC{$ \Delta, \Pi,  \formtwo \vdash \formthree$}

		\RightLabel{cut}
		\BinaryInfC{$\Delta, \Pi,\Gamma, \vdash \formthree$}			
		\DisplayProof 
&
$\rightarrow$
&
	\AxiomC{$\Pi \vdash   \form$}
		\AxiomC{$\Gamma \vdash   \formtwo$}
		\AxiomC{$\Delta,  \form,  \formtwo \vdash \formthree$}

		\RightLabel{cut}
		\BinaryInfC{$\Delta,   \form, \Gamma \vdash \formthree$}

		\RightLabel{cut}
		\BinaryInfC{$\Delta,\Pi , \Gamma \vdash \formthree$}			
		\DisplayProof 
\end{tabular}

%% file: figure-bad-traditional-cut-elim.tex
% !TEX root = main.tex
\small
\tablinesep=3pt
 \begin{tabular}{c}
		\AxiomC{}
		%\RightLabel{$\ax$}
		\UnaryInfC{$ \form \vdash  \form$}
		\AxiomC{}
		%\RightLabel{$\ax$}
		\UnaryInfC{$ \formtwo \vdash  \formtwo$}
		\RightLabel{$\lolli_l$}
		\BinaryInfC{$ \form,  \form\lolli\formtwo \vdash  \formtwo$}

		\AxiomC{}
		%\RightLabel{$\ax$}
		\UnaryInfC{$ \formtwo \vdash  \formtwo$}
		\AxiomC{}
		%\RightLabel{$\ax$}
		\UnaryInfC{$ \formthree \vdash \formthree$}
		\RightLabel{$\lolli_l$}
		\BinaryInfC{$ \formtwo,   \formtwo\lolli \formthree \vdash  \formthree$}
		\RightLabel{cut}
		\BinaryInfC{$  \form,  \form\lolli\formtwo,\formtwo\lolli \formthree \vdash  \formthree$}			
		\DisplayProof  

			\\[6pt]
		$\mbox{}_+{\swarrow}\ \ \ \ \searrow_+$
		\\[6pt]
		\begin{tabular}{c@{\hspace{30pt}}c}
		\AxiomC{}
		%\RightLabel{$\ax$}
		\UnaryInfC{$ \form \vdash  \form$}
		\AxiomC{}
		%\RightLabel{$\ax$}
		\UnaryInfC{$ \formtwo \vdash  \formtwo$}
		\AxiomC{}
		%\RightLabel{$\ax$}
		\UnaryInfC{$ \formthree \vdash  \formthree$}
		\RightLabel{$\lolli_l$}
		\BinaryInfC{$\formtwo,  \formtwo\lolli\formthree \vdash  \formthree$}
		\RightLabel{$\lolli_l$}
		\BinaryInfC{$ \form,  \form\lolli\formtwo, \formtwo\lolli\formthree \vdash  \formthree$}
		\DisplayProof  
&
	\AxiomC{}
		%\RightLabel{$\ax$}
		\UnaryInfC{$ \form \vdash \form$}
		\AxiomC{}
		%\RightLabel{$\ax$}
		\UnaryInfC{$ \formtwo \vdash  \formtwo$}
		\RightLabel{$\lolli_l$}
		\BinaryInfC{$ \form,  \form\lolli\formtwo \vdash  \formtwo$}
		\AxiomC{}
		%\RightLabel{$\ax$}
		\UnaryInfC{$ \formthree \vdash  \formthree$}
		\RightLabel{$\lolli_l$}
		\BinaryInfC{$\form,  \form\lolli\formtwo,  \formtwo\lolli\formthree \vdash  \formthree$}
		\DisplayProof  
		\end{tabular}	
		\end{tabular}

%% file: 03-Sub-Term_Property_and_Size_Explosion.tex
% !TeX spellcheck = en_US
% !TEX root = main.tex
%%%%%%%%%%%%%%%%%%%%%%%%
%%%%%%%%%%%%%%%%%%%%%%%%
\section{Sub-Term Property and Size Explosion}
\label{sect:sub-term-size-explosion}
%%%%%%%%%%%%%%%%%%%%%%%%
%%%%%%%%%%%%%%%%%%%%%%%%
The sub-term property is used in many works on abstract machines, or the LSC, or reasonable cost models for the $\l$-calculus. Here we aim at providing evidence of its relevance. We first discuss it in the $\l$-calculus and then see an example of size explosion in proof nets. 

We give  a new and more accurate definition of the property than the one in the introduction. We actually define three variants of the property.

\begin{defi}[Sub-term property]
A rewriting system $(S,\to)$ has the sub-term property if, for every $\tm_0\in S$ and every sequence $\tm_0 \to^n \tm_n$, every step of the sequence either only involves a constant number of constructors or it duplicates or erases a term:
\begin{itemize}
\item \emph{Literal sub-term property}: which is a sub-term of $\tm_0$;
\item \emph{Structural sub-term property}: which is a sub-term of $\tm_0$ up to variables renaming;
\item \emph{Quantitative sub-term property}: of size bound by $\size{\tm_0}$.
\end{itemize}
In the case of duplications, the number of copies is also bound by $\size{\tm_0}$ (for all notions).
\end{defi}

The literal variant is the strongest formulation, it provides the intuition, and gives the name to the property. It is usually found in abstract machines with local environments such as the Krivine abstract machine (KAM) or the CEK. The literal sub-term property is crucial for studying the logarithmic space complexity of $\l$-terms, as it allows one to only duplicate a pointer (of logarithmic size) to the sub-term of $\tm_0$, rather than the sub-term itself (which would have linear size). See Accattoli et al. for more details \cite{DBLP:conf/lics/AccattoliLV22}. We shall not deal with the literal variant here, since we are not dealing with logarithmic space nor abstract machines (the adaptation of the KAM to our setting would have the literal sub-term property).

The structural variant is often found in strategies of the LSC and in abstract machines with global environments (such as the Milner abstract machine, or MAM), and it is the one used for time analyses. The \emph{up to variables renaming} weakening of the property is due to $\alpha$-equivalence, which might change names in sub-terms, but not their structure. For time analyses, in fact, one only needs the quantitative version of the property. The two formulations (structural and quantitative) are two sides of the same concept and most of the time we shall simply refer to the sub-term property, without further specification. 

\paragraph{$\l$-Calculus $vs$ the Sub-Term Property} In the $\l$-calculus, no evaluation strategy has the structural sub-term property---as we now show---because its operational semantics is somewhat too coarse. Let $\tau_{\tm} \defeq \la\vartwo\vartwo\tm\tm$ and $\Id \defeq \la\var\var$. Then:
\begin{equation}
\begin{array}{lllllllll}
\tmtwo &\defeq &(\la\var\var(\la\varthree\tau_{\varthree})\tau_{\var}) \Id
&\tob&
\Id(\la\varthree\tau_{\varthree})\tau_{\Id}
 &\tob
& 
(\la\varthree\tau_{\varthree})\tau_{\Id}
 &\tob&\tau_{\tau_{\Id}}
 \label{eq:no-subterm}
\end{array}\end{equation}
Note that the term $\tau_{\Id}$ duplicated by the third step is not a sub-term of $\tmtwo$. Moreover, $\tmtwo$ is closed and each term in the sequence has at most one $\beta$-redex, which is out of abstractions and the argument of which is a value. Therefore, the example affects also closed weak evaluation, and even its call-by-value variant.

What actually happens is that in the $\l$-calculus duplicated sub-terms are \emph{combinations} of sub-terms of the initial term, as it is the case for $\tau_\Id$ in the example above. The size of such combinations can grow exponentially with the number of $\beta$-steps, a degeneracy called \emph{size explosion} and affecting \emph{all} strategies in the $\l$-calculus, as we shall now recall, suggesting that the number of $\beta$-steps might \emph{not} be a polynomial cost model for \emph{any} strategy. 

\paragraph{Inevitable Size Exploding Family} In the $\l$-calculus, there is a size exploding family that explodes no matter the evaluation strategy. Its definition is in two steps. Firstly, define the following \emph{pre-family} $\{\tm_n\}_{1\leq n\in\nat}$, where, again, $\tau_{\tm} \defeq \la\vartwo\vartwo\tm\tm$:
\begin{center}
$\begin{array}{rcl@{\hspace{2cm}}rcl}
 \tm_1 & \defeq & \la{\var}\tau_\var
 &
  \tm_{n+1}  &\defeq &\la{\var}\tm_n \tau_\var
\end{array}$  
\end{center}

The size exploding family is 
actually given by $\{\tm_n \Id\}_{1\leq n\in\nat}$ , that is, it is obtained by applying the term $\tm_n$ 
of the pre-family to the identity $\Id \defeq \la\vartwo\vartwo$.  Note that it is a family of closed terms.
We also define the family of results $\{\tmfour_n\}_{n\in\nat}$ as follows:
$$\begin{array}{ccc@{\hspace{2cm}}ccc}
\tmfour_0  &\defeq&  \Id &
\tmfour_{n+1}  &\defeq&  \tau_{\tmfour_{n}}
\end{array}$$
Last, let $\towh$ be the weak head strategy, which reduces only head redexes out of abstractions and is the simplest meaningful reduction strategy in the $\l$-calculus.

\begin{prop}[Closed and strategy-independent size explosion]
\label{prop:abs-size-explosion}
  Let 
  $n \!>\! 0$. Then 
  \begin{enumerate}
  \item \emph{Size explosion for the weak head strategy}: $\tm_n \Id \towh^n  \tmfour_n$ and the $i$-th step of the sequence makes two copies of $\tmfour_{i-1}$. 
  \item \emph{Strategy independent}: any other reduction sequence $\tm_n \Id \towh^*  \tmfour_n$ has length $n$, and all redexes in these sequences are $\betav$ redexes.
  \end{enumerate}
  Moreover, $\size{\tm_n \Id} = \bigo(n)$, $\size{\tmfour_n} = \Omega(2^n)$, 
$\tm_n \Id$ is closed, and $\tmfour_n$ is a $\beta$ normal form.
\end{prop}
\pagebreak
\begin{proof}
The bounds on the sizes are straightforward inductions.
\begin{enumerate}
\item We prove a more general statement: $\tm_n \tmfour_m \towh^n  \tmfour_{n+m}$ (note that $\tmfour_1 = \Id$). By induction on $n$. If $n=1$ then $(\la\var\tau_\var) \tmfour_m \towh \tau_{\tmfour_m} = \tmfour_{m+1}$. For $n+1$, $\tm_{n+1} \tmfour_m = (\la\var\tm_n \tau_\var) \tmfour_m \towh \tm_n \tau_{\tmfour_m} = \tm_n \tmfour_{m+1}$. By \ih,  $\tm_n \tmfour_{m+1} \towh^n  \tmfour_{n+m+1}$.
\item It is easily seen that the $n$ redexes in $\tm_n \Id$ are all hereditarily independent and by-value: they cannot duplicate/erase each other, their arguments are abstractions, and they all produce terms with the same property. Then all local confluence diagrams from $\tm_n \Id$, as well as from the terms reachable from it, are diamonds and all reductions to normal form have the same length.\qedhere
\end{enumerate}
\end{proof}
Size explosion is the fact that the linear size term $\tm_n \Id$ reduces in $n$ steps to the exponential size result $\tmfour_n$, thus doing an exponential amount of work in a linear number of steps. The first point specifies that $\tm_n \Id$ explodes while using the weak head strategy, \emph{and} that in this case each step of the sequence (but the first one) breaks the structural sub-term property, as it duplicates $\tmfour_n$ which is \emph{not} a sub-term of the initial term $\tm_n \Id$ (up to variables renaming). Thus, breaking the sub-term property leads to size explosion for the weak head strategy, which breaks the \emph{quantitative} sub-term property, as, for $n$ big enough, the last duplication of the sequence involves a term bigger than the initial one.

The second point says that the family evaluates to the same result and in the same number of steps according to any other evaluation strategy, no matter whether strong or weak, call-by-name or call-by-value. It is easily seen that the family is also typable with simple types, although of types that grow in size exponentially in $n$. Since the family is also closed, size explosion affects all the main variants and dialects of the $\l$-calculus.

Let us show a final observation about size explosion. Consider the third element $\tm_3 \Id$ of the given family, and let us reduce it in a \emph{innermost way}:
\begin{equation}
\arraycolsep=4pt
\begin{array}{cccccccccccc}
\tm_3 \Id &  = & (\la{\var}(\la{\var}(\la{\var}\tau_{\var}) \tau_{\var}) \tau_{\var}) \Id
& \tob & (\la{\var}(\la{\var}\tau_{\tau_{\var}}) ) \tau_{\var}) \Id
& \tob &  (\la{\var}\tau_{\tau_{\tau_{\var}}}) \Id
& \tob & \tau_{\tau_{\tau_{ \Id}}}
\end{array}
\label{eq:innerm-size-exp}
\end{equation}
Note that the first two steps duplicate $\tau_\var$ and that the last step duplicates $\Id$: it looks like the structural sub-term property is not broken. Here it plays a role the final clause in the definition of the property:
\emph{in the case of duplications, the number of copies is also bound by $\size{\tm_0}$.}
Generalizing \refeq{innerm-size-exp}  to $n>3$, indeed, the number of copies of the duplicated sub-term grows exponentially (note that $\tau_{\tau_{\tau_{\var_3}}}$ has eight occurrences of $\var_3$), thus still breaking the sub-term property.

\paragraph{Small-Steps $vs$ Micro-Steps}
The reader might be aware of the fact that, despite what just explained, the number of steps of some strategies of the $\l$-calculus do provide polynomial cost models. How is that possible?

There is a way out, but it requires deviating from the $\l$-calculus. The lack of the sub-term property in the $\l$-calculus stems from its \emph{small-step} operational semantics, $\beta$-reduction, which is based on meta-level substitution. To retrieve the property, one has to add ESs (or some other form of sharing) to the $\l$-calculus  and decompose $\beta$-reduction in \emph{micro-steps} performing \emph{one} variable replacement at a time. Then, there are strategies for the refined setting that have the sub-term property, of which the number of steps thus is a  polynomial cost model. The micro-step refinement can be formulated in various ways,  the LSC being an elegant one. Without going into too much detail, the result is then transferred to the $\l$-calculus, by showing that there is a polynomial simulation (up to some form of sharing) of a small-step strategy by a micro-step one with the sub-term property. Therefore, small-step polynomiality is inherited from the micro-step formalism (as in Blelloch and Greiner \cite{DBLP:conf/fpca/BlellochG95}, Sands et al. \cite{DBLP:conf/birthday/SandsGM02}, Dal Lago and Martini \cite{DBLP:journals/corr/abs-1208-0515}, Accattoli and Dal Lago \cite{DBLP:conf/rta/AccattoliL12,DBLP:journals/corr/AccattoliL16}, Accattoli et al. \cite{DBLP:conf/lics/AccattoliCC21}, Biernacka et al. \cite{DBLP:conf/ppdp/BiernackaCD21}). 

Summing up, the sub-term property seems to be \emph{mandatory} for the study of cost models, because even when the cost model strategy does not have the property, the proof of its polynomiality rests on an auxiliary strategy that does indeed have it. Moreover, the micro-step aspect seems \emph{unavoidable} for unlocking polynomial time cost models. 

\paragraph{Size Explosion in Linear Logic}
We now explain how to obtain a size exploding family for linear logic proof nets, taking as notion of evaluation \emph{cut elimination at level 0} (the level of a link is the number of $!$ boxes in which it is contained), noted $\Rew{l_0}$, which is, roughly, the analogous of head reduction for proof nets (all head reduction redexes are redexes at level 0, when $\l$-terms are translated to proof nets). We shall need only the two reduction rules for proof nets in \reffig{pn-rules} (page \pageref{fig:pn-rules}). We shall use $\Rew{\contr 0}$ and $\Rew{\boxbox 0}$ for when those rules are applied to cuts at level $0$.

\paragraph{Size Exploding Family.} For the time being, we consider untyped proof nets. We now define various families of proof nets. The $n$-th net $\tau_{n}$ has the following shape:
\begin{center}
\begin{tikzpicture}[ocenter]
\node at (0,0) [nospace](bang){};
\node at (bang.center) [etic, above left = 6pt and 11pt](name){$\tau_{n}$};
\node at (bang.center) [etic, left = \stlar](auxr){};
\node at (auxr.center) [etic, left = 1.4*\hstlar](auxl){};
\node at (auxr.center) [etic, below = 0.7*\hstalt](auxrghost){};
\node at (auxl.center) [etic, below = 0.7*\hstalt](auxlghost){};
\node at \med{auxl}{auxr} [etic, below = 0.4*\hstalt](dots){$...$};
\node at \med{auxlghost}{auxrghost}[below=\sepbox, rotate=90, nospace, anchor = center](bracket){\huge \{};
\node at (bracket.center)[ below=\sepbox+2pt, nospace](bracket2){\footnotesize $2^{n}$};

\abox{bang}{exbox}{30pt}{5pt}{15pt}
\node at (bang.center)[etic] (bangsym){$!$};
\draw[nopol](auxl)to(auxlghost);
\draw[nopol](auxr)to(auxrghost);

\node at (bang.center) [etic, below = 0.7*\hstalt](bangghost){};
\draw[nopol](bangsym)to(bangghost);
\end{tikzpicture}
\end{center}
And it is defined as follows:
\begin{center}
\begin{tabular}{ccccc}
\input{figure-size-exploding-seed}
&&&
\input{figure-size-exploding-result-no-contrs}
\end{tabular}
\end{center}
Then we define the family $\pi_{n}$, which is obtained from $\tau_{n}$ by contracting all its $2^{n}$ auxiliary conclusions via a tree of contractions:
\begin{center}
$\pi_{n} \ \defeq\ \ $
\begin{tikzpicture}[ocenter]
\node at (0,0) [nospace](bang){};
\node at (bang.center) [etic, above left = 6pt and 11pt](name){$\tau_{n}$};
\node at (bang.center) [etic, left = \stlar](auxr){};
\node at (auxr.center) [etic, left = 1.4*\hstlar](auxl){};
\node at (auxr.center) [etic, below = 0.7*\hstalt](auxrghost){};
\node at (auxl.center) [etic, below = 0.7*\hstalt](auxlghost){};
\node at \med{auxl}{auxr} [etic, below = 0.4*\hstalt](dots){$...$};
\inetcell[at=\med{auxlghost}{auxrghost},inductiveTrSmall,anchor=center,below=-1pt](contr){$\contr$};
\node at (contr.pal) [etic, below = 0.7*\hstalt](contrghost){};
\draw[nopol](contr)to(contrghost);

\abox{bang}{exbox}{30pt}{5pt}{15pt}
\node at (bang.center)[etic] (bangsym){$!$};
\draw[nopol](auxl)to(auxlghost);
\draw[nopol](auxr)to(auxrghost);

\node at (bang.center) [etic, below = 0.7*\hstalt](bangghost){};
\draw[nopol](bangsym)to(bangghost);
\end{tikzpicture}
\ also noted \ 
\begin{tikzpicture}[ocenter]
\node at (0,0) [nospace](bang){};
\node at (bang.center) [etic, above left = 6pt and 5pt](name){$\pi_{n}$};
\node at (bang.center) [nospace, left = \stlar](auxr){};
\node at (auxr.center) [etic, below = 0.7*\hstalt](auxrghost){};

\abox{bang}{exbox}{20pt}{5pt}{15pt}
\node at (bang.center)[etic] (bangsym){$!$};
\draw[nopol](auxr)to(auxrghost);

\node at (bang.center) [etic, below = 0.7*\hstalt](bangghost){};
\draw[nopol](bangsym)to(bangghost);
\end{tikzpicture}
, for instance:
\input{figure-size-exploding-seed-with-contraction}.
\end{center}
Finally, define the following family of nets:
\begin{center}
$\begin{array}{lllllllllll}
\rho_1 & \defeq & \pi_{1}
&&&&&
\rho_{n+1} & \defeq & \input{figure-size-exploding-family}
\end{array}$
\end{center}

\begin{prop}[Size explosion in MELL]
\label{prop:size-exp}
$\rho_{n} \Rew{l_0}^{3(n-1)} \pi_{n}$, $\size{\rho_{n}}= \bigo(n)$, and $\size{\pi_{n}}= \Omega(2^{n})$.
\end{prop}
\begin{proof}
By induction on $n$. The bounds on the sizes follow immediately from the definition and the \ih
For $n=1$, we have $\rho_{1} = \pi_{1}$, so the statement holds. For $n+1$, we have that by \ih $\rho_{n} \Rew{l_0}^{3(n-1)} \pi_{n}$. Therefore, we obtain:
\begin{center}
\input{figure-size-exploding-family-proof} 
\end{center}
\end{proof}

The size exploding family of \refprop{size-exp} is a chain of what Roversi and Vercelli in \cite{DBLP:conf/fopara/RoversiV09} call \emph{spindles}, which they identify as the problematic configurations for representing polynomial time.

 \paragraph{Typing the Size Exploding Family} The proof net $\rho_{1} = \pi_{1}$ used to build the family can be typed by assigning type $\form$, for an arbitrary type $\form$, to the right conclusion of the axioms (thus having $\form^{\bot}$ on the left conclusion of the axioms), and $\form_{1} \defeq !(\form \tens\form)$ for the right conclusions of $\rho_{1}$ (and $?\form^\bot$ on the left conclusion of $\rho_{1}$). To type the second term of the family $\rho_{2}$, and assuming that the left occurrence of $\pi_{1}$ in $\rho_{2}$ is typed as just described, the right one needs to be typed using $\form_1$ as type of the axioms, thus obtaining as type of the right conclusion, $\form_{2} \defeq !(\form_1\tens \form_1) =  !(!(\form \tens\form)\tens !(\form \tens\form))$. It is then easily seen that the type $\form_{n+1}$ of the right conclusion of $\rho_{n+1}$ is given by $\form_{n+1} = !(\form_{n}\tens \form_{n})$, that is, the type size grows exponentially with $n$.
 
\paragraph{Size Explosion and Elementary Linear Logic} The proofs $\rho_{n}$ of the size exploding family actually belong to \emph{elementary linear logic} (ELL): they use derelictions---which are forbidden in ELL---but those derelictions can be seen as the auxiliary ports of functorial promotions associated to the $!$s, which are allowed in ELL. Therefore, size explosion affects ELL as well. We mention this fact because ELL has some interesting properties with respect to \levy's \emph{optimal reduction}. On the one hand, ELL provides a much simpler implementation of optimal reduction, because the technical \emph{oracle} part of the implementation is not needed in ELL, as  pointed out by Asperti \cite{DBLP:conf/lics/Asperti98}. On the other hand, the complex behaviour of duplication responsible for the \emph{unreasonability} of optimal reduction is already fully present in ELL, and it is thus independent of the oracle, as shown by Asperti et al.  \cite{DBLP:conf/popl/AspertiCM00}. For the study of time cost models, restricting to E(ME)LL does not seem to simplify the problem, as size explosion affects both MELL and EMELL with respect to reduction at level 0. This is rather  in line with the results of \cite{DBLP:conf/popl/AspertiCM00}.

%% file: figure-size-exploding-seed.tex
% !TEX root = main.tex
$\tau_{1} \defeq$
\begin{tikzpicture}[ocenter]
\node at (0,0) [etic, left = 2*\stlar](belowax){\scriptsize $\ax$};
\node at (belowax.center) [etic, above = \hstalt ](axSym){\scriptsize $\ax$};
\node at (axSym) [etic, below right = 1*\stalt and \stlar](tens){\scriptsize $\tens$};
\draw[nopol, out=0, in=45](axSym)to(tens);
\draw[nopol, out = 0, in=135](belowax)to(tens);

\node at (tens) [etic, left = 2*\stlar](derright){$\der$};
\draw[nopol, out = 180, in=90](belowax)to(derright);
\node at (derright) [etic, left = \stlar](derleft){$\der$};
\draw[nopol, out = 180, in=90](axSym)to(derleft);

\node at (derleft.center) [etic, below = .8*\stalt](derleftghost){};
\node at (derright.center) [etic, below = .8*\stalt](derrightghost){};
\draw[nopol](derleft)to(derleftghost);
\draw[nopol](derright)to(derrightghost);

\node at (tens) [nospace, below = \hstalt](bang){};
\abox{bang}{exbox}{55pt}{10pt}{42pt}
\node at (bang.center)[etic] (bangsym){$!$};
\draw[nopol](tens)to(bangsym);

\node at (bang.center) [etic, below = \hstalt](bangghost){};
\draw[nopol](bangsym)to(bangghost);

%\node at (axRightConclusion) [etic, below left = .7*\stalt and 3.5*\stlar](axRightConclusion2){};
%\node at (axRightConclusion2.center) [etic, left = 1.4*\stlar](der2){\scriptsize $\der$};
%\node at \med{axRightConclusion2}{der2} [etic, above = \hstalt ](axSym2){\scriptsize $\ax$};
%\node at (axSym2) [etic, below = 1.4*\stalt](contr2){\scriptsize $\csym$};
%\draw[nopol, out=0, in=45](axSym2)to(contr2);
%\draw[nopol, out=180, in=90](axSym2)to(der2);
%\draw[nopol, out = -90, in=135](der2)to(contr2);
%
%\node at \med{bangsym}{contr2} [etic, below = \hstalt ](cut){\scriptsize $\cut$};
%\draw[nopol, out=-90, in=0](bangsym)to(cut);
%\draw[nopol, out = -90, in=180](contr2)to(cut);
\end{tikzpicture}

%% file: figure-size-exploding-result-no-contrs.tex
% !TEX root = main.tex
$\tau_{n+1} \defeq$
\begin{tikzpicture}[ocenter]
\node at (0,0) [etic, left = 2*\stlar](belowax){\scriptsize $\ax$};
\node at (belowax.center) [etic, above = \hstalt ](axSym){\scriptsize $\ax$};
\node at (axSym) [etic, below right = \stalt and \stlar](tens){\scriptsize $\tens$};
\draw[nopol, out=0, in=45](axSym)to(tens);
\draw[nopol, out = 0, in=135](belowax)to(tens);

\node at (tens) [etic, left = 2*\stlar](derright){$\der$};
\draw[nopol, out = 180, in=90](belowax)to(derright);
\node at (derright) [etic, left = \stlar](derleft){$\der$};
\draw[nopol, out = 180, in=90](axSym)to(derleft);

\node at (derleft.center) [left = 1.6*\stlar] (bang){};
\node at \med{derleft}{bang} [etic, below = 0.8*\hstalt](cut){$\cut$};
\node at (bang.center) [etic, above left = 6pt and 11pt](name){$\tau_{n}$};
\node at (bang.center) [etic, left = \stlar](auxr){};
\node at (auxr.center) [etic, left = 1.4*\hstlar](auxl){};
\node at (auxr.center) [etic, below = 1.8*\stalt](auxrghost){};
\node at (auxl.center) [etic, below = 1.8*\stalt](auxlghost){};
\node at \med{auxl}{auxr} [etic, below = 0.4*\hstalt](dots){$...$};
\node at \med{auxlghost}{auxrghost}[below=\sepbox, rotate=90, nospace, anchor = center](bracket){\huge \{};
\node at (bracket.center)[ below=\sepbox+2pt, nospace](bracketlab){\footnotesize $2^{n}$};

\abox{bang}{exbox}{30pt}{5pt}{15pt}
\node at (bang.center)[etic] (bangsym){$!$};
\draw[nopol](auxl)to(auxlghost);
\draw[nopol](auxr)to(auxrghost);

\node at (bang.center) [etic, below = 0.7*\hstalt](bangghost){};
\draw[nopol, out = -90, in=0](derleft)to(cut);
\draw[nopol, out = -90, in=180](bangsym)to(cut);

\node at (bang.center) [nospace, left= 2.6*\stlar](bang2){};
\node at (bang2.center) [etic, above left = 6pt and 11pt](name2){$\tau_{n}$};
\node at (bang2.center) [etic, left = \stlar](auxr2){};
\node at (auxr2.center) [etic, left = 1.4*\hstlar](auxl2){};
\node at (auxr2.center) [etic, below = 1.8*\stalt](auxrghost2){};
\node at (auxl2.center) [etic, below = 1.8*\stalt](auxlghost2){};
\node at \med{auxl2}{auxr2} [etic, below = 0.4*\hstalt](dots){$...$};
\node at \med{auxlghost2}{auxrghost2}[below=\sepbox, rotate=90, nospace, anchor = center](bracket2){\huge \{};
\node at (bracket2.center)[ below=\sepbox+2pt, nospace](bracketlab2){\footnotesize $2^{n}$};

\abox{bang2}{exbox}{30pt}{5pt}{15pt}
\node at (bang2.center)[etic] (bangsym2){$!$};
\draw[nopol](auxl2)to(auxlghost2);
\draw[nopol](auxr2)to(auxrghost2);

\node at \med{derright}{bang2} [etic, below = 2*\hstalt](cut2){$\cut$};
\draw[nopol, out = -90, in=0](derright)to(cut2);
\draw[nopol, out = -90, in=180](bangsym2)to(cut2);

\node at (tens) [nospace, below = 1.4*\stalt](bang3){};
\abox{bang3}{exbox}{150pt}{10pt}{62pt}
\node at (bang3.center)[etic] (bangsym3){$!$};
\draw[nopol](tens)to(bangsym3);

\node at (bang) [etic, below = 0.7*\hstalt](bangghost){};
\draw[nopol](bangsym)to(bangghost);

\end{tikzpicture}

%% file: figure-size-exploding-seed-with-contraction.tex
% !TEX root = main.tex
\begin{tabular}{cccc}

$\pi_{1} \ =$
&
\begin{tikzpicture}[ocenter]
\node at (0,0) [nospace](bang){};
\node at (bang.center) [etic, above left = 6pt and 11pt](name){$\tau_{1}$};
\node at (bang.center) [etic, left = \stlar](auxr){};
\node at (auxr.center) [etic, left = 1.4*\hstlar](auxl){};

\abox{bang}{exbox}{30pt}{5pt}{15pt}
\node at (bang.center)[etic] (bangsym){$!$};

\node at (bang.center) [etic, below = \hstalt](bangghost){};
\draw[nopol](bangsym)to(bangghost);

\node at \med{auxl}{auxr} [etic, below = 0.7*\hstalt](contr){$\csym$};
\node at (contr.center) [etic, below = \hstalt](contrghost){};
\draw[nopol, out = -90, in=135](auxl)to(contr);
\draw[nopol, out = -90, in=45](auxr)to(contr);
\draw[nopol](contr)to(contrghost);
\end{tikzpicture}
&
=
&
\begin{tikzpicture}[ocenter]
\node at (0,0) [etic, left = 2*\stlar](belowax){\scriptsize $\ax$};
\node at (belowax.center) [etic, above = \hstalt ](axSym){\scriptsize $\ax$};
\node at (axSym) [etic, below right = 1*\stalt and \stlar](tens){\scriptsize $\tens$};
\draw[nopol, out=0, in=45](axSym)to(tens);
\draw[nopol, out = 0, in=135](belowax)to(tens);

\node at (tens) [etic, left = 2*\stlar](derright){$\der$};
\draw[nopol, out = 180, in=90](belowax)to(derright);
\node at (derright) [etic, left = \stlar](derleft){$\der$};
\draw[nopol, out = 180, in=90](axSym)to(derleft);

\node at \med{derleft}{derright} [etic, below = .8*\stalt](contr){$\csym$};
\node at (contr.center) [etic, below = \hstalt](contrghost){};
\draw[nopol, out = -90, in=135](derleft)to(contr);
\draw[nopol, out = -90, in=45](derright)to(contr);
\draw[nopol](contr)to(contrghost);

\node at (tens) [nospace, below = \hstalt](bang){};
\abox{bang}{exbox}{55pt}{10pt}{42pt}
\node at (bang.center)[etic] (bangsym){$!$};
\draw[nopol](tens)to(bangsym);

\node at (bang.center) [etic, below = \hstalt](bangghost){};
\draw[nopol](bangsym)to(bangghost);

%\node at (axRightConclusion) [etic, below left = .7*\stalt and 3.5*\stlar](axRightConclusion2){};
%\node at (axRightConclusion2.center) [etic, left = 1.4*\stlar](der2){\scriptsize $\der$};
%\node at \med{axRightConclusion2}{der2} [etic, above = \hstalt ](axSym2){\scriptsize $\ax$};
%\node at (axSym2) [etic, below = 1.4*\stalt](contr2){\scriptsize $\csym$};
%\draw[nopol, out=0, in=45](axSym2)to(contr2);
%\draw[nopol, out=180, in=90](axSym2)to(der2);
%\draw[nopol, out = -90, in=135](der2)to(contr2);
%
%\node at \med{bangsym}{contr2} [etic, below = \hstalt ](cut){\scriptsize $\cut$};
%\draw[nopol, out=-90, in=0](bangsym)to(cut);
%\draw[nopol, out = -90, in=180](contr2)to(cut);
\end{tikzpicture}
\end{tabular}

%% file: figure-size-exploding-family.tex
% !TEX root = main.tex
\begin{tikzpicture}[ocenter]
\node at (0,0) [nospace](bang){};
\node at (bang.center) [etic, above left = 6pt and 5pt](name){$\pi_{1}$};
\node at (bang.center) [nospace, left = \stlar](auxr){};

\abox{bang}{exbox}{20pt}{5pt}{15pt}
\node at (bang.center)[etic] (bangsym){$!$};

\node at (bang.center) [etic, below = 0.7*\hstalt](bangghost){};
\draw[nopol](bangsym)to(bangghost);

\node at (auxr.center) [etic, left = 2*\stlar](bang2){};
\node at (bang2.center) [etic, above left = 6pt and 5pt](name2){$\rho_{n}$};
\node at (bang2.center) [etic, left = \stlar](auxr2){};
\node at (auxr2.center) [etic, below = 0.7*\hstalt](auxrghost2){};

\abox{bang2}{exbox}{20pt}{5pt}{15pt}
\node at (bang2.center)[etic] (bangsym2){$!$};
\draw[nopol](auxr2)to(auxrghost2);

\node at \med{bang2}{auxr} [etic, below = 0.7*\hstalt](cut){$\cut$};

\draw[nopol, out = -90, in=0](auxr)to(cut);
\draw[nopol, out = -90, in=180](bangsym2)to(cut);
\end{tikzpicture}

%% file: figure-size-exploding-family-proof.tex
% !TEX root = main.tex
\begin{tabular}{ccccccccccc}
$\rho_{n+1}$ 
&=& 
\input{figure-size-exploding-family}
&
$\Rew{l_0}^{3(n-1)} $
&
\begin{tikzpicture}[ocenter]
\node at (0,0) [nospace](bang){};
\node at (bang.center) [etic, above left = 6pt and 5pt](name){$\pi_{1}$};
\node at (bang.center) [nospace, left = \stlar](auxr){};

\abox{bang}{exbox}{20pt}{5pt}{15pt}
\node at (bang.center)[etic] (bangsym){$!$};

\node at (bang.center) [etic, below = 0.7*\hstalt](bangghost){};
\draw[nopol](bangsym)to(bangghost);

\node at (auxr.center) [etic, left = 2*\stlar](bang2){};
\node at (bang2.center) [etic, above left = 6pt and 5pt](name2){$\pi_{n}$};
\node at (bang2.center) [etic, left = \stlar](auxr2){};
\node at (auxr2.center) [etic, below = 0.7*\hstalt](auxrghost2){};

\abox{bang2}{exbox}{20pt}{5pt}{15pt}
\node at (bang2.center)[etic] (bangsym2){$!$};
\draw[nopol](auxr2)to(auxrghost2);

\node at \med{bang2}{auxr} [etic, below = 0.7*\hstalt](cut){$\cut$};

\draw[nopol, out = -90, in=0](auxr)to(cut);
\draw[nopol, out = -90, in=180](bangsym2)to(cut);
\end{tikzpicture}
&
=
&
\begin{tikzpicture}[ocenter]
\node at (0,0) [etic, left = 2*\stlar](belowax){\scriptsize $\ax$};
\node at (belowax.center) [etic, above = \hstalt ](axSym){\scriptsize $\ax$};
\node at (axSym) [etic, below right = 1*\stalt and \stlar](tens){\scriptsize $\tens$};
\draw[nopol, out=0, in=45](axSym)to(tens);
\draw[nopol, out = 0, in=135](belowax)to(tens);

\node at (tens) [etic, left = 2*\stlar](derright){$\der$};
\draw[nopol, out = 180, in=90](belowax)to(derright);
\node at (derright) [etic, left = \stlar](derleft){$\der$};
\draw[nopol, out = 180, in=90](axSym)to(derleft);

\node at \med{derleft}{derright} [etic, below = .8*\stalt](contr){$\csym$};
\draw[nopol, out = -90, in=135](derleft)to(contr);
\draw[nopol, out = -90, in=45](derright)to(contr);

\node at (tens) [nospace, below = \hstalt](bang){};
\abox{bang}{exbox}{55pt}{10pt}{42pt}
\node at (bang.center)[etic] (bangsym){$!$};
\draw[nopol](tens)to(bangsym);

\node at (bang.center) [etic, below = \hstalt](bangghost){};
\draw[nopol](bangsym)to(bangghost);
%%%%%

\node at (contr.center) [etic, left = 2*\stlar](bang2){};
\node at (bang2.center) [etic, above left = 6pt and 11pt](name2){$\tau_{n}$};
\node at (bang2.center) [etic, left = \stlar](auxr2){};
\node at (auxr2.center) [etic, left = 1.4*\hstlar](auxl2){};
\node at (auxr2.center) [etic, below = 0.7*\hstalt](auxrghost2){};
\node at (auxl2.center) [etic, below = 0.7*\hstalt](auxlghost2){};
\node at \med{auxl2}{auxr2} [etic, below = 0.4*\hstalt](dots2){$...$};
\inetcell[at=\med{auxlghost2}{auxrghost2},inductiveTrSmall,anchor=center,below=-1pt](contr2){$\contr$};
\node at (contr2.pal) [etic, below = 0.7*\hstalt](contrghost2){};
\draw[nopol](contr2)to(contrghost2);

\abox{bang2}{exbox}{30pt}{5pt}{15pt}
\node at (bang2.center)[etic] (bangsym2){$!$};
\draw[nopol](auxl2)to(auxlghost2);
\draw[nopol](auxr2)to(auxrghost2);

%%%

\node at \med{bang2}{contr} [etic, below = 0.7*\hstalt](cut){$\cut$};

\draw[nopol, out = -90, in=0](contr)to(cut);
\draw[nopol, out = -90, in=180](bangsym2)to(cut);
\end{tikzpicture}
\end{tabular}

\begin{tabular}{ccccccccccc}
$\Rew{\contr 0}$
&
\begin{tikzpicture}[ocenter]
\node at (0,0) [etic, left = 2*\stlar](belowax){\scriptsize $\ax$};
\node at (belowax.center) [etic, above = \hstalt ](axSym){\scriptsize $\ax$};
\node at (axSym) [etic, below right = 1*\stalt and \stlar](tens){\scriptsize $\tens$};
\draw[nopol, out=0, in=45](axSym)to(tens);
\draw[nopol, out = 0, in=135](belowax)to(tens);

\node at (tens) [etic, left = 2*\stlar](derright){$\der$};
\draw[nopol, out = 180, in=90](belowax)to(derright);
\node at (derright) [etic, left = \stlar](derleft){$\der$};
\draw[nopol, out = 180, in=90](axSym)to(derleft);

\node at (tens) [nospace, below = \hstalt](bang3){};
\abox{bang3}{exbox}{55pt}{10pt}{42pt}
\node at (bang3.center)[etic] (bangsym3){$!$};
\draw[nopol](tens)to(bangsym3);

\node at (bang3.center) [etic, below = \hstalt](bangghost3){};
\draw[nopol](bangsym3)to(bangghost3);
%%%%%

%\node at (contr.center) [etic, left = 2*\stlar](bang2){};
%\node at (bang2.center) [etic, above left = 6pt and 11pt](name2){$\tau_{n}$};
%\node at (bang2.center) [etic, left = \stlar](auxr2){};
%\node at (auxr2.center) [etic, left = 1.4*\hstlar](auxl2){};
%\node at (auxr2.center) [etic, below = 0.7*\hstalt](auxrghost2){};
%\node at (auxl2.center) [etic, below = 0.7*\hstalt](auxlghost2){};

\node at (derleft|-bang3) [left = 1.6*\stlar] (bang){};
\node at \med{derleft}{bang} [etic, below = 1.2*\hstalt](cut){$\cut$};
\node at (bang.center) [etic, above left = 6pt and 11pt](name){$\tau_{n}$};
\node at (bang.center) [etic, left = \stlar](auxr){};
\node at (auxr.center) [etic, left = 1.4*\hstlar](auxl){};
\node at (auxr.center) [etic, below = 1.8*\stalt](auxrghost){};
\node at (auxl.center) [etic, below = 1.8*\stalt](auxlghost){};
\node at \med{auxl}{auxr} [etic, below = 0.4*\hstalt](dots){$...$};

\abox{bang}{exbox}{30pt}{5pt}{15pt}
\node at (bang.center)[etic] (bangsym){$!$};

\draw[nopol, out = -90, in=0](derleft)to(cut);
\draw[nopol, out = -90, in=180](bangsym)to(cut);

\node at (bang.center) [nospace, left= 2.6*\stlar](bang2){};
\node at (bang2.center) [etic, above left = 6pt and 11pt](name2){$\tau_{n}$};
\node at (bang2.center) [etic, left = \stlar](auxr2){};
\node at (auxr2.center) [etic, left = 1.4*\hstlar](auxl2){};
\node at (auxr2.center) [etic, below = 1.8*\stalt](auxrghost2){};
\node at (auxl2.center) [etic, below = 1.8*\stalt](auxlghost2){};
\node at \med{auxl2}{auxr2} [etic, below = 0.4*\hstalt](dots){$...$};

\abox{bang2}{exbox}{30pt}{5pt}{15pt}
\node at (bang2.center)[etic] (bangsym2){$!$};

\node at (cut.center) [etic, below = 0.8*\hstalt](cut2){$\cut$};
\draw[nopol, out = -90, in=0](derright)to(cut2);
\draw[nopol, out = -75, in=180](bangsym2)to(cut2);
\node at \med{auxl2}{auxr2} [etic, below = 0.4*\hstalt](dots2){$...$};

\node at \med{auxl}{auxl2} [etic,  below left =  \stalt and 2pt](contrL){$\csym$};
\draw[nopol, out = -90, in=135](auxl2)to(contrL);
\draw[nopol, out = -90, in=45](auxl)to(contrL);

\node at \med{auxr}{auxr2} [etic,  below right=  \stalt and 2pt](contrR){$\csym$};
\draw[nopol, out = -90, in=135](auxr2)to(contrR);
\draw[nopol, out = -90, in=45](auxr)to(contrR);

\node at \med{contrL}{contrR} [etic](dotsContr){$...$};

\inetcell[at=\med{contrL}{contrR},inductiveTrSmall,anchor=center,below=\hstalt](contrTree){$\contr$};
\node at (contrTree.pal) [etic, below = 0.7*\hstalt](contrTreeGhost){};
\draw[nopol](contrTree)to(contrTreeGhost);
\draw[nopol,out = -90, in =70](contrR)to(contrTree.left pax);
\draw[nopol,out = -90, in =110](contrL)to(contrTree.right pax);

\node at \med{contrTree.left pax}{contrTree.right pax} [etic, above = 2pt](dotsTree){$...$};
\end{tikzpicture}
&
$\Rew{\boxbox 0}\Rew{\boxbox 0}$
&
\begin{tikzpicture}[ocenter]
\node at (0,0) [etic, left = 2*\stlar](belowax){\scriptsize $\ax$};
\node at (belowax.center) [etic, above = \hstalt ](axSym){\scriptsize $\ax$};
\node at (axSym) [etic, below right = \stalt and \stlar](tens){\scriptsize $\tens$};
\draw[nopol, out=0, in=45](axSym)to(tens);
\draw[nopol, out = 0, in=135](belowax)to(tens);

\node at (tens) [etic, left = 2*\stlar](derright){$\der$};
\draw[nopol, out = 180, in=90](belowax)to(derright);
\node at (derright) [etic, left = \stlar](derleft){$\der$};
\draw[nopol, out = 180, in=90](axSym)to(derleft);

\node at (derleft.center) [left = 1.6*\stlar] (bang){};
\node at \med{derleft}{bang} [etic, below = 0.8*\hstalt](cut){$\cut$};
\node at (bang.center) [etic, above left = 6pt and 11pt](name){$\tau_{n}$};
\node at (bang.center) [etic, left = \stlar](auxr){};
\node at (auxr.center) [etic, left = 1.4*\hstlar](auxl){};
\node at \med{auxl}{auxr} [etic, below = 0.4*\hstalt](dots){$...$};

\abox{bang}{exbox}{30pt}{5pt}{15pt}
\node at (bang.center)[etic] (bangsym){$!$};

\node at (bang.center) [etic, below = 0.7*\hstalt](bangghost){};
\draw[nopol, out = -90, in=0](derleft)to(cut);
\draw[nopol, out = -90, in=180](bangsym)to(cut);

\node at (bang.center) [nospace, left= 2.6*\stlar](bang2){};
\node at (bang2.center) [etic, above left = 6pt and 11pt](name2){$\tau_{n}$};
\node at (bang2.center) [etic, left = \stlar](auxr2){};
\node at (auxr2.center) [etic, left = 1.4*\hstlar](auxl2){};
\node at \med{auxl2}{auxr2} [etic, below = 0.4*\hstalt](dots){$...$};

\abox{bang2}{exbox}{30pt}{5pt}{15pt}
\node at (bang2.center)[etic] (bangsym2){$!$};

\node at \med{derright}{bang2} [etic, below = 2*\hstalt](cut2){$\cut$};
\draw[nopol, out = -90, in=0](derright)to(cut2);
\draw[nopol, out = -90, in=180](bangsym2)to(cut2);

\node at (tens) [nospace, below = 1.4*\stalt](bang3){};
\abox{bang3}{exbox}{150pt}{10pt}{62pt}
\node at (bang3.center)[etic] (bangsym3){$!$};
\draw[nopol](tens)to(bangsym3);

\node at \med{auxl}{auxl2} [etic,  below left =  1.8*\stalt and 2pt](contrL){$\csym$};
\draw[nopol, out = -90, in=135](auxl2)to(contrL);
\draw[nopol, out = -90, in=45](auxl)to(contrL);

\node at \med{auxr}{auxr2} [etic,  below right=  1.8*\stalt and 2pt](contrR){$\csym$};
\draw[nopol, out = -90, in=135](auxr2)to(contrR);
\draw[nopol, out = -90, in=45](auxr)to(contrR);

\node at \med{contrL}{contrR} [etic](dotsContr){$...$};

\inetcell[at=\med{contrL}{contrR},inductiveTrSmall,anchor=center,below=\hstalt](contrTree){$\contr$};
\node at (contrTree.pal) [etic, below = 0.7*\hstalt](contrTreeGhost){};
\draw[nopol](contrTree)to(contrTreeGhost);
\draw[nopol,out = -90, in =70](contrR)to(contrTree.left pax);
\draw[nopol,out = -90, in =110](contrL)to(contrTree.right pax);

\node at \med{contrTree.left pax}{contrTree.right pax} [etic, above = 2pt](dotsTree){$...$};
\end{tikzpicture}
\end{tabular}

\hfill\begin{tabular}{cccc}
=
&
\begin{tikzpicture}[ocenter]
\node at (0,0) [nospace](bang){};
\node at (bang.center) [etic, above left = 6pt and 5pt](name){$\tau_{n+1}$};
\node at (bang.center) [etic, left = \stlar](auxr){};
\node at (auxr.center) [etic, left = 1.4*\hstlar](auxl){};
\node at (auxr.center) [etic, below = 0.7*\hstalt](auxrghost){};
\node at (auxl.center) [etic, below = 0.7*\hstalt](auxlghost){};
\node at \med{auxl}{auxr} [etic, below = 0.4*\hstalt](dots){$...$};
\inetcell[at=\med{auxlghost}{auxrghost},inductiveTrSmall,anchor=center,below=-1pt](contr){$\contr$};
\node at (contr.pal) [etic, below = 0.7*\hstalt](contrghost){};
\draw[nopol](contr)to(contrghost);

\abox{bang}{exbox}{30pt}{5pt}{15pt}
\node at (bang.center)[etic] (bangsym){$!$};
\draw[nopol](auxl)to(auxlghost);
\draw[nopol](auxr)to(auxrghost);

\node at (bang.center) [etic, below = 0.7*\hstalt](bangghost){};
\draw[nopol](bangsym)to(bangghost);
\end{tikzpicture}
& = & 
\begin{tikzpicture}[ocenter]
\node at (0,0) [nospace](bang){};
\node at (bang.center) [etic, above left = 6pt and -2pt](name){$\pi_{n+1}$};
\node at (bang.center) [nospace, left = \stlar](auxr){};
\node at (auxr.center) [etic, below = 0.7*\hstalt](auxrghost){};

\abox{bang}{exbox}{20pt}{5pt}{15pt}
\node at (bang.center)[etic] (bangsym){$!$};
\draw[nopol](auxr)to(auxrghost);

\node at (bang.center) [etic, below = 0.7*\hstalt](bangghost){};
\draw[nopol](bangsym)to(bangghost);
\end{tikzpicture}\ .
\end{tabular} \hfill\qedhere

%% file: 04-The_Linear_Substitution_Calculus.tex
% !TEX root = main.tex

\section{The Linear Substitution Calculus}
\label{sect:lsc}
Here we recall the basics of the linear substitution calculus, which is the footprint of the exponential substitution calculus that shall be introduced in the next section.

There are various ingredients. First of all, the syntax of the $\l$-calculus is extended with explicit substitutions $\tm\esub\var\tmtwo$ (shortened to ESs), which is a constructor binding $\var$ in $\tm$ (meant to be substituted by $\tmtwo$)---and a compact notation for $\letin\var\tmtwo\tm$---while we use $\tm\isub\var\tmtwo$ for  meta-level substitution. 
\begin{center}
$\begin{array}{r\colspace\colspace rlll}
\textsc{LSC Terms} & \tm,\tmtwo, \tmthree & \grameq &\var \mid \la\var\tm \mid \tm\tmtwo 
\mid \tm \esub\var\tmtwo 
\end{array}$
\end{center}
The LSC is based on the \emph{rejection} of commuting rules for ESs. In traditional $\l$-calculi with ESs, a term such as $(\tm\tmtwo)\esub\var\tmthree$ would rewrite to $\tm\esub\var\tmthree\tmtwo\esub\var\tmthree$ by commuting the ES and the application. In the LSC it does not: ESs interact directly with variable occurrences, using \emph{contexts} to specify a commutation-free rewriting system. 
Contexts, used pervasively in the study of the LSC, are terms with a \emph{hole} $\ctxhole$ intuitively standing for a removed sub-term. Here, we need (general) contexts $\ctx$ and the special case of substitution contexts $\lctx$ (standing for \emph{L}ist of substitutions).
\begin{center}
$\begin{array}{r\colspace\colspace rlll}
\textsc{Contexts} & \ctx & \grameq &\ctxhole \mid \ctx \tm \mid \tm \ctx \mid \la{\var}{\ctx} \mid %
\ctx \esub\var\tm \mid \tm \esub \var \ctx \\
\textsc{Substitution contexts} &\lctx & \grameq &\ctxhole \mid \lctx 
\esub\var\tm
\end{array}$
\end{center}
Replacing the hole of a context $\ctx$ with a term $\tm$ (or another context $\ctxtwo$) is called \emph{plugging} and noted $\ctxp\tm$ (resp. $\ctxp\ctxtwo$).
In general, plugging does capture free variables, that is, $(\la\var\ctxhole\var)\ctxholep{\var} = \la\var\var\var$. We use $\ctxfp\tm$ for a plugging with on-the-fly renaming of bound variables of $\ctx$ to avoid capture in $\tm$. For instance, $(\la\var\ctxhole\var)\ctxholefp{\var} = \la\vartwo\var\vartwo$.

The LSC has three rewriting rules. As usual, the root rules are extended to be applied everywhere in a term via a context closure. An unusual aspect is that two of the rules also use contexts in the definition of the root rules. 
\begin{center}
$\begin{array}{rrll}
\multicolumn{4}{c}{\textsc{LSC Root rewriting rules}}\\
    \textsc{Beta at a distance} & \lctxp{\la\var\tm}\tmtwo &  \rtodb  & \lctxp{\tm\esub{\var}{\tmtwo}} 
    \\
    \textsc{Linear Substitution}  & \ctxfp\var\esub\var{\tmtwo} &  \rtols  & \ctxfp\tmtwo\esub\var\tmtwo
    \\
    \textsc{Garbage Collection}  & \tm\esub\var\tmtwo &  \rtogc  & \tm \ \ \ \mbox{if }\var\notin\fv\tm

\end{array}$\smallskip
		
		\begin{tabular}{ccc}
\textsc{Contextual closure}:
			&
			\multirow{2}{*}{
			\AxiomC{$\tm \rootRew{a} \tm'$}
	\UnaryInfC{$  \ctxp{\tm} \Rew{a} \ctxp{\tm'}$}	
	\DisplayProof}
		\\
		($a \in \set{\db,\ls,\gc}$)
	\end{tabular}
\end{center}
For $\tols$, the notation means that given $\tm\esub\var\tmtwo$ with $\var\in\fv\tm$ we consider a decomposition $\tm=\ctxfp\var$ isolating a free occurrence of $\var$ (one among possibly many). The rule is inspired by replication in the $\pi$-calculus. Examples of steps: 
\begin{center}$\begin{array}{rcl}
(\la\var\var\vartwo)\esub\vartwo\tm \tmtwo \tmthree 
& \todb & (\var\vartwo)\esub\var\tmtwo \esub\vartwo\tm \tmthree
\\
\la\var((\la\vartwo\var\var)\esub\varthree\tm)
&\togc&
 \la\var\la\vartwo\var\var
 \\
 (\la\vartwo\var\var)\esub\var\tm \esub\varthree\tmtwo
 &\tols& 
 (\la\vartwo\var\tm)\esub\var\tm \esub\varthree\tmtwo
 \\
 (\la\vartwo\var\var)\esub\var\tm \esub\varthree\tmtwo
 &\tols& (\la\vartwo\tm\var)\esub\var\tm \esub\varthree\tmtwo
 \end{array}$\end{center}
 The LSC is a very well behaved rewriting system (see \cite{DBLP:conf/rta/Accattoli12,DBLP:conf/popl/AccattoliBKL14,DBLP:conf/rta/BarenbaumB17}), conservatively extending the $\l$-calculus, and improving it in various respects.
  
\paragraph{Micro-Step Simulates Small-Step} Let $\sizep\tm\var$ the number of free occurrences of $\var$ in $\tm$. The LSC simulates meta-level substitution as follows: 
\begin{center}
$\begin{array}{cccccccc}
\tm\esub\var\tmtwo &\tols^{\sizep\tm\var}& \tm\isub\var\tmtwo\esub\var\tmtwo &\togc& \tm\isub\var\tmtwo
\end{array}$
\end{center}
that is, by first replacing one by one the occurrences of $\var$ using $\tols$, and then collecting $\esub\var\tmtwo$ when none are left. Therefore, the LSC simulates $\beta$ reduction, as $(\la\var\tm)\tmtwo \todb \tm\esub\var\tmtwo \tols^{*}\togc \tm\isub\var\tmtwo$.

 \paragraph{Strong Normalization} We define strong normalization as it is relevant for the next property, and shall be used also later on. Given a reduction relation $\to$, the predicate \emph{$\tm$ is $\snsubst$}, also noted $\tm \in \snsubst$, is defined \emph{inductively} as follows:
\begin{itemize}
	\item $\tm$ is $\snsubst$ if $\tm$ is $\to$-normal, \ie, has no $\to$ redexes.
	\item $\tm$ is $\snsubst$ if $\tmtwo$ is $\snsubst$ for all $\tmtwo$ such that $\tm \to \tmtwo$.
\end{itemize}
We use $\snsubst$ for the set of terms that are $\sn{\to}$, and say that $\to$ is strongly normalizing (or $\sn{}$) if $\tm\in\sn\to$ for all terms $\tm$.
 
 \paragraph{Local Termination} While terms in the LSC can be divergent, because the LSC simulates the $\l$-calculus, an important property is that the rewriting rules of the LSC are strongly normalizing when considered separately.
 
 \begin{prop}[Local termination, \cite{DBLP:conf/rta/Accattoli12}]
 Let $\to \in \set{\todb, \tols, \togc}$. Then $\to$ is strongly normalizing.
 \end{prop}
 The only non-trivial point is the strong normalization of $\tols$, which requires a measure. It easily extends to the one of $\tols\cup\togc$, that is, of the sub-system rewriting ESs. The ESC shall preserve this key property, as we shall prove in \refsect{local-termination}.

\paragraph{Structural Equivalence} The LSC is often enriched with the following equivalence.

\begin{defi}[Structural equivalence]
\label{def:streq} Define head contexts as follows:
\begin{center}
$\begin{array}{r\colspace\colspace rlll}
\textsc{Head contexts} & \hctx & \grameq &\ctxhole \mid \hctx \tm  \mid \la{\var}{\hctx} \mid \hctx \esub\var\tmtwo
\end{array}$
\end{center}
If $\var\notin\fv\hctx$ and $\hctx$ does not capture variables in $\fv\tmtwo$, set:
\begin{center}
$\begin{array}{cccccccc}
\hctxp\tm\esub\var\tmtwo & \sim  & \hctxp{\tm\esub\var\tmtwo}
\end{array}$
\end{center}
\emph{Structural equivalence} $\eqstruct$ is the closure of $\sim$ by reflexivity, symmetry, transitivity, and contexts $\ctx$.
\end{defi}

Its key property is that it commutes with evaluation in the following strong sense.

\begin{prop}[$\eqstruct$ is a strong bisimulation wrt $\tolsc$ \cite{DBLP:conf/popl/AccattoliBKL14}]
	\label{prop:bisimulation}
	Let $a\in\set{\db,\ls,\gc}$. If $\tm \eqstruct\tmtwo \Rew{a} \tmthree$ then exists $\tmfour$ such that $\tm \Rew{a}\tmfour \eqstruct\tmthree$.
\end{prop}

Essentially, $\equiv$ never creates redexes, it can be postponed, and vanishes on normal forms (that have no ESs). Accattoli shows that $\equiv$ is exactly the quotient induced by the (call-by-name) translation of $\l$-terms with ESs to proof nets \cite{DBLP:conf/ictac/Accattoli18}.

\paragraph{Linear Head Reduction} One of the key properties of the LSC is that it admits a neat definition of \emph{linear head reduction}, a notion first studied by Danos and Regnier \cite{Danos04headlinear} and Mascari and Pedicini \cite{DBLP:journals/tcs/MascariP94}. Before the introduction of the LSC, the presentations of linear head reduction (in \cite{Danos04headlinear,DBLP:journals/tcs/MascariP94}) were very technical and hard to manage. To define it in the LSC, we need the head contexts of \refdef{streq}. They are used twice. Firstly, head contexts are used to define the root linear head substitution rule:
\begin{center}
$\begin{array}{r\colspace\colspace rllll}
    \textsc{Linear head substitution}  & \hctx\ctxholefp\var\esub\var{\tmtwo} &  \rtolhs  & \hctx\ctxholefp\tmtwo\esub\var\tmtwo
   \end{array}$
\end{center}
Secondly, the root rules $\rtodb$, $\rtolhs$, and $\rtogc$ are closed by \emph{head} contexts:
\begin{center}
	\begin{tabular}{c}
	\begin{tabular}{cccc}
		\multicolumn{3}{c}{\textsc{Contextual closures}}
		\\
		\AxiomC{$\tm \rootRew{\db} \tm'$}
		\UnaryInfC{$  \hctxp{\tm} \Rew{\symfont{lh}\db} \hctxp{\tm'}$}	
		\DisplayProof
		&
		\AxiomC{$\tm \rootRew{\lhs} \tm'$}
		\UnaryInfC{$  \hctxp{\tm} \tolhs \hctxp{\tm'}$}	
		\DisplayProof
		&
		\AxiomC{$\tm \rootRew{\gc} \tm'$}
		\UnaryInfC{$  \hctxp{\tm} \Rew{\symfont{lh}\gc} \hctxp{\tm'}$}	
		\DisplayProof
	\end{tabular}
	\\[6pt]
	$\begin{array}{c\colspace ccc}
\textsc{Linear head reduction} & \tolh & \defeq &\Rew{\symfont{lh}\db} \cup \tolhs \cup \Rew{\symfont{lh}\gc}
	\end{array}$
	\end{tabular}
\end{center}
Linear head reduction has many interesting property. For the present paper, the most relevant one is the \emph{sub-term property}, which, for the given presentation, is proved by Accattoli and Dal Lago in \cite{DBLP:conf/rta/AccattoliL12} (it was already used by Danos and Regnier and Mascari and Pedicini, but not exploited for complexity analyses).

\begin{thm}[Sub-term property for $\tolh$, \cite{DBLP:conf/rta/AccattoliL12}]
All the sub-terms duplicated by $\tolhs$ along a linear head evaluation sequence $\tm \tolh^* \tmtwo$ are sub-terms of $\tm$ (up to variables renaming).
\end{thm}

The good strategy for IMELL that we shall define in \refsect{strategy} generalizes linear head reduction to the whole of IMELL and preserves its sub-term property. For that, note that $\tolh$ never evaluates inside ESs. Similarly, the good strategy shall not evaluate inside cuts. Note also that $\tolh$ does not compute normal forms, as it never evaluates arguments. Since we want the good strategy to compute cut-free proofs, we shall have to evaluate the analogous of arguments in our framework (namely, the left premises of subtractions). This shall be obtained by entering into arguments only under some conditions, similarly to how $\tolh$ is generalized to enter inside arguments by Accattoli and co-authors \cite{DBLP:journals/corr/AccattoliL16,DBLP:conf/aplas/AccattoliBM15}, obtaining \emph{linear leftmost-outermost reduction}.

%% file: 05-Towards_the_IMELL_calculus.tex
% !TEX root = main.tex
\section{Towards Exponentials as Substitutions}
\label{sect:towards}
This section is an informal and hopefully intuitive introduction to our cut elimination for IMELL. The logical rules we refer to are standard (but decorated with terms), see \reffig{typing}. 

\paragraph{Splitting Terms} For the moment, we prefer to avoid giving grammars, but we want nonetheless to fix some terminologies and intuitions. We distinguish variables into multiplicatives, noted $\mvar, \mvartwo, \mvarthree$, and exponentials, noted $\evar, \evartwo, \evarthree$, and use $\var,\vartwo,\varthree$ for variables of unspecified kind. Proof terms for axioms and right introduction rules shall be called \emph{values}, and noted $\val$. We borrow from the LSC the notation $\lctx$ (here standing for \emph{L}eft context) and generalize it to a context corresponding to (the term annotations induced by) a sequence of left rules of the sequent calculus. Exactly as IMELL proofs can be seen as ending on a sequence of left rules following a right rule or an axiom, every proof term $\tm$ shall uniquely (on-the-fly) decompose as $\lctxp\val$, called its \emph{splitting}.

\paragraph{Meta-Level Substitution} For the design of our system, the first step is understanding the subtle interplay between linearity and substitution. Consider a cut on an exponential formula, and a first attempt at decorating it with terms:
\begin{center}
\AxiomC{$\multiForm   \vdash \tmtwo \hastype \bang\form$}
	\AxiomC{$\multiFormtwo, \evar\hastype \bang\form \vdash \tm\hastype\formtwo$}
 	\RightLabel{$\cut$}
 	\BinaryInfC{$  \multiForm, \multiFormtwo\vdash \cuta\tmtwo\evar\tm\hastype\formtwo$} 	
	\DisplayProof	
\end{center}
Note that the notation $\cuta\tmtwo\evar\tm$ for cuts is similar but opposite to the one for ESs in the LSC, as the arrow goes from left to right and it binds on the right (in $\tm$), to better reflect the structure of sequent proofs. We seek a notion of meta-level substitution $\cutsub\tmtwo\evar\tm$ and a small-step rule such as:
\begin{equation}
\label{eq:naive-sub}
\begin{array}{cccccccc}
\cuta\tmtwo\evar\tm & \to& \cutsub\tmtwo\evar\tm
\end{array}
\end{equation}
Perhaps surprisingly, linear logic does not validate this rule. Despite the $\bang\form$ type, not all of $\tmtwo$, indeed, might be allowed to be duplicated. Let's refine the example and consider the case in which the last rule contributing to $\tmtwo$ is, say, the left rule for $\tens$, what we shall call a \emph{par} and annotate with $\para\mvar\var\vartwo$:
\begin{center}
\AxiomC{$\multiForm, \var\hastype\formtwo, \vartwo\hastype\formthree   \vdash \tmtwo' \hastype \bang\form$}
	\RightLabel{$ \tensLeftRule $}
\UnaryInfC{$\multiForm, \mvar\hastype\formtwo\tens\formthree   \vdash \para\mvar\var\vartwo\tmtwo' \hastype \bang\form$}
	\AxiomC{$\multiFormtwo, \evar\hastype \bang\form \vdash \tm\hastype\formtwo$}
 	\RightLabel{$\cut$}
 	\BinaryInfC{$  \multiForm, \mvar\hastype\formtwo\tens\formthree,\multiFormtwo\vdash \cuta{\para\mvar\var\vartwo\tmtwo'}\evar\tm\hastype\formtwo$} 	
	\DisplayProof	
\end{center}
Now we can see why \refeq{naive-sub} is not valid: the annotation $\para\mvar\var\vartwo$ being linear, it should not be duplicated. The proof nets approach to this issue is marking with a \emph{box} the sub-term/proof net of $\tmtwo'$ that can be duplicated. The usual sequent calculus approach, instead, is commuting the cut with $ \tensLeftRule $, which on proof terms corresponds to have a rule such as:
\begin{center}$
\begin{array}{ccccc}
\cuta{\para\mvar\var\vartwo\tmtwo'}\evar\tm &\to& \para\mvar\var\vartwo\cuta{\tmtwo'}\evar\tm
\end{array}
$\end{center}
and then keep commuting the cut with left rules until an introduction of a $\bang$ connective on the right---what we call \emph{an exponential value}, noted $\exval$---is reached, which can then trigger the substitution. But these commutative rules are a burden, and exactly what the LSC philosophy rejects.

The idea is to exploit the splitting of terms into values and left contexts, so as to import on terms the proof nets approach. A term $\tmtwo$ of exponential type, in fact, shall always be splittable as $\lctxp\exval$, that is, as a left context and an exponential value. We then refine \refeq{naive-sub} as follows:
\begin{center}
%\label{eq:split-sub}
$\begin{array}{ccccc}
\cuta\tmtwo\evar\tm &=& \cuta{\lctxp\exval}\evar\tm& \to& \lctxp{\cutsub\exval\evar\tm}
\end{array}
$\end{center}
Meta-level substitution thus concerns only exponential values $\exval$---playing the role of proof nets boxes---and \emph{not}  the surrounding left context $\lctx$. This same mechanism is used in Accattoli and Paolini's \emph{value substitution calculus} \cite{DBLP:conf/flops/AccattoliP12}.

\paragraph{Split Cuts} In fact, we adopt a further refinement. We restrict cuts to have only values as left sub-terms, that is, of being of shape $\cuta\val\var\tm$. Since we want to stick to the standard sequent calculus for IMELL, we do not restrict the typing rule, we do instead modify its decoration:
\begin{center}
\AxiomC{$\multiForm   \vdash \tmtwo =\lctxp\val \hastype \bang\form$}
	\AxiomC{$\multiFormtwo, \evar\hastype \bang\form \vdash \tm\hastype\formtwo$}
 	\RightLabel{$\cut$}
 	\BinaryInfC{$  \multiForm, \multiFormtwo\vdash \lctxp{\cuta\val\evar\tm} \hastype\formtwo$} 	
	\DisplayProof	
\end{center}
And we shall do the same for the left rule $\lolliLeftRule$ for $\lolli$ (see \reffig{typing}). This refinement simplifies various technical points, in particular it enables a simple definition of the good strategy.

\paragraph{Micro-Step Substitution} For micro-step substitution we simply borrow the LSC approach, adopting rules such as:
\begin{center}
$\begin{array}{ccccc}
\cuta\exval\evar\ctxfp\evar & \to&  \cuta\exval\evar\ctxfp\exval
\end{array}$
\end{center}
While in line with the LSC, this is \emph{not} how cut elimination is usually performed in both proof nets and the sequent calculus. There are three main differences:
\begin{enumerate}
\item \emph{LSC-style duplications}: in our case duplication happens also when there are no contractions involved, for instance we shall have $\cuta\exval\evar\evar \to \cuta\exval\evar\exval$, and there shall be a rule ($\tobder$) duplicating and interacting with derelictions at the same time (see the next section). Not only this approach is perfectly sound, it actually has better rewriting properties than the traditional one. 
%%%%
\item \emph{No commutations}: there are no commutations with $\bang$. Even in proof nets, where most commutative cuts vanish, micro-step cut elimination usually relies on a commutative rule for  $\bang$-boxes (namely $\Rew{\boxbox}$ in \reffig{classical-diverging}), that on terms looks as follows (when $\evar\in\fv\tm$):
\begin{center}
$\begin{array}{cccccccc}
\cuta\exval\evar\bang\tm &\to& \bang\cuta\exval\evar\tm
\end{array}$
\end{center}
The LSC philosophy rejects such a commutative rule. %Even without LSC-style duplication,  rejecting this rule leads to simpler proofs of strong normalization, see Accattoli \cite{DBLP:conf/rta/Accattoli13}.
%%%%
\item \emph{Axioms}: in the proof nets literature, cut elimination of exponential axioms is identical to the one for multiplicative axioms (there is in fact no distinction between the two kinds of axioms) and simply amounts to the removal of the cut and the axiom, in both the small-step (\emph{\`a la} Regnier) and the micro-step (\emph{\`a la} Girard) approaches. In (asymmetric) intuitionistic settings, axioms have two rewriting rules, depending on whether they are substituted, or substituted upon. In our syntax, the traditional approach to cut elimination for exponential axioms would take the following form:
\begin{center}
$\begin{array}{cccccccc}
\cuta\evartwo\evar\tm &  \to & \cutsub\evartwo\evar\tm
\end{array}$
\end{center}
When written with terms, one immediately notices that such a rule is small-step. At the end of  \refsect{confluence}, we shall show that such a small-step rule generates a closable but unpleasant local confluence diagram with the other micro-step exponential rules, which is avoided by adopting a LSC-duplication style also for exponential axioms/variables.
\end{enumerate}

%% file: 06-The_IMELL_substitution_calculus.tex
% !TEX root = main.tex
\section{The Exponential Substitution Calculus}
\label{sect:calculus}
 \begin{figure*}[t]
	\input{figure-IMELL-SC-grammars}
	\caption{Grammars of the exponential substitution calculus.}
\label{fig:grammars}
\end{figure*}
\paragraph{Values and Terms} The grammars of the ESC are in \reffig{grammars}. First of all, a disclaimer. The constructors of the calculus are  an untyped \emph{minimalist encoding} of the sequent calculus rules, and \emph{not} an intuitively readable calculus. For that aim, one would rather use $\symfont{let\mbox{-}in}$ notations. Variables are of two disjoint kinds, multiplicative and exponential, and we often refer to variables of unspecified kind using the notations $\var,\vartwo,\varthree$. Values are the proof terms associated to axioms or to the right rules and, beyond variables, are \emph{abstractions} $\la\var\tm$, \emph{tensor pairs} $\pair\tm\tmtwo$, and \emph{promotions} $\bang\tm$. The proof terms decorating left rules are \emph{pars} $\para\mvar\var\vartwo\tm$, \emph{subtractions} $\suba\mvar\val\var\tm$, \emph{derelictions} $\dera\evar\var\tm$,  and \emph{cuts} $\cuta\val\var\tm$, which is in red because of its special role.  Note that cuts and subtractions are \emph{split}, that is, have values (rather than terms) as left sub-terms. Back to readability, subtraction might also be decorated with $\letin \var {\mvar\val}\tm$, but note that this notation is less faithful to the logical rule because the two premises are represented asymmetrically, and that it uses an application constructor that rather belongs to natural deduction.

In $\la\var\tm$, $\suba\mvar\val\var\tm$, $\dera\evar\var\tm$, and $\cuta\val\var\tm$ we have that $\var$ is bound in $\tm$, and in $\para\mvar\var\vartwo\tm$ both $\var$ and $\vartwo$ are bound in $\tm$. We identify terms up to $\alpha$-renaming. Free variables (resp.  multiplicative/exponential variables) are defined as expected, and noted $\fv\tm$ (resp. $\mfv\tm$ and $\efv\tm$). We use $\sizep\tm\var$ for the number of free occurrences of $\var$ in $\tm$.
There is a notion of \emph{proper term} ensuring the linearity of multiplicative variables. The only relevant case is: \emph{$\bang\tm$ is \proper if $\tm$ is \proper and $\mfv\tm = \emptyset$}. 

\begin{defi}[Proper terms]
\label{def:proper-terms}
Proper terms are defined by induction on $\tm$ as follows:
\begin{itemize}
\item \emph{Variables}: $\mvar$ and $\evar$ are \proper.
\item \emph{Tensor}: $\pair\tm\tmtwo$ is \proper if $\tm$ is \proper, $\tmtwo$ is \proper, and $\mfv\tm \cap \mfv\tmtwo = \emptyset$
\item \emph{Par}: $\para\mvar\var\vartwo \tm$ is \proper if $\tm$ is \proper, $\mvar \notin(\mfv\tm\setminus{\var,\vartwo})$, $\var\neq \vartwo$, and if $\var$ (resp. $\vartwo$) is a multiplicative variable then $\var \in\mfv\tm$ (resp. $\vartwo \in\mfv\tm$).
\item \emph{Implication}: $\la\var\tm$ is \proper if $\tm$ is \proper and if $\var$ is a multiplicative variable then $\var \in\mfv\tm$.
\item \emph{Subtraction}: $\suba\mvar\val\var \tm$ is \proper if $\tm$ is \proper, $\val$ is \proper, $\mfv\val \cap (\mfv\tm\setminus\set{\var}) = \emptyset$, $\mvar \notin\mfv\tm$, and if $\var$ is a multiplicative variable then $\var \in\mfv\tm$.
\item \emph{Bang}: $\bang\tm$ is \proper if $\tm$ is \proper and $\mfv\tm = \emptyset$.
\item \emph{Dereliction}: $\dera\evar\var \tm$ is \proper if $\tm$ is \proper and if $\var$ is a multiplicative variable then $\var \in\mfv\tm$.
\item \emph{Cut}: $\cuta\val\var\tm$ is \proper if $\tm$ is \proper, $\val$ is \proper, $\mfv\val \cap (\mfv\tm\setminus\set\var) = \emptyset$, and if $\var$ is a multiplicative variable then $\var \in\mfv\tm$.
\end{itemize}
\end{defi}

 \begin{figure*}[t]
	\input{figure-IMELL-SC-typing}
	\caption{IMELL typing rules for the ESC, where $\multiForm\#\multiFormtwo$ is a shortcut for $\dom\multiForm\cap\dom\multiFormtwo = \emptyset$.}
\label{fig:typing}
\end{figure*}
	
\paragraph{Typing System} The typing rules are in \reffig{typing}. The formulas of IMELL and the deductive rules of the sequent calculus are standard, but for the decoration with proof terms and the side conditions about variable names of the form $\multiForm\#\multiFormtwo$, which is a shortcut for $\dom\multiForm\cap\dom\multiFormtwo = \emptyset$. Linear implication $\lolli$ is also referred to as \emph{lolli}. The only atomic formula that we consider, $\aform$, is multiplicative. There is no multiplicative unit because in presence of the exponentials $1$ can be simulated by $!(\aform \lolli\aform)$. We distinguish between multiplicative and exponential axioms, in order to decorate them with the corresponding kind of variable.

Note that the weakening and contraction rules do not add constructors to terms. This is crucial in order to keep the calculus manageable. Note also that the decorations of the cut and $\lolliLeftRule$ rules are \emph{split}, as explained in the previous section. Clearly, typed terms are proper.  We use $\tderiv \pof \multiForm \vdash \tm\hastype \form$ for a proof/typing derivation $\tderiv$ ending in that sequent.

\paragraph{Contexts and Plugging} The broadest notion of context that we consider is \emph{general contexts} $\ctx$, which simply allow the hole $\ctxhole$ to replace any sub-term in a term. Because of split cuts and subtractions, the definition relies on the auxiliary notion of \emph{value context} $\vctx$. The definition also uses \emph{left contexts} $\lctx$, which are contexts under left constructors (or, for binary left constructors, under the right sub-term) that play a key role in the system---their use in defining $\ctx$ is just to keep the grammar compact (to fit it in the figure). 

A fact used pervasively is that every term $\tm$ writes, or \emph{splits} uniquely as $\tm = \lctxp\val$. For instance $\dera\evar\mvar\cuta\val\mvartwo\pair\mvar\mvartwo$ splits as $\lctx = \dera\evar\mvar\cuta\val\mvartwo\ctxhole$ and $\val =\pair\mvar\mvartwo$. We also have \emph{multiplicative contexts} $\mctx$, that the proof nets literature would call \emph{level 0} contexts, that is, contexts with the hole out of all $\bang$, and which also rely on their own value contexts $\vmctx$.

Because of split cuts and subtractions, plugging is slightly tricky, as it has to preserve the split shape.  For instance, we have $(\cuta\ctxhole\mvarthree\tm)\ctxholep{\dera\evar\mvar\pair\mvar\mvartwo} = \dera\evar\mvar\cuta{\pair\mvar\mvartwo}\mvarthree\tm$.
 A similar approach to plugging is also used by Accattoli et al. in \cite{DBLP:conf/ppdp/AccattoliCGC19}. The full definition follows. It is mostly as expected, the only two subtle cases are for $\cuta\ctxhole\var\tm$ and $\suba\mvar\ctxhole\var$.
 
 \begin{defi}[Plugging]
\label{def:plugging}
Plugging $\ctxp\tm$ is defined as follows:
\begin{center}
$\begin{array}{rll@{\hspace{.4cm}}rlll}
\multicolumn{6}{c}{\textsc{Plugging}}
\\
\ctxhole\ctxholep\tm & \defeq & \tm
&
(\la\var\ctx)\ctxholep\tm& \defeq & \la\var\ctxp\tm
\\
\pair\ctx\tmtwo\ctxholep\tm & \defeq & \pair{\ctxp\tm}\tmtwo
&
\pair\tmtwo\ctx\ctxholep\tm & \defeq & \pair\tmtwo{\ctxp\tm} 
 \\
 (\bang\ctx)\ctxholep\tm& \defeq & \bang\ctxp\tm
 &
 (\dera\evar\var\ctx)\ctxholep\tm & \defeq & \dera\evar\var\ctxp\tm
 \\
(\para\mvar\var\vartwo \ctx)\ctxholep\tm& \defeq & \para\mvar\var\vartwo \ctxp\tm
&
(\suba\mvar\val\var\ctx)\ctxholep\tm& \defeq & \suba\mvar\val\var\ctxp\tm
\\
(\cuta\val\var\ctx)\ctxholep\tm & \defeq &  \cuta\val\var\ctxp\tm
\end{array}$\medskip

$\begin{array}{rlll|l|l|l|l|l|l|l|l|l}
(\suba\mvar\vctx\var\tmtwo)\ctxholep\tm& \defeq & 
\begin{cases}
\lctxp{\suba\mvar\val\var\tmtwo}
 & \mbox{if }\vctx = \ctxhole\mbox{ and }\tm = \lctxp\val
\\
\suba\mvar{\vctxp\tm}\var\tmtwo  & \mbox{otherwise}
\end{cases}
\\
(\cuta\vctx\var \tmtwo)\ctxholep\tm& \defeq & 
\begin{cases}
\lctxp{\cuta\val\var \tmtwo} & \mbox{if }\vctx = \ctxhole\mbox{ and }\tm = \lctxp\val
\\
\cuta{\vctxp\tm}\var \tmtwo & \mbox{otherwise}
\end{cases}
\end{array}$
\end{center}
\end{defi}

 Plugging $\ctxp\ctxtwo$ of contexts is defined similarly, and plugging is extended to all forms of contexts by seeing them as general contexts. As for the LSC, plugging can capture variables and we use $\ctxfp\tm$ when we want to prevent it.

With our notion of plugging, a left constructor in a term might be identified via plugging of a term in a context in various ways, \eg for $\tm \defeq \pair{\dera\evar\mvar\cuta\mvar\mvartwo\mvartwo}{\evartwo}$ we have $\tm = \ctxp{\dera\evar\mvar\mvar}$ with $\ctx \defeq \pair{\cuta\ctxhole\mvartwo\mvartwo}{\evartwo}$, and $\tm = \ctxtwop{\dera\evar\mvar\cuta\mvar\mvartwo\mvartwo}$ with $\ctxtwo \defeq\pair{\ctxhole}{\evartwo}$. For most uses, in particular for defining the rewriting rules, this is harmless. It shall play a role, however, for defining the good strategy, because it shall be important whether the hole of the context is inside or outside a cut. To this purpose, note that in the first decomposition the term $\dera\evar\mvar\mvar$ is not a sub-term of $\tm$. This suggest the following notion. 
\begin{defi}[Positions]
\label{def:position}
A \emph{position} in a term $\tm$ is a decomposition $\tm=\ctxp\tmtwo$ such that $\tmtwo$ is a sub-term of $\tm$ (as for $\ctxtwop{\dera\evar\mvar\cuta\mvar\mvartwo\mvartwo}$ above).
\end{defi}

We can now state the property ensured by properness. 
\begin{toappendix}
\begin{lem}[Structural linearity]
\label{l:proper-linearity}
Let $\tm$ be a \proper term and $\mvar\in\mfv\tm$. Then $\varmeas\tm\mvar = 1$ and $\tm = \mctxfp\mvar$ for some $\mctx$.
\end{lem}
\end{toappendix}

\begin{proof}
By induction on $\tm$.
\end{proof}

 \begin{figure*}[t]
 \begin{center}
	\begin{tabular}{cc}
	\input{figure-IMELL-SC-rewriting_rules}
	\end{tabular}
\end{center}
%\vspace{-8pt}
\caption{Rewriting rules of the ESC.}
\label{fig:rewriting-rules}
\end{figure*}
\paragraph{Multiplicative Cut Elimination Rules} The rewriting rules are in \reffig{rewriting-rules}. The ESC has four multiplicative rules, in particular two for axioms, depending on whether they are acted upon ($\toaxmone$) or used to rename another multiplicative (thus linear) variable ($\toaxmtwo$). Rule $\toaxmone$ is expressed generically for multiplicative values $\mval$ (that is, multiplicative variables $\mvar$, abstractions $\la\mvar\tm$, and tensor pairs $\pair\tm\tmtwo$). For $\toaxmtwo$, we abuse notations and use a form of meta-level substitution $\cutsub\mvartwo\mvar\tm$ for the renaming operation. There is also a slight superposition between $\toaxmone$ and $\toaxmtwo$, for instance in $\cuta\mvar\mvartwo \la\evar\mvartwo \to \la\evar\mvar$, which can be both kinds of steps. To disambiguate, one should refine $\toaxmtwo$ into two rules:
\begin{center}$\begin{array}{rlll}
\cuta{\mvartwo}\mvar \mctxp{\para\mvar\var\vartwo \tm}
& \Rew{\axmtwo'} & 
\mctxp{\para\mvartwo\var\vartwo \tm}
\\[3pt]
\cuta{\mvartwo}\mvar \mctxp{\suba\mvar\val\var  \tm}
& \Rew{\axmtwo''} & 
\mctxp{\suba\mvartwo\val\var \tm}
\end{array}$\end{center}
But the ambiguity is harmless and the notation is convenient. Note that in $\toaxmone$, $\totens$, and $\tololli$ (and $\Rew{\axmtwo'/\axmtwo''}$), the cut acts on a sub-term inside a multiplicative context $\mctx$, that is, out of $!$. Moreover, it is silently assumed that $\mctx$ does not capture $\mvar$ in $\totens$ and $\tololli$ (and $\Rew{\axmtwo'/\axmtwo''}$), but it might capture other variables in $\tm$ and $\val$.

In both $\totens$ and $\tololli$ the rule has to respect split cuts, which is why, for writing the reduct, the sub-terms $\tmtwo$ and $\tmthree$ get split on-the-fly. We suggest to have a look at the proof of subject reduction on page \pageref{thm:sub-red}, which shows how some of the rewriting steps act on the decorated proof.
An example of $\tololli$ step follows:
\begin{center}$\begin{array}{rlll}
\cuta{\la\evar\dera\evar\mvar\mvar}\mvartwo\suba\mvartwo{\bang\evartwo}\mvarthree\mvarthree 
& \tololli &
\cuta{\bang\evartwo}\evar\dera\evar\mvar\cuta\mvar\mvarthree \mvarthree.
\end{array}$\end{center}

\paragraph{Micro-Step Exponential Rules} There are also four micro-step exponential rules, with again two rules for axioms. Replacement of variables ($\toaxeone$) and erasure ($\tobw$) are expressed generically for exponential values $\exval$ (that is, exponential variables $\evartwo$ and promotions $\bang\tm$), interaction with derelictions (in $\toaxetwo$ and $\tobder$) instead requires inspecting $\exval$. 

Rule $\tobder$ removes the dereliction, copies the promotion body, and puts it in a cut---in proof nets jargon, \emph{it opens the box}. To preserve the split shape, the body of the promotion is split and only the value is cut. An example: 
\begin{center}$\begin{array}{rlll}
\cuta{\bang\cuta\evartwo\evarthree \evarthree}\evar \la\mvar\dera\evar{\evar'} \pair{\evar'}\mvar 
& \tobder &  
\cuta{\bang\cuta\evartwo\evarthree \evarthree}\evar\la\mvar\cuta\evartwo\evarthree \cuta{\evarthree}{\evar'} \pair{\evar'}{\mvar}.
\end{array}$\end{center}
 Note that $\tobder$ entangles \emph{interaction with a dereliction} and \emph{duplication}, which is not what proof nets usually do, as mentioned in \refsect{towards}.  It is silently assumed that $\ctx$ does not capture $\evar$ in $\toaxetwo$ and $\tobder$ but it might capture other variables in $\tm$.

\paragraph{Meta-Level Substitution and the Small-Step Exponential Rule} For defining the small-step exponential cut elimination we need meta-level substitution $\cutsub\exval\evar\tm$, which is defined only for exponential values. The long definition is mostly as expected, but for the case of derelictions, explained after the definition. 
\begin{defi}[Meta-level exponential substitution]
\label{def:meta-sub}
The meta-level (exponential) substitution $\cutsub\exval\evar\tm$ of the exponential value $\exval$ for the free occurrences of $\evar$ in $\tm$ is defined by induction on $\tm$ as follows (assuming on-the-fly $\alpha$-renaming to avoid capture of free variables, omitted for ease of reading and manipulation):
\begin{center}
\begin{tabular}{c@{\hspace{.1cm}}|@{\hspace{.1cm}}c}
\multicolumn{2}{c}{\textsc{Meta-level substitution}}
\\[5pt]
$\begin{array}{r@{\hspace{.2cm}}l@{\hspace{.2cm}}ll}
\cutsub{\exval}\evar \mvar & \defeq & \mvar
\\
 \cutsub\exval\evar \evartwo & \defeq & \evartwo 
\\
\cutsub{\exval}\evar \evar & \defeq & \exval 
\\
\cutsub{\exval}\evar \pair\tm\tmthree & \defeq & \pair{\cutsub{\exval}\evar\tm}{\cutsub{\exval}\evar\tmthree}
\\
\cutsub{\exval}\evar \la\var\tm & \defeq & \la\var \cutsub{\exval}\evar \tm
\\
\cutsub{\exval}\evar \bang\tm& \defeq & \bang\cutsub{\exval}\evar\tm
\end{array}$

&

$\begin{array}{r@{\hspace{.2cm}}l@{\hspace{.2cm}}ll}
\cutsub{\exval}\evar\para\mvar\var\vartwo \tm & \defeq & \para\mvar\var\vartwo  \cutsub{\exval}\evar \tm
\\
\cutsub{\exval}\evar \suba\mvar\val\var \tm & \defeq & \suba\mvar {\cutsub{\exval}\evar\val}\var \cutsub{\exval}\evar\tm
\\
\cutsub{\exval}\evar \dera\evartwo\var \tm & \defeq & \dera\evartwo\var \cutsub{\exval}\evar \tm 
\\
\cutsub{\evartwo}\evar \dera\evar\var \tm & \defeq & \dera\evartwo\var \cutsub{\evartwo}\evar\tm 
\\
\cutsub{\bang\lctxp\val}\evar \dera\evar\var \tm & \defeq & \lctxp{\cuta\val\var \cutsub{\bang\lctxp\val}\evar\tm} 
\\
\cutsub{\exval}\evar \cuta\val\var\tm & \defeq & \cuta{\cutsub{\exval}\evar\val}\var \cutsub{\exval}\evar\tm
\end{array}$
\end{tabular}
\end{center}
\end{defi}
The key cases are the substitution of a value for a dereliction (the second and third to last in the right column), where the definition depends on the shape of the value. For promotions, the definition mimics the micro-step rule $\tobder$, making a copy of the content of the box and splitting it on-the-fly. An example: 
\begin{center}$\begin{array}{rlll}
\cutsub{\bang\cuta\evartwo\evarthree \evarthree}\evar \dera\evar{\evar'} \pair{\evar'}\evar 
& = & 
\cuta\evartwo\evarthree \cuta{\evarthree}{\evar'} \pair{\evar'}{\bang\cuta\evartwo\evarthree \evarthree}.\end{array}$\end{center}

Our definition verifies the expected properties of meta-level substitution. Some are given by the next two lemmas.
\begin{lem}[Basic properties of meta-level substitution]
\label{l:bang-subs-prop} % \reflemmap{bang-subs-prop}{zero}
\hfill
\begin{enumerate}
\item If $\exval$ and $\tm$ are proper then $\cutsub\exval\evar\tm$ is proper.
\item \label{p:bang-subs-prop-values} $\cutsub\exval\evar\exvaltwo$ is an exponential value and $\cutsub\exval\evar\mval$ is a multiplicatie value.

\item \label{p:bang-subs-prop-zero} If $\ctx$ does not capture free variables of $\exval$ then $\cutsub\exval\evar \ctxp\tm = (\cutsub\exval\evar \ctx)\ctxholep{\cutsub\exval\evar \tm}$. 

\item \label{p:bang-subs-prop-three} $\cutsub\exval\evar \tmtwo = \tmtwo$ if $\evar\notin\fv\tmtwo$. 

\item \label{p:bang-subs-prop-two} $\cutsub{\exval}\evartwo \cutsub{\exvaltwo}\evar \tm =  \cutsub{\cutsub{\exval}\evartwo\exvaltwo}\evar \cutsub{\exval}\evartwo\tm$.

\item \label{p:bang-subs-prop-four} $\cutsub\exval\evar \cutsub\exvaltwo\evartwo \tmthree = \cutsub\exvaltwo\evartwo \cutsub\exval\evar \tmthree$ if $\evar\notin\fv\tmtwo$ and $\evar\notin\fv\tmthree$.

\item \label{p:bang-subs-prop-four-mult} $\cutsub\exval\evar \cutsub\mvartwo\mvar \tmthree = \cutsub\mvartwo\mvar \cutsub\exval\evar \tmthree$.

\item \label{p:bang-subs-prop-one} $\cutsub{\bang\tmthree}\evar \cutsub{\bang\tmthree}\evartwo \tm = \cutsub{\bang\tmthree}\evar\cutsub\evar\evartwo\tm$.

\end{enumerate}
\end{lem}

\begin{lem}[Stability of steps under substitution]
\label{l:substitutivity-of-red}
\hfill
\begin{enumerate}
\item \label{p:substitutivity-of-red-four}
If $\tm \Rew{a} \tmtwo$ then $\cutsub{\exval}\evar\tm \Rew{a} \cutsub{\exval}\evar\tmtwo$ for $a\in\set{\axmone,\axmtwo,\tens,\lolli,\ess, \ems}$.

\item \label{p:substitutivity-of-red-fourb}
If $\exval \Rew{a} \exvaltwo$ then $\cutsub{\exval}\evar\tm \Rewn{a}\cutsub{\exvaltwo}\evar\tm$ for $a\in\set{\axmone,\axmtwo,\tens,\lolli,\ess, \axeone,\axetwo,\bdersym,\wsym}$.
\end{enumerate}
\end{lem}

\withproofs{\input{./\proofspath/calculus/substitutivity}}

The small-step exponential rule $\toess$ then simply turns a cut $\cuta\exval\evar\tm$ into the corresponding meta-level substitution.

\paragraph{Linear Implication and Cut Creations} Rule $\tololli$ is special in that it is the only rule of the ESC that \emph{creates} cuts. Note indeed that in the root $\rtololli$ rule:
\begin{center}
$\begin{array}{rcl}
\cuta{\la\vartwo\lctxp\val}\mvar \mctxp{\suba\mvar\valtwo\var \tm}
& \rtololli & 
\mctxp{ \cuta\valtwo\vartwo \lctxp{\cuta\val\var \tm }}  
\end{array}$
\end{center} 
the cut $\cuta\valtwo\vartwo$ in the reduct is not a cut of the reducing term. One might argue that also $\rtotens$ and $\rtobang$ do create cuts, because in their root rules, recalled here:
\begin{center}
$\begin{array}{rcl}
\cuta{\pair{\lctxp\val}{\lctxtwop\valtwo}}\mvar \mctxp{\para\mvar\var\vartwo \tm}  
&\rtotens&
\mctxp{\lctxp{\cuta\val\var \lctxtwop{\cuta\valtwo\vartwo \tm}}}
\\[6pt]
\cuta{\bang\lctxp\val}\evar\ctxp{\dera\evar\var \tm}
& \rtobang &
\cuta{\bang\lctxp\val}\evar\ctxp{\lctxp{\cuta\val\var \tm}} 
\end{array}$
\end{center}
the cuts $\cuta\val\var$ and $\cuta\valtwo\vartwo$ of the reducts are not in the reducing terms. The difference is that $\val$ and $\valtwo$ \emph{do occur as sub-terms of cuts} in the reducing terms of $\rtotens$ and $\rtobang$, while in $\rtololli$ the value $\valtwo$ does not. We then say that $\cuta\valtwo\vartwo$ is \emph{created} in $\rtololli$.
This fact shall play a key role for \emph{local termination} (\refthm{local-termination} below) and for the good strategy. Being specific to $\tololli$, cut creations are a feature of the intuitionistic setting. That is, they are \emph{invisible} in classical linear logic.

\paragraph{Variable Occurrences and Redex Positions} For later defining the good strategy, we identify a redex with its position, which is a context, as it is done for the LSC by Accattoli and Dal Lago  \cite{DBLP:journals/corr/AccattoliL16}. Every step $\tm \tomsnw \tmtwo$ reduces a redex of shape $\tm=\ctxp{\cuta\val\var\ctxtwop{\tm_{\var}}}$ where $\tm_{\var}$ is an \emph{occurrence} of $\var$, \ie a sub-term of $\tm$ of shape $\var$, $\para\var\vartwo\varthree\tmtwo$, $\suba\var\val\vartwo\tmtwo$, or $\dera\var\vartwo\tmtwo$. The \emph{redex position of $\tomsnw$ steps} is the context $\ctxp{\cuta\exval\var\ctxtwo}$. The \emph{redex position of $\tobw$ and $\toess$ steps} is the context closing the root step. We write $\ctx:\tm\Rew{\mssym/\sssym}\tmtwo$ for a redex of position $\ctx$ in $\tm$, the reduction of which produces $\tmtwo$.\medskip

We conclude the section with the fact that proper terms are stable by reduction.

\begin{lem}
\label{l:proper-replacement}
Let $\ctxp\tm$ be a \proper term. Then:
\begin{enumerate}
\item $\tm$ is \proper.
\item If $\tmtwo$ is \proper and such that $\mfv\tmtwo = \mfv\tm$ then $\ctxp\tmtwo$ is \proper.
\end{enumerate}
\end{lem}

\begin{proof}
By induction on $\ctx$.
\end{proof}

\begin{lem}
Let $\tm$ be a \proper term and $\tm \toms \tmtwo$. Then $\tmtwo$ is \proper.
\end{lem}

\begin{proof}
Let $\tm = \ctxp\tmthree \toms \ctxp\tmfour = \tmtwo$. By \reflemma{proper-replacement}, $\tmthree$ is \proper, and to prove the statement it is enough to prove that if $\tmthree \toms \tmfour$ is a root step then $\tmfour$ is proper and  $\mfv\tmfour = \mfv\tmthree$. Note that all the  root  rules preserve the set of multiplicative variables, that is, $\mfv\tmfour = \mfv\tmthree$ holds. Proving that the root rules preserve being \proper is straightforward.
\end{proof}
Preservation of properness by $\toess$ follows from \emph{full composition} (\refprop{ms-simulates-ss}),  just below.

%% file: figure-IMELL-SC-grammars.tex
% !TEX root = main.tex
%\fbox{
  \begin{tabular}{c@{\hspace{.1cm}}cc}
% \begin{tabular}{c@{\hspace{.5cm}}c}
$\begin{array}{rrllll}
\textsc{Multiplicative vars} & \mvar,\mvartwo,\mvarthree & \in & \mvars
\\
\textsc{Exponential vars} & \evar,\evartwo,\evarthree & \in & \evars
\\
\textsc{Variables} & \var,\vartwo,\varthree & \in & \vars \defeq \mvars\uplus\evars
\end{array}$
\\[7pt]
$\begin{array}{rrllll}
\textsc{Multiplicative values} & \mval & \grameq & \mvar \in\mvars\mid \pair\tm\tmtwo \mid \la\var\tm
\\
\textsc{Exponential values} & \exval & \grameq & \evar \in\evars \mid  \bang\tm
\\
\textsc{Values} & \val & \grameq & \mval \mid  \exval
\end{array}$
%\end{tabular}
\\[7pt]
 \begin{tabular}{c}
$\begin{array}{rrllll}
\textsc{Terms} & \tm,\tmtwo,\tmthree & \grameq & %\star & 
 \val  \mid
\cuta\val\var\tm \mid \para\mvar\var\vartwo \tm \mid \suba\mvar\val\var \tm \mid \dera\evar\var \tm
\end{array}$
\\[7pt]
$\begin{array}{rrllll}
\textsc{Left ctxs} &\lctx & \grameq & \ctxhole %& \lctx\rsub{}\mvar
\mid \cuta\val\var\lctx \mid \para\mvar\var\vartwo \lctx \mid \suba\mvar\val\var\lctx  \mid \dera\evar\var\lctx
\\
\textsc{Value ctxs} & \vctx & \grameq & \ctxhole \mid \pair\ctx\tmtwo \mid \pair\tm\ctx \mid \la\var\ctx \mid \bang\ctx %& \ctx\rsub{}\mvar
\\
\textsc{General ctxs} & \ctx & \grameq & \vctx \mid \cuta\vctx\var \tm  \mid \suba\mvar\vctx\var\tm \mid \lctxp\ctx
\\
\textsc{Mult. value ctxs} & \vmctx & \grameq & \ctxhole  \mid \pair\mctx\tmtwo \mid \pair\tm\mctx  \mid \la\var\mctx  %& \ctx\rsub{}\mvar
\\
\textsc{Mul. ctxs} &\mctx & \grameq & \vmctx \mid\cuta\vmctx\var\tm \mid\suba\mvar\vmctx\var \tm \mid \lctxp\mctx
\end{array}$
\end{tabular}
\end{tabular}
%}

%% file: figure-IMELL-SC-typing.tex
% !TEX root = main.tex
%\fbox{
\begin{tabular}{c}	
\begin{tabular}{c}
$\begin{array}{rclll}
\textsc{Formulas} & \form,\formtwo,\formthree & \grameq & \aform \mid \form \tens \formtwo \mid \form \lolli \formtwo \mid \bang\form
\end{array}$

%%%%%%%
	\\[5pt]\hline\\[-7pt]
	%%%%%%%
\begin{tabular}{c}	
\textsc{Multiplicative right rules}
\\[9pt]
	\AxiomC{$\form \neq \bang \formtwo$}
	\RightLabel{$\ax_{\msym}$}
	\UnaryInfC{$  \mvar\hastype\form \vdash \mvar\hastype \form$}	
	\DisplayProof

	\\[9pt]

	\AxiomC{$  \multiForm\vdash \tm\hastype\form$}
 	\AxiomC{$  \multiFormtwo\vdash \tmtwo\hastype\formtwo$}
	\AxiomC{$\multiForm\#\multiFormtwo$}
 	\RightLabel{$ \tensRightRule $}
 	\TrinaryInfC{$  \multiForm, \multiFormtwo\vdash \pair\tm\tmtwo \hastype\form\otimes \formtwo$}
	\DisplayProof
	
	\\[9pt]

	\AxiomC{$  \var\hastype\form,\multiForm \vdash \tm\hastype\formtwo$}
	\RightLabel{$ \lolliRightRule $}
	\UnaryInfC{$  \multiForm \vdash \la\var\tm\hastype\form \lolli \formtwo$}
	\DisplayProof
\end{tabular}
	%%%%%%%
	\\[8pt]\hline\\[-7pt]
	%%%%%%%
\begin{tabular}{c}	
\textsc{Multiplicative left rules}
\\[9pt]
	\AxiomC{$\multiForm   \vdash \lctxp\val \hastype \form$}
	\AxiomC{$\multiFormtwo, \var\hastype \form \vdash \tm\hastype\formtwo$}
	\AxiomC{$\multiForm\#(\multiFormtwo, \var\hastype \form)$}
 	\RightLabel{$\cut$}
 	\TrinaryInfC{$  \multiForm, \multiFormtwo\vdash \lctxp{\cuta\val\var\tm} \hastype\formtwo$} 	
	\DisplayProof	

	\\[9pt]

	\AxiomC{$  \multiForm, \var\hastype\form, \vartwo\hastype\formtwo \vdash \tm\hastype\formthree$}
	\AxiomC{$\mvar$ fresh}
	\RightLabel{$ \tensLeftRule $}
	\BinaryInfC{$  \multiForm, \mvar\hastype\form \tens \formtwo \vdash \para\mvar\var\vartwo \tm \hastype \formthree$}	
	\DisplayProof
	
	\\[9pt]

	\AxiomC{$  \multiForm \vdash \lctxp\val \hastype\form$}
	\AxiomC{$  \multiFormtwo, \var\hastype\formtwo \vdash \tm\hastype\formthree$}
		\AxiomC{$\multiForm\#(\multiFormtwo, \var\hastype\formtwo)$, $\mvar$ fresh}
	\RightLabel{$ \lolliLeftRule $}
	\TrinaryInfC{$  \multiForm, \multiFormtwo, \mvar\hastype \form \lolli \formtwo \vdash \lctxp{\suba\mvar\val\var \tm} \hastype\formthree$}
	\DisplayProof

\end{tabular}
%%%%%%%
	\\[8pt]\hline\\[-7pt]
	%%%%%%%
\begin{tabular}{c@{\hspace*{.9cm}}c}	
\multicolumn{2}{c}{\textsc{Exponential rules}}
\\[9pt]
	&
	\AxiomC{}
	\RightLabel{$\ax_{\esym}$}
	\UnaryInfC{$  \evar\hastype\bang\form \vdash \evar\hastype \bang\form$}	
	\DisplayProof
	\\[9pt]
	%%%%%%%
	\AxiomC{$  \multiForm, \var\hastype\form \vdash \tm\hastype\formtwo$}
	\AxiomC{$\evar$ fresh}
	\RightLabel{$ \bangLeftRule$}
	\BinaryInfC{$  \multiForm, \evar\hastype\bang\form \vdash \dera\evar\var\tm \hastype\formtwo$}
	\DisplayProof
	&
		\AxiomC{$  \bang\multiForm \vdash \tm\hastype\form$}
	\RightLabel{$ \bangRightRule$}
	\UnaryInfC{$  \bang\multiForm \vdash \bang\tm\hastype\bang\form$}
	\DisplayProof
	%%%%%%%
	\\[9pt]
	%%%%%%%

	\AxiomC{$  \multiForm \vdash \tm\hastype\form$}
				\AxiomC{$\evar$ fresh}
	\RightLabel{$ \weakRule $}
	\BinaryInfC{$  \multiForm, \evar\hastype\bang\formtwo \vdash \tm\hastype \form$}
	\DisplayProof
&
		\AxiomC{$  \multiForm, \evar\hastype\bang\formtwo, \evartwo\hastype\bang\formtwo  \vdash \tm\hastype\form$}
	\RightLabel{$ \contrRule $}
	\UnaryInfC{$  \multiForm, \evar\hastype\bang\formtwo \vdash \cutsub\evar\evartwo \tm \hastype \form$}
	\DisplayProof
\end{tabular}
\end{tabular}
\end{tabular}
%}

%% file: figure-IMELL-SC-rewriting_rules.tex
% !TEX root = main.tex
%\fbox{
 \begin{tabular}{c}
$\begin{array}{rlll}
\multicolumn{3}{c}{\textsc{Root multiplicative rules}}
\\[9pt]
\cuta{\mval}\mvar\mctxfp\mvar & \rtoaxmone & \mctxfp\mval
\\
\cuta{\mvartwo}\mvar\tm & \rtoaxmtwo & \cutsub\mvartwo\mvar\tm
\\
%\ctxp{\tm\rsub{}\mvar}\esub\mvar{\lctxp{\star}}
%& \Rew{1} & 
%\lctxp{\ctxp{\tm}}
%\\
\cuta{\pair\tmtwo\tmthree}\mvar \mctxp{\para\mvar\var\vartwo\tm}
& \rtotens & 
\mctxp{\lctxp{\cuta\val\var \lctxtwop{\cuta\valtwo\vartwo \tm}}} 

\\&& \ \ \ \ \ \mbox{with }\tmtwo = \lctxp{\val}\mbox{ and }\tmthree = \lctxtwop{\valtwo}
\\[3pt]
\cuta{\la\vartwo\tmtwo}\mvar \mctxp{\suba\mvar\val\var  \tm}
& \rtololli & 
\mctxp{ \cuta\val\vartwo \lctxp{\cuta\valtwo\var \tm }} 
\\
&& \ \ \ \ \ \mbox{with }\tmtwo = \lctxp{\valtwo}
\end{array}$
%%%
%%%%%%%
\\[9pt]\hline\\[-6pt]
%%%%%%%
%%%
$\begin{array}{rlll}
\multicolumn{3}{c}{\textsc{Micro-step root exponential rules}}
\\[9pt]
\cuta{\exval} \evar \ctxfp{\evar}
& \rtoaxeone & 
\cuta{\exval}\evar\ctxfp{\exval}
\\
\cuta{\evartwo}\evar\ctxp{\dera\evar\var \tm}
& \rtoaxetwo &
\cuta{\evartwo}\evar\ctxp{\dera\evartwo\var \tm}
\\
\cuta{\bang\tmtwo}\evar\ctxp{\dera\evar\var \tm}
& \rtobder &
\cuta{\bang\tmtwo}\evar\ctxp{\lctxp{\cuta\val\var \tm}} 
\\
&& \ \ \ \ \ \mbox{with }\tmtwo = \lctxp{\val}
\\[3pt]
\cuta{\exval}\evar \tm
& \rtow & 
\tm \ \ \ \  \mbox{if }\evar\notin\fv\tm
\end{array}$
%%%
%%%%%%%
\\[9pt]\hline\\[-6pt]
%%%%%%%
%%%
\textsc{Small-step root exponential rule}
\\[9pt]
$\begin{array}{lll}
\cuta\exval\evar \tm & \rtoess & \cutsub\exval\evar\tm
\end{array}$
%%%
%%%%%%%
\\[9pt]\hline\\[-6pt]
%%%%%%%
%%%
	\begin{tabular}{ccc}
\multicolumn{2}{c}{\textsc{Contextual closure}}
\\[5pt]
\AxiomC{$\tm \rootRew{a} \tmtwo$}
\UnaryInfC{$\ctxp{\tm} \Rew{a} \ctxp{\tmtwo}$}
\DisplayProof
		& 
		for $a \in \left\{
  \begin{aligned}
  & \axmone,\axmtwo,\tens,\lolli,\\
  & \axeone,\axetwo,\bdersym,\wsym,\ess
  \end{aligned}
\right\}$
	\end{tabular}
%%%
%%%%%%%
\\[9pt]\hline\\[-6pt]
%%%%%%%
%%%
$\begin{array}{rlllc}
\multicolumn{4}{c}{\textsc{Notations}}
\\
%\toaxm  &\defeq & \toaxmtwo \cup \toaxmone
%\\
\textsc{Multiplicative} & \tom  &\defeq &  \toaxmone\cup \toaxmtwo \cup \totens \cup \tololli
\\
\textsc{Exponential micro-step} & \toems  &\defeq & \toaxeone \cup \toaxetwo \cup \tobder \cup \tow
\\
\textsc{Small-step} & \toss  &\defeq & \tom \cup \toess
\\
\textsc{Micro-step} &\toms  &\defeq & \tom \cup \toems
\\
\textsc{Non-erasing micro-step} &\tomsnw  &\defeq & \tom \cup \toaxeone \cup \toaxetwo \cup \tobder
\\
\textsc{Non-lolli micro-step} &\tomsnlolli  &\defeq & \toaxmone\cup \toaxmtwo \cup \totens \cup \toems
\\
\textsc{Non-lolli small-step} &\tossnlolli  &\defeq & \toaxmone\cup \toaxmtwo \cup \totens \cup \toess
	\end{array}$
\end{tabular}
%}

%% file: proofs/calculus/substitutivity.tex
% !TEX root = ../../main.tex
\begin{proof}
\hfill
\begin{enumerate}
\item By induction on $\tm \Rew{a} \tmtwo$. Details are in \cite{DBLP:journals/corr/abs-lics}. The only subtle case is the one for $\rtoaxetwo$, which has a sub-case where the step might become a $\rtobder$ step after subtitution. Let $\tm = \cuta{\evarthree}\evartwo\ctxp{\dera\evartwo\var \tmthree}
 \rtoaxetwo 
\cuta{\evarthree}\evartwo\ctxp{\dera\evarthree\var \tmthree} = \tmtwo$,  $\evarthree = \evar$, and $\exval = \bang\lctxp\val$. Then:
\begin{center}
	$\begin{array}{rllll}
	\cutsub{\exval}\evarthree\tm 
	& = &
	\cutsub{\exval}\evarthree \cuta{\evarthree}\evartwo\ctxp{\dera\evartwo\var \tmthree}
	\\
	& = &
  	\cuta{\exval}\evartwo \cutsub{\exval}\evarthree \ctxp{\dera\evartwo\var \tmthree}
	\\
	& =_{\reflemmapeq{bang-subs-prop}{zero}} &
	\cuta{\exval}\evartwo (\cutsub{\exval}\evarthree \ctx)	\ctxholep{\dera\evartwo\var \cutsub{\exval}\evarthree \tmthree}
	\\
	& \rtobder & 
	 \cuta{\exval}\evartwo (\cutsub{\exval}\evarthree \ctx)\ctxholep{\lctxp{\cuta\val\var \cutsub{\exval}\evarthree \tmthree}}
  	\\
	& = &
  	\cuta{\exval}\evartwo (\cutsub{\exval}\evarthree \ctx)\ctxholep{\cutsub{\exval}\evarthree \dera{\evarthree}\var  \tmthree}
	\\
	&=_{\reflemmapeq{bang-subs-prop}{zero}} &
  	\cuta{\exval}\evartwo \cutsub{\exval}\evarthree \ctxp{\dera{\evarthree}\var \tmthree}
	\\
	 &=& 
	\cutsub{\exval}\evarthree \cuta{\evarthree}\evartwo\ctxp{\dera\evarthree\var \tmthree}
	
	 &= & 
	\cutsub{\exval}\evarthree \tmtwo
	\end{array}$
	\end{center}
\item By induction on $\tm$.\qedhere
\end{enumerate}
\end{proof}

%% file: 07-Basic_Properties.tex
% !TEX root = main.tex
\section{Basic Properties}
\label{sect:basic-prop}
We now prove various basic properties of the ESC. An expected property that we do not prove is the simulation of the LSC by the ESC. It can be obtained via Girard's (call-by-name) encoding of intuitionistic logic into linear logic. It is not surprising and yet it is technical, for reasons (the encoding of applications via subtractions) that are intrinsic to the translation of natural deduction to sequent calculus and orthogonal to the features of the ESC, which is why it is omitted.

\paragraph{Full Composition.} First of all, micro-step exponential rules simulate the small-step one. The ESs literature calls this fact \emph{full composition}.

\begin{prop}[Full composition, or micro-step simulates small-step]
\label{prop:ms-simulates-ss}
If $\tm\toess \tmtwo$ then $\tm\toems^{+}\tmtwo$.
\end{prop}

\begin{proof}
Let 
$\tm = \ctxp{\cuta\exval\evar\tmthree} \toess \ctxp{\cutsub\exval\evar\tmthree} = \tmtwo$. We treat the case $\exval = \bang\tmfour$, the case $\exval = \evartwo$ is analogous. The proof is by induction on $k \defeq \varmeas{\tmthree}\evar$. If $k = 0$ then $\ctxp{\cuta\exval\evar\tmthree} \tow \ctxp{\tmthree} = \ctxp{\cutsub\exval\evar\tmthree}$. Otherwise, $\tmthree$ writes has $\ctxtwop{\evar}$ or $\ctxtwop{\dera\evar\var \tmthree'}$. Consider the first case. Then:
\begin{center}
$\begin{array}{ccccccc}
\ctxp{\cuta{\bang\tmthree'}\evar\ctxtwop{\evar}} 
& \toaxeone &
\ctxp{\cuta{\bang\tmthree'}\evar\ctxtwop{\bang\tmthree'}} 
& =: & \tmtwo'
\end{array}$
\end{center} 
Now note that $\tmtwo' \toess \tmtwo$ and that in $\tmtwo'$ we have $\varmeas{\ctxp{\bang\tmthree'}}\evar < k$. Then we can apply the \ih, obtaining $\tmtwo' \toems^{+}\tmtwo$, that is, $\tm \toems^{+}\tmtwo$. The case $\tmthree=\ctxtwop{\dera\evar\var \tmthree'}$ is analogous, simply using $\tobder$ instead of $\toaxeone$.
\end{proof}

\begin{cor}
Let $\tm$ be a \proper term and $\tm \toss \tmtwo$. Then $\tmtwo$ is \proper.
\end{cor}

\paragraph{Subject Reduction.} The next property is that the defined untyped cut elimination respects the typing system, that is, subject reduction holds. We first need a lemma.

\begin{lem}
\label{l:subj-red-aux}
Let $\tderiv \pof \multiform \vdash \tm \hastype \form$.  We have $\form = \bang\formtwo$ if and only if $\tm = \lctxp\exval$.
\end{lem}

\begin{proof}
If $\form =  \bang\formtwo$, then consider $\tm = \lctxp\val$. The proof is by induction on $\lctx$. Clearly, if $\lctx=\ctxhole$ then the last rule of $\tderiv$ is necessarily an exponential axiom or a right introduction of $\bang$, that is, $\val = \evar$ or $\val = \bang\tmtwo$, that is, $\val$ is an exponential value. If $\lctx \neq \ctxhole$ then it follows from the \ih

The reasoning is analogous if $\form \neq \bang\formtwo$.
\end{proof}
%We suggest to have a look at the proof in \ben{the extended version}, to see how contextual cut elimination acts on proofs. 
%One case of the proof is in \reffig{mult-cut-on-proofs}.
\begin{toappendix}
\begin{thm}[Subject reduction]
\label{thm:sub-red}
Let $\multiform\vdash \tm\hastype \form$ and $\tm \toms \tmtwo$. Then $\multiform\vdash \tmtwo\hastype \form$.
\end{thm}
\end{toappendix}

\withproofs{\input{\proofspath/calculus/subject_reduction}}

Full composition extends subject reduction to $\toess$.

\paragraph{Clashes.} The presence of many constructors in an untyped setting gives rise to \emph{clashes}, that is, irreducible cuts.
\begin{defi}[Clash, clash-free terms]
A \emph{clash} is a term of the form $\cuta\mval\evar\tm$, $\cuta\exval\mvar\tm$, $\cuta{\pair\tm\tmtwo}\mvar \mctxp{\suba\mvar\val\vartwo\tmthree}$, or $\cuta{\la\var\tm}\mvar \mctxp{\para\mvar\vartwo\varthree \tmthree}$ and in these cases we also say that the root cut is \emph{clashing}. A term $\tm$ is \emph{clash-free} if, co-inductively,
\begin{enumerate}
	\item No sub-terms of $\tm$ are clashes, and
	\item If $\tm \toms \tmtwo$ then $\tmtwo$ is clash-free.
\end{enumerate}
\end{defi}
Note that there are no purely exponential clashes. The second point of the next lemma uses subject reduction.

\begin{lem}
\label{l:clashfree-implies-cutfree}
Let $\tm$ be a proper term.
\begin{enumerate}
\item If $\tm$ has no clashes and $\tm \not\toms$ then $\tm$ is cut-free.
\item If $\tm$ is typable then it is clash-free.
\end{enumerate}
\end{lem}

\begin{proof}
\hfill
\begin{enumerate}
\item Suppose by contradiction that $\tm$ has a cut. Then $\tm = \ctxp{\cuta\val\var\tmtwo}$. If $\var$ is exponential then, since there are no clashes, $\val$ is also exponential. Then there is a $\toess$ redex, and, by full composition, a $\toems$ redex, absurd. If $\var=\mvar$ is multiplicative then so is $\val$. By structural linearity (\reflemma{proper-linearity}), $\tmtwo$ has one of the following shapes: 
\begin{itemize}
\item $\mctxp\mvar$. Then there is a $\toaxmone$ redex, absurd.
\item $\mctxp{\para\mvar\vartwo\varthree\tmtwo}$. Since $\tm$ has no clashes, $\val$ is either a multiplicative variable, and then there is a $\toaxmtwo$ redex, or a tensor pair, and then there is a $\totens$ redex, absurd.
\item $\mctxp{\suba\mvar\val\vartwo\tmtwo}$. Since $\tm$ has no clashes, $\val$ is either a multiplicative variable, and then there is a $\toaxmtwo$ redex, or an abstraction, and then there is a $\tololli$ redex, absurd.
\end{itemize}
%%%%%%%%%

\item By co-induction on the definition of being clash-free. It is easily seen that clashes are not typable. So if $\tm$ is typable then it has no clashes. Now, if $\tm \toms \tmtwo$ then by subject reduction (\refthm{sub-red}) $\tmtwo$ is typable. By the co-inductive hypothesis, $\tmtwo$ is clash-free. Therefore, $\tm$ is clash-free.\qedhere
\end{enumerate}
\end{proof}

\paragraph{Postponement of Garbage Collection} The micro-step rule $\tow$, that models garbage collection, can be postponed. The same happens in the LSC,
 but \emph{not} in the $\l$-calculus: for instance, the first step in the sequence $(\la\var\la\vartwo\vartwo)\tmtwo\tmthree \tob (\la\vartwo\vartwo)\tmthree \tob \tmthree$ is erasing but cannot be postponed.

\begin{lem}[Local postponement of garbage collection]
\label{l:local-post-gc}
 If $\tm \tow \tomsnw \tmtwo$ then $\tm\tomsnw \tow^{+}\tmtwo$.
\end{lem}
\withproofs{\input{\proofspath/calculus/GC-postponement}}

\begin{prop}[Postponement of garbage collection]
\label{prop:postp-gc} % \refpropp{postp-gc}{global}
if $\tm \toms^{*} \tmtwo$ then $\tm \tomsnw^{*} \tow^{*} \tmtwo$.
\end{prop}
\begin{proof}
It is an instance of a well-known rewriting property: the local swap in \reflemma{local-post-gc} is an instance of Hindley's strong postponement property, which implies postponement.
\end{proof}

\paragraph{Cut Equivalence} Similarly to the structural equivalence of the LSC, we consider a notion of cut equivalence, which is a strong bisimulation and it is postponable.

\begin{defi}[Cut equivalence]
If $\var\notin\fv\mctx$ and $\mctx$ does not capture variables in $\fv\val$, define:
\begin{center}
$\begin{array}{ccccccc}
\cuta\val\var \mctxp\tm
& \rcuteq & 
\mctxp{\cuta\val\var \tm}. 
\end{array}$
\end{center}
\emph{Cut equivalence} $\cuteq$ is the closure of $\rcuteq$ by reflexivity, symmetry, transitivity, and general contexts.
\end{defi}

To prove that cut equivalence is a strong bisimulation in the small-step case, we need the following substitutivity property.

\begin{lem}[Stability by substitution of $\cuteq$]
\label{l:cuteq-subs} % \reflemmap{cuteq-subs}{}
\hfill
\begin{enumerate}
\item \label{p:cuteq-subs-one}
If $\tm \cuteq \tmtwo$ then $\cutsub\exval\evar\tm \cuteq \cutsub\exval\evar\tmtwo$ for all $\exval$.
\item \label{p:cuteq-subs-two}
If $\exval \cuteq \exvaltwo$ then $\cutsub\exval\evar\tm \cuteq \cutsub\exvaltwo\evar\tm$.
\end{enumerate}
\end{lem}

\withproofs{
\input{\proofspath/calculus/subs-for-cut-equivalence}
}

\begin{prop}
\label{prop:cuteq-bisim} % \refpropp{cuteq-bisim}{sn}
Let $\to\in\set{\toss,\toms}$. 
\begin{enumerate}
\item \emph{Strong bisimulation}: %\label{p:cuteq-bisim-sb} 
 $\cuteq$ is a strong bisimulation with respect to $\to$, preserving also the kind of step.
\item \label{p:cuteq-bisim-sn} 
\emph{Structural stability of SN}: if $\tm \cuteq \tmtwo$ and $\tm\in\snsubst$ then $\tmtwo\in\snsubst$.
\item \emph{Postponement of $\cuteq$}: if $\tm (\cuteq\to\cuteq)^k \tmtwo$ then $\tm \to^k \cuteq \tmtwo$.
\end{enumerate}
\end{prop}

\withproofs{\input{\proofspath/calculus/cut-equivalence}}

Cut equivalence shall play a role in the following sections. Note that its postponement property means that it is \emph{never needed} for cut elimination to progress. And on cut-free terms, obviously, it vanishes.

There is also a larger \emph{left equivalence} $\streq_{\symfont{left}}$ defined as $\cuteq$ but including every left constructor (cut, par, subtraction, and dereliction), which is also a strong bisimulation. As mentioned in \refsect{lsc}, $\equiv$ is the quotient induced by proof nets on the LSC. Intuitively, $\streq_{\symfont{left}}$ is the quotient induced by proof nets on the ESC. There are, however, no notions of proof nets realizing such a quotient, as weakenings require non-canonical \emph{jumps} for IMELL, that can instead be avoided for the LSC, as shown by Accattoli \cite{DBLP:conf/ictac/Accattoli18}.

A possibly interesting point is that we shall need cut equivalence, in particular for the proof of strong normalization, while there is no need of left equivalence for any of the properties studied in this paper.

\paragraph{Stability Under Renaming} The last basic property that we prove is the stability of the rules and of their strong normalization under renaming.
\begin{lem}[Renaming]
\label{l:deformations} % \reflemmap{deformations}{split}
Let $\to\in\set{\toss,\toms}$ and $\var$ and $\vartwo$ be two variables of the same 
multiplicative/exponential kind such that  $\cutsub\var\vartwo\tm$ is proper.
\begin{enumerate}
%%%%%%
 %%%%%%%%%%
 \item \label{p:deformations-rename} \emph{Stability of steps under renaming}:
if $\ctx:\tm \Rew{a} \tmtwo$ then $\cutsub\vartwo\var\ctx:\cutsub\vartwo\var\tm \Rew{a} \cutsub\vartwo\var\tmtwo$ for $a\in\set{\axmone,\axmtwo,\tens,\lolli, \axeone,\axetwo,\bdersym,\wsym,\ess}$.

  \item \label{p:deformations-sn} \emph{Stability of SN under renaming}:
  if $\tm \in \sn\to$  then $\cutsub\var\vartwo\tm \in \sn\to$.
\end{enumerate}
\end{lem}

\begin{proof}
\hfill
\begin{enumerate}
\item Straightforward induction on the rewriting step.
\item By induction on $\tm\in\sn\to$. Since renaming cannot create, erase, duplicate, or change the kind of redexes, 
if $\cutsub\var\vartwo\tm \to \tmtwo$ then there exists $\tmthree$ such that $\tm \to \tmthree$ and 
$\cutsub\var\vartwo\tmthree = \tmtwo$. Then by \ih on $\tmthree$ we obtain $\cutsub\var\vartwo\tmthree=\tmtwo\in\sn\to$.
\qedhere
\end{enumerate}
\end{proof}

%% file: proofs/calculus/subject_reduction.tex
% !TEX root = ../../main.tex
\begin{proof}
By induction on $\tderiv \pof \multiform \vdash \tm \hastype \form$. If $\tderiv$ is an axiom then $\tm$ cannot reduce. If the last rule is a cut then $\tm = \lctxp{\cuta\val\var\tmtwo}$. If it is the root cut that it is reduced there are a number of cases given next. Otherwise (that is, if the last rule is not a cut, or if the reduced cut is not the one of the last rule) the statement follows by the \ih 

By \reflemma{subj-red-aux}, if the cut formula is $\bang\formtwo$ for some $\formtwo$ then $\var = \evar$ and $\val = \exval$, and otherwise $\var = \mvar$ and $\val = \mval$. We show two representative cases, the $\tololli$ one, which shows how to interpret the splitting of terms into value and left contexts on proofs, and $\toaxetwo$, which shows how LSC-style duplication acts on proofs. The other cases are similar.
\begin{itemize}
\item $\rtololli$. 
\begin{center}
\begin{tabular}{ccc}
\multicolumn{2}{c}{
\AxiomC{$\multiForm \vdash  \val\hastype\formtwo$}
\doubleLine\dashedLine
 	\RightLabel{$\lctx$}
\UnaryInfC{$\multiForm', \vartwo\hastype\form \vdash  \lctxp\val\hastype\formtwo$}
 	\RightLabel{$\lolliRightRule$}
	\UnaryInfC{$\multiForm'   \vdash  \la\vartwo\lctxp\val\hastype \form{\lolli}\formtwo$}
\doubleLine\dashedLine
 	\RightLabel{$\lctxtwo$}
	\UnaryInfC{$\multiForm''   \vdash  \lctxtwop{\la\vartwo\lctxp\val}\hastype \form{\lolli}\formtwo$}%			
%	\AxiomC{$\multiFormthree,  \mvartwo\hastype \form, \mvarthree\hastype \formtwo \vdash \tm\hastype\formthree$}
%	 	\RightLabel{$\tensLeftRule$}
%		\UnaryInfC{$\multiFormthree,  \mvar\hastype \form{\tens}\formtwo \vdash \subpar\mvar{\mvartwo,\mvarthree}\tm\hastype\formthree$}
%\doubleLine\dashedLine
% 	\RightLabel{$\mctx$}
	
\AxiomC{$\multiFormtwo   \vdash  \valtwo\hastype\form$}
\doubleLine\dashedLine
 	\RightLabel{$\lctxthree$}
\UnaryInfC{$\multiFormtwo'   \vdash  \lctxthreep\valtwo\hastype\form$}

	\AxiomC{$  \multiFormthree, \var\hastype\formtwo \vdash \tm\hastype\formthree'$}
	\RightLabel{$ \lolliLeftRule $}
	
			\BinaryInfC{$\multiFormtwo', \multiFormthree,  \mvar\hastype \form{\lolli}\formtwo \vdash \lctxthreep{\suba\mvar{\valtwo,\var}\tm} \hastype\formthree'$}
			\doubleLine\dashedLine
 	\RightLabel{$\mctx$}
\UnaryInfC{$\multiFormthree',  \mvar\hastype \form{\lolli}\formtwo \vdash \mctxp{\lctxthreep{\suba\mvar{\valtwo,\var}\tm}} \hastype\formthree'$}

 	\RightLabel{$\cut$}
 	\BinaryInfC{$  \multiForm'', \multiFormthree'\vdash \lctxtwop{\cuta{\la\vartwo\lctxp\val}\mvar \mctxp{\lctxthreep{\suba\mvar{\valtwo,\var}\tm}}} \hastype\formthree'$} 	

	\DisplayProof
}
\\[6pt]
$\tololli$
 &
\AxiomC{$\multiFormtwo   \vdash  \valtwo\hastype\form$}
\doubleLine\dashedLine
 	\RightLabel{$\lctxthree$}
\UnaryInfC{$\multiFormtwo'   \vdash  \lctxthreep\valtwo\hastype\form$}

\AxiomC{$\multiForm \vdash  \val\hastype\formtwo$}
\doubleLine\dashedLine
 	\RightLabel{$\lctx$}
\UnaryInfC{$\multiForm', \vartwo\hastype\form \vdash  \lctxp\val\hastype\formtwo$}

	\AxiomC{$\multiFormthree,  \vartwo\hastype \form, \var\hastype \formtwo \vdash \tm\hastype\formthree$}

 	\RightLabel{$\cut$}
 	\BinaryInfC{$  \multiFormtwo', \vartwo\hastype \form,\multiFormthree\vdash \lctxp{\cuta\val\var \tm} \hastype\formthree'$} 	
 	\RightLabel{$\cut$}

 	\BinaryInfC{$  \multiForm',\multiFormtwo', \multiFormthree\vdash \lctxthreep{\cuta\valtwo\vartwo \lctxp{\cuta\val\var \tm}}$} 	
\doubleLine\dashedLine
 	\RightLabel{$\mctx$}
			\UnaryInfC{$\multiform',\multiFormthree' \vdash \mctxp{\lctxthreep{\cuta\valtwo\vartwo \lctxp{\cuta\val\var \tm}} }\hastype\formthree'$}
				\doubleLine\dashedLine
 	\RightLabel{$\lctxtwo$}
			\UnaryInfC{$\multiform'',\multiFormthree' \vdash \lctxtwop{\mctxp{\lctxthreep{\cuta\valtwo\vartwo \lctxp{\cuta\val\var \tm}} }} \hastype\formthree'$}

	\DisplayProof		
	\end{tabular}
	\end{center}
	
	\item $\toaxetwo$. The step is $\lctxp{\cuta{\evartwo}\evar\ctxp{\dera\evar\var \tm}} \hastype\formtwo \toaxetwo \lctxp{\cuta{\evartwo}\evar\ctxp{\dera\evartwo\var \tm}} \hastype\formtwo$. To avoid analyzing many similar (simple) sub-cases, let us assume that $\lctx$ does not capture $\evartwo$ and that $\evartwo\notin\fv\lctx$.
Now there still are two cases, depending on the multiplicity of $\evar$ in $\ctxp{\dera\evar\var \tm}$.
\begin{itemize}
\item $\evar\notin\fv\ctx$. Then:
\begin{center}
\begin{tabular}{ccc}
\multicolumn{2}{c}{
\AxiomC{}
 	\RightLabel{$\ax_{\esym}$}
	\UnaryInfC{$\evartwo\hastype \bang\form   \vdash  \evartwo\hastype \bang\form$}
\doubleLine\dashedLine
 	\RightLabel{$\lctx$}
	\UnaryInfC{$\multiForm,\evartwo\hastype \bang\form   \vdash  \lctxp{\evartwo}\hastype \form$}
	%%%%
	\AxiomC{$  \multiFormthree, \var\hastype\form \vdash \tm\hastype\formtwo$}
	\RightLabel{$ \bangLeftRule $}
			\UnaryInfC{$\multiFormthree,  \evar\hastype \bang\form \vdash \dera\evar\var\tm \hastype\formtwo$}
			\doubleLine\dashedLine
 	\RightLabel{$\ctx$}
\UnaryInfC{$\multiFormtwo, \evar\hastype \bang\form \vdash \ctxp{\dera\evar\var\tm} \hastype\formtwo$}
%%%
 	\RightLabel{$\cut$}
 	\BinaryInfC{$  \multiForm, \evartwo\hastype \bang\form, \multiFormtwo\vdash \lctxp{\cuta{\evartwo}\evar\ctxp{\dera\evar\var \tm}} \hastype\formtwo$} 	
	\DisplayProof	
}
\\[6pt]
$\toaxetwo$ 
&
\AxiomC{}
 	\RightLabel{$\ax_{\esym}$}
	\UnaryInfC{$\evartwo\hastype \bang\form   \vdash  \evartwo\hastype \bang\form$}
\doubleLine\dashedLine
 	\RightLabel{$\lctx$}
	\UnaryInfC{$\multiForm,\evartwo\hastype \bang\form   \vdash  \lctxp{\evartwo}\hastype \form$}
	%%%%
	\AxiomC{$  \multiFormthree, \var\hastype\form \vdash \tm\hastype\formtwo$}
	\RightLabel{$ \bangLeftRule $}
			\UnaryInfC{$\multiFormthree,  \evarthree\hastype \bang\form \vdash \dera\evarthree\var\tm \hastype\formtwo$}
			\doubleLine\dashedLine
 	\RightLabel{$\ctx$}
\UnaryInfC{$\multiFormtwo, \evarthree\hastype \bang\form \vdash \ctxp{\dera\evarthree\var\tm} \hastype\formtwo$} 	
\RightLabel{$\weakRule$}
\UnaryInfC{$\multiFormtwo, \evarthree\hastype \bang\form, \evar\hastype \bang\form \vdash \ctxp{\dera\evarthree\var\tm} \hastype\formtwo$}

%%%
 	\RightLabel{$\cut$}
 	\BinaryInfC{$  \multiForm, \evartwo\hastype \bang\form, \evarthree\hastype \bang\form, \multiFormtwo\vdash \lctxp{\cuta{\evartwo}\evar\ctxp{\dera\evarthree\var \tm}} \hastype\formtwo$} 	
	\RightLabel{$\contrRule$}
\UnaryInfC{$ \multiForm, \evartwo\hastype \bang\form, \multiFormtwo\vdash \lctxp{\cuta{\evartwo}\evar\ctxp{\dera\evartwo\var \tm}} \hastype\formtwo$}
	\DisplayProof	
	\end{tabular}
	\end{center}
%%%%%%%
\item $\evar\in\fv\ctx$. Then in the sequence of rules corresponding to the context $\ctx$ there is a first contraction acting on the dereliction which splits the sequence of rules in two, that is, there exists $\ctxtwo$ and $\ctxthree$ such that $\ctx=\ctxtwop\ctxthree$ and:
\begin{center}
\begin{tabular}{ccc}
\multicolumn{2}{c}{
\AxiomC{}
 	\RightLabel{$\ax_{\esym}$}
	\UnaryInfC{$\evartwo\hastype \bang\form   \vdash  \evartwo\hastype \bang\form$}
\doubleLine\dashedLine
 	\RightLabel{$\lctx$}
	\UnaryInfC{$\multiForm,\evartwo\hastype \bang\form   \vdash  \lctxp{\evartwo}\hastype \form$}
	%%%%
	\AxiomC{$  \multiFormthree, \var\hastype\form \vdash \tm\hastype\formtwo$}
	\RightLabel{$ \bangLeftRule $}
			\UnaryInfC{$\multiFormthree,  \evarthree\hastype \bang\form \vdash \dera\evarthree\var\tm \hastype\formtwo$}
			\doubleLine\dashedLine
 	\RightLabel{$\ctxthree$}
\UnaryInfC{$\multiFormtwo', \evar'\hastype \bang\form, \evarthree\hastype \bang\form \vdash \ctxthreep{\dera\evarthree\var\tm} \hastype\formtwo$}
\RightLabel{$\contrRule$}
\UnaryInfC{$\multiFormtwo', \evar'\hastype \bang\form\vdash \ctxthreep{\dera{\evar'}\var\tm} \hastype\formtwo$}
 	\RightLabel{$\ctxtwo$}
	\UnaryInfC{$\multiFormtwo, \evar\hastype \bang\form \vdash \ctxtwop{\ctxthreep{\dera\evar\var\tm}} \hastype\formtwo$}

%%%
 	\RightLabel{$\cut$}
 	\BinaryInfC{$  \multiForm, \evartwo\hastype \bang\form, \multiFormtwo\vdash \lctxp{\cuta{\evartwo}\evar\ctxp{\dera\evar\var \tm}} \hastype\formtwo$} 	
	\DisplayProof	
	}
\\[6pt]

$\toaxetwo$ & \AxiomC{}
 	\RightLabel{$\ax_{\esym}$}
	\UnaryInfC{$\evartwo\hastype \bang\form   \vdash  \evartwo\hastype \bang\form$}
\doubleLine\dashedLine
 	\RightLabel{$\lctx$}
	\UnaryInfC{$\multiForm,\evartwo\hastype \bang\form   \vdash  \lctxp{\evartwo}\hastype \form$}
	%%%%
\AxiomC{$  \multiFormthree, \var\hastype\form \vdash \tm\hastype\formtwo$}
	\RightLabel{$ \bangLeftRule $}
			\UnaryInfC{$\multiFormthree,  \evarthree\hastype \bang\form \vdash \dera\evarthree\var\tm \hastype\formtwo$}
			\doubleLine\dashedLine
 	\RightLabel{$\ctxthree$}
\UnaryInfC{$\multiFormtwo', \evar'\hastype \bang\form, \evarthree\hastype \bang\form \vdash \ctxthreep{\dera\evarthree\var\tm} \hastype\formtwo$}
 	\RightLabel{$\ctxtwo$}
	\UnaryInfC{$\multiFormtwo, \evar\hastype \bang\form,\evarthree\hastype \bang\form   \vdash \ctxtwop{\ctxthreep{\dera\evarthree\var\tm}} \hastype\formtwo$}

%%%
 	\RightLabel{$\cut$}
 	\BinaryInfC{$  \multiForm, \evartwo\hastype \bang\form, \evarthree\hastype \bang\form, \multiFormtwo\vdash \lctxp{\cuta{\evartwo}\evar\ctxp{\dera\evarthree\var \tm}} \hastype\formtwo$} 	
	\RightLabel{$\contrRule$}
\UnaryInfC{$ \multiForm, \evartwo\hastype \bang\form, \multiFormtwo\vdash \lctxp{\cuta{\evartwo}\evar\ctxp{\dera\evartwo\var \tm}} \hastype\formtwo$}
	\DisplayProof	
	\end{tabular}
	\end{center}\vspace*{-\baselineskip}\qedhere
\end{itemize}

\end{itemize}
\end{proof}

%% file: proofs/calculus/GC-postponement.tex
% !TEX root = ../../main.tex
\begin{proof}
 By induction on $\tm \tow \tmthree$ and case analysis of $\tmthree \tomsnw \tmtwo$. Cases:
 \begin{itemize}
 	\item \emph{Root}: $\tm = \cuta{\exval}\evar \tmthree \tow \tmthree \tomsnw \tmtwo$. Then $\cuta{\exval}\evar \tmthree \tomsnw \cuta{\exval}\evar \tmtwo \tow \tmtwo$.
	
	\item \emph{Inductive cases different from cut}: we treat the case of tensor, the others are similar. Let $\tm = \pair\tmfour\tmfive \tow \pair{\tmfour'}\tmfive = \tmthree$. If the $\tomsnw$ step takes place in $\tmfour'$ then we apply the \ih, otherwise $\pair{\tmfour'}\tmfive \tomsnw \pair{\tmfour'}{\tmfive'} = \tmtwo$ and we simply have $\pair\tmfour\tmfive \tomsnw \pair{\tmfour}{\tmfive'} \tow \pair{\tmfour'}{\tmfive'}$.
	
		\item \emph{Inductive case left of cut}: $\tm = \cuta\val\var\tmfive \tow \cuta{\valtwo}\var\tmfive = \tmtwo$. If the $\tomsnw$ step takes place in $\valtwo$ we use the \ih, and if it is in $\tmfive$ then the two steps are disjoint and swap. The last possibility is that the root cut itself is the active one in the $\tomsnw$ redex. Now, the key observation is that swapping the step $\tmtwo \tomsnw \tmthree$ before the $\tow$ one amounts to add a cut (the erased one) to both $\tmthree$ and $\tmtwo$, and to remove it with the postponed $\tow$ step, which can be done smoothly because the rewriting rules do not depend on the cuts surrounding the sub-terms of the redex nor those surrounding the redex. The only cases where something interesting happens are those for $\toaxeone$ and $\tobder$ when the $\tow$ step is in the duplicated exponential value. For instance, if $\evartwo\notin\fv\tm$ we have the following diagram.
		 \begin{center}
\begin{tikzpicture}[ocenter]
		\node at (0,0)[align = center](source){\normalsize$  \cuta{\vctxp{\cuta\exval\evartwo\tm}} \evar \ctxfp{\evar}$};
		\node at (source.center)[right = 260pt](source-right){\normalsize $\cuta{\vctxp\tm} \evar \ctxfp{\evar}$};	
		\node at (source.center)[below = 25pt](source-down){\normalsize $ \cuta{\vctxp{\cuta\exval\evartwo\tm}} \evar \ctxfp{\vctxp{\cuta\exval\evartwo\tm}}$};
		
		\node at (source-right|-source-down)(target){\normalsize $ \cuta{\vctxp\tm} \evar \ctxfp{\vctxp\tm}$};
		\node at \med{source-down.east}{target.west}(fifthnode){\normalsize $ \cuta{\vctxp\tm} \evar \ctxfp{\vctxp{\cuta\exval\evartwo\tm}}$};
		
		\draw[->](source) to node[above] {\scriptsize $\wsym $} (source-right);
		\draw[->](source-right) to node[right] {\scriptsize $\axeone $}(target);
		
		\draw[->, dotted](source) to node[left] {\scriptsize $\axeone $}(source-down);
		\draw[->, dotted](source-down) to node[above] {\scriptsize $\wsym $} (fifthnode);
		\draw[->, dotted](fifthnode) to node[above] {\scriptsize $\wsym $} (target);
	\end{tikzpicture}
\end{center}
			
		\item \emph{Inductive case right of cut}: $\tm = \cuta\val\var\tmfive \tow \cuta{\val}\var\tmfive' = \tmtwo$.  If the $\tomsnw$ step takes place in $\tmfive'$ we use the \ih, and if it is in $\val$ then the two steps are disjoint and swap. The last possibility is that the root cut itself is the active one in the $\tomsnw$ redex. As in the previous case the swap happens smoothly because the rewriting rules do not depend on the cuts surrounding the subterms of the redex nor those surrounding the redex. There are no interesting cases.\qedhere

 \end{itemize}
 \end{proof}

%% file: proofs/calculus/subs-for-cut-equivalence.tex
% !TEX root = ../../main.tex
\begin{proof}
Easy inductions on $\tm \cuteq \tmtwo$ and $\tm$. Details in \cite{DBLP:journals/corr/abs-lics}.
\end{proof}

%% file: proofs/calculus/cut-equivalence.tex
% !TEX root = ../../main.tex
\begin{proof}
\hfill
\begin{enumerate}
\item It is a very long but otherwise straightforward check of all possible diagrams, spelled out in \cite{DBLP:journals/corr/abs-lics}. The following two diagrams for $\toess$ are where \reflemma{cuteq-subs} is used.
\begin{center}
\begin{tabular}{c\colspace\colspace\colspace c}
\begin{tikzpicture}[ocenter]
		\node at (0,0)[align = center](source){\normalsize$ \cuta\val\vartwo\tmtwo $};
		\node at (source.east)[right = 35pt](source-right){\normalsize $\cutsub\val\vartwo\tmtwo$};	
		\node at (source.center)[below = 20pt](source-down){\normalsize $\cuta\val\vartwo\tmtwo'$};
		
		\node at (source-right|-source-down)(target){\normalsize $\cutsub\val\vartwo\tmtwo'$};

		\node at \med{source.center}{source-down.center}(cuteq-left){\normalsize $\cuteq$};
		\node at \med{source-right.center}{target.center}(cuteq-right){\normalsize $\cuteq$};
		\draw[|->](source) to node[above] {\scriptsize $\ess $} (source-right);
		\draw[|->, dotted](source-down) to node[above] {\scriptsize $\ess $} (target);
	\end{tikzpicture}
	&
	\begin{tikzpicture}[ocenter]
		\node at (0,0)[align = center](source){\normalsize$ \cuta\val\vartwo\tmtwo $};
		\node at (source.east)[right = 35pt](source-right){\normalsize $\cutsub\val\vartwo\tmtwo$};	
		\node at (source.center)[below = 20pt](source-down){\normalsize $\cuta{\val'}\vartwo\tmtwo$};
		
		\node at (source-right|-source-down)(target){\normalsize $\cutsub{\val'}\vartwo\tmtwo$};

		\node at \med{source.center}{source-down.center}(cuteq-left){\normalsize $\cuteq$};
		\node at \med{source-right.center}{target.center}(cuteq-right){\normalsize $\cuteq$};
		\draw[|->](source) to node[above] {\scriptsize $\ess $} (source-right);
		\draw[|->, dotted](source-down) to node[above] {\scriptsize $\ess $} (target);
	\end{tikzpicture}
\end{tabular}
\end{center}
	
	\item By induction on $\tm\in\snsubst$. One simply uses the fact that $\cuteq$ is a strong bisimulation, given by the previous point, and the \ih
	\item By induction on $k$, using the strong bisimulation property.\qedhere
\end{enumerate}
\end{proof}

%% file: 08-Local_Termination.tex
% !TEX root = main.tex
%%%%%%%%%%%%%
%%%%%%%%%%%%%
\section{Local termination}
\label{sect:local-termination} 
%%%%%%%%%%%%%
%%%%%%%%%%%%%
Here we are going to prove that the ESC has the same local termination property of the LSC, thus showing that one can claim that it realizes the slogan \emph{exponentials as substitutions}. 

\paragraph{Divergence} Let us first show that evaluation in the untyped ESC can diverge. We adapt the usual looping combinator $\Omega$ of the $\l$-calculus. Let $\delta \defeq \la\evar\dera\evar\mvar\suba\mvar\evar\mvartwo\mvartwo$ and consider:
\begin{center}$
\begin{array}{lllllll}
\Omega &\defeq &\cuta\delta\mvarthree \suba\mvarthree{\bang\delta}{\mvarthree'}\mvarthree' 
& \tololli &
\cuta{\bang\delta}\evar \dera\evar\mvar\suba\mvar\evar\mvartwo\cuta\mvartwo{\mvarthree'}\mvarthree'
\\
&&& \toaxeone &
\cuta{\bang\delta}\evar \dera\evar\mvar\suba\mvar{\bang\delta}\mvartwo\cuta\mvartwo{\mvarthree'}\mvarthree'
\\
&&& \tobder &
\cuta{\bang\delta}\evar \cuta\delta\mvar \suba\mvar{\bang\delta}\mvartwo\cuta\mvartwo{\mvarthree'}\mvarthree'
\\
&&& \tow &
\cuta\delta\mvar \suba\mvar{\bang\delta}\mvartwo\cuta\mvartwo{\mvarthree'}\mvarthree'
\\
&&& \toaxmone &
\cuta\delta\mvar \suba\mvar{\bang\delta}\mvartwo\mvartwo
&=_\alpha & \Omega
\end{array}
$\end{center}

\paragraph{Local Termination} A key feature of the untyped ESC (inherited from the LSC) is \emph{local termination}, that is, every rewriting rule is strongly normalizing (SN) when considered separately (not valid in the $\l$-calculus which has only one rule). Additionally, the groups of multiplicative rules $\tom$ and of micro-step exponential rules $\toems$ are also SN separately, which is instead not valid in untyped classical proof nets, where the exponential rule can diverge as shown by the untyped classical proof net in \reffig{classical-diverging} (page \pageref{fig:classical-diverging}).

Even the union of \emph{all} the micro-step rules except $\tololli$, noted $\tomsnlolli$, is SN. This is due to the fact that $\tololli$ is the only rule that \emph{creates} cuts, as explained in \refsect{calculus}. Strong normalization of $\tomsnlolli$ is a deep property, pointed out here for the first time. It shows that, for expressivity, rule $\tololli$ is as important as duplication, and thus, non-linearity. It also shows that $\lolli$ and $\tens$ have completely different computational properties. This is remarkable because in classical linear logic \emph{à la} Girard, the two are collapsed and managed by a single cut elimination rule, the one for $\parr/\tens$, and creations are not visible. Together with (non-)termination of untyped exponentials, it is a sign that the intuitionistic setting is sharper than the classical one for studying cut elimination.

\paragraph{Proving Local Termination.} The proofs of local termination and of the termination of $\tomsnlolli$ require various definitions. In particular, we need an auxiliary notion of \emph{variable potential}, which is used to define a \emph{measure}. The measure is then showed to decrease with every micro-step rewriting rule but $\tololli$. A number of auxiliary lemmas are also required. 

The intuition behind the potential $\potmul\tm\var$ of a variable $\var$ in $\tm$ is that it is the maximum number of occurrences of $\var$ that can appear during the cut elimination of $\tm$, plus the same number for all the variables that are recursively 'nested' under $\var$, where for instance $\var$ nests under $\evar$ in $\dera\evar\var\tm$. The key clause is the one for cuts $\cuta\val\var\tm$, which multiplies the potential in $\val$ for the one of $\var$ in $\tm$.
\begin{defi}[Variable potential]
The potential $\potmul\tm\var$ of a variable $\var$ in $\tm$ is given by:
\begin{center}
\begin{tabular}{c}
\begin{tabular}{cc}
$\begin{array}{rll}
\potmul \vartwo \var & \defeq & 0  
\\
\potmul \var \var & \defeq & 1 
\\\potmul{ \pair\tm\tmtwo}\var & \defeq & \potmul\tm\var + \potmul\tmtwo\var
\\
 \potmul{ \bang \tm}\var& \defeq & \potmul{  \tm}\var
\\
\potmul{ \la\vartwo\tm}\var & \defeq &  \potmul\tm\var

 \end{array}$
 &
 $\begin{array}{rlll}
\potmul{\para\mvar\vartwo\varthree \tm}\var & \defeq & 
 \begin{cases}
		1+ \potmul\tm\vartwo + \potmul\tm\varthree & \mbox{if }\var=\mvar\\
		\potmul{  \tm}\var & \mbox{otherwise.}
		\end{cases}
\\
  \potmul{\suba\mvar\val\vartwo  \tm}\var & \defeq & \
  		\begin{cases}
		1 & \mbox{if }\var=\mvar\\
		\potmul\val\var + \potmul\tm\var & \mbox{otherwise.}
		\end{cases}
 \\
 \potmul{\dera\evar\vartwo \tm}\var & \defeq & 
 		\begin{cases}
		1 + \potmul\tm\var + \potmul\tm\vartwo & \mbox{if }\var=\evar\\
		\potmul\tm\var & \mbox{otherwise.}
		\end{cases}
\end{array}$
\end{tabular}
\\
$\begin{array}{rll@{\hspace{1.3cm}}rll}
 \potmul{\cuta\val\vartwo\tm}\var & \defeq & \potmul\tm\var + \potmul\val\var\cdot(\potmul\tm\vartwo +1)
 \end{array}$
 \end{tabular}
\end{center}
The potential is extended to contexts by considering $\ctxhole$ as a (fresh) exponential variable, that is, having $\potmul\ctxhole\var \defeq 0$, $\potmul\var\ctxhole \defeq 0$, and $\potmul\ctxhole\ctxhole \defeq 1$, plus all the expected inductive cases for contexts.
\end{defi}

In the definition of the potential, the clause for cut multiplies $\potmul\val\var$ for $\potmul\tm\vartwo +1$ where, for exponential variables, the $+1$ accounts for the garbage copy of $\val$ that survives the replacement of all the occurrences of $\var$ with $\val$. Such a $+1$ is not needed in the case of multiplicative cuts, but the clause does not distinguish between cuts on multiplicative or exponential variables, treating all variables as if they were exponential. Such an over-estimate is harmless and provides a more compact definition of the potential.

Note also the clause for subtractions: when $\var=\mvar$, one could define it as $0$, because anyway cuts on subtractions are ignored, given that we shall prove termination of all rules \emph{but} $\tololli$. But it simpler to set it to $1$, as to obtain the useful property of \reflemmap{potmul-prop}{onebis} below.

\begin{defi}[Termination measure]
\label{def:termination-meas}
The termination measure $\ltmeas\tm$ of a term $\tm$ is defined as follows (and extended to contexts by defining it also for $\ctxhole$):
\begin{center}
     \begin{tabular}{cc}
$\begin{array}{rll\colspace \colspace rll}
\ltmeas\var & \defeq & 1
&
\ltmeas\ctxhole & \defeq & 0
\\
\ltmeas{\para\mvar\var\vartwo \tm} & \defeq & \ltmeas{  \tm} +1
&
\ltmeas{ \pair\tm\tmtwo} & \defeq & \ltmeas\tm + \ltmeas\tmtwo
\\
  \ltmeas{\suba\mvar\val\var \tm} & \defeq & \ltmeas\val + \ltmeas\tm +1
&
\ltmeas{ \la\var\tm} & \defeq &  \ltmeas\tm
  \\
 \ltmeas{\dera\evar\var \tm} & \defeq & \ltmeas\tm+1 
 & 
  \ltmeas{ \bang \tm}& \defeq & \ltmeas{  \tm}
\\
 \ltmeas{\cuta{\val}\var\tm} & \defeq & \ltmeas\val\cdot(\potmul\tm\var +1) + \ltmeas\tm
\end{array}$

\end{tabular}
\end{center}
\end{defi}

The idea behind the measure is that it counts all variable occurrences (that is, variables, pars, derelictions, and subtractions), multiplying for the potential of $\var$ the occurrences appearing in a value cut on $\var$, as for the potential. Each rewriting step removes a variable occurrence, which makes the measure decrease. An exponential rewriting step can also duplicate a sub-term and increase the number of occurrences, but this increment is anticipated by the measure by means of the potential.

For proving that the measure decreases, we need some basic properties of the potential and of the measure plus some lemmas about their decomposition with respect to contexts.

\begin{lem}[Basic properties]
\label{l:potmul-prop} % \reflemmaeqp{potmul-prop}{four}
\hfill
\begin{enumerate}
	\item \label{p:potmul-prop-one} if $\var\notin\fv\tm$ then $\potmul\tm\var =0 $.
	\item \label{p:potmul-prop-onebis} if $\var\in\fv\tm$ then $\potmul\tm\var \geq 1$.
	\item \label{p:potmul-prop-two} $\potmul\ctx\ctxhole \geq 1$. 
	\item \label{p:potmul-prop-four} $\ltmeas{\tm}\geq 1$.
	\item \label{p:potmul-prop-five} 
	$\ltmeas{\tm} = \ltmeas{\cutsub\mvartwo\mvar\tm}$.
\end{enumerate}
\end{lem}
\begin{proof}
Points 1, 2, 4, and 5 are easy inductions on $\tm$. Point 3 follows from Point 2, since, when looking at the hole $\ctxhole$ as a variable, one has $\ctxhole\in \fv\ctx$ for every context $\ctx$.
\end{proof}

\begin{lem}[Properties for the $\toaxeone$ and $\toaxmone$ cases]
\label{l:expmeas-is-contextual} % \reflemmaeqp{expmeas-is-contextual}{one}
\hfill
\begin{enumerate}
	\item \label{p:expmeas-is-contextual-one} $\potmul{\ctxfp\tm}\var = \potmul\ctx\var + \potmul\ctx{\ctxhole}\cdot\potmul\tm\var$.
	\item \label{p:expmeas-is-contextual-two} $\ltmeas{\ctxfp\tm} = \ltmeas\ctx + \potmul\ctx{\ctxhole}\cdot\ltmeas\tm$.
\end{enumerate}
\end{lem}

\withproofs{\input{\proofspath/local-termination/expmeas-is-contextual}}

\begin{lem}[Properties for the $\totens$, $\toaxetwo$, $\tobder$ cases]
\label{l:expmeas-auxone} % \reflemmaeqp{expmeas-auxone}{two}
\hfill
\begin{enumerate}
		\item \label{p:expmeas-auxone-two} If $\ctx$ does not capture variables in $\fv\lctx$ and $\var\notin\fv\lctx$ then $\potmul{\ctxp{\lctxp\tm}}\var = \potmul{\ctxp{ \tm}}\var$.	
	\item %\label{p:expmeas-auxone-one} 
	If $\ctx$ does not capture $\evar$ and $\evartwo$ then: 
	\begin{enumerate}
		\item \label{p:expmeas-auxone-der-pm} \emph{Potential}: 
		$\potmul{\ctxp{\dera\evar\var \tm}}\evar = \potmul{\ctxp{ \tm}}\evar + (1+\potmul\tm\var)\cdot\potmul\ctx\ctxhole$ and $\potmul{\ctxp{\dera\evartwo\var \tm}}\evar < \potmul{\ctxp{\dera\evar\var \tm}}\evar$.	
		\item \label{p:expmeas-auxone-der-meas} \emph{Measure}: 
		$\ltmeas{\ctxp{\dera\evar\var \tm}}>\ltmeas{\ctxp{\tm}}$ and $\ltmeas{\ctxp{\dera\evar\var \tm}} = \ltmeas{\ctxp{\dera\evartwo\var \tm}}$.
	\end{enumerate}
	\item %\label{p:expmeas-auxone-par} 
	If $\mctx$ does not capture $\mvar$ then: 
	\begin{enumerate}
		\item \label{p:expmeas-auxone-par-pm} \emph{Potential}: 
		$\potmul{\mctxp{\para\mvar\var\vartwo \tm}}\mvar = (1+ \potmul\tm\var + \potmul\tm\vartwo)\cdot\potmul\mctx\ctxhole$.	
		\item \label{p:expmeas-auxone-par-meas} \emph{Measure}: 
		$\ltmeas{\mctxp{\para\mvar\var\vartwo \tm}}> \ltmeas{\mctxp{\tm}}$.	
	\end{enumerate}
\end{enumerate}
\end{lem}

\withproofs{\input{proofs/local-termination/aux-properties-one}}

The fourth point of the next lemma shall be used in the proof that the measure decreases for the the $\totens$ and $\tobder$ cases. The first three point are intermediate properties to obtain the fourth one.

\begin{lem}[Inequalities for the $\totens$ and $\tobder$ cases]
\label{l:expmeas-auxtwo} % \reflemmaeqp{expmeas-auxtwo}{four}
Let $\ctx$ be a context that does not capture free variables in $\lctx$ nor $\val$, and $\lctx$ does not capture variables of $\tm$.
\begin{enumerate}
	\item \label{p:expmeas-auxtwo-one} $\potmul{\lctxp{\cuta\val\var\tm}}\vartwo \leq \potmul{\tm}\vartwo +\potmul{\lctxp\val}\vartwo \cdot(1+\potmul\tm\var)$.		
	\item \label{p:expmeas-auxtwo-two} $\ltmeas{\lctxp{\cuta\val\var\tm}} \leq \ltmeas{\tm} +\ltmeas{\lctxp\val} \cdot(1+\potmul\tm\var)$.
		\item \label{p:expmeas-auxtwo-three} $\potmul{\ctxp{\lctxp{\cuta\val\var\tm}}}\vartwo \leq \potmul{\ctxp{\tm}}\vartwo +\potmul{\lctxp\val}\vartwo \cdot(1+\potmul\tm\var)\cdot\potmul\ctx\ctxhole$.

	\item \label{p:expmeas-auxtwo-four} $\ltmeas{\ctxp{\lctxp{\cuta\val\var\tm}}} \leq \ltmeas{\ctxp{\tm}} +\ltmeas{\lctxp\val} \cdot(1+\potmul\tm\var)\cdot\potmul\ctx\ctxhole$.
\end{enumerate}
\end{lem}

\withproofs{\input{proofs/local-termination/aux-properties-two}}

\begin{prop}[Measure decreases]
\label{prop:meas-decreases}
If $\tm \Rew{a} \tmtwo$ with $a\in\set{\axmone,\axmtwo,\tens,\axeone,\axetwo,\bdersym,\wsym}$ then $\ltmeas\tm > \ltmeas\tmtwo$.
\end{prop}

\withproofs{
\input{\proofspath/local-termination/measure-decreases}
}

\begin{thm}[Local termination]
\label{thm:local-termination}
Let $a \in \set{
\msym, \ems,\ess,\msnlolli, \ssnlolli}$. Then $\Rew{a}$ is strongly normalizing.
\end{thm}

\begin{proof}
For $\toems$ and $\tomsnlolli$ it is an immediate consequence of the fact that the measure decreases (\refprop{meas-decreases}) and that it is always positive (\reflemmap{potmul-prop}{four}). The result is extended to $\toess$ and $\tossnlolli$ by full composition (\refprop{ms-simulates-ss}). For $\tom$, the statement is obvious, as at each step $\tom$ decreases the number of constructors.
\end{proof}

The literature contains other partitions of linear logic cut elimination in strongly normalizing reductions: Danos \cite{Danos:Thesis:90} and Joinet \cite{joinetthesis} split it into axiom and non-axiom rewriting rules, Pagani and Tortora de Falco \cite{DBLP:journals/tcs/PaganiF10} into \emph{structural} (a subset of exponential) and \emph{logical} (\ie non-structural). We are interested in the lolli/non-lolli partition because it mimics what happens in $\l$-calculi with ESs, and also because of the future work described in \emph{the next step} paragraph of the conclusions, for which the mentioned alternative partitions in the literature would not work.

\paragraph{Next} Three independent topics follow: confluence (\refsect{confluence}), the good strategy (\refsect{strategy}), and strong normalization  (Sections \ref{sect:PSN} and \ref{sect:SN}). They can be read in any order. First, however, we need to develop some technical tools for dealing with contexts in proofs.

%% file: proofs/local-termination/expmeas-is-contextual.tex
% !TEX root = ../../main.tex
\begin{proof}
Both points are by induction on $\ctx$, and the second one uses the first one in the cut case. Details are in \cite{DBLP:journals/corr/abs-lics}.
\end{proof}

%% file: proofs/local-termination/aux-properties-one.tex
% !TEX root = ../../main.tex
\begin{proof}
All points are by induction on the context of the statement. The only non-trivial case, for each point, is the one for cut, which is treated in \cite{DBLP:journals/corr/abs-lics}. The others follow immediately from the \ih
\end{proof}

%% file: proofs/local-termination/aux-properties-two.tex
% !TEX root = ../../main.tex
\begin{proof}
Point 1 is by induction on $\lctx$, point 2 uses point 1, point 3 and 4 are by induction on $\ctx$ using point 1 and 2 in the base case. Details are in \cite{DBLP:journals/corr/abs-lics}.
\end{proof}

%% file: proofs/local-termination/measure-decreases.tex
% !TEX root = ../../main.tex
\begin{proof}
By induction on $\tm \Rew{a} \tmtwo$. All the inductive cases follow immediately from the \ih Root cases:
\begin{itemize}

	%%%%%%%%%%%%%%%
	\item $\rtoaxmone$, that is, $\cuta\mval \mvar \mctxfp{\mvar} \rtoaxmone  \mctxfp\mval$. Then:
\begin{center}$\begin{array}{llllllll}
	\ltmeas{\cuta\mval \mvar \mctxfp{\mvar}}
	 = &
	\ltmeas\mval\cdot(\potmul{\mctxfp{\mvar}}\mvar +1) + \ltmeas{\mctxfp{\mvar}}
	 & =_{\reflemmaeqp{expmeas-is-contextual}{one}} \\
	&\ltmeas\mval\cdot(\underbrace{\potmul\mctx\mvar}_{=_{\reflemmaeqp{potmul-prop}{one}}0} + \potmul\mctx{\ctxhole}\cdot\underbrace{\potmul{\mvar}\mvar}_{=1} +1) + \ltmeas{\mctxfp{\mvar}}
	 & = \\
	&\ltmeas\mval\cdot(\potmul\mctx{\ctxhole} +1) + \ltmeas{\mctxfp{\mvar}}
	 & =_{\reflemmaeqp{expmeas-is-contextual}{two}} \\
	&\underbrace{\ltmeas\mval}_{=_{\reflemmaeqp{potmul-prop}{four}}1}\cdot(\potmul\mctx{\ctxhole} +1) + 	
	\ltmeas\mctx + \underbrace{\potmul\mctx{\ctxhole}}_{\geq_{\reflemmaeqp{potmul-prop}{two}}1}\cdot\ltmeas{\mvar}
		 & > \\
		 	&\ltmeas\mval\cdot\potmul\mctx{\ctxhole} + \ltmeas\mctx
	 & =_{\reflemmaeqp{expmeas-is-contextual}{two}} \\
&	 	\ltmeas{ \mctxfp\mval}
	\end{array}$\end{center}
	
\item $\rtoaxmtwo$, that is, $\cuta\mvartwo \mvar \tmthree \rtoaxmtwo  \cutsub\mvartwo \mvar \tmthree$. Then:
\begin{center}$\begin{array}{lllllllll}
	\ltmeas{\cuta\mvartwo \mvar \tmthree}
	& = &
	\underbrace{\ltmeas\mvartwo}_{=1}\cdot(\potmul{\tmthree}\mvar +1) + \ltmeas{\tmthree}
	& = \\
&&	\potmul{\tmthree}\mvar +1 + \ltmeas{\tmthree}
	& > &
	\ltmeas{\tmthree}
	 & =_{\reflemmaeqp{potmul-prop}{five}} &
	\ltmeas{\cutsub\mvartwo \mvar\tmthree}
	\end{array}$\end{center}

	%%%%%%%%%%%%%%%%
	\item $\rtotens$, that is, $\cuta{\pair\tmthree\tmfour}\mvar \mctxp{\para\mvar\var\vartwo \tmfive}
 \rtotens  
\mctxp{\lctxp{\cuta\val\var \lctxtwop{\cuta\valtwo\vartwo \tmfive}}}$
with $\tmthree = \lctxp{\val}$ and $\tmfour = \lctxtwop{\valtwo}$. Then:
\begin{center}\footnotesize
\renewcommand{\arraystretch}{1.1}
$\begin{array}{lllll}
	\ltmeas{\cuta{\pair\tmthree\tmfour}\mvar \mctxp{\para\mvar\var\vartwo \tmfive}}
		& = \\
	\ltmeas{\mctxp{\para\mvar\var\vartwo \tmfive}} + \ltmeas{\pair\tmthree\tmfour}\cdot(\potmul{\mctxp{\para\mvar\var\vartwo \tmfive}}\mvar +1) 
		& >_{\reflemmaeqp{expmeas-auxone}{par-meas}} \\
	\ltmeas{\mctxp{\tmfive}} + \ltmeas{\pair\tmthree\tmfour}\cdot(\potmul{\mctxp{\para\mvar\var\vartwo \tmfive}}\mvar +1)
		& = \\
		\ltmeas{\mctxp{\tmfive}} + (\ltmeas\tmthree + \ltmeas\tmfour)\cdot(\potmul{\mctxp{\para\mvar\var\vartwo \tmfive}}\mvar +1)
		& =_{\reflemmaeqp{expmeas-auxone}{par-pm}} \\
		\ltmeas{\mctxp{\tmfive}} + (\underbrace{\ltmeas\tmthree + \ltmeas\tmfour}_{\geq_{\reflemmaeqp{potmul-prop}{four}} 1})\cdot((1+ \potmul\tmfive\var + \potmul\tmfive\vartwo)\cdot\potmul\mctx\ctxhole +1)
		& > \\
		\ltmeas{\mctxp{\tmfive}} + (\ltmeas\tmthree + \ltmeas\tmfour)\cdot(1+ \potmul\tmfive\var + \potmul\tmfive\vartwo)\cdot\potmul\mctx\ctxhole
		& = \\
		\ltmeas{\mctxp{\tmfive}} + \ltmeas\tmthree\cdot(1+ \potmul\tmfive\var + \potmul\tmfive\vartwo)\cdot\potmul\mctx\ctxhole + \ltmeas\tmfour\cdot(1+ \potmul\tmfive\var + \potmul\tmfive\vartwo)\cdot\potmul\mctx\ctxhole
		& \geq \\
		\ltmeas{\mctxp{\tmfive}} + \ltmeas\tmthree\cdot(1+  \potmul\tmfive\vartwo)\cdot\potmul\mctx\ctxhole + \ltmeas\tmfour\cdot(1+ \potmul\tmfive\var )\cdot\potmul\mctx\ctxhole
				& = \\
		\underbrace{\ltmeas{\mctxp{\tmfive}} +\ltmeas{\lctxtwop\valtwo} \cdot(1+\potmul\tmfive\vartwo)\cdot\potmul\mctx\ctxhole}_{=_{\reflemmaeqp{expmeas-auxtwo}{four}} \ltmeas{\mctxp{\lctxtwop{\cuta\valtwo\vartwo \tmfive}}}}
		+\ltmeas{\lctxp\val} \cdot(1+\potmul{ \tmfive}\var)\cdot\potmul\mctx\ctxhole
		& = \\		\ltmeas{\mctxp{\lctxtwop{\cuta\valtwo\vartwo \tmfive}}} +\ltmeas{\lctxp\val} \cdot(1+\potmul{\tmfive}\var)\cdot\potmul\mctx\ctxhole
		& =_{\reflemmaeqp{expmeas-auxone}{two}}  \\		\ltmeas{\mctxp{\lctxtwop{\cuta\valtwo\vartwo \tmfive}}} +\ltmeas{\lctxp\val} \cdot(1+\potmul{\lctxtwop{\cuta\valtwo\vartwo \tmfive}}\var)\cdot\potmul\mctx\ctxhole
		& =_{\reflemmaeqp{expmeas-auxtwo}{four}} \\
	\ltmeas{\mctxp{\lctxp{\cuta\val\var \lctxtwop{\cuta\valtwo\vartwo \tmfive}}}}
	\end{array}$\end{center}
		
	%%%%%%%%%%%%%%%%%%
	\item $\rtoaxeone$, that is, $\cuta\exval \evar \ctxfp{\evar} \rtoaxeone  \cuta\exval\evar\ctxfp\exval$. 
	\begin{center}$\begin{array}{lllll}
	\ltmeas{\cuta\exval \evar \ctxfp{\evar}}
	& = \\
	\ltmeas\exval\cdot(\potmul{\ctxfp{\evar}}\evar +1) + \ltmeas{\ctxfp{\evar}}
	 & =_{\reflemmaeqp{expmeas-is-contextual}{one}} \\
	\ltmeas\exval\cdot(\potmul\ctx\evar + \potmul\ctx{\ctxhole}\cdot\underbrace{\potmul{\evar}\evar}_{=1} +1) + \ltmeas{\ctxfp{\evar}}
	 & = \\
	\ltmeas\exval\cdot(\potmul\ctx\evar + \potmul\ctx{\ctxhole} +1) + \ltmeas{\ctxfp{\evar}}
	 & = \\
	\ltmeas\exval\cdot(\potmul\ctx\evar  +1) + \ltmeas\exval\cdot\potmul\ctx{\ctxhole} + \ltmeas{\ctxfp{\evar}} & =_{\reflemmaeqp{expmeas-is-contextual}{two}} \\
	\ltmeas\exval\cdot(\potmul\ctx\evar  +1) + \ltmeas\exval\cdot\potmul\ctx{\ctxhole} +  	
	\ltmeas\ctx + \potmul\ctx{\ctxhole}\cdot\underbrace{\ltmeas{\evar}}_{=1}
	 & = \\
	\ltmeas\exval\cdot(\potmul\ctx\evar  +1) + \underbrace{\ltmeas\exval\cdot\potmul\ctx{\ctxhole} +  	
	\ltmeas\ctx}_{=_{\reflemmaeqp{expmeas-is-contextual}{two}} \ltmeas{\ctxp\exval}} + \potmul\ctx{\ctxhole}
		& =\\
	\ltmeas\exval\cdot(\potmul\ctx\evar  +1) + \ltmeas{\ctxp\exval} + \underbrace{\potmul\ctx{\ctxhole}}_{>_{\reflemmaeqp{potmul-prop}{two}} 0}
	& >\\
	\ltmeas\exval\cdot(\potmul\ctx\evar  +1) + \ltmeas{\ctxp\exval}
	& =_{\reflemmaeqp{potmul-prop}{one}} \\
%	\end{array}$\end{center}
%	
%	\begin{center}$\begin{array}{lllll}

		\ltmeas\exval\cdot(\potmul\ctx\evar + \potmul\ctx{\ctxhole}\cdot\underbrace{\potmul\exval\evar}_{=_{\reflemmaeqp{potmul-prop}{one}}0} +1) + \ltmeas{\ctxfp\exval}
	 & =_{\reflemmaeqp{expmeas-is-contextual}{one}} \\
\ltmeas\exval\cdot(\potmul{\ctxfp\exval}\evar +1) + \ltmeas{\ctxfp\exval}
	& = \\
		\ltmeas{\cuta\exval \evar \ctxfp\exval}
	\end{array}$\end{center}

%%%%%%%%%%%%
\item $\rtoaxetwo$, that is, $ \cuta{\evartwo}\evar\ctxp{\dera\evar\var \tm} \rtoaxetwo \cuta{\evartwo}\evar\ctxp{\dera\evartwo\var \tm}$. Then:
\begin{center}$\begin{array}{lllllll}
	\ltmeas{\cuta{\evartwo} \evar \ctxp{\dera\evar\var \tm}}
	& = &
	\ltmeas{\evartwo}\cdot(\potmul{\ctxp{\dera\evar\var \tm}}\evar +1) + \ltmeas{\ctxp{\dera\evar\var \tm}}
	 &>_{\reflemmaeqp{expmeas-auxone}{der-pm}} \\
		&&\ltmeas{\evartwo}\cdot(\potmul{\ctxp{\dera\evartwo\var \tm}}\evar +1) + \ltmeas{\ctxp{\dera\evar\var \tm}}
	 &=_{\reflemmaeqp{expmeas-auxone}{der-meas}} \\
		&&\ltmeas{\evartwo}\cdot(\potmul{\ctxp{\dera\evartwo\var \tm}}\evar +1) + \ltmeas{\ctxp{\dera\evartwo\var \tm}}
		 & =\\
		&&\ltmeas{\cuta{\evartwo} \evar \ctxp{\dera\evartwo\var \tm}}
	\end{array}$\end{center}

%%%%%%%%
\item $\rtobder$, that is, $\cuta{\bang\tmtwo}\evar\ctxp{\dera\evar\var \tm} \rtobder  \cuta{\bang\tmtwo}\evar\ctxp{\lctxp{\cuta\val\var \tm}}$ with $\tmtwo = \lctxp\val$. Then:
\begin{center}$\begin{array}{lllll}
	\ltmeas{\cuta{\bang\tmtwo} \evar \ctxp{\dera\evar\var \tm}}
	& =  \\
	
	\ltmeas{\bang\tmtwo}\cdot(\potmul{\ctxp{\dera\evar\var \tm}}\evar +1) + \ltmeas{\ctxp{\dera\evar\var \tm}}
	 & >_{\reflemmaeqp{expmeas-auxone}{der-meas}}  \\
	\ltmeas{\bang\tmtwo}\cdot(\potmul{\ctxp{\dera\evar\var \tm}}\evar +1) + \ltmeas{\ctxp{ \tm}}
	 & =_{\reflemmaeqp{expmeas-auxone}{der-pm}}  \\
	\ltmeas{\bang\tmtwo}\cdot(\potmul{\ctxp{ \tm}}\evar + (1+\potmul\tm\var)\cdot\potmul\ctx\ctxhole +1) + \ltmeas{\ctxp{ \tm}}
	 & =\\
	\ltmeas{\bang\tmtwo}\cdot(\potmul{\ctxp{ \tm}}\evar + 1) + \ltmeas{\bang\tmtwo}\cdot(1+\potmul\tm\var)\cdot\potmul\ctx\ctxhole + \ltmeas{\ctxp{ \tm}}
	 & =\\
	\ltmeas{\bang\tmtwo}\cdot(\potmul{\ctxp{ \tm}}\evar + 1) + \ltmeas{\tmtwo}\cdot(1+\potmul\tm\var)\cdot\potmul\ctx\ctxhole + \ltmeas{\ctxp{ \tm}}
	 & =\\
	\ltmeas{\bang\tmtwo}\cdot(\potmul{\ctxp{ \tm}}\evar + 1) + \underbrace{\ltmeas{\lctxp\val}\cdot(1+\potmul\tm\var)\cdot\potmul\ctx\ctxhole + \ltmeas{\ctxp{ \tm}}}_{=_{\reflemmaeqp{expmeas-auxtwo}{four}} \, \ltmeas{\ctxp{\lctxp{\cuta\val\var \tm}}}}

		 & =\\
	\ltmeas{\bang\tmtwo}\cdot(\potmul{\ctxp{\tm}}\evar +1) + \ltmeas{\ctxp{\lctxp{\cuta\val\var \tm}}}
		 & =_{\reflemmaeqp{expmeas-auxone}{two}}\\
	\ltmeas{\bang\tmtwo}\cdot(\potmul{\ctxp{\lctxp{\cuta\val\var \tm}}}\evar +1) + \ltmeas{\ctxp{\lctxp{\cuta\val\var \tm}}}
		 & =\\
		\ltmeas{\cuta{\bang\tmtwo}\evar\ctxp{\lctxp{\cuta\val\var \tm}}}
	\end{array}$\end{center}
	%%%%%%%%%%%%
\item $\rtow$, that is, $\cuta{\exval}\evar \tm \rtow  \tm$ with $\evar\notin\fv\tm$. Then:
\begin{center}$
	\ltmeas{\cuta{\exval} \evar \tm}
	\ = \
	\ltmeas{\exval}\cdot(\potmul{\tm}\evar +1) + \ltmeas{\tm}
	\  =_{\reflemmaeqp{potmul-prop}{one}} \
	 \ltmeas{\exval} + \ltmeas{\tm}
	 \ >_{\reflemmaeqp{potmul-prop}{four}} \   
	\ltmeas{\tm}. \hfill\qedhere
	$\end{center}
\end{itemize}
\end{proof}

%% file: 09-Some_Technicalities.tex
% !TEX root = main.tex
\section{Some Technicalities}
\label{sect:technicalities}
This section develops definitions and technical tools that are used in the proofs of the following sections to manage contexts, namely the outside-in order on contexts, 
contexts with \emph{two} holes (called \emph{double contexts}), and a technical but important \emph{deformation lemma}. At a first (or even second) reading, the reader can skip this section.

\begin{defi}[Outside-in context order, disjoint contexts]
We define the partial \emph{outside-in} order $\outer$ over contexts as follows:
\begin{center}
\begin{tabular}{cccccccc}
\AxiomC{}
 	\UnaryInfC{$\ctxhole \outer \ctx$} 	
	\DisplayProof	
&&&
\AxiomC{$\ctx \outer\ctxtwo$}
 	\UnaryInfC{$\ctxthreep\ctx \outer \ctxthreep\ctxtwo$} 	
	\DisplayProof	
\end{tabular}
\end{center}
And say that $\ctx$ is \emph{outer} than $\ctxtwo$ if $\ctx\outer \ctxtwo$. If $\ctx\not\outer \ctxtwo$ and 
$\ctxtwo\not\outer \ctx$ we say that $\ctx$ and $\ctxtwo$ are \emph{disjoint}, and write $\ctx\parallel \ctxtwo$.
\end{defi}

\paragraph{Double Contexts.} Double contexts shall be used to compare two contexts on the same term. They have as base 
cases binary constructors (that is, tensor pairs, cuts, and subtractions) having contexts replacing their subterms, and 
as inductive cases they are simply closed by an ordinary context. 
\begin{defi}[Double contexts]
Double contexts $\cctx$ are defined by the following grammar.
\[\begin{array}{l\colspace\colspace llll}
\textsc{Double contexts} & \cctx& \defeq & \pair\ctx\ctxtwo \mid \suba\mvar\ctx\var\ctxtwo \mid \cuta\ctx\var\ctxtwo 
\mid \ctxp\cctx
\end{array}\]
\end{defi}
Some easy facts about double contexts. 
\begin{itemize}
\item \emph{Plugging}: the plugging $\cctxp{\tm,\tmtwo}$ of two terms $\tm$ and $\tmtwo$ into a double context $\cctx$ 
is defined as expected and gives a term. The two ways of plugging one term $\cctxp{\tm,\ctxhole}$ and 
$\cctxp{\ctxhole,\tmtwo}$ into a double context give instead a context. We often use $\cctxp{\tm,\cdot}$ and 
$\cctxp{\cdot,\tmtwo}$ as compact notations for $\cctxp{\tm,\ctxhole}$ and $\cctxp{\ctxhole,\tmtwo}$.

\item \emph{Pairs of disjoint positions and double contexts}: every pair of positions $\ctxp\tm=\ctxtwop\tmtwo$ which 
are disjoint, that is, such that $\ctx \parallel\ctxtwo$, gives rise to a double context $\cctx_{\ctx,\ctxtwo}$ such 
that $\cctx_{\ctx,\ctxtwo}\ctxholep{\cdot,\tmtwo} = \ctx$ and $\cctx_{\ctx,\ctxtwo}\ctxholep{\tm,\cdot} = \ctxtwo$. 
Conversely, every double context $\cctx$ and every pair of terms $\tm$ and $\tmtwo$ gives rise to two context 
$\cctx_{l}^{\tm,\tmtwo} \defeq \cctxp{\cdot, \tmtwo}$ and $\ctx_{r}^{\tm,\tmtwo} \defeq \cctxp{\tm,\cdot}$. Note that 
the two constructions are inverse of each other, that is, given $\cctx$, we have 
$\cctx_{\cctx_{l}^{\tm,\tmtwo},\ctx_{r}^{\tm,\tmtwo} } = \cctx$, and, conversely, given $\ctxp\tm=\ctxtwop\tmtwo$ we 
have $(\cctx_{\ctx,\ctxtwo})_{l}^{\tm,\tmtwo} = \ctx$ and $(\cctx_{\ctx,\ctxtwo})_{r}^{\tm,\tmtwo} = \ctxtwo$.
\end{itemize}

\begin{lem}
\label{l:double-ctx-comp}
$\ctxp{\cctxp{\tm,\cdot}} = \ctxp{\cctx}\ctxholep{\tm,\cdot}$ and $\ctxp{\cctxp{\cdot,\tm}} = 
\ctxp{\cctx}\ctxholep{\cdot,\tm}$.
\end{lem}

\begin{proof}
By induction on $\ctx$.
\end{proof}

\paragraph{Deformation Lemma} We need a set of \emph{deformation lemmas}, collected together in 
\reflemma{aux-props-loc-confl-strong-comm}, that say what happens to a variable occurrence under the action 
of a rewriting step: for instance, given a step $\mctxfp\mvar \toms \tmtwo$, one lemma states that $\tmtwo$ has shape $\mctxtwofp\mvar$ unless the step reduces a cut containing $\mvar$. The idea is that then $\mctxtwo$ is the residual of $\mctx$, but the statement(s) avoid using any form of residuals of steps, sub-terms, or constructors. The residual aspect is captured by specifying the preservation modulo a parameter---in 
the case mentioned above, this is obtained stating that there is a step $\mctxfp\mval \Rew{a} 
\mctxtwofp\mval$ for every $\mval$.

The deformation lemmas shall  be used in the following sections, namely in the proofs of local confluence (where preservation modulo a parameter is exploited), for the diamond property of the good strategy (for which the deformation lemmas shall have to be strengthened), and for the root cut expansion property (\refprop{root-sn-expansion}) at work in both the preservation of untyped strong normalization and typed normalization. 
%The \emph{moreover} part of each statement, however, is used only here for the good diamond, while the case of $\to = \toss$ shall be used in the proof of root cut expansion and only for some of the deformation lemmas.

\begin{lem}[Deformation properties]
\label{l:aux-props-loc-confl-strong-comm} % \reflemmap{aux-props-loc-confl-strong-comm}{five}
Let $\to\in\set{\toss,\toms}$. 
\begin{enumerate}
%%%%%%
 %%%%%%%%%%
 \item \label{p:aux-props-loc-confl-strong-comm-two}
Let $\ctx: \tm = \mctxfp\mvar \Rew{a} \tmtwo$ and $\Rew{a}$ be a $\to$ step. Then:\begin{enumerate}
\item \label{p:aux-props-loc-confl-strong-comm-two-one}
either $\tmtwo = \mctxtwofp{\mvar}$ for some 
$\mctxtwo$ such that $\mvar\notin\fv\mctxtwo$ and there is a step $\ctxtwo_{\mval}: \mctxfp\mval \Rew{a} 
\mctxtwofp\mval$ for every $\mval$,
\item  or $\mctx = \mctxtwop{\cuta\ctxhole\mvartwo\tmthree}$ for some 
$\mctxtwo$ and $\ctx$ reduces the cut on $\mvartwo$.
\end{enumerate}
%%%%%%%%%
\item \label{p:aux-props-loc-confl-strong-comm-axmtwo}
Let $\ctx: \tm = \mctxfp{\tm_\mvar} \Rew{a} \tmtwo$ where $\tm_\mvar$ is a par or a subtraction of conclusion $\mvar$, 
$\mctxtwo$ does not capture $\mvartwo$, and $\Rew{a}$ be a $\to$ step. Then  $\tmtwo = \mctxtwofp{\tmtwo_\mvar}$ where 
$\tmtwo_\mvar$ is an occurrence of $\mvar$ of the same kind of $\tm_\mvar$,  
$\mctxtwo$ does not capture $\mvar$ and $\mvar\notin\fv\mctxtwo$, and there is a step $\ctxtwo_{\mvartwo}: 
\mctxfp{\cutsub\mvartwo\mvar \tm_\mvar} \Rew{a} \mctxtwofp{\cutsub\mvartwo\mvar \tmtwo_\mvar}$ for every fresh 
multiplicative variable $\mvartwo$.
%%%%%%%%%
 \item \label{p:aux-props-loc-confl-strong-comm-three} If $\ctx: \mctxp{\para\mvar\var\vartwo \tm} \Rew{a} \tmtwo$ is a 
$\to$-step and $\mctx$ does not 
capture $\mvar$ then\\ $\tmtwo = \mctxtwop{\para\mvar\var\vartwo \tm'}$ with $\mctxtwo$ not capturing $\mvar$ and 
such that there is a step $\ctxtwo_{\ctxthree}: \mctxp{\ctxthreep\tm} \Rew{a} \mctxtwop{\ctxthreep{\tm'}}$ for every 
context $\ctxthree$ capturing no more than $\var$ and $\vartwo$ in $\tm$ and such that $\mctx$ does not capture 
variables of $\ctxthree$.
 %%%%%%%%%%%
 \item \label{p:aux-props-loc-confl-strong-comm-four} If $\ctx: \mctxp{\suba\mvar\val\var \tmfour} \Rew{a} \tmtwo$ is a 
$\to$-step and $\mctx$ does not 
capture $\mvar$ then either 
\begin{itemize}
\item $\tmtwo = \mctxtwop{\suba\mvar\val\var \tmfour'}$ with 
$\ctxtwo_{\ctxthree}:\mctxp{\ctxthreep\tmfour} \Rew{a} 
\mctxtwop{\ctxthreep{\tmfour'}}$ for every context $\ctxthree$ capturing no more than $\var$ in $\tmfour$  and such 
that 
$\mctx$ does not capture in $\ctxthree$ more variables than those in $\fv\val$, or 
\item $\tmtwo = \mctxtwop{\suba\mvar\valtwo\var \tmfour}$ with $\ctxtwo_{\ctxthree}:\mctxp{\ctxthreep\val} \Rew{a} 
\mctxtwop{\ctxthreep{\valtwo}}$ for every context $\ctxthree$ not capturing variables of $\val$ and such that $\mctx$ 
does not capture in $\ctxthree$ more variables than those in $\fv\tmfour \setminus\set{\var}$.
\end{itemize}
In both cases $\mctxtwo$ does not capture $\mvar$.
%%%%%%%
\item \label{p:aux-props-loc-confl-strong-comm-five}
Let $\ctx: \tm = \ctxtwofp\evar \Rew{a} \tmtwo$ and $\Rew{a}$  be a $\toms$ step.  Then:
\begin{enumerate}
\item \label{p:aux-props-loc-confl-strong-comm-five-a}
either $\tmtwo = \ctxtwo'\ctxholefp{\evar}$ for some 
$\ctxthree$ such that there is a step $\ctx^\bullet_{\exval}:\ctxtwofp\exval \Rew{a} \ctxtwo'\ctxholefp\exval$ for 
every 
$\exval$,
\item  or $\ctxtwo = \ctxthreep{\cuta\ctxfour\evartwo\tmthree}$ for some 
$\ctxthree$ and $\ctxfour$, and $\ctx$ reduces a redex where the  active cut is the one on $\evartwo$.
\end{enumerate}
%%%%%%%%%
\item \label{p:aux-props-loc-confl-strong-comm-six}
Let $\ctx: \tm = \ctxtwop{\dera\evar\var\tmthree} \Rew{a} \tmtwo$ with $\ctxtwo$ not capturing $\evar$ and $\Rew{a}$  
be 
a $\toms$ step.  Then:
\begin{enumerate}
\item \label{p:aux-props-loc-confl-strong-comm-six-a}
either $\tmtwo = \ctxtwo'\ctxholep{\dera\evar\var\tmthree'}$ for some 
$\ctxtwo'$ such that there is a step $\ctx^\bullet_{\ctxthree}:\ctxtwop{\ctxthreep{\tmthree}} \Rew{a} 
\ctxtwo'\ctxholep{\ctxthreep{\tmthree'}}$ for every $\ctxthree$ capturing no more than $\var$ in $\tmthree$ and 
$\tmthree'$,
\item  or $\ctxtwo = \ctxthreep{\cuta\ctxfour\evartwo\tmthree}$ for some $\ctxthree$ and
$\ctxfour$, and $\ctx$ reduces a redex given by the cut on $\evartwo$.
\end{enumerate}
\end{enumerate}
\end{lem}

\begin{proof}
The proof is by induction on the rewriting step in each point, and is a tedious check. In \cite{DBLP:journals/corr/abs-lics}, we give all the (many!) details for the first deformation property. The proofs of the others are minor variations.
\end{proof}

% \withproofs{
% \input{./\proofspath/strategy/deformations}
% }

%% file: 10-Untyped_Confluence.tex
% !TEX root = main.tex

\section{Untyped Confluence}
\label{sect:confluence}
Here we prove confluence for the untyped ESC, using an elegant technique based on local diagrams and local termination. The technique is the Hindley-Rosen method. In our case, it amounts to prove that the multiplicative and exponential rules $\tom$ and $\toems$ are confluent separately, proved by local termination and Newman lemma, and commute, proved by local termination and Hindley's strong commutation. Confluence then follows by Hindley-Rosen lemma, for which the union of confluent and commuting reductions is confluent.

The Hindley-Rosen method is a modular technique often used for confluence of extensions of the $\l$-calculus, for instance in \cite{ArrighiD17,DBLP:conf/lics/FaggianR19,DBLP:conf/csl/Saurin08,DBLP:conf/fossacs/CarraroG14,Revesz92,DBLP:journals/lmcs/BucciarelliKR21,AriolaFMOW95,DBLP:conf/flops/AccattoliP12}. It is also used for untyped proof nets by Pagani and Tortora de Falco \cite{DBLP:journals/tcs/PaganiF10}. %Usually, one pairs a possibly divergent reduction ($\beta$) with a strongly normalizing one (the extension). What is slightly unusual here is that we pair two strongly normalizing reductions ($\tom$ and $\toems/\toess$), the union of which is possibly divergent. Similar uses of the method are in \cite{AriolaFMOW95,DBLP:journals/tcs/PaganiF10,DBLP:conf/flops/AccattoliP12}. 
Confluence for untyped proof nets is also proved by Danos \cite{Danos:Thesis:90} and Tranquilli \cite{DBLP:conf/csl/Tranquilli09} via finite developments. Essentially, local termination internalizes finite developments.

\paragraph{The Glitch} For the untyped ESC there is a slight flaw due to clashes. The following local diagram, indeed, can be closed only if $\cuta\val\mvar$ is not a clashing cut (precisely, when $\val$ is not an exponential value nor an abstraction), or with cut equivalence $\cuteq$:
		\begin{center}
\begin{tikzpicture}[ocenter]
		\node at (0,0)[align = center](source){\normalsize$\cuta{\val}\mvar \mctxtwop{\cuta\mvar\mvartwo \mctxp{\para\mvartwo\var\vartwo \tm} }$};
		\node at (source.center)[right = 120pt](source-right){\normalsize $\mctxtwop{\cuta\val\mvartwo \mctxp{\para\mvartwo\var\vartwo \tm}}$};	
		\node at (source.center)[below = 25pt](source-down){\normalsize $\cuta{\val}\mvar \mctxtwop{\mctxp{\para\mvar\var\vartwo \tm}}  $};

		\draw[->](source) to node[above] {\scriptsize $\axmone $} (source-right);
		\draw[->](source) to node[left] {\scriptsize $\axmtwo $}(source-down);
	\end{tikzpicture}
\end{center}
A second similarly problematic diagram is obtained by replacing the par with a subtraction.

Now, if there is a clash, the fact that confluence holds only up
to $\cuteq$ is irrelevant: the clash is a bigger issue. Moreover:
\begin{enumerate}
\item \emph{(Recursive) types remove clashes:} clashes are ruled out by our typing (\reflemma{clashfree-implies-cutfree}), but also by recursive types such as those used for typing the untyped $\l$-calculus, or even the more general weak recursive typing considered by Ehrhard and Regnier in \cite{DBLP:journals/tcs/EhrhardR06}. 
\item \emph{Proof nets are not better}: because of weakenings, IMELL proof nets need \emph{jumps} and jump \emph{rewiring rules}. Confluence should then be proved \emph{up to jump rewiring}, which is analogous to confluence up to $\cuteq$. 
\item \emph{Good steps are glitch-free}: anticipating from the next section, the glitch does not affect \emph{good steps}, as the $\axmone$ steps in the diagram above is \emph{bad}.
\end{enumerate}

\paragraph{The Proof} We first prove the local properties.

\begin{prop}[Local confluence]
\label{prop:loc-confluence}
\hfill
\begin{enumerate}
\item 
\emph{Multiplicative clash-free diamond}: 
$\tom$ is diamond on clash-free terms.
\item \emph{Exponential local confluence}: $\toess$ and $\toems$ are both
locally confluent.
\end{enumerate}
\end{prop}
\withproofs{
\input{\proofspath/confluence/local-confluence}
}

In contrast to confluence, commutation of two reductions $\Rew{1}$ and $\Rew{2}$ does \emph{not} follow from their \emph{local} commutation and strong normalization. Here however the rules verify a \emph{linear} form of Hindley's \emph{strong} (local) commutation \cite{HindleyPhD} of $\Rew{1}$ over $\Rew{2}$, defined as:
\begin{center}
\begin{tabular}{ccccccccccccc}
\begin{tikzpicture}[ocenter]
		\node at (0,0)[align = center](source){\normalsize$\tm$};
		\node at (source.center)[right = 25pt](source-right){\normalsize $\tmtwo_1$};	
		\node at (source.center)[below = 25pt](source-down){\normalsize $\tmtwo_2$};
		
		\draw[->](source) to node[above] {\scriptsize $1 $} (source-right);
		\draw[->](source) to node[left] {\scriptsize $2 $}(source-down);
\end{tikzpicture}
&
implies that there exists $\tmthree$ such that 
&
\begin{tikzpicture}[ocenter]
		\node at (0,0)[align = center](source){\normalsize$\tm$};
		\node at (source.center)[right = 25pt](source-right){\normalsize $\tmtwo_1$};	
		\node at (source.center)[below = 25pt](source-down){\normalsize $\tmtwo_2$};
		
		\node at (source-right|-source-down)(target){\normalsize $\tmthree$};
		
		\draw[->](source) to node[above] {\scriptsize $1 $} (source-right);
		\draw[->](source) to node[left] {\scriptsize $2 $}(source-down);

		\draw[->, dotted, labelEndRight=*](source-right) to node[right] {\scriptsize $2 $}(target);
		\draw[->, dotted](source-down) to node[above] {\scriptsize $1 $} (target);
\end{tikzpicture}
\end{tabular}
\end{center}
That is, $\totwo$ cannot duplicate nor erase $\toone$. Linear commutation and strong normalization do imply commutation.

\begin{prop}[Linear commutation]
\label{prop:strong-commutation}
$\toess$ and $\toems$ both linearly commute over $\tom$.
\end{prop}
\withproofs{
\input{\proofspath/confluence/strong-commutation}
}

Then we lift the local properties using local termination and conclude using Hindley-Rosen lemma.

\begin{thm}[Confluence]
\label{thm:confluence}
The relations $\toss$ and $\toms$ are confluent on clash-free terms.
\end{thm}
\begin{proof}
We treat $\toms$, the case of $\toss$ is identical. By Newman lemma, local confluence (\refprop{loc-confluence}), and local termination (\refthm{local-termination}),  $\tom$ and $\toems$ are confluent separately. By a result of Hindley \cite{HindleyPhD}, linear commutation  (\refprop{strong-commutation}) and local termination imply that $\tom$ and $\toems$ commute. By Hindley-Rosen lemma, $\toms = \tom \cup \toems$ is confluent.
\end{proof}
\input{figure-strategy}

\paragraph{On the Traditional Proof Nets Rewriting Rule for Axioms.} As it was mentioned at the end of \refsect{towards}, our micro-step rules for exponential axioms are different from the one at work in proof nets. The proof net micro-step rule, when formulated on terms, looks as a small-step rule:
\begin{center}
$\begin{array}{ccccccc}
\cuta\evartwo\evar\tm & \to & \cutsub\evartwo\evar\tm
\end{array}$
\end{center}
If considered with our other micro-step rules, in particular our $\tobder$, it generates a closable but unpleasant local confluence diagram:
\begin{center}
\begin{tikzpicture}[ocenter]
		\node at (0,0)[align = center](source){\normalsize$\cuta{\bang\val}\evartwo \ctxp{\cuta\evartwo\evar \tm}$};
		\node at (source.center)[right = 110pt](source-right){\normalsize $\cuta{\bang\val}\evartwo \ctxp{\cuta{\bang\val}\evar \tm}$};	
		\node at (source-right.center)[right = 110pt](source-right-right){\normalsize $\cuta{\bang\val}\evartwo \ctxp{\cutsub{\bang\val}\evar \tm}$};	
		
		\node at (source.center)[below = 25pt](source-down){\normalsize $\cuta{\bang\val}\evartwo \ctxp{\cutsub\evartwo\evar \tm}$};
		
		\node at (source-right-right|-source-down)(target){\normalsize $\cutsub{\bang\val}\evartwo \ctxp{\cutsub{\bang\val}\evar \tm}$};
		\node at \med{source-down.center}{target.center}(fifthnode){\normalsize $\cutsub{\bang\val}\evartwo \ctxp{\cutsub\evartwo\evar \tm} $};
		
		\draw[|->](source) to node[above] {\scriptsize $\bdersym $} (source-right);
		\draw[->](source) to (source-down);

		\draw[->, dotted, labelEndAbove=+](source-right) to node[above] {\scriptsize $\ems $} (source-right-right);
		\draw[->, dotted, labelEndRight=+](source-right-right) to node[right] {\scriptsize $\ems $}(target);
		\draw[->, dotted, labelEndAbove=+](source-down) to node[above] {\scriptsize $\ems $} (fifthnode);
		\draw[double](fifthnode) to node[above] {\scriptsize  \reflemmaeq{bang-subs-prop}} (target);
	\end{tikzpicture}
\end{center}
In order to close the local diagram one needs to \emph{fully develop} some exponential cuts (note $\dasharrow_{\ems}^+$), which is a small-step concept and not a micro-step one. This is why we adopt different axiom rules, for which local confluence diagrams never need small-step developments. %Curiously, this is exactly the confluence issue that Accattoli \cite{DBLP:conf/lics/Accattoli13} has with MELLP proof nets at a distance, the solution of which is now clear: it is enough to switch to LSC-duplication of exponential axioms rather than the usual approach. The moral of the story is that proof nets might be misleading in distinguishing between small-step and micro-step cut elimination rules.

%% file: proofs/confluence/local-confluence.tex
% !TEX root = ../../main.tex
\begin{proof}
\hfill
\begin{enumerate}
%%%%%%%%% Multiplicative diamond on clash-free terms
\item \input{\proofspath/confluence/local-confluence-multiplicative}

%%%%%%%%% Multiplicative diamond-clash
%\item \input{\proofspath/confluence/local-confluence-multiplicative-clash}

%%%%%%%%% Small-step exponential local confluence
\item \input{\proofspath/confluence/local-confluence-exp-micro}
\end{enumerate}
\end{proof}

%% file: proofs/confluence/local-confluence-multiplicative.tex
% !TEX root = ../../main.tex
The formal statement is: $\tm_{1} \lRew{\msym} \tm_{0} \tom \tm_{2}$ with $\tm_{1}\neq\tm_{2}$ implies that there exists $\tm_3$ such that $\tm_{1} \tom \tm_{3} \lRew{\msym} \tm_{2}$.  Let $a,b \in \set{\axmtwo,\axmone,\tens,\lolli}$, $\tm_{0} \Rew{a} \tm_{1}$, and $\tm_{0} \Rew{b} \tm_{2}$. The proof is by induction on $\tm_{0} \Rew{a} \tm_{1}$, that is, by induction on the context $\ctx$ closing the root rule $\rootRew{a}$, and case analysis of $\tm_{0} \Rew{b} \tm_{2}$. Here we show only the case for $\rtoaxmone$, which is the most interesting one, requiring the deformation lemma and allowing us to discuss the case of clashes. The other cases are in \cite{DBLP:journals/corr/abs-lics}.

We have $\tm_{0} = \cuta{\mval}\mvar\mctxfp\mvar  \rtoaxmone  \mctxfp\mval = 
\tm_{1}$. 
		If $\Rew{b}$ takes place entirely in $\mval$ then the diagram closes in one step on both 
sides. 
	If $\Rew{b}$ takes place in $\mctxfp\mvar$, that is, $\mctxfp\mvar \Rew{b}\tmtwo$ then by deformation (\reflemmap{aux-props-loc-confl-strong-comm}{two}) there are two mutually exclusive cases:
	\begin{enumerate}
		\item $\tmtwo = \mctxtwofp{\mvar}$ for some 
$\mctxtwo$ such that $\mvar\notin\fv\mctxtwo$ and $\mctxfp\mval \Rew{b} \mctxtwofp\mval$ for every $\mval$. Then:
\begin{center}
\begin{tikzpicture}[ocenter]
		\node at (0,0)[align = center](source){\normalsize$ \cuta{\mval}\mvar\mctxfp\mvar $};
		\node at (source.east)[right = 35pt](source-right){\normalsize $\mctxfp\mval$};	
		\node at (source.center)[below = 20pt](source-down){\normalsize $\cuta{\mval}\mvar\mctxtwofp\mvar $};
		
		\node at (source-right|-source-down)(target){\normalsize $\mctxtwofp\mval$};
		
		\draw[|->](source) to node[above] {\scriptsize $\axmone $} (source-right);
		\draw[->](source) to node[left] {\scriptsize $b $}(source-down);

		\draw[->, dotted](source-right) to node[right] {\scriptsize $b $}(target);
		\draw[|->, dotted](source-down) to node[above] {\scriptsize $\axmone $} (target);
	\end{tikzpicture}
\end{center}

		\item $\mctx = \mctxtwop{\cuta\ctxhole\mvartwo\tmthree}$ and $\Rew{b}$ reduces the cut on $\mvartwo$, which then is a $\toaxmtwo$ step. Then $\tmthree = \mctxthreep{\tmthree_{\mvartwo}}$ and diagram depends on the occurrence $\tmthree_{\mvartwo}$ that interacts with the cut:
	\begin{itemize}
	\item \emph{Variable}, that is, $\tmthree_{\mvartwo} = \mvartwo$. Then:
\begin{center}
\begin{tikzpicture}[ocenter]
		\node at (0,0)[align = center](source){\normalsize$\cuta{\mval}\mvar \mctxtwop{\cuta\mvar\mvartwo \mctxthreep\mvartwo }$};
		\node at (source.east)[right = 35pt](source-right){\normalsize $\mctxtwop{\cuta\mval\mvartwo \mctxthreep\mvartwo }$};	
		\node at (source.center)[below = 20pt](source-down){\normalsize $\cuta{\mval}\mvar \mctxtwop{\mctxthreep\mvar}  $};
		
		\node at (source-right|-source-down)(target){\normalsize $\mctxtwop{\mctxthreep\mval} $};
		
		\draw[|->](source) to node[above] {\scriptsize $\axmone $} (source-right);
		\draw[->](source) to node[left] {\scriptsize $\axmtwo $}(source-down);

		\draw[->, dotted](source-right) to node[right] {\scriptsize $\axmone $}(target);
		\draw[|->, dotted](source-down) to node[above] {\scriptsize $\axmone $} (target);
	\end{tikzpicture}
\end{center}
Note a curious fact: the diagram turns a $\toaxmtwo$ step into a $\toaxmone$ step, that is, $\toaxmtwo$ and $\toaxmone$ do not commute. Note however that the $\toaxmtwo$ step can also be seen as a $\toaxmone$ step, as the two rules superpose in this case.

	\item \emph{Par}, that is, $\tmthree_{\mvartwo} = \para\mvartwo\var\vartwo \tmfour$, and the span is:
		\begin{center}
\begin{tikzpicture}[ocenter]
		\node at (0,0)[align = center](source){\normalsize$\cuta{\mval}\mvar \mctxtwop{\cuta\mvar\mvartwo \mctxthreep{\para\mvartwo\var\vartwo \tmfour} }$};
		\node at (source.east)[right = 35pt](source-right){\normalsize $\mctxtwop{\cuta\mval\mvartwo \mctxthreep{\para\mvartwo\var\vartwo \tmfour}}$};	
		\node at (source.center)[below = 20pt](source-down){\normalsize $\cuta{\mval}\mvar \mctxtwop{\mctxthreep{\para\mvar\var\vartwo \tmfour}}  $};

		\draw[|->](source) to node[above] {\scriptsize $\axmone $} (source-right);
		\draw[->](source) to node[left] {\scriptsize $\axmtwo $}(source-down);
	\end{tikzpicture}
\end{center}
By clash-freeness, $\mval = \pair{\lctxp\val}{\lctxtwop\valtwo}$, and the diagram is:
		\begin{center}
\begin{tikzpicture}[ocenter]
		\node at (0,0)[align = center](source){\normalsize$\cuta{\mval}\mvar \mctxtwop{\cuta\mvar\mvartwo \mctxthreep{\para\mvartwo\var\vartwo \tmfour} }$};
		\node at (source.east)[right = 35pt](source-right){\normalsize $\mctxtwop{\cuta\mval\mvartwo \mctxthreep{\para\mvartwo\var\vartwo \tmfour}}$};	
		\node at (source.center)[below = 20pt](source-down){\normalsize $\cuta{\mval}\mvar \mctxtwop{\mctxthreep{\para\mvar\var\vartwo \tmfour}}  $};
		
		\node at (source-right|-source-down)(target){\normalsize $\mctxtwop{\mctxthreep{ \lctxp{\cuta\val\var  \lctxtwop{\cuta\valtwo\vartwo \tmfour}}} }$};
		
		\draw[|->](source) to node[above] {\scriptsize $\axmone $} (source-right);
		\draw[->](source) to node[left] {\scriptsize $\axmtwo $}(source-down);

		\draw[->, dotted](source-right) to node[right] {\scriptsize $\tens $}(target);
		\draw[|->, dotted](source-down) to node[above] {\scriptsize $\tens $} (target);
	\end{tikzpicture}
\end{center}
Note that the diagram closes using $\totens$ rather than $\toaxmtwo$ or $\toaxmone$, showing that $\toaxmtwo$ and $\toaxmone$ do not commute. In absence of clash-freeness, note that the diagram closes also with $\cuteq$.

	\item \emph{Subtraction}, that is, $\tmthree_{\mvartwo} = \suba\mvartwo\val\var \tmfour$, and the span is:
	\begin{center}
\begin{tikzpicture}[ocenter]
		\node at (0,0)[align = center](source){\normalsize$\cuta{\mval}\mvar \mctxtwop{\cuta\mvar\mvartwo \mctxthreep{\suba\mvartwo\val\var \tmfour} }$};
		\node at (source.east)[right = 35pt](source-right){\normalsize $\mctxtwop{\cuta\mval\mvartwo \mctxthreep{\suba\mvartwo\val\var \tmfour}}$};	
		\node at (source.center)[below = 20pt](source-down){\normalsize $\cuta{\mval}\mvar \mctxtwop{\mctxthreep{\suba\mvar\val\var\tmfour}}  $};

		\draw[|->](source) to node[above] {\scriptsize $\axmone $} (source-right);
		\draw[->](source) to node[left] {\scriptsize $\axmtwo $}(source-down);
	\end{tikzpicture}
\end{center}
By clash-freeness, $\mval = \la\vartwo\lctxp\valtwo$, and the diagram is:
\begin{center}
\begin{tikzpicture}[ocenter]
		\node at (0,0)[align = center](source){\normalsize$\cuta{\mval}\mvar \mctxtwop{\cuta\mvar\mvartwo \mctxthreep{\suba\mvartwo\val\var \tmfour} }$};
		\node at (source.east)[right = 35pt](source-right){\normalsize $\mctxtwop{\cuta\mval\mvartwo \mctxthreep{\suba\mvartwo\val\var \tmfour}}$};	
		\node at (source.center)[below = 20pt](source-down){\normalsize $\cuta{\mval}\mvar \mctxtwop{\mctxthreep{\suba\mvar\val\var\tmfour}}  $};

		\node at (source-right|-source-down)(target){\normalsize $\mctxtwop{\mctxthreep{\cuta\val\vartwo \lctxp{\cuta\valtwo\var\tmfour}} }$};

		\draw[|->](source) to node[above] {\scriptsize $\axmone $} (source-right);
		\draw[->](source) to node[left] {\scriptsize $\axmtwo $}(source-down);
		
		\draw[->, dotted](source-right) to node[right] {\scriptsize $\lolli $}(target);
		\draw[|->, dotted](source-down) to node[above] {\scriptsize $\lolli $} (target);

	\end{tikzpicture}
\end{center}
Note that the diagram closes using $\tololli$ rather than $\toaxmtwo$ or $\toaxmone$. In absence of clash-freeness, note that the diagram closes also with $\cuteq$.
	\end{itemize}
\end{enumerate}

%% file: proofs/confluence/local-confluence-exp-micro.tex
% !TEX root = ../../main.tex
The proof follows the same structure as for $\tom$. Again, we show only the root case of $\toaxeone$, that is, $\rtoaxeone$, which is the most interesting one, in particular because it shows that $\toaxeone$ and $\toaxetwo$ do not commute. The other cases are in \cite{DBLP:journals/corr/abs-lics}.
 We have $\tm_{0} =  \cuta{\exval} \evar \ctxfp{\evar}
 \rtoaxeone 
\cuta{\exval}\evar\ctxfp{\exval}
 = \tm_{1}$. Three cases:
 \begin{enumerate}
 \item \emph{$\Rew{b}$ takes place in $\exval$}. Then the diagram is (note that it is not a diamond diagram):
 \begin{center}
\begin{tikzpicture}[ocenter]
		\node at (0,0)[align = center](source){\normalsize$ \cuta{\exval} \evar \ctxfp{\evar}$};
		\node at (source.center)[right = 170pt](source-right){\normalsize $\cuta{\exval}\evar\ctxfp{\exval}$};	
		\node at (source.center)[below = 20pt](source-down){\normalsize $\cuta{\exvaltwo} \evar \ctxfp{\evar}$};
		
		\node at (source-right|-source-down)(target){\normalsize $\cuta{\exvaltwo} \evar \ctxfp{\exval}$};
		\node at \med{source-down.east}{target.west}(fifthnode){\normalsize $\cuta{\exvaltwo}\evar\ctxfp{\exvaltwo} $};
		
		\draw[|->](source) to node[above] {\scriptsize $\axeone $} (source-right);
		\draw[->](source) to node[left] {\scriptsize $b $}(source-down);

		\draw[->, dotted](source-right) to node[right] {\scriptsize $b $}(target);
		\draw[|->, dotted](source-down) to node[above] {\scriptsize $\axeone $} (fifthnode);
		\draw[->, dotted](target) to node[above] {\scriptsize $b $} (fifthnode);
	\end{tikzpicture}
\end{center}

\item \emph{$\Rew{b}$ takes place in $\ctxfp\evar$}. Then by the deformation lemma (\reflemmap{aux-props-loc-confl-strong-comm}{five}) there are two mutually exclusive cases:
 \begin{itemize}
\item $\tmtwo = \ctxtwofp{\evar}$ for some 
$\ctxtwo$ such that $\ctxfp\exval \Rew{a} \ctxtwofp\exval$ for every $\exval$. Then:
\begin{center}
\begin{tikzpicture}[ocenter]
		\node at (0,0)[align = center](source){\normalsize$ \cuta{\exval} \evar \ctxfp{\evar}$};
		\node at (source.east)[right = 40pt](source-right){\normalsize $\cuta{\exval}\evar\ctxfp{\exval}$};	
		\node at (source.center)[below = 20pt](source-down){\normalsize $\cuta{\exval} \evar \ctxtwofp{\evar}$};
		
		\node at (source-right|-source-down)(target){\normalsize $\cuta{\exval}\evar\ctxtwofp{\exval}$};
		
		\draw[|->](source) to node[above] {\scriptsize $\axeone $} (source-right);
		\draw[->](source) to node[left] {\scriptsize $b $}(source-down);

		\draw[->, dotted](source-right) to node[right] {\scriptsize $b $}(target);
		\draw[|->, dotted](source-down) to node[above] {\scriptsize $\axeone $} (target);
	\end{tikzpicture}
\end{center}

\item  or $\ctx = \ctxtwop{\cuta\ctxthree\evartwo\tmthree}$ for some 
$\ctxtwo$ and $\ctxthree$, and $\Rew{b}$ reduces a redex given by the cut on $\evartwo$. Cases of the reduced redex:
\begin{itemize}
\item $\toaxeone$. Then we have $\ctxfp\evar = \ctxtwop{\cuta{\evar} \evartwo \ctxfourfp{\evartwo}}$ for some $\ctxfour$ (and $\ctxthree$ is empty), and the diagram closes as follows (it is not diamond):
\begin{center}
\begin{tikzpicture}[ocenter]
		\node at (0,0)[align = center](source){\normalsize$\cuta{\exval} \evar \ctxtwop{\cuta{\evar} \evartwo \ctxfourfp{\evartwo}}$};
		\node at (source.center)[right = 140pt](source-right){\normalsize $\cuta{\exval} \evar \ctxtwop{\cuta{\exval} \evartwo \ctxfourfp{\evartwo}}$};	
		\node at (source.center)[below = 20pt](source-down){\normalsize $\cuta{\exval} \evar \ctxtwop{\cuta{\evar} \evartwo \ctxfourfp{\evar}}$};
		
		\node at (source-right|-source-down)(target){\normalsize $\cuta{\exval} \evar \ctxtwop{\cuta{\exval} \evartwo \ctxfourfp{\exval}}$};
		\node at \med{source-down.east}{target.west}[below=20pt](fifthnode){\normalsize $\cuta{\exval} \evar \ctxtwop{\cuta{\exval} \evartwo \ctxfourfp{\evar}} $};
		
		\draw[|->](source) to node[above] {\scriptsize $\axeone $} (source-right);
		\draw[->](source) to node[left] {\scriptsize $\axeone $}(source-down);

		\draw[->, dotted](source-right) to node[right] {\scriptsize $\axeone $}(target);
		\draw[|->, dotted](source-down) to node[below left = -3pt and 0pt] {\scriptsize $\axeone $} (fifthnode);
		\draw[|->, dotted](fifthnode) to node[below right = -3pt and 0pt] {\scriptsize $\axeone $} (target);
	\end{tikzpicture}
\end{center}

\item $\toaxetwo$. Then we have $\ctxfp\evar = \ctxtwop{\cuta{\evar} \evartwo \ctxfourp{\dera\evartwo\var\tm}}$ for some $\ctxfour$ (and $\ctxthree$ is empty). Two sub-cases, depending on $\exval$:
\begin{enumerate}
\item \emph{$\exval$ is a variable $\evarthree$}:
\begin{center}
\begin{tikzpicture}[ocenter]
		\node at (0,0)[align = center](source){\normalsize$\cuta{\evarthree} \evar \ctxtwop{\cuta{\evar} \evartwo \ctxfourp{\dera\evartwo\var\tm}}$};
		\node at (source.center)[right = 140pt](source-right){\normalsize $\cuta{\evarthree} \evar \ctxtwop{\cuta{\evarthree} \evartwo \ctxfourp{\dera\evartwo\var\tm}}$};	
		\node at (source.center)[below = 20pt](source-down){\normalsize $\cuta{\evarthree} \evar \ctxtwop{\cuta{\evar} \evartwo \ctxfourp{\dera\evar\var\tm}} $};
		
		\node at (source-right|-source-down)(target){\normalsize $\cuta{\evarthree} \evar \ctxtwop{\cuta{\evarthree} \evartwo \ctxfourp{\dera\evarthree\var\tm}}$};
		\node at \med{source-down.east}{target.west}[below=20pt](fifthnode){\normalsize $\cuta{\evarthree} \evar \ctxtwop{\cuta{\evarthree} \evartwo \ctxfourp{\dera\evar\var\tm}} $};
		
		\draw[|->](source) to node[above] {\scriptsize $\axeone $} (source-right);
		\draw[->](source) to node[left] {\scriptsize $\axetwo $}(source-down);

		\draw[->, dotted](source-right) to node[right] {\scriptsize $\axetwo $}(target);
		\draw[|->, dotted](source-down) to node[below left = -3pt and 0pt] {\scriptsize $\axeone $} (fifthnode);
		\draw[|->, dotted](fifthnode) to node[below right = -3pt and 0pt] {\scriptsize $\axetwo $} (target);
	\end{tikzpicture}
\end{center}

\item \emph{$\exval$ is a promotion $\bang\tmtwo=\bang\lctxp\exvaltwo$}:
\begin{center}
\begin{tikzpicture}[ocenter]
		\node at (0,0)[align = center](source){\normalsize$\cuta{\bang\tmtwo} \evar \ctxtwop{\cuta{\evar} \evartwo \ctxfourp{\dera\evartwo\var\tm}}$};
		\node at (source.center)[right = 120pt](source-right){\normalsize $\cuta{\bang\tmtwo} \evar \ctxtwop{\cuta{\bang\tmtwo} \evartwo \ctxfourp{\dera\evartwo\var\tm}}$};	
		\node at (source.center)[below = 20pt](source-down){\normalsize $\cuta{\bang\tmtwo} \evar \ctxtwop{\cuta{\evar} \evartwo \ctxfourp{\dera\evar\var\tm}} $};
		
		\node at (source-right|-source-down)(target){\normalsize $\cuta{\bang\tmtwo} \evar \ctxtwop{\cuta{\bang\tmtwo} \evartwo \ctxfourp{\lctxp{\cuta\exvaltwo\var\tm}}}$};
		\node at \med{source-down.east}{target.west}[below=20pt](fifthnode){\normalsize $\cuta{\bang\tmtwo} \evar \ctxtwop{\cuta{\bang\tmtwo} \evartwo \ctxfourp{\dera\evar\var\tm}} $};
		
		\draw[|->](source) to node[above] {\scriptsize $\axeone $} (source-right);
		\draw[->](source) to node[left] {\scriptsize $\axetwo $}(source-down);

		\draw[->, dotted](source-right) to node[right] {\scriptsize $\bang $}(target);
		\draw[|->, dotted](source-down) to node[below left = -3pt and 0pt] {\scriptsize $\axeone $} (fifthnode);
		\draw[|->, dotted](fifthnode) to node[below right = -3pt and 0pt] {\scriptsize $\bang $} (target);
	\end{tikzpicture}
\end{center}
\end{enumerate}
Note that the two sub-cases above show that $\toaxeone$ and $\toaxetwo$ do not commute.

\item $\tobder$. Then we have $\ctxfp\evar = \ctxtwop{\cuta{\bang\tmfour} \evartwo \ctxfourp{\dera\evartwo\var\tmthree}}$ for some $\ctxfour$ and with $\evar$ occurring in $\tmfour$. To write down the rewriting step we have to decompose $\tmfour$ as a left context and a value, that is, $\tmfour= \lctxp\val$. Now, $\evar$ can be in either $\lctx$ or $\val$. We consider the latter case, which is easier to spell out in symbols, the former case is analogous, just heavier to write down. Then assume that $\val = \ctxthreefp\evar$, so that $\tmfour = \lctxp{\ctxthreefp\evar}$. We have the following (non-diamond) diagram:
\begin{center}
\begin{tikzpicture}[ocenter]
		\node at (0,0)[align = center](source){\scriptsize$\cuta{\exval} \evar \ctxtwop{\cuta{\bang\lctxp{\ctxthreefp\evar}} \evartwo \ctxfourp{\dera\evartwo\var\tmthree}}$};
		\node at (source.center)[right = 110pt](source-right){\scriptsize $\cuta{\exval} \evar \ctxtwop{\cuta{\bang\lctxp{\ctxthreefp{\exval}}} \evartwo \ctxfourp{\dera\evartwo\var\tmthree}}$};	
		\node at (source.center)[below = 20pt](source-down){\scriptsize $\cuta{\exval} \evar \ctxtwop{\cuta{\bang\lctxp{\ctxthreefp{\evar}}} \evartwo \ctxfourp{\lctxp{\cuta{\ctxthreefp{\evar}}\var\tmthree}}}$};
		
		\node at (source-right|-source-down)(target){\scriptsize $\cuta{\exval} \evar \ctxtwop{\cuta{\bang\lctxp{\ctxthreefp{\exval}}} \evartwo \ctxfourp{\lctxp{\cuta{\ctxthreefp\exval}\var\tmthree}}}$};
		\node at \med{source-down.east}{target.west}[below=20pt](fifthnode){\scriptsize $\cuta{\exval} \evar \ctxtwop{\cuta{\bang\lctxp{\ctxthreefp\exval}} \evartwo \ctxfourp{\lctxp{\cuta{\ctxthreefp\evar}\var\tmthree}}} $};
		
		\draw[|->](source) to node[above] {\scriptsize $\axeone $} (source-right);
		\draw[->](source) to node[left] {\scriptsize $\bdersym $}(source-down);

		\draw[->, dotted](source-right) to node[right] {\scriptsize $\bang $}(target);
		\draw[|->, dotted](source-down) to node[below left = -3pt and 0pt] {\scriptsize $\axeone $} (fifthnode);
		\draw[|->, dotted](fifthnode) to node[below right = -3pt and 0pt] {\scriptsize $\axeone $} (target);
	\end{tikzpicture}
\end{center}

\item $\tobw$. Then $\ctxfp\evar = \ctxtwop{\cuta{\ctxthreep\evar}\evartwo \tmtwo} \tobw \ctxtwop{\tmtwo}$ and we have the following (non-diamond) diagram:
\begin{center}
\begin{tikzpicture}[ocenter]
		\node at (0,0)[align = center](source){\normalsize$\cuta{\exval} \evar \ctxtwop{\cuta{\ctxthreep\evar}\evartwo \tmtwo}$};
		\node at (source.east)[right = 40pt](source-right){\normalsize $\cuta{\exval}\evar \ctxtwop{\cuta{\ctxthreep\exval}\evartwo \tmtwo}$};	
		\node at (source.center)[below = 20pt](source-down){\normalsize $\cuta{\exval} \evar \ctxtwop{\tmtwo}$};
		
		\draw[|->](source) to node[above] {\scriptsize $\axeone $} (source-right);
		\draw[->](source) to node[left] {\scriptsize $\wsym $}(source-down);

		\draw[->, dotted](source-right) to node[below] {\scriptsize $\wsym$}(source-down);
	\end{tikzpicture}
\end{center}

%If $\Rew{b}$ involves both $\lctx$ and $\exval$ then it has the shape $\lctxp\exval \Rew{b} \lctxtwop\exvaltwo$ and by \reflemmap{loc-confluence-aux}{six} we have:
%	\begin{center}$\begin{array}{cccccc}
%\cuta{\lctxp\exval}\evar \ctxfp\evar
%& \rtoaxeone &
%\lctxp{\cuta{\exval}\evar \ctxfp\exval}
%\\
% &&\downarrow_{b}
%\\ 
%\downarrow_{b}& & 
%\lctxtwop{\cuta{\exvaltwo}\evar \ctxfp\exval}
%\\
%&&\downarrow_{b}
%\\
%\cuta{\lctxtwop\exvaltwo}\evar \ctxfp\evar
%& \rtoaxeone    &
%\lctxtwop{\cuta{\exvaltwo}\evar \ctxfp\exvaltwo}
%\end{array}$\end{center}
%Note that this diagram is not diamond.
\end{itemize}
\end{itemize}

\item \emph{$\Rew{b}$ involves the same root cut of $\Rew{a}$}. Then $\tm_{0} =  \cuta{\exval}\evar\ctxfp{\evar}= \cuta{\exval}\evar\ctxtwop{\tm_\evar}$ for a context $\ctxtwo\neq \ctx$ that does not capture $\evar$ and $\tm_\evar$ is an occurrence of $\evar$. Two sub-cases:
\begin{itemize}
\item \emph{Variable occurrence}, that is, $\tm_\evar = \evar$. Then $\Rew{b}$ is also a $\rtoaxeone$ step. Note that $\ctx$ and $\ctxtwo$ are disjoint, that is, $\ctx\parallel\ctxtwo$, so that there is a double context $\cctx_{\ctx,\ctxtwo}$ such that $\cctx_{\ctx,\ctxtwo}\ctxholep{\cdot,\evar} = \ctx$ and $\cctx_{\ctx,\ctxtwo}\ctxholep{\evar,\cdot} = \ctxtwo$. Then $\tm_0 = \cuta{\exval} \evar \cctx_{\ctx,\ctxtwo}\ctxholep{\evar,\evar}$ and the diagram closes as follows:
\begin{center}
\begin{tikzpicture}[ocenter]
		\node at (0,0)[align = center](source){\normalsize$ \cuta{\exval} \evar \cctx_{\ctx,\ctxtwo}\ctxholep{\evar,\evar}$};
		\node at (source.east)[right = 40pt](source-right){\normalsize $\cuta{\exval} \evar \cctx_{\ctx,\ctxtwo}\ctxholep{\exval,\evar}$};	
		\node at (source.center)[below = 20pt](source-down){\normalsize $\cuta{\exval} \evar \cctx_{\ctx,\ctxtwo}\ctxholep{\evar,\exval}$};
		
		\node at (source-right|-source-down)(target){\normalsize $\cuta{\exval} \evar \cctx_{\ctx,\ctxtwo}\ctxholep{\exval,\exval}$};
		
		\draw[|->](source) to node[above] {\scriptsize $\axeone $} (source-right);
		\draw[|->](source) to node[left] {\scriptsize $\axeone $}(source-down);

		\draw[|->, dotted](source-right) to node[right] {\scriptsize $\axeone $}(target);
		\draw[|->, dotted](source-down) to node[above] {\scriptsize $\axeone $} (target);
	\end{tikzpicture}
\end{center}

\item \emph{Dereliction occurrence}, that is, $\tm_\evar = \dera\evar\var\tmtwo$. Two sub-cases, depending on the relationship between $\ctx$ and $\ctxtwo$.
\begin{enumerate}
\item \emph{$\ctx$ nests under $\ctxtwo$}, that is, $\ctxtwo \outer \ctx$. Then, there exists $\ctxthree$ such that $\ctx = \ctxtwop{\dera\evar\var\ctxthree}$. Two sub-cases, depending on $\exval$. 
\begin{enumerate}
\item \emph{$\exval$ is a variable}, that is, $\exval = \evartwo$. Then:
\begin{center}
\begin{tikzpicture}[ocenter]
		\node at (0,0)[align = center](source){\normalsize$ \cuta{\evartwo} \evar \ctxtwop{\dera\evar\var\ctxthreefp\evar}$};
		\node at (source.east)[right = 40pt](source-right){\normalsize $\cuta{\evartwo} \evar \ctxtwop{\dera\evar\var\ctxthreefp\evartwo}$};	
		\node at (source.center)[below = 20pt](source-down){\normalsize $\cuta{\evartwo} \evar \ctxtwop{\dera\evartwo\var\ctxthreefp\evar}$};
		
		\node at (source-right|-source-down)(target){\normalsize $\cuta{\evartwo} \evar \ctxtwop{\dera\evartwo\var\ctxthreefp\evartwo}$};
		
		\draw[|->](source) to node[above] {\scriptsize $\axeone $} (source-right);
		\draw[|->](source) to node[left] {\scriptsize $\axetwo $}(source-down);

		\draw[|->, dotted](source-right) to node[right] {\scriptsize $\axetwo $}(target);
		\draw[|->, dotted](source-down) to node[above] {\scriptsize $\axeone $} (target);
	\end{tikzpicture}
\end{center}

\item \emph{$\exval$ is a promotion}, that is, $\exval = \bang\tmtwo=\bang\lctxp\exvaltwo$. Then:
\begin{center}
\begin{tikzpicture}[ocenter]
		\node at (0,0)[align = center](source){\normalsize$ \cuta{\bang\tmtwo} \evar \ctxtwop{\dera\evar\var\ctxthreefp\evar}$};
		\node at (source.center)[right = 110pt](source-right){\normalsize $\cuta{\bang\tmtwo} \evar \ctxtwop{\dera\evar\var\ctxthreefp{\bang\tmtwo}}$};	
		\node at (source.center)[below = 20pt](source-down){\normalsize $\cuta{\bang\tmtwo} \evar \ctxtwop{\lctxp{\cuta\exvaltwo\var\ctxthreefp\evar}}$};
		
		\node at (source-right|-source-down)(target){\normalsize $\cuta{\bang\tmtwo} \evar \ctxtwop{\lctxp{\cuta\exvaltwo\var\ctxthreefp{\bang\tmtwo}}}$};
		
		\draw[|->](source) to node[above] {\scriptsize $\axeone $} (source-right);
		\draw[|->](source) to node[left] {\scriptsize $\bdersym $}(source-down);

		\draw[|->, dotted](source-right) to node[right] {\scriptsize $\bdersym $}(target);
		\draw[|->, dotted](source-down) to node[above] {\scriptsize $\axeone $} (target);
	\end{tikzpicture}
\end{center}
\end{enumerate}

\item \emph{$\ctx$ and $\ctxtwo$ are disjoint}. Then, the reasoning goes similarly to the case where $\tm_\evar$ is a variable occurrence, resting on a double context, and it has two straightforward sub-cases depending on $\exval$, as when $\ctx$ nests under $\ctxtwo$. \qedhere
\end{enumerate}
\end{itemize}
\end{enumerate}

%% file: proofs/confluence/strong-commutation.tex
% !TEX root = ../../main.tex
\begin{proof}
The proof is an analysis of cases very similar to the one for local confluence. Some details are in \cite{DBLP:journals/corr/abs-lics}.
\end{proof}

%% file: figure-strategy.tex
% !TEX root = main.tex
 \begin{figure*}[t]

 \newcases{nullcases}
    {\ }
    {$##$\hfil} {$##$\hfil}
    {\lbrace} {.}
     \arraycolsep=2pt
  \tabcolsep=2pt
%\fbox{
% \begin{center} 
\begin{tabular}{c}
  \begin{tabular}{c@{\hspace{.3cm}}cccc}
\multicolumn{2}{c}{\textsc{Dominating free variables of contexts}}&\\
{\small$\begin{array}{rlll}
\dfv\ctxhole & \defeq & \emptyset
\\
\dfv{ \pair\ctx\val} & \defeq & \dfv\ctx
\\
\dfv{ \pair\val\ctx} & \defeq & \dfv\ctx
\\
\dfv{ \la\var\ctx} & \defeq &  
		\dfv\ctx \setminus \set{\var}
\\
 \dfv{ \bang \ctx}& \defeq & \dfv{  \ctx}
\\
 \dfv{\cuta\val\var\ctx} & \defeq &
		\dfv\ctx \setminus \set{\var} 
\\
 \dfv{\cuta{\vctx}\var\tm} & \defeq & \dfv\vctx
\end{array}$
}
&
{\small
$\begin{array}{rlll}
\dfv{\para\mvar\var\vartwo \ctx} & \defeq & \begin{nullcases}
		\set\mvar \cup (\dfv\ctx \setminus \set{\var,\vartwo}) & \mbox{if $\var\in \dfv\ctx$}
		\\
		&\mbox{or $\vartwo \in \dfv\ctx$}
		\\
		\dfv\ctx & \mbox{otherwise.}
		\end{nullcases}
\\
  \dfv{\suba\mvar\vctx\var \tm} & \defeq & \set\mvar\cup \dfv\vctx
\\

  \dfv{\suba\mvar\val\var \ctx} & \defeq & \begin{nullcases}
		\set\mvar\cup (\dfv\ctx \setminus \set{\var}) & \mbox{if $\var \in \dfv\ctx$}\\
		\dfv\ctx & \mbox{otherwise.}
		\end{nullcases}
\\
 \dfv{\dera\evar\var \ctx} & \defeq & \begin{nullcases}
		\set\evar\cup (\dfv\ctx \setminus \set{\var}) & \mbox{if $\var \in \dfv\ctx$}\\
		\dfv\ctx & \mbox{otherwise.}
		\end{nullcases}

\end{array}$
}
\end{tabular}
%%%
%%%%%%%
\\[9pt]\hline\\[-6pt]
%%%%%%%
%%%
$\begin{array}{rcl}
\multicolumn{3}{c}{\textsc{Good value contexts}}
\\
 \vgctx & \grameq & \ctxhole \,\mid\, \pair\gctx\tm \,\mid\, \pair\tm\gctx \,\mid\, \la\var\gctx \,\mid\, \bang\gctx
 \end{array}$
 \\[6pt]
 $\begin{array}{rcl}
\multicolumn{3}{c}{\textsc{Good contexts}}
\\
 \gctx & \grameq & \vgctx \,\mid\, \para\mvar\var\vartwo \gctx \,\mid\, \suba\mvar\val\var \gctx
 \,\mid\, \suba\mvar\vgctx\var \tm \,\mid\,  
\dera\evar\var\gctx \,\mid\,  
 \cuta\val\var\gctx \mbox{ if }\var\notin\dfv\gctx
 \end{array}$
\\[6pt]
 $\begin{array}{rcl}
\multicolumn{3}{c}{\textsc{Bad contexts}}
\\
 \bctx & \grameq & \cuta\vctx\var\tm \,\mid\, \cuta\val\var\ctx \mbox{ if }\var\in\dfv\ctx 
\,\mid\, \ctxp\bctx
\end{array}$
\end{tabular}
%}
%\end{center}

\caption{Definitions for the good strategy: dominating free variables, good and bad contexts.}
\label{fig:strategy}
\end{figure*}

%% file: 11-The_Good_Strategy.tex
% !TEX root = main.tex

\section{The Good Strategy}
\label{sect:strategy}
In this section, we define the \emph{good cut elimination strategy} $\tog$ for the micro-step untyped ESC, show various of its properties, including the sub-term property, and prove that it provides a polynomial cost model.

\paragraph{Breaking the Sub-Term Property} When does the sub-term property not hold? One has to duplicate an exponential value $\exval$ \emph{touched} by previous steps. In our setting, \emph{touched} can mean two things. Either a redex \emph{fully} contained in $\exval$ is reduced, obtaining $\exvaltwo$, and then $\exvaltwo$ is duplicated (or erased), as in the step marked with $\bigstar$ in the following diagram (the other, dashed path of which has the sub-term property):
\begin{center}
\begin{tikzpicture}[ocenter]
		\node at (0,0)[align = center](source){\normalsize$  \cuta{\exval} \evar \ctxfp{\evar}$};
		\node at (source.center)[right = 170pt](source-right){\normalsize $\cuta{\exvaltwo} \evar \ctxfp{\evar}$};
		\node at (source.center)[below = 25pt](source-down){\normalsize $\cuta{\exval}\evar\ctxfp{\exval}$};
		
		\node at (source-right|-source-down)(target){\normalsize $\cuta{\exvaltwo}\evar\ctxfp{\exvaltwo}$};
		\node at \med{source-down.east}{target.west}(fifthnode){\normalsize $\cuta{\exvaltwo} \evar \ctxfp{\exval}$};
		
		\draw[->](source) to node[above] {\scriptsize $\mssym $} (source-right);
		\draw[->](source-right) to node[right] {\scriptsize $\axeone$} node[right=20pt] {\scriptsize $\bigstar$}(target);

		\draw[->, dotted](source) to node[left] {\scriptsize $\axeone$}(source-down);
		\draw[->, dotted](source-down) to node[above] {\scriptsize $\mssym$} (fifthnode);
		\draw[->, dotted](fifthnode) to node[above] {\scriptsize $\mssym$} (target);
	\end{tikzpicture}
\end{center}
The other way of \emph{touching} $\exval$ is when a cut $\cuta{\exvaltwo} \evar$ \emph{external} to $\exval$ acts on some exponential variable $\evar$ in $\exval=\vctxfp\evar$, as in the following diagram:
\begin{center}
\begin{tikzpicture}[ocenter]
		\node at (0,0)[align = center](source){\footnotesize$  \cuta{\exvaltwo} \evar \ctxp{\cuta{\vctxfp\evar} \evartwo \ctxtwofp{\evartwo}}$};
		\node at (source.center)[right = 240pt](source-right){\footnotesize $\cuta{\exvaltwo} \evar \ctxp{\cuta{\vctxfp\exvaltwo} \evartwo \ctxtwofp{\evartwo}}$};
		\node at (source.center)[below = 25pt](source-down){\footnotesize $\cuta{\exvaltwo} \evar \ctxp{\cuta{\vctxfp\evar} \evartwo \ctxtwofp{\vctxfp\evar}}$};
		
		\node at (source-right|-source-down)(target){\footnotesize $\cuta{\exvaltwo} \evar \ctxp{\cuta{\vctxfp\exvaltwo} \evartwo \ctxtwofp{\vctxfp\exvaltwo}}$};
		\node at \med{source-down.east}{target.west}(fifthnode){\footnotesize $\cuta{\exvaltwo} \evar \ctxp{\cuta{\vctxfp\exvaltwo} \evartwo \ctxtwofp{\vctxfp\evar}}$};
		
		\draw[->](source) to node[above] {\scriptsize $\axeone $} (source-right);
		\draw[->](source-right) to node[right] {\scriptsize $\axeone$} node[right=20pt] {\scriptsize $\bigstar$}(target);

		\draw[->, dotted](source) to node[left] {\scriptsize $\axeone$}(source-down);
		\draw[->, dotted](source-down) to node[above] {\scriptsize $\axeone$} (fifthnode);
		\draw[->, dotted](fifthnode) to node[above] {\scriptsize $\axeone$} (target);
	\end{tikzpicture}
\end{center}
Similar diagrams can be obtained using $\tobder$ and $\toaxetwo$ rather than $\toaxeone$, and these are all the local confluence diagrams for $\toms$ that are not squares nor triangles (see \refprop{loc-confluence} and \refprop{strong-commutation}).

Preventing these situations from happening, thus forcing evaluation to follow the other (dashed) side of the diagram, is easy. It is enough to forbid the position of the reduced redex to be inside the left sub-term of a cut---we say inside a \emph{cut value} for short. It is however not enough, because cuts are also \emph{created}. Consider:
\begin{center}$
\begin{array}{clllllcl}
\cuta{\la\evar\pair\evar\evar}\mvar\suba\mvar\exval\var\tm &\toms&
\cuta{\la\evar\pair\evar\evar}\mvar\suba\mvar\exvaltwo\var\tm &\tololli
\\
&&\cuta\exvaltwo\evar\cuta{\pair\evar\evar}\var\tm &\overset\bigstar\rightarrow_{\axeone} &
\cuta\exvaltwo\evar\cuta{\pair\exvaltwo\evar}\var\tm
\end{array}
$\end{center}
Reducing inside the subtraction value $\exval$ leads to a \emph{later} breaking of the sub-term property by the $\axeone$ step, because the $\tololli$ step creates a cut with $\exvaltwo$ inside.
Preventing these cases is tricky, because forbidding reducing subtraction values leads to cut elimination stopping too soon, without producing a cut-free term. In the $\l$-calculus, it corresponds to forbidding reducing inside arguments, which leads to \emph{head} reduction, that does not compute normal $\l$-terms. We shall then forbid reducing only subtraction values which are \emph{at risk} of becoming cuts. It might look like the dangerous positions are those in $\val$ in $\cuta\mval\mvar\mctxp{\suba\mvar\val\var\tm}$, i.e. those where the surrounding subtraction is involved in a cut. They are in fact more general, as a chain of dependencies can be involved. Consider the positions inside $\val$ in $\tm$ here:
\begin{center}$
\begin{array}{clllllcl}
\tm 
&\defeq& 
\cuta{\bang\mval}\evar\dera\evar\mvar\suba\mvar\val\var\tmtwo
& \tobder& 
\cuta{\bang\mval}\evar\cuta\mval\mvar\suba\mvar\val\var\tmtwo
\end{array}
$\end{center}
they are also dangerous, because they reduce to those of the previous kind. Thus, we need to ensure that the conclusion $\mvar$ of the subtraction is not \emph{hereditarely} involved  in a cut. 

\paragraph{Dominating Variables.} The key notion is the one of \emph{dominating (free) variables} $\dfv\ctx$ of a context (where $\ctx$ is meant to be the position of a redex), defined in \reffig{strategy}, the base case of which is for $\suba\mvar\vctx\var\tm$. If $\ctx$ is a position and $\var\in\dfv\ctx$ then $\cuta\val\var\ctx$ turns $\ctx$ into a dangerous position, that is, a redex of position $\cuta\val\var\ctx:\tm\toms\tmtwo$ might lead to a breaking of the sub-term property later on during cut elimination. In the example, $\evar$ belongs to $\dfv\ctx$ for every context $\ctx\defeq \dera\evar\mvar\suba\mvar\vctx\var\tmtwo$ of $\dera\evar\mvar\suba\mvar\val\var\tmtwo$, for every $\vctx$.

\paragraph{Good and Bad Contexts and Steps} The previous considerations  lead to the notions of \emph{good} and \emph{bad contexts} in \reffig{strategy}. A good context forbids the two ways of breaking the sub-term property: its hole cannot be in a cut value (note the absence of the production $\cuta{\vgctx}\var\tm$) nor in a subtraction value such that one of its dominating variables is cut (because of the production $\cuta\val\var\gctx$ if $\var\notin\dfv\gctx
$). Every other step is allowed. The next lemma ensures that bad and good contexts are indeed complementary concepts.

\begin{lem}[Good/bad partition]
\label{l:good-bad-partition}
Let $\ctx$ be a context. Then $\ctx$ is either good or bad.
\end{lem}
\begin{proof}
By induction on $\ctx$. The empty context $\ctxhole$ is good and not bad. For all inductive cases but cut, it follows from the \ih The cut cases:
\begin{itemize}
	\item $\ctx = \cuta\vctx\var\tm$. Then $\ctx$ is bad and not good.
	\item $\ctx = \cuta\val\var\ctxtwo$. By \ih, $\ctxtwo$ is either good or bad. If $\ctxtwo$ is bad then $\ctx$ is bad, and $\ctx$ is not good because $\ctxtwo$ is not good. If $\ctxtwo$ is good, consider whether $\var\in \dfv\ctxtwo$. If it does, then $\ctx$ is bad and not good. If instead  $\var\notin \dfv\ctxtwo$ then $\ctx$ is good and not bad.\qedhere
\end{itemize}
\end{proof}

\begin{defi}[Good/bad step, good strategy]
A micro step $\ctx : \tm \toms \tmtwo$ is \emph{good} if its position $\ctx$ is good. In such a case, we write $\tm \tog \tmtwo$. The \emph{good cut elimination strategy} is simply $\tog$. A step $\tm \toms \tmtwo$ is \emph{bad} if it is not good, and we then write $\tm\tobad \tmtwo$. We also use $\togp a$ and $\tobadp a$ to stress that the good/bad step is of kind $a\in\set{\axmone,\axmtwo,\tens,\lolli, \axeone,\axetwo,\bdersym,\wsym}$. 
\end{defi}

\paragraph{A Technical Remark About Goodness} The notions of good and bad steps are inherently micro-step, as the next example shows. Consider the following term:
\begin{center}$
\begin{array}{lll}
\tm 
& \defeq &
 \cuta{\evartwo}\evar \dera\evar\evarthree \cuta\mvarthree\mvar \suba\mvar\evar\mvartwo \mvartwo
\end{array}
$\end{center}
In $\tm$, the cut $\cuta{\evartwo}\evar$ gives rise to \emph{two} redexes, one for each occurrence of $\evar$. The redex concerning $\dera\evar\evarthree$ is good while the one concerning $\suba\mvar\evar\mvartwo $ is bad, because its position $\cuta{\evartwo}\evar \dera\evar\evarthree \cuta\mvarthree\mvar  \suba\mvar\ctxhole\mvartwo \mvartwo$ is a bad context ($\mvar$ dominates the hole, and there is a cut on $\mvar$). Thus, being good/bad is not a property of the cut $\cuta{\evartwo}\evar$, that is, it is not a small-step concept, because not all the micro-step redexes in which a cut is involved share the same character. Consider now:
\begin{center}$\begin{array}{lll}
\tmtwo
& \defeq &
\cuta{\evartwo}\evar \dera\evar\mvar \suba\mvar\evar\mvartwo \mvartwo
\end{array}$\end{center}
In $\tmtwo$, again, there are two redexes on $\evar$, a good one and a bad one. The difference is that now it is the good occurrence of $\evar$ that turns the other occurrence into a bad one (by dominating $\mvar$, which dominates the second occurrence).

\paragraph{Basic Properties of Good Contexts.} The next three lemmas collect basic facts about good contexts that are used in the proof of the rest of the section.

\begin{lem}[Good context decomposition]
\label{l:good-ctx-decomp}
Let $\ctxp\ctxtwo$ be a good context. Then $\ctx$ and $\ctxtwo$ are good contexts.
\end{lem}

\begin{lem}[Bad cannot outer good]
\label{l:b-not-outer-g}
Let $\gctxp\tm= \bctxp\tmtwo$. Then $\bctx\not\outer\gctx$, that is, either $\gctx\outer\bctx$ or $\gctx \parallel \bctx$.
\end{lem}

\begin{proof}
By induction on $\gctx$. Note that $\bctx$ cannot be empty, otherwise it would be good. Cases:
\begin{itemize}
\item $\gctx = \ctxhole$. Then $\tm = \bctxp\tmtwo$ and so $\bctx\not\outer \gctx$.
\item The abstraction, bang, par, and dereliction cases follow from the \ih
\item $\gctx = \pair\gctxtwo\tmthree$. If $\bctx = \pair\bctxtwo\tmthree$ then it follows from the \ih If $\bctx = \pair\tmfour\bctxtwo$ then $\bctx \parallel \gctx$.
\item $\gctx = \pair\tmthree\gctxtwo$. It goes as the previous case.
\item The subtraction cases are as the tensor cases.

\item $\gctx = \cuta\val\var\gctxtwo$. If $\bctx = \cuta\val\var\bctxtwo$ it follows from the \ih If $\bctx = \cuta\bctxtwo\var\tmthree$ then $\bctx \parallel \gctx$.\qedhere
\end{itemize}
\end{proof}

\begin{lem}[Double contexts and replacements]
\label{l:double-ctx-repl}
Let $\cctx$ be a double context and $\tm$ and $\tmtwo$ be two terms. Then:
 \begin{enumerate}
 \item $\cctxp{\cdot,\tm}$ is as capturing as $\cctxp{\cdot,\tmtwo}$;
 \item $\dfv{\cctxp{\cdot,\tm}} = \dfv{\cctxp{\cdot,\tmtwo}}$;
 \item $\cctxp{\cdot,\tm}$ is good if and only if $\cctxp{\cdot,\tmtwo}$ is good.
 \end{enumerate}
\end{lem}

Note that the last point of \reflemma{double-ctx-repl} could have be equivalently stated with respect to \emph{being bad}.

\begin{proof}
By induction on $\cctx$. The base cases are straightforward, and the inductive cases follow from the \ih In the inductive case $\cctx = \cuta\val\var\cctxtwo$, the third point uses the second one.
\end{proof}

\paragraph{Sub-Term Property.} We are ready for the key point of the paper, the sub-term property of the good strategy. The proof is based on a natural local invariant about \emph{bad values}, which are the sub-terms at risk of being duplicated or erased.
\begin{defi}[Bad values]
Given $\tm = \ctxp\val$, $\val$ is a \emph{bad value} of $\tm$ if $\ctx$ is a bad context.
\end{defi}

The local invariant shows that one good step cannot create bad values. The sub-term property follows by induction on the length of cut elimination sequences.

Since strict equality of sub-terms is not preserved by reduction (because of on-the-fly $\alpha$-renaming), the local invariant concerns the \emph{size} of bad values, which is what is important for cost analyses. In other words, we prove the quantitative form of the sub-term property, but one could equivalently prove the structural form and then obtain the quantitative one as a corollary (the literal sub-term property instead does not hold).

\begin{lem}
\label{l:good-ctx-mutilation}
\hfill
\begin{enumerate}
\item \label{p:good-ctx-mutilation-good} If $\ctxp{\cuta\val\var\ctxtwo}$ is good then $\ctxp{\ctxtwo}$ is good.
%%%
\item \label{p:good-ctx-mutilation-bad}If $\ctxp{\ctxtwo}$ is bad then $\ctxp{\cuta\val\var\ctxtwo}$ is bad.
\end{enumerate}
\end{lem}

\begin{proof}
Point 1 is by induction on $\ctx$. Since a context is either good or bad (\reflemma{good-bad-partition}), Point 2 is simply the contrapositive of Point 1. 
\end{proof}

\begin{prop}[Local sub-term invariant]
\label{prop:local-sub-tm}
Let $\tm \tog \tmtwo = \bctxp\val$. Then $\tm = \bctxtwop\valtwo$ for some $\bctxtwo$ and $\valtwo$ such that $\size\val = \size\valtwo$.
\end{prop}
\input{\proofspath/strategy/sub-term}

The statement of the next theorem expresses the quantitative sub-term property for $\toms$ because the sub-terms duplicated and erased by $\toms$ are values, and in particular are values that are left sub-terms of cuts (which are bad contexts), that is, they are bad values.

\begin{thm}[Sub-term property]
\label{thm:sub-term-prop}
Let $\tm \tog^{*} \tmtwo$ and $\val$ be a bad value of $\tmtwo$. Then $\size\val\leq \size\tm$.
\end{thm}
\begin{proof}
By induction on the length $k$ of the reduction $\tm \tog^{k} \tmtwo$. If $k=0$ then the statement trivially holds. If $k>0$ then  consider the last step $\tmthree \tog \tmtwo$ of the sequence. By the local sub-term invariant (\refprop{local-sub-tm}), the size of every bad value $\val$ of $\tmtwo$ is bound by the size of a bad value $\valtwo$ of $\tmthree$, which by \ih satisfy the statement. 
\end{proof}

\paragraph{Good Diamond} The good strategy is not deterministic, as for instance if $\tm \tog \tmtwo$ then:
\begin{center}$\begin{array}{llllllll}
\pair\tmtwo\tm
&  \lRew{\Gsym}& 
\pair\tm\tm 
& \tog & 
\pair\tm\tmtwo.
\end{array}$\end{center}
Note also, however, that the local diagrams that are not squares nor triangles are forbidden by design for good steps (see above). Triangles are also forbidden: it is easily seen that they all involve touching the cut value of a $\tobw$ step, which is bad. Intuitively, this is the reason why all local diagrams for good steps are diamonds.

\begin{prop}[Good diamond]
\label{prop:good-diamond}
$ \tmtwo_{1} \lRew{\Gsym}\tm \tog \tmtwo_{2}$ and $\tmtwo_{1} \neq \tmtwo_{2}$ then $ \tmtwo_{1} \tog \tmthree \lRew{\Gsym}  \tmtwo_{2}$ for some $\tmthree$.
\end{prop}

\input{\proofspath/strategy/good-diamond}

Two well known consequences of being diamond are \emph{uniform normalization}, that is, if there is a normalizing reduction
sequence then there are no diverging sequences, and \emph{random
descent}, that is, when a term is normalizable, all sequences to
normal form \emph{have the same length}. These two consequences are the reason why the diamond property
embodies a more liberal form of determinism. Additionally, these properties are essential for the study of cost models, as otherwise the number of steps of the strategy is an ambiguously defined measure.

\paragraph{Fullness.} The sub-term property expresses the correctness of our design. We also need to prove a simple form of completeness, ensuring that the good strategy does not stop when there is still work left to do. This is given by \emph{fullness}, that is, the fact that as long as $\tm$ is not normal then $\tm$ has a good redex. We need to assume that $\tm$ has no clashes, because terms with clashes such as $\cuta{(\cuta\evar\evartwo\evartwo,\evar)}\mvar \suba\mvar\evar\mvartwo\mvartwo$, where the tensor pair and the subtraction clash, are not normal (note that $\cuta\evar\evartwo\evartwo$ is a redex) but have no good redexes, as all their redexes are inside cut values of clashing cuts (a tensor pair cut on a subtraction is a clash).

We recall that an \emph{occurrence} of $\var$ in $\tm$ is given by a position $\tm = \ctxp{\tm_{\var}}$  where $\tm_{\var}$ is a sub-term of $\tm$ of shape $\var$, $\para\var\vartwo\varthree\tmtwo$, $\suba\var\val\vartwo\tmtwo$, or $\dera\var\vartwo\tmtwo$. Note that the shapes of $\tm_{\var}$ are the kind of terms with which a cut on $\var$ can interact at a distance to form a redex.

\begin{defi}[Outermost variable occurrence]
An occurrence $\tm = \ctxp{\tm_{\var}}$ of $\var$ in $\tm$ is \emph{outermost} if $\ctxtwo \not\outer\ctx$ for all other occurrences $\tm = \ctxtwop{\tm'_{\var}}$ of $\var$.
\end{defi}

\begin{lem}
\label{l:outermost-and-domination} Let $\tm = \ctxp{\tm_{\var}}$ be an outermost occurrence of $\var$ in $\tm$ with $\ctx$ not capturing $\var$. Then $\var \notin\dfv\ctx$.
\end{lem}

\begin{proof}
By induction on the number $\sizep\ctx\var$ of free occurrences of $\var$ in $\ctx$. If $\sizep\ctx\var = 0$ then $\var\notin \fv\ctx \supset \dfv\ctx$. If $\sizep\ctx\var > 0$ consider another occurrence $\ctxtwop{\tm_{\var}'}$ of $\var$ in $\tm$ such that $\ctx\not\outer\ctxtwo$, which exists because $\sizep\ctx\var > 0$ (indeed if $\ctx\outer\ctxtwo$ then $\tm_{\var}'$ is a sub-term of $\tm_{\var}$ and does not contribute to $\sizep\ctx\var$). Since $\tm_{\var}$ is outermost, we have $\ctxtwo \parallel \ctx$. Now, consider the double context $\cctx_{\ctx,\ctxtwo}$ such that $\cctx_{\ctx,\ctxtwo}\ctxholep{\cdot, \tm_{\var}'}=\ctx$ and $\cctx_{\ctx,\ctxtwo}\ctxholep{\tm_{\var},\cdot}=\ctxtwo$. By \reflemma{double-ctx-repl}, the context $\ctxthree \defeq \cctx_{\ctx,\ctxtwo}\ctxholep{\cdot, \la\mvar\mvar}$ verifies $\dfv\ctxthree = \dfv\ctx$ and $\sizep\ctxthree\var \leq \sizep\ctx\var -1$. Then by \ih we have $\var\notin\dfv\ctxthree = \dfv\ctx$.
\end{proof}

\begin{lem}
\label{l:dominating-cases} % \reflemmap{dominating-cases}{two}
Let $\var\in\dfv\ctx$.  Then one of the following cases hold:
\begin{enumerate}
\item  $\ctx = \mctxp{\para\mvar\vartwo\varthree\ctxtwo}$ and $\var = \mvar$;
\item $\ctx = \mctxp{\suba\mvar\val\vartwo\ctxtwo}$ and $\var = \mvar$;
\item $\ctx = \mctxp{\suba\mvar\vctx\vartwo\tm}$ and $\var = \mvar$;
\item \label{p:dominating-cases-onefour}$\ctx = \ctxtwop{\dera\evar\vartwo\ctxthree}$, $\var = \evar$, and $\evar \notin\dfv\ctxtwo$.
\end{enumerate}
\end{lem}

\begin{proof}
The first three points are simply the observation that if $\var=\mvar$ is multiplicative and $\mvar\in\dfv\ctx$ then 
\begin{itemize}
\item $\mvar\in\fv\ctx$ and so the position of $\mvar$ has to be a multiplicative context $\mctx$;
\item the occurrence has to be a par or subtraction occurrence by definition of $\dfv\ctx$ (that is, it cannot be a variable occurrence).
\end{itemize}
The fourth point is similar but it says something more so we need a bit more. It is by induction on $\ctx$. If $\var=\evar$  is exponential and $\evar\in\dfv\ctx$ then $\ctx = \ctxtwo'\ctxholep{\dera\evar\vartwo\ctxthree'}$. Now, if $\evar\notin\dfv{\ctxtwo'}$ the statement holds with $\ctxtwo \defeq \ctxtwo'$. Otherwise, by \ih, $\ctxtwo' = \ctxtwo''\ctxholep{\dera\evar\varthree\ctxthree''}$ with $\evar\notin\dfv{\ctxtwo''}$. Then $\ctx =  \ctxtwo''\ctxholep{\dera\evar\varthree\ctxthree''\ctxholep{\dera\evar\vartwo\ctxthree'}}$, which satisfies the statement with respect to $\ctxtwo \defeq \ctxtwo''$.
\end{proof}

\begin{prop}[Fullness]
\label{prop:good-fullness}
Let $\tm$ be clash-free. If $\tm$ is not $\toms$-normal then $\tm \tog\tmtwo$ for some $\tmtwo$.
\end{prop}
\input{\proofspath/strategy/fullness}

\paragraph{Summing Up.} Together with typed strong normalization (anticipating from \refsect{SN}), the properties of the good strategy imply our main result.

\begin{thm}[Good polynomial cost model]
\label{thm:cost-model}
Let $\multiform\vdash \tm\hastype \form$ be a typed term. Then there exist  $k$ and a cut-free term $\tmtwo$ such that $\tm \tog^{k} \tmtwo$. Moreover, such a reduction sequence is implementable on a random access machine in time polynomial in $k$ and $\size\tm$.
\end{thm}

\begin{proof}
By typed strong normalization (\refcoro{sn}), $\tog$ terminates on $\tm$, let $\tmtwo$ be its $\tog$ normal form. Since typable implies clash-freeness (\reflemma{clashfree-implies-cutfree}), $\tm$ is clash-free, and thus so is $\tmtwo$. By fullness (\refprop{good-fullness}) and $\tog$ normality, $\tmtwo$ is $\toms$-normal, which with clash-freeness implies that it is cut-free (\reflemma{clashfree-implies-cutfree}).

Now, about the cost. Let $\tm = \tm_{0} \tog \tm_{1} \tog\ldots \tog \tm_{k} = \tmtwo$. All steps $\tm_{i} \tog \tm_{i+1}$ of the sequence that are multiplicative, as well as all $\toaxetwo$ step, do only manipulation of a constant number of constructors of the term. Therefore, they can clearly be implemented in time polynomial in $\size{\tm_{i}}$, as searching for a redex and checking that its position is good are clearly polynomial problems, and there is at most a linear number of redexes in a term. By a straightforward induction using the sub-term property (\refthm{sub-term-prop}), one obtains that $\size{\tm_{i}} \leq i\cdot\size{\tm_{0}}$, and so all those steps have  cost polynomial in $k$ and $\size{\tm_{0}}$. Again by the sub-term property, the cost of duplication/erasure of values in non-linear steps (namely, $\toaxeone$, $\tobder$, and $\tobw$ steps) is bound by $\tm_{0}$, so they have polynomial cost in $k$ and $\size{\tm_{0}}$ as well.
\end{proof}

About the degree of the polynomial, Accattoli et al. \cite{DBLP:conf/icfp/AccattoliBM14,DBLP:conf/aplas/AccattoliBM15} show that strategies with the sub-term property are usually implementable via abstract machines in time \emph{linear} in both $k$ and $\size\tm$. We expect the linear bound to hold here as well, but we leave the design of an abstract machine to future work.

%In the untyped case, the good strategy can be shown (omitted for lack of space) to be \emph{normalizing} (on clash-free terms), \ie, it reaches a cut-free reduct whenever there is one.

%% file: proofs/strategy/sub-term.tex
% !TEX root = ../../main.tex
%%%%%%%%%%%%%%%
%%%%%%%%%%%%%%% 
\begin{proof}
Cases of $\tm \tog \tmtwo = \bctxp\val$:
\begin{itemize}
	%%%%%%%%%%%%%%
	%%%%%%%%%%%%%%
	\item $\toaxmone$. We have $\tm = \ctxp{\cuta{\mval}\mvar\mctxfp\mvar}  \togp\axmone  \ctxp{\mctxfp\mval} = 
\tmtwo = \bctxp\val$. Since $\ctxp{\cuta{\mval}\mvar\mctx}$ is good, $\ctxp\mctx$ is good by \reflemma{good-ctx-mutilation}, and  by \reflemma{good-ctx-decomp} both $\ctx$ and $\mctx$ are good. By \reflemma{b-not-outer-g}, $\bctx \not\outer \ctx$ and there are two cases:
\begin{itemize}
\item $\ctx \parallel \bctx$. Then there is a double context $\cctx_{\ctx,\bctx}$ such that $\cctx_{\ctx,\bctx}\ctxholep{\cdot, \val} = \ctx$ and\\ $\cctx_{\ctx,\bctx}\ctxholep{\mctxfp\mval,\cdot} = \bctx$. By \reflemma{double-ctx-repl}, the context $\bctxtwo \defeq \cctx_{\ctx,\bctx}\ctxholep{\cuta{\mval}\mvar\mctxfp\mvar,\cdot}$ is bad. Note that $\bctxtwop\val = \cctx_{\ctx,\bctx}\ctxholep{\cuta{\mval}\mvar\mctxfp\mvar,\val} =\ctxp{\cuta{\mval}\mvar\mctxfp\mvar}=\tm$, and so $\bctxtwo$ isolates $\val$ in $\tm$.

\item $\ctx \outer \bctx$. By \reflemma{b-not-outer-g}, $\bctx \not\outer \ctxp\mctx$ and there are two sub-cases:
\begin{itemize}
\item $\ctxp\mctx \parallel \bctx$. Then $\bctx = \ctxp\ctxtwo$ for some $\ctxtwo$ such that $\mctx \parallel \ctxtwo$. There is a double context $\cctx_{\mctx,\ctxtwo}$ such that $\cctx_{\mctx,\ctxtwo}\ctxholep{\cdot,\val} = \mctx$ and $\cctx_{\mctx,\ctxtwo}\ctxholep{\mval,\cdot} = \ctxtwo$. Note that $\ctxp{\cctx_{\mctx,\ctxtwo}}\ctxholep{\mval,\cdot} =_{\reflemmaeq{double-ctx-comp}} \ctxp{\cctx_{\mctx,\ctxtwo}\ctxholep{\mval,\cdot}}=\ctxp\ctxtwo=\bctx$. By \reflemma{double-ctx-repl}, the context $\bctxtwo \defeq \ctxp{\cctx_{\mctx,\ctxtwo}}\ctxholep{\mvar,\cdot}$ is bad. Now, $\bctxtwo =_{\reflemmaeq{double-ctx-comp}} \ctxp{\cctx_{\mctx,\ctxtwo}\ctxholep{\mvar,\cdot}}$ and let $\ctxthree \defeq \cctx_{\mctx,\ctxtwo}\ctxholep{\mvar,\cdot}$, so that  $\bctxtwo = \ctxp\ctxthree$.
Note that $\bctxtwop\val= \ctxp{\cctx_{\mctx,\ctxtwo}\ctxholep{\mvar,\val}}=\ctxp{\mctxp{\mvar}}$ and also $\bctxtwop\val = \ctxp{\ctxthreep\val}$. Then $\tm = \ctxp{\cuta{\mval}\mvar\mctxp{\mvar}} = \ctxp{\cuta{\mval}\mvar\ctxthreep\val}$. Thus,  $\ctxp{\cuta{\mval}\mvar\ctxthree}$ isolates $\val$ in $\tm$ and it is bad because of \reflemmap{good-ctx-mutilation}{bad} and the fact that $\bctxtwo = \ctxp\ctxthree$ is bad.

% \cuta{\mval}\mvar

\item $\ctxp\mctx \outer \bctx$. Then $\bctx = \ctxp{\mctxp\ctxtwo}$ for $\ctxtwo$ such that $\mval = \ctxtwop\val$. The context $\ctxp{\cuta{\ctxtwo}\mvar\mctxfp\mvar}$ is bad and isolates $\val$ in $\tm$.
\end{itemize}

\end{itemize}

	\item $\toaxmtwo$. We have $\tm = \ctxp{\cuta{\mvar}\mvartwo\tmfour}  \togp\axmtwo  \ctxp{\cutsub\mvar\mvartwo\tmfour} = \tmtwo= \bctxp\val$. All bad values of $\tmtwo$ coming from $\ctx$ are obviously also in $\tm$. Let $\bctx = \ctxp{\cutsub\mvar\mvartwo\ctxtwo}$ be the position of a bad value $\val$ in $\tmtwo$. Two cases:
	\begin{itemize}
 \item \emph{$\ctxp\ctxtwo$ is a bad context}. Then $\val$ occurs in $\tm$ as a bad value at position $\ctxp{\cuta\mval\mvar\ctx}$ (which is bad by \reflemmap{good-ctx-mutilation}{bad}). 
 \item \emph{$\ctxp\ctxtwo$ is a good context}. Since it is the addition of $\cutsub\mvar\mvartwo$ that makes $\bctx$ bad, necessarily $\mvar \in \dfv\ctxtwo$. Now, note that then $\cuta\mval\mvar\ctxtwo$ is bad, and so is $\ctxp{\cuta\mval\mvar\ctxtwo}$, which isolates $\val$ in $\tm$, making it a bad value.
 \end{itemize}

%%%%%%%%%%%%%%
	\item $\totens$. We have $\tm = \ctxp{\cuta{\mval}\mvar \mctxp{\para\mvar\var\vartwo \tmfour}}
 \togp\tens  
\ctxp{ \mctxp{ \lctxp{\cuta\valtwo\var \lctxtwop{\cuta\valthree\vartwo \tmfour}}}}  = \tmtwo$ with $\mval = \pair{\lctxp\valtwo}{\lctxtwop\valthree}$ and $\tmtwo = \bctxp\val$. 
 If the hole of $\bctx$ falls in $\ctx$, $\mctx$, or $\lctxp{\cuta\valtwo\var \lctxtwop{\cuta\valthree\vartwo\ctxhole}}$ then the reasoning is as in the $\toaxmone$ case. The only potentially delicate point is about the bad values in $\tmfour$. Let $\bctx = \ctxp{\mctxp{ \lctxp{\cuta\valtwo\var \lctxtwop{\cuta\valthree\vartwo \ctxtwo}}}}$ be the position of one such bad value $\val$. Two cases:
 \begin{itemize}
 \item \emph{$\ctxp{\mctxp \ctxtwo}$ is a bad context}. Then $\val$ occurs in $\tm$ as bad value at position $\ctxp{\cuta\mval\mvar\mctxp\ctxtwo}$. 

 \item \emph{$\ctxp{\mctxp \ctxtwo}$ is a good context}. Note that the only variables of $\ctxtwo$ that can be captured by $\lctxp{\cuta\valtwo\var \lctxtwop{\cuta\valthree\vartwo\ctxhole}}$ are $\var$ and $\vartwo$. Since it is the addition of $\lctxp{\cuta\valtwo\var \lctxtwop{\cuta\valthree\vartwo\ctxhole}}$ that makes $\bctx$ bad, necessarily $\var \in \dfv\ctxtwo$ or $\vartwo \in \dfv\ctxtwo$. Now, note that then $\mvar\in\dfv{\mctxp{\para\mvar\var\vartwo \ctxtwo}}$ and so that $\ctxp{\cuta\mval\mvar\mctxp{\para\mvar\var\vartwo \ctxtwo}}$ is a bad context isolating $\val$ in $\tm$.
 \end{itemize}
%%%%%%%%%%%
\item $\tololli$. As for $\totens$.

%%%%%%%%%
\item $\toaxeone$.   We have $\tm = \ctxp{\cuta{\exval} \evar \ctxtwofp{\evar}}
 \togp\axeone
\ctxp{\cuta{\exval}\evartwo\ctxtwofp{\exval}} = \tmtwo= \bctxp\val
$. Since $\ctxp{\cuta{\exval}\evar\ctxtwo}$ is good, by \reflemma{b-not-outer-g}, $\bctx \not\outer \ctxp{\cuta{\exval}\evar\ctxtwo}$ and there are two cases:
\begin{itemize}
\item $\ctxp{\cuta{\exval}\evar\ctxtwo} \parallel \bctx$. Then there is a double context $\cctx_{\ctxp{\cuta{\exval}\evar\ctxtwo},\bctx}$ such that\\ $\cctx_{\ctxp{\cuta{\exval}\evar\ctxtwo},\bctx}\ctxholep{\exval,\cdot} = \bctx$. By \reflemma{double-ctx-repl}, the context $\bctxtwo \defeq \cctx_{\ctxp{\cuta{\exval}\evar\ctxtwo},\bctx}\ctxholep{\evar,\cdot}$ is bad. Note that $\bctxtwo$ isolates $\val$ in $\tm$.

\item $\ctxp{\cuta{\exval}\evar\ctxtwo} \outer \bctx$.  Then $\bctx = \ctxp{\cuta{\exval}\evar\ctxtwop\ctxthree}$ for $\ctxthree$ such that $\exval = \ctxthreep\val$. Note that the context $\ctxp{\cuta{\ctxthree}\evar\ctxtwofp\evar}$ is bad and isolates $\val$ in $\tm$.
\end{itemize}

%%%%%%%%%%%%%
	\item $\toaxetwo$.  We have $\tm = \ctxp{\cuta{\evartwo}\evar\ctxtwop{\dera\evar\var \tmthree}}
 \togp\axetwo 
\ctxp{\cuta{\evartwo}\evar\ctxtwop{\dera\evartwo\var \tmthree}} = \tmtwo= \bctxp\val$ with $\ctxtwo$ not capturing $\evar$ nor $\evartwo$.
If the hole of $\bctx$ falls in $\ctxp{\cuta{\evartwo}\evar\ctxtwo}$ then the reasoning is as in the $\toaxeone$ case. The only potentially delicate point is about the bad values in $\tmthree$. Let $\bctx = \ctxp{\cuta{\evartwo}\evar\ctxtwop{\dera\evartwo\var \ctxthree}}$ be the position of one such bad value $\val$. The same value is isolated in $\tm$ by $\bctxtwo = \ctxp{\cuta{\evartwo}\evar\ctxtwop{\dera\evar\var \ctxthree}}$. Note that  $\bctxtwo$ cannot be good: since the only thing that has changed is the renaming of $\evartwo$ as $\evar$ in the dereliction, then:
\begin{itemize}
\item either $\evartwo$ is irrelevant for the badness of $\bctx$, and so $\bctxtwo$ is bad, 
\item or $\bctx$ is bad because $\evartwo\in\dfv{\dera\evartwo\var \ctxthree}$ and $\evartwo$ is cut in $\ctx$. Note that then $\evar\in\dfv{\dera\evar\var \ctxthree}$ and $\evar$ is cut in $\bctxtwo$, so that $\bctxtwo$ is bad.
\end{itemize} 
%%%%%%%%%%%%%
	\item $\tobder$.  We have $\tm =  \ctxp{\cuta{\bang\tmthree}\evar\ctxtwop{\dera\evar\var \tmfour}}
 \togp\bder 
\ctxp{\cuta{\bang\tmthree}\evar\ctxtwop{\lctxp{\cuta\valtwo\var \tmfour}}}
 = \tmtwo= \bctxp\val$ with $\tmthree = \lctxp\valtwo$. If the hole of $\bctx$ falls in $\ctxp{\cuta{\bang\tmthree}\evar\ctxtwo}$ then the reasoning is as in the $\toaxeone$ case. If it falls in $\lctxp{\cuta\valtwo\var\ctxhole}$, then a bad context with the hole in $\tmthree$ isolates the same value in $\tm$. The only potentially delicate point is about the bad values in $\tmfour$. Let $\bctx = \ctxp{\cuta{\bang\tmthree}\evar\ctxtwop{\lctxp{\cuta\valtwo\var \ctxthree}}}$ be the position of one such bad value $\val$. The same value is isolated in $\tm$ by $\bctxtwo = \ctxp{\cuta{\bang\tmthree}\evar\ctxtwop{\dera\evar\var \ctxthree}}$ which cannot be good because:
 \begin{itemize}
\item either $\var$ is irrelevant for the badness of $\bctx$, and so $\bctxtwo$ is bad because $\var$ is the only variable of $\ctxthree$ that can be captured by $\lctxp{\cuta\valtwo\var\ctxhole}$ (which is the only difference between $\bctx$ and $\bctxtwo$);
\item or $\bctx$ is bad because $\var\in\dfv{\ctxthree}$ and $\var$ is cut in $\bctx$. Note that then $\evar\in\dfv{\dera\evar\var \ctxthree}$ and $\evar$ is cut in $\bctxtwo$, so that then $\bctxtwo$ is also bad.
\end{itemize} 

%%%%%%%%%%%%%
	\item $\tow$.  We have $\tm =  \ctxp{\cuta{\exval}\evar \tmfour}
 \togp\bw  
\ctxp{\tmfour}
 = \tmtwo= \bctxp\val$ with $\evar\notin\fv\tmfour$. Since $\ctxp{\cuta{\exval}\evar\ctxhole}$ is good, $\ctx$ is good. By \reflemma{b-not-outer-g}, $\bctx \not\outer \ctx$ and there are two cases:
\begin{itemize}
\item $\ctx \parallel \bctx$. Then there is a double context $\cctx_{\ctx,\bctx}$ such that $\cctx_{\ctx,\bctx}\ctxholep{\tmfour,\ctxhole} = \bctx$. By \reflemma{double-ctx-repl}, the context $\bctxtwo \defeq \cctx_{\ctx,\bctx}\ctxholep{\cuta{\exval}\evar\tmfour,\cdot}$ is bad. Note that $\bctxtwo$ isolates $\val$ in $\tm$.

\item $\ctx \outer \bctx$.  Then $\bctx = \ctxp\ctxthree$ for $\ctxthree$ such that $\tmfour = \ctxthreep\val$. Note that the context $\ctxp{\cuta\exval\evar\ctxthree}$ is bad and isolates $\val$ in $\tm$.\qedhere
\end{itemize}

\end{itemize}
\end{proof}

%% file: proofs/strategy/good-diamond.tex
% !TEX root = ../../main.tex
%%%%%%%%%%%%%%%
%%%%%%%%%%%%%%% 
\begin{proof}
For the proof to go through, we need to slightly strengthen the statement adding a clause about dominating free variables, for the inductive cases. We also reformulate the statement using different meta variables. The new statement is:

\begin{center}
If  $\gctx_1: \tm_{0} \tog \tm_{1}$ and  $\gctx_2: \tm_{0} \tog \tm_{2}$ with $\tm_{1}\neq\tm_{2}$ then that there exists $\tm_3$ such that $\gctxtwo_2:\tm_{1} \tog \tm_{3}$ with $\dfv{\gctxtwo_2} = \dfv{\gctx_2}$ and $\gctxtwo_1:\tm_{2} \tog \tm_{3}$ with $\dfv{\gctxtwo_1} = \dfv{\gctx_1}$.
\end{center}

Let $a,b \in \set{\axmtwo,\axmone,\tens,\lolli, \axeone, \axetwo, \bang, \wsym}$, $\tm_{0} \togp{a} \tm_{1}$, and $\tm_{0} \togp{b} \tm_{2}$. The proof is by induction on $\tm_{0} \togp{a} \tm_{1}$, that is, by induction on the context closing the root rule $\rootRew{a}$, and case analysis of $\tm_{0} \togp{b} \tm_{2}$. It is a check of diagrams similar to the proof of local confluence. As for local confluence, the deformation lemmas are used, but their statements have to be strengthened because they have to preserve goodness and dominating variables. About the diagrams, there are less cases, none of which is surprising. Details are in \cite{DBLP:journals/corr/abs-lics}.
\end{proof}

%% file: proofs/strategy/fullness.tex
% !TEX root = ../../main.tex
\begin{proof}
By induction on $\tm$. For variables the statement trivially holds, because they are $\toms$-normal. All inductive cases but cut immediately follow from the \ih Then let $\tm = \cuta\val\var\tmthree$. Two cases:
\begin{enumerate}
\item \emph{$\var\notin\fv\tmthree$}. Then $\tm = \cuta\val\var\tmthree \tobw \tmthree$. Note that this step is good, so the statement holds with $\tmtwo \defeq \tmthree$.
\item \emph{$\var\in\fv\tmthree$}. Two sub-cases:
\begin{enumerate}
\item \emph{$\tmthree$ is $\toms$-normal}. Let $\tmthree = \ctxp{\tmthree_{\var}}$ with $\tmthree_{\var}$ an outermost occurrence of $\var$ in $\tmthree$ (clearly $\ctx$ does not capture $\var$). We show that the context $\cuta\val\var\ctx$ is good. By \reflemma{clashfree-implies-cutfree}, $\tmthree$ is cut-free (because $\tm$ is clash-free and $\toms$-normal). Since there are no cuts in $\ctx$, $\ctx$ is good. By \reflemma{outermost-and-domination}, $\var\notin\dfv\ctx$. Then $\cuta\val\var\ctx$ is good. By clash-freeness, $\cuta\val\var\ctx$ is the position of a redex, which is then good.

\item \emph{$\tmthree$ is not $\toms$-normal}. By \ih, $\tmthree$ has a good redex. Let its position be $\gctx$, with $\tmthree = \gctxp\tmfour$. If $\cuta\val\var\gctx$ is a good context then it induces the redex for the statement. If instead it is a bad context then it must be that $\var\in\dfv\gctx$. By \reflemma{dominating-cases}, $\gctx$ can have 4 possible shapes:
\begin{enumerate}
\item \emph{$\var$ has a par occurrence}: $\gctx = \mctxp{\para\mvar\vartwo\varthree\ctx}$ and $\var = \mvar$. By \reflemma{good-ctx-decomp} and the fact that $\gctx$ is good, $\mctx$ is good. Then $\cuta\val\mvar\mctx$ is good, because (being a multiplicative variable) $\mvar\notin\fv\mctx \supseteq\dfv\mctx$. Since $\tm$ has no clashes, $\val$ is a tensor pair and $\cuta\val\mvar\mctx$ is the position of a good redex.

\item \emph{$\var$ has a subtraction occurrence 1}: $\gctx = \mctxp{\suba\mvar\val\vartwo\ctx}$ and $\var = \mvar$. Analogously to the par occurrence.
\item \emph{$\var$ has a subtraction occurrence 2}: $\gctx = \mctxp{\suba\mvar\ctx\vartwo\tm}$ and $\var = \mvar$. Analogously to the par occurrence.
\item \emph{$\var$ has a dereliction occurrence}: $\gctx = \ctxp{\dera\evar\ctxtwo}$, $\var = \evar$, and $\evar\notin\dfv\ctx$. This case is also similar to the par occurrence, the only difference is that $\evar\notin\dfv\ctx$ is given by \reflemmap{dominating-cases}{onefour} instead that obtained by linearity as in the multiplicative cases.\qedhere
\end{enumerate}
\end{enumerate}

\end{enumerate}
\end{proof}

%% file: 12-Untyped_PSN.tex
% !TEX root = main.tex
\section{Untyped Preservation of Strong Normalization}
\label{sect:PSN}
In the ESs literature, the crucial property for a calculus with ESs is \emph{preservation of $\beta$ strong normalization} (shortened to PSN) on untyped terms, that is, that if a $\l$-term is $\sn\beta$ then it is SN when evaluated with ESs. In our setting, it is rephrased as $\tm\in\sn\toss$ implies $\tm\in\sn\toms$. The property ensures that switching to micro-steps does not introduces unexpected divergence, as it was surprisingly the case, historically, for the first $\l$-calculus with ESs by Abadi et al. \cite{DBLP:journals/jfp/AbadiCCL91}, as shown by \mellies \cite{DBLP:conf/tlca/Mellie95}.

Proofs of PSN are often demanding, see Kesner's survey \cite{DBLP:conf/csl/Kesner07}. Rewriting rules at a distance enable short proofs, as first showed by Accattoli and Kesner \cite{DBLP:conf/csl/AccattoliK10}. We adapt and further simplify that proof, that here rests on just two natural properties of micro steps, here referred to as \emph{extension} and \emph{root cut expansion} (Kesner calls the latter \emph{IE property}), the latter resting on \emph{structural stability of SN}, that is, the fact that cut equivalence $\cuteq$ preserves SN (\refpropp{cuteq-bisim}{sn}). 

Extension and root cut expansion shall be used also in the next section, for proving typed strong normalization for $\toss$, which is why the statements are given as to cover both $\toms$ (relevant here for untyped PSN) and $\toss$ (relevant in the next section for typed SN).

\paragraph{Extension.} The \emph{extension} property is the easy fact that, for all root constructs but (non clashing) cuts, SN follows from SN of the root sub-terms.

\begin{lem}[Extension]
\label{l:extension} % \reflemmap{substitutivity-of-red}{four}
Let $\to\in\set{\toss,\toms}$ and $\tm, \tmtwo, \val \in \snsubst$. 
Then:
\begin{enumerate}
\item \emph{Neutral}:
 $\pair\tm\tmtwo$, $\la\var\tm$, $\bang\tm$, $\para\mvar\var\vartwo \tm$,  $\suba\mvar\val\var  \tm$, and $\dera\evar\var \tm$ are in $\snsubst$;
 \item \emph{Clash}: if the root cut of $\cuta\val\var\tm$ is clashing then $\cuta\val\var\tm\in\snsubst$.
 \end{enumerate}
\end{lem}
\input{\proofspath/psn/extension-psn}

\paragraph{Root Cut Expansion.} The \emph{root cut expansion} property is the less obvious fact that if the reduct of a root \emph{small-step} cut is $\sn\toms$ then the reducing term also is. Kesner isolated the importance of this property for PSN \cite{DBLP:journals/corr/abs-0905-2539}, calling it \emph{IE property}. Accattoli and Kesner \cite{DBLP:conf/csl/AccattoliK10,DBLP:conf/rta/Accattoli13} show that, with rules \emph{at a distance}, it is proved via simple inductions (here even simpler than in \cite{DBLP:conf/csl/AccattoliK10,DBLP:conf/rta/Accattoli13} thanks to LSC-style duplication and structural stability), while commutative rules (such as Girard's box commutation $\Rew{\boxbox}$ for proof nets) break those inductions and require more complex techniques.

\begin{prop}[Root cut expansion]
\label{prop:root-sn-expansion} % \reflemmap{sn-expansion}{six}
Let $\to\in\set{\toss,\toms}$
\begin{enumerate}
\item \emph{Multiplicative}: if $\tm \rtom \tmtwo$ and $\tmtwo\in \sn\to$ then $\tm\in\sn\to$.

\item \emph{Exponential}: %\label{p:sn-expansion-six}
 if $\exval\in\sn\to$ and $\cutsub\exval\evar\tm\in\sn\to$ then 
$\cuta{\exval}\evar \tm \in\sn\to$.
\end{enumerate}
\end{prop}

\input{\proofspath/psn/root-expansion}

\paragraph{PSN} Now, PSN is proved via a very easy induction over $\tm\in\sn\toss$ and the size $ \size\tm$ of $\tm$, using the two properties.

\begin{thm}[Untyped PSN]
\label{thm:PSN}
If $\tm\in\sn{\toss}$ then $\tm\in\sn{\toms}$.
\end{thm}

\input{\proofspath/psn/PSN}

The next section proves typed SN for $\toss$, which then transfers to $\toms$ by PSN. The relevance of PSN, however, is that it holds in the \emph{untyped} setting, thus it transfers SN also if proved with respect to other type systems, such as intersection types, polymorphic types, and so on.

%% file: proofs/psn/extension-psn.tex
% !TEX root = ../../main.tex
\begin{proof}
By induction on (in the cases with more than one induction the order between the components is irrelevant):
\begin{itemize}
\item $\tm \in \sn\to$ for $\la\var\tm$, $\bang\tm$, $\para\mvar\var\vartwo \tm$, and $\dera\evar\var 
\tm$,
\item $(\tm\in \sn\to, \tmtwo\in \sn\to)$ for $\pair\tm\tmtwo$,
\item $(\tm\in \sn\to, \val \in \sn\to)$ for $\suba\mvar\val\var\tm$,
\item $(\tm\in \sn\to, \mval \in \sn\to)$ for $\cuta\mval\evar\tm$,
\item $(\tm\in \sn\to, \exval \in \sn\to)$ for $\cuta\exval\mvar\tm$.
\end{itemize}
In each case one shows that all reducts are in $\snsubst$, which follows immediately from the \ih, because there cannot be interaction between the immediate sub-terms.
\end{proof}

%% file: proofs/psn/root-expansion.tex
% !TEX root = ../../main.tex
\begin{proof}
\hfill
\begin{enumerate}
\item By induction on $\tmtwo\in\sn\to$, showing that any reduct of $\tm$ is in $\sn\to$. Cases of the multiplicative root step:
\begin{itemize}

\item $\rtoaxmone$, that is, $\tm = \cuta{\mval}\mvar\mctxfp\mvar \rtoaxmone \mctxfp\mval = \tmtwo$. If $\tm \to \tm'$ by reducing the root cut then $\tm'=\tmtwo\in\sn\to$ by hypothesis. The other cases of $\tm \to\tm'$ are:
\begin{itemize}
\item \emph{Reduction of a cut of $\mctxfp\mvar$}: that is, $\cuta{\mval}\mvar\mctxfp\mvar \to \cuta{\mval}\mvar\tmthree$ because $\mctxfp\mvar \to \tmthree$. By deformation (\reflemmap{aux-props-loc-confl-strong-comm}{two}), there are two cases:
\begin{enumerate}
\item $\tmthree = \mctxtwofp\mvar$ and  $\mctxfp{\mval} \to \mctxtwofp{\mval}$. Then we have the following diagram:
\begin{center}
\begin{tikzpicture}[ocenter]
		\node at (0,0)[align = center](source){\normalsize$\cuta{\mval}\mvar\mctxfp\mvar $};
		\node at (source.east)[right = 40pt](source-right){\normalsize $\mctxfp\mval$};	
		\node at (source.center)[below = 20pt](source-down){\normalsize $\cuta{\mval}\mvar\mctxtwofp\mvar $};
		
		\node at (source-right|-source-down)(target){\normalsize $\mctxtwofp{\mval}$};
		
		\draw[|->](source) to node[above] {\scriptsize $\axmone $} (source-right);
		\draw[->](source) to node[left] {\scriptsize $ $}(source-down);

		\draw[->, dotted](source-right) to node[right] {\scriptsize $ $}(target);
		\draw[|->, dotted](source-down) to node[above] {\scriptsize $\axmone $} (target);
		
		\node at (source.west)[above=1pt, anchor=east](source-eq){\normalsize $\tm\ =$};	
		\node at (source-right.east)[anchor=west](source-right-eq){\normalsize $=\ \tmtwo$};	
		\node at (source-down.west)[above=1pt, anchor=east](source-down-eq){\normalsize $\tm'\ =$};	
		\node at (target.east)[anchor=west](target-eq){\normalsize $=:\ \tmtwo'$};	
	\end{tikzpicture}
\end{center}
Note that $\tmtwo'\in\sn\to$ because $\tmtwo\in \sn\to$. By \ih applied to the bottom side of the diagram,  $\tm'  \in\sn\to$.

\item $\mctx = \mctxtwop{\cuta\ctxhole\mvartwo\tmthree'}$ for some 
$\mctxtwo$ and the step $\mctxfp\mvar \to \tmthree$ reduces the cut on $\mvartwo$, that is, $\tmthree' = \mctxthreep{\tmthree_{\mvartwo}}$ for some $\mctxthree$ and variable occurrence $\tmthree_{\mvartwo}$, and $\mctxfp\mvar = \mctxtwop{\cuta\mvar\mvartwo\mctxthreep{\tmthree_{\mvartwo}}} \toaxmtwo \mctxtwop{\mctxthreep{\cutsub\mvar\mvartwo\tmthree_{\mvartwo}}} =\tmthree$. We have:
\begin{center}
\begin{tikzpicture}[ocenter]
		\node at (0,0)[align = center](source){\normalsize$\cuta{\mval}\mvar\mctxtwop{\cuta\mvar\mvartwo\mctxthreep{\tmthree_{\mvartwo}}}$};
		\node at (source.east)[right = 40pt](source-right){\normalsize $\mctxtwop{\cuta\mval\mvartwo\mctxthreep{\tmthree_{\mvartwo}}}$};	
		\node at (source.center)[below = 20pt](source-down){\normalsize $\cuta{\mval}\mvar\mctxtwop{\mctxthreep{\cutsub\mvar\mvartwo\tmthree_{\mvartwo}}}$};
		
		\draw[|->](source) to node[above] {\scriptsize $\axmone $} (source-right);
		\draw[->](source) to node[left] {\scriptsize $ $}(source-down);
		
		\node at (source.west)[above=1pt, anchor=east](source-eq){\normalsize $\tm\ =$};	
		\node at (source-right.east)[anchor=west](source-right-eq){\normalsize $=\ \tmtwo$};	
		\node at (source-down.west)[above=1pt, anchor=east](source-down-eq){\normalsize $\tm'\ =$};	
	\end{tikzpicture}
\end{center}
Note that $\tm' \cuteq \tmtwo$. By $\tmtwo \in \sn\to$ and structural stability of SN (\refprop{cuteq-bisim}), we obtain $\tm' \in \sn\to$.
%
%Since $\tm$ is clash-free, $\tm'$ can reduce the cut on $\mvar$ and $\tmtwo$ can reduce the cut on $\mvartwo$. Since the reduct of a multiplicative cut does not depend on the position of the cut nor on the name of the cut variable, both step reduce to the same term $\tmfour$ using the same kind of rewriting rule. Then:
%\begin{center}$\begin{array}{cccccccccc}
%\tm &= &\cuta{\mval}\mvar\mctxtwop{\cuta\mvar\mvartwo\mctxthreep{\tmthree_{\mvartwo}}}
%& \rtoaxmone &
%\mctxtwop{\cuta\mval\mvartwo\mctxthreep{\tmthree_{\mvartwo}}}
%& = & \tmtwo
%\\
%&&\downarrow_{\axmtwo} & & \downarrow_{\msym}
%\\
%\tm'&=&\cuta{\mval}\mvar\mctxtwop{\mctxthreep{\cutsub\mvar\mvartwo\tmthree_{\mvartwo}}}
%&  \tom & \tmfour
%\end{array}$\end{center}
%By \ih applied to the bottom side of the diagram, we obtain $\cuta{\mval}\mvar\mctxtwop{\mctxthreep{\cutsub\mvar\mvartwo\tmthree_{\mvartwo}}}  \in\sn\toms$.
\end{enumerate}

\item \emph{Reduction of a cut of $\mval$}: then $\cuta{\mval}\mvar\mctxfp\mvar  \to  \cuta{\mvaltwo}\mvar\mctxfp\mvar$ with $\mval \to \mvaltwo$. We have the following diagram:
\begin{center}
\begin{tikzpicture}[ocenter]
		\node at (0,0)[align = center](source){\normalsize$\cuta\mval\mvar \mctxfp\mvar  $};
		\node at (source.east)[right = 40pt](source-right){\normalsize $\mctxfp\mval$};	
		\node at (source.center)[below = 20pt](source-down){\normalsize $\cuta\mvaltwo\mvar \mctxfp\mvar$};
		
		\node at (source-right|-source-down)(target){\normalsize $\mctxfp\mvaltwo$};
		
		\draw[|->](source) to node[above] {\scriptsize $\axmone $} (source-right);
		\draw[->](source) to node[left] {\scriptsize $ $}(source-down);

		\draw[->, dotted](source-right) to node[right] {\scriptsize $ $}(target);
		\draw[|->, dotted](source-down) to node[above] {\scriptsize $\axmone $} (target);
		
		\node at (source.west)[above=1pt, anchor=east](source-eq){\normalsize $\tm\ =$};	
		\node at (source-right.east)[anchor=west](source-right-eq){\normalsize $=\ \tmtwo$};	
		\node at (source-down.west)[above=1pt, anchor=east](source-down-eq){\normalsize $\tm'\ =$};	
		\node at (target.east)[anchor=west](target-eq){\normalsize $=:\ \tmtwo'$};
	\end{tikzpicture}
\end{center}
Note that $\tmtwo'\in\sn\to$ because $\tmtwo\in \sn\to$. By \ih applied to the bottom side of the diagram,  $\tm'  \in\sn\to$.
\end{itemize}

%%%%%%%%%%%
\item  $\rtoaxmtwo$, that is, $\tm = \cuta\vartwo\var\tmthree \rtoaxmtwo \cutsub\vartwo\var\tmthree = \tmtwo$. If $\tm \to \tm'$ by reducing the root cut then $\tm'=\tmtwo\in\sn\to$ by hypothesis. The other cases of $\tm \to\tm'$ are:
\begin{itemize}
\item \emph{Reduction of a cut of $\tmthree$}, that is $\cuta\vartwo\var\tmthree \to \cuta\vartwo\var\tmthree'$ because $\tmthree \to \tmthree'$. We have:
\begin{center}
\begin{tikzpicture}[ocenter]
		\node at (0,0)[align = center](source){\normalsize$\cuta\vartwo\var\tmthree$};
		\node at (source.east)[right = 40pt](source-right){\normalsize $\cutsub\vartwo\var\tmthree$};	
		\node at (source.center)[below = 20pt](source-down){\normalsize $\cuta\vartwo\var\tmthree'$};
		
		\node at (source-right|-source-down)(target){\normalsize $\cutsub\vartwo\var\tmthree'$};
		
		\draw[|->](source) to node[above] {\scriptsize $\axmtwo $} (source-right);
		\draw[->](source) to node[left] {\scriptsize $ $}(source-down);

		\draw[->, dotted](source-right) to node[right] {\scriptsize $ $}(target);
		\draw[|->, dotted](source-down) to node[above] {\scriptsize $\axmtwo $} (target);
		
		\node at (source.west)[above=1pt, anchor=east](source-eq){\normalsize $\tm\ =$};	
		\node at (source-right.east)[anchor=west](source-right-eq){\normalsize $=\ \tmtwo$};
		\node at (source-down.west)[above=1pt, anchor=east](source-down-eq){\normalsize $\tm'\ =$};	
		\node at (target.east)[anchor=west](target-eq){\normalsize $=:\ \tmtwo'$};	
	\end{tikzpicture}
\end{center}
where the step from $\tmtwo$ is given by the stability of steps under renaming (\reflemmap{deformations}{rename}). Note that $\tmtwo'\in\sn\to$ because $\tmtwo\in \sn\to$. By \ih applied to the bottom side of the diagram,  $\tm'  \in\sn\to$.
\end{itemize}

%%%%%%%%%%%%
\item $\rtotens$, that is, $\tm = \cuta\mval\mvar \mctxp{\para\mvar\var\vartwo \tmfive}
 \rtotens 
\mctxp{\lctxp{\cuta\val\var \lctxtwop{\cuta\valtwo\vartwo \tmfive}}}$
with $\mval = \pair\tmthree\tmfour = \pair{\lctxp\val}{\lctxtwop\valtwo}$. If $\tm \to \tm'$ by reducing the root cut then $\tm'=\tmtwo\in\sn\to$ by hypothesis. The other cases of $\tm \to\tm'$ are:
\begin{itemize}
\item \emph{Reduction of a cut of $\mctxp{\para\mvar\var\vartwo \tmfive}$}: then $\tm = \cuta\mval\mvar \mctxp{\para\mvar\var\vartwo \tmfive} \to \cuta\mval\mvar \tmsix$ because $\mctxp{\para\mvar\var\vartwo \tmfive} \to \tmsix$. By deformation (\reflemmap{aux-props-loc-confl-strong-comm}{three}), $\tmsix = \mctxtwop{\para\mvar\var\vartwo \tmfive'}$ for some $\mctxtwo$ and $\tmfive'$ and also $\mctxp{\lctxp{\cuta\val\var \lctxtwop{\cuta\valtwo\vartwo \tmfive}}} \to \mctxtwop{\lctxp{\cuta\val\var \lctxtwop{\cuta\valtwo\vartwo \tmfive'}}}$. Then we have the following diagram:
\begin{center}
\begin{tikzpicture}[ocenter]
		\node at (0,0)[align = center](source){\normalsize$\cuta\mval\mvar \mctxp{\para\mvar\var\vartwo \tmfive}$};
		\node at (source.east)[right = 40pt](source-right){\normalsize $\mctxp{\lctxp{\cuta\val\var \lctxtwop{\cuta\valtwo\vartwo \tmfive}}}$};	
		\node at (source.center)[below = 20pt](source-down){\normalsize $\cuta\mval\mvar \mctxtwop{\para\mvar\var\vartwo \tmfive'}$};
		
		\node at (source-right|-source-down)(target){\normalsize $\mctxtwop{\lctxp{\cuta\val\var \lctxtwop{\cuta\valtwo\vartwo \tmfive'}}}$};
		
		\draw[|->](source) to node[above] {\scriptsize $\tens $} (source-right);
		\draw[->](source) to node[left] {\scriptsize $ $}(source-down);

		\draw[->, dotted](source-right) to node[right] {\scriptsize $ $}(target);
		\draw[|->, dotted](source-down) to node[above] {\scriptsize $\tens $} (target);
		
		\node at (source.west)[above=1pt, anchor=east](source-eq){\normalsize $\tm\ =$};	
		\node at (source-right.east)[anchor=west](source-right-eq){\normalsize $=\ \tmtwo$};
		\node at (source-down.west)[above=1pt, anchor=east](source-down-eq){\normalsize $\tm'\ =$};	
		\node at (target.east)[anchor=west](target-eq){\normalsize $=:\ \tmtwo'$};	
	\end{tikzpicture}
\end{center}
Note that $\tmtwo'\in\sn\to$ because $\tmtwo\in \sn\to$. By \ih applied to the bottom side of the diagram,  $\tm'  \in\sn\to$.

\item \emph{Reduction of a cut of $\mval$}: if the redex affects only $\val$, $\valtwo$, $\lctx$, or $\lctxtwo$, then the obvious diagram together with the \ih gives the statement. 

If the redex is given by a cut in $\lctx$ acting on a variable occurrence in $\val$, that is, $\mval = \pair{\lctxp\val}{\lctxtwop\valtwo} \to \pair{\lctxthreep\valthree}{\lctxtwop\valtwo} = \mvaltwo$ then we have $\lctxp{\cuta\val\var \lctxtwop{\cuta\valtwo\vartwo \tmfive'}} \to \lctxthreep{\cuta\valthree\var \lctxtwop{\cuta\valtwo\vartwo \tmfive'}}$. The diagram is as follows:
\begin{center}
\begin{tikzpicture}[ocenter]
		\node at (0,0)[align = center](source){\normalsize$\cuta\mval\mvar \mctxp{\para\mvar\var\vartwo \tmfive}$};
		\node at (source.east)[right = 60pt](source-right){\normalsize $\mctxp{\lctxp{\cuta\val\var \lctxtwop{\cuta\valtwo\vartwo \tmfive}}}$};	
		\node at (source.center)[below = 20pt](source-down){\footnotesize $\cuta{\pair{\lctxthreep\valthree}{\lctxtwop\valtwo}}\mvar \mctxtwop{\para\mvar\var\vartwo \tmfive'}$};
		
		\node at (source-right|-source-down)(target){\normalsize $\lctxthreep{\cuta\valthree\var \lctxtwop{\cuta\valtwo\vartwo \tmfive'}}$};
		
		\draw[|->](source) to node[above] {\scriptsize $\tens $} (source-right);
		\draw[->](source) to node[left] {\scriptsize $ $}(source-down);

		\draw[->, dotted](source-right) to node[right] {\scriptsize $ $}(target);
		\draw[|->, dotted](source-down) to node[above] {\scriptsize $\tens $} (target);
		
		\node at (source.west)[above=1pt, anchor=east](source-eq){\normalsize $\tm\ =$};	
		\node at (source-right.east)[anchor=west](source-right-eq){\normalsize $=\ \tmtwo$};
		\node at (source-down.west)[above=1pt, anchor=east](source-down-eq){\normalsize $\tm'\ =$};	
		\node at (target.east)[anchor=west](target-eq){\normalsize $=:\ \tmtwo'$};	
	\end{tikzpicture}
\end{center}
Note that $\tmtwo'\in\sn\to$ because $\tmtwo\in \sn\to$. By \ih applied to the bottom side of the diagram,  $\tm'  \in\sn\to$.

If the redex is given by a cut in $\lctxtwo$ acting on a variable occurrence in $\valtwo$ the reasoning is analogous.
\end{itemize}

%%%%%%%%%%%%
\item $\rtololli$: as the previous case, simply using \reflemmap{aux-props-loc-confl-strong-comm}{four} instead of \reflemmap{aux-props-loc-confl-strong-comm}{three} at the corresponding point of the proof.

\end{itemize}

%%%%%%%%%%%%
%% EXPONENTIAL
%%%%%%%%%%%%%%%
\item We first treat the case of $\toss$. By induction on $(\exval \in \sn{\toss}, \cutsub{\exval}\evar \tm \in \sn{\toss})$, proving that any reduct of $\cuta{\exval}\evar \tm$ is in $\sn{\toss}$. Cases:
\begin{itemize}
\item \emph{Reduction of the root cut}: $\cuta{\exval}\evar \tm \rtobang \cutsub{\exval}\evar \tm$, which is in $\sn{\toss}$ by hypothesis.
\item \emph{Reduction of a cut of $\tm$}: that is, $\cuta{\exval}\evar \tm \toss \cuta{\exval}\evar \tm'$ with $\tm \toss \tm'$. By stability under substitution (\reflemmap{substitutivity-of-red}{four}), $\cutsub{\exval}\evar \tm \toss \cutsub{\exval}\evar \tm'$ and so $\cutsub{\exval}\evar \tm' \in\sn\to$. By \ih (2nd component, the 1st is unchanged), $\cuta{\exval}\evar \tm'\in\sn{\toss}$.

\item \emph{Reduction of a cut of $\exval$}: that is, $\cuta{\exval}\evar \tm \toss \cuta{\exvaltwo}\evar \tm$ with $\exval \toss \exvaltwo$. By \reflemmap{substitutivity-of-red}{four}, $\cutsub{\exval}\evar \tm \toss^{*} \cutsub{\exvaltwo}\evar \tm$, and so $\cutsub{\exvaltwo}\evar \tm \in\sn{\toss}$. By \ih (1st component), $\cuta{\exvaltwo}\evar \tm \in\sn{\toss}$.
\end{itemize}
Now, the case of $\toms$, which is similar. By induction on $(\exval \in \sn{\toms}, \cutsub{\exval}\evar \tm \in \sn{\toms}, \varmeas\tm\evar)$, proving that any reduct of $\cuta{\exval}\evar \tm$ is in $\sn{\toms}$. The cases where reduction takes place in $\tm$ or $\exval$ go exactly as for $\toss$, because by full composition (\refprop{ms-simulates-ss}) $\toms$ simulates $\toss$. What changes is when the root cut is reduced. Cases:
\begin{itemize}
\item $\rtoaxeone$, that is, $\cuta{\exval}\evar \tm = \cuta{\exval}\evar \ctxfp\evar \rtoaxeone \cuta{\exval}\evar \ctxfp{\exval}$. Note that we can assume that $\evar\notin\fv\exval$. Since $\cutsub{\exval}\evar \ctxfp\evar = \cutsub{\exval}\evar \ctxfp{\exval}$ (by \reflemmap{bang-subs-prop}{zero}) and $\varmeas{\ctxfp{\exval}}\evar < \varmeas{\ctxfp{\evar}}\evar$, we can apply the \ih to $\cutsub{\exval}\evar \ctxfp{\exval}$ (3rd component decreases, 1st and 2nd unchanged), obtaining $\cuta{\exval}\evar \ctxfp{\exval} \in\sn{\toms}$.

\item $\rtoaxetwo$, that is, $\exval = \evartwo$ and the step is
$\cuta\evartwo\evar \tm = \cuta\evartwo\evar \ctxp{\dera\evar\var\tmthree} \rtoaxetwo \cuta\evartwo\evar \ctxp{\dera\evartwo\var\tmthree}$. 
Since $\cutsub\evartwo\evar \ctxp{\dera\evar\var\tmthree} = \cutsub\evartwo\evar \ctxp{\dera\evartwo\var\tmthree}$ (by \reflemmap{bang-subs-prop}{zero}) and $\varmeas{\ctxp{\dera\evartwo\var\tmthree}}\evar < \varmeas{\ctxp{\dera\evar\var\tmthree}}\evar$, we can apply the \ih to $\cutsub\evartwo\evar \ctxp{\dera\evartwo\var\tmthree}$ (3rd component decreases, 1st and 2nd unchanged), obtaining $\cuta\evartwo\evar \ctxp{\dera\evartwo\var\tmthree} \in\sn{\toms}$.

\item $\rtobder$, that is, $\exval = \bang\tmtwo = \bang\lctxp\val$ and the step is:
\begin{center}$\begin{array}{ccccc}
\cuta{\bang\tmtwo}\evar \tm &=& \cuta{\bang\tmtwo}\evar \ctxp{\dera\evar\var\tmthree} & \rtobder&\cuta{\bang\tmtwo}\evar \ctxp{\lctxp{\cuta\val\var\tmthree}}
\end{array}$\end{center}
Note that we can assume  $\evar\notin\fv\exval$.
Since $\cutsub{\bang\tmtwo}\evar \ctxp{\dera\evar\var\tmthree} = \cutsub{\bang\tmtwo}\evar \ctxp{\lctxp{\cuta\val\var\tmthree}}$ (by \reflemmap{bang-subs-prop}{zero}) and $\varmeas{\ctxp{\lctxp{\cuta\val\var\tmthree}}}\evar < \varmeas{\ctxp{\dera\evar\var\tmthree}}\evar$, we can apply the \ih to $\cutsub{\bang\tmtwo}\evar \ctxp{\lctxp{\cuta\val\var\tmthree}}$ (3rd component decreases, 1st and 2nd unchanged), obtaining $\cuta{\bang\tmtwo}\evar \ctxp{\lctxp{\cuta\val\var\tmthree}} \in\sn{\toms}$.

\item $\rtow$, that is, $\evar\notin\fv\tm$ and $\cuta{\exval}\evar \tm \rtow \tm$. Note that  $\tm = \cutsub{\exval}\evar \tm$, which is in $\sn{\toms}$ by hypothesis.\qedhere
\end{itemize}
\end{enumerate}
\end{proof}

%% file: proofs/psn/PSN.tex
% !TEX root = ../../main.tex
\begin{proof}
By induction on $(\tm\in\sn\toss, \size\tm)$. Cases of $\tm$:
\begin{itemize}
\item \emph{All cases but cut}: they follow from the \ih and neutral extension (\reflemma{extension}). We give the details of one case. Let $\tm = \suba\mvar\val\var \tmtwo$. By \ih (2nd component), $\val,\tmtwo\in\sn{\toms}$. By extension, $\tm\in\sn{\toms}$.

\item \emph{Root cut}, that is, $\tm = \cuta\val\var\tmtwo$. There are three cases, handled by either clashing extension (\reflemma{extension}) or root cut expansion (\refprop{root-sn-expansion}):
\begin{enumerate}
	\item \emph{The root cut is clashing}. By \ih (2nd component), $\val,\tmtwo\in\sn{\toms}$. By clashing extension, $\tm\in\sn{\toms}$.
	
	\item \emph{The root cut is multiplicative and not clashing}. Then $\tm \rtom \tmthree$. By \ih (1st component), $\tmthree \in \sn{\toms}$. By multiplicative root cut expansion, $\tm \in \sn{\toms}$.
	
	\item \emph{The root cut is exponential}. By \ih (2nd component), $\val\in\sn\toms$. We have $\tm = \cuta{\exval}\evar\tmtwo \toess \cutsub{\exval}\evar\tmtwo$. By \ih (1st component), $\cutsub{\exval}\evar\tmtwo \in \sn{\toms}$. By exponential root cut expansion, $\cuta{\exval}\evar\tmtwo \in \sn{\toms}$. \qedhere
\end{enumerate}
\end{itemize}
\end{proof}

%% file: 13-Typed_Strong_Normalization.tex
% !TEX root = main.tex
\section{Typed Strong Normalization}
\label{sect:SN}
Here we prove strong normalization (SN) of the IMELL-typed small-step ESC (then transferred to the micro-step ESC by PSN) 
using the reducibility method. We do it following the \emph{bi-orthogonal} schema for proof nets  by Girard 
\cite{DBLP:journals/tcs/Girard87}, see also Pagani and Tortora de Falco \cite{DBLP:journals/tcs/PaganiF10}. We actually 
adopt here the variant by Accattoli \cite{DBLP:conf/rta/Accattoli13}, the main point of which is that without 
commutative exponential rules (as it is also the case here) there is a simpler proof of \emph{adequacy} (the key step of 
the method) based on the \emph{root cut expansion property} of the previous section. There are three differences between the 
present proof and \cite{DBLP:conf/rta/Accattoli13} (beyond a neater technical development here):
\begin{enumerate}
\item \emph{Formal}: the study in \cite{DBLP:conf/rta/Accattoli13} uses proof nets and an \emph{informal} term notation 
with explicit and meta-level substitutions. That informal notation is here replaced by a formal one, the ESC;
\item \emph{Exponential rules}: both studies use rules without commutative cases, but here we use LSC-style 
duplications, obtaining a simpler proof of root cut expansion.
\item \emph{Quotient $vs$ cut equivalence}: the quotient on proofs given by proof nets plays a subtle role in the proof 
of adequacy. The argument is here  replaced by structural stability of SN, that is, the stability by cut equivalence $\cuteq$ of SN.
\end{enumerate}
 We are also inspired by Riba's dissection of reducibility \cite{riba:hal-00779623}, and refer to 
\cite{DBLP:conf/rta/Accattoli13} for extensive discussions and references about SN in linear logic. 

To our knowledge, the literature contains only two studies of SN for a linear term calculus: 
\begin{itemize}
\item Benton \cite{DBLP:journals/jfp/Benton95}, who, as he points out, does not deal with commutative cases---that are 
nonetheless crucial in his presentation---thus not really proving SN;
\item P\'erez et al. \cite{DBLP:journals/iandc/PerezCPT14}, where cut elimination is seriously restricted because they 
use it to model process communication, which does not take place under prefixes. Logically, it means that their notion 
of cut elimination does not compute cut-free proofs.
\end{itemize}
 
 \paragraph{Key Properties.} The reducibility method requires a number of definitions, detailed in the next paragraphs. Because of the many technicalities, it is easy to loose sight of what are the crucial concepts at work in the proof. From a high-level perspective, our proof is based only on three key properties of $\toss$:
\begin{enumerate}
	\item \emph{Structural stability of SN}: if $\tm\in\sn\toss$ and $\tm\cuteq\tmtwo$ then $\tmtwo\in\sn\toss$ 
(\refprop{cuteq-bisim}),
	\item \emph{Extension} (\reflemma{extension}), and 
	\item
	\emph{Root cut expansion} (\refprop{root-sn-expansion}).
\end{enumerate}
Since they also hold  for micro-step cut elimination $\toms$, our proof works with only minor changes also for 
micro-steps, without really needing PSN.

\paragraph{Elimination Contexts and Duality.} 
The bi-orthogonal technique we follow is based on a notion of duality defined via \emph{elimination contexts}, that are contexts of the form $\lctxp{\cuta{\ctxhole}\var\tm}$, noted $\elctx$. Types can be extended to contexts by considering $\ctxhole$ as a free variable and typing it via an axiom, as follows:
\begin{center}
\AxiomC{}
	\RightLabel{$\ax_{\text{ctx}}$}
	\UnaryInfC{$  \ctxhole\hastype\form \vdash \ctxhole\hastype \form$}	
	\DisplayProof
\end{center}
Let us set some notations:
\begin{itemize}
\item $\lltermsform$ for the set of terms $\tm$ of type $\form$, that is, such that $\multiForm\vdash \tm\hastype\form$, 
and also  $\vars_{\form} \defeq  \lltermsform \cap \vars$.  
 
\item $\elctxsform \defeq \set{\elctx \ 
|\ \multiForm, \ctxhole\hastype \form \vdash \elctx\hastype \formtwo}$ the set of typed elimination contexts with hole 
of type $\form$, and say that $\elctx$ has \emph{co-type} $\form$. 
\item $\elvars_{\form} \defeq \set{\elctx\in\elctxsform \ |\  \elctx = \cuta\ctxhole\var\var}$.
\item $\elctx \in \sn\toss$ if $\elctx = \lctxp{\cuta\ctxhole\var\tm}$ and $\lctxp{\cutsub\vartwo\var\tm}\in\sn\toss$ 
for every variable $\vartwo$ of the same kind as $\var$ which is not captured by $\lctx$.
%\item $\elctx \in \snA$ if $\elctx \in \sn\toss\cap\elctxsform$.
\end{itemize}
\begin{rem}
\label{rem:sn-renaming}
Checking that $\elctx=\lctxp{\cuta\ctxhole\var\tm}$ is in $\sn\toss$ amounts to prove that \\
$\lctxp{\cutsub\vartwo\var\tm} = \cutsub\vartwo\var\lctxp{\tm}\in\sn\toss$ for every appropriate $\vartwo$. By the stability of SN by renamings (\reflemmap{deformations}{sn}), it is enough to prove that $\lctxp{\tm}\in\sn\toss$.
\end{rem}

\begin{defi}[Duality]
\label{def:duality}
Given a set $\settone \subseteq \lltermsform$ of terms of type $\form$, the \emph{dual} set 
$\settone^\bot\subseteq\elctxsform$ contains the elimination contexts $\elctx$ of co-type 
$\form$ such that  $\elctxfp{\tm}$ is \proper and in $\sn\toss$ for every 
$\tm\in\settone$. The dual of a set of elimination contexts $\setecone \subseteq\elctxsform$ of co-type $\form$ is a set of terms $\setecone^\bot$
defined symmetrically. 
\end{defi}
Note the 
use of $\ctxholef$ in the definition of duality: the plugging in contexts at work in duality does \emph{not} capture the variables of the plugged 
term, this is crucial, and standard in the reducibility method. 

The following properties of duality are standard.

\begin{lem}[Basic properties of duality]
\label{l:basic-prop-duality}
If $\settone \subseteq \lltermsform$ or $\settone\subseteq\elctxsform$ then: 
 \begin{enumerate}
 \item \emph{Closure}: $\settone\subseteq\settone^{\bot\bot}$;
 \item \emph{Bi-orthogonal}: $\settone^{\bot\bot\bot}=\settone^{\bot}$.
% \item \emph{Covariance}: if $\settone\subseteq\setttwo$ then $\setttwo^\bot\subseteq\settone^\bot$.
 \end{enumerate}
 \end{lem}

\paragraph{Generators, Candidates, and Formulas.} The definition of reducibility candidates comes together with a notion 
of generator, justified by the proposition that follows (which requires an auxiliary lemma).

 \begin{defi}[Generators and candidates]
A \emph{generator} of type $\form$ (resp. co-type $\form$) is a sub-set $\settone\subseteq\lltermsform$ (resp. 
$\settone\subseteq\elctxsform$) such that: 
\begin{enumerate}
\item \emph{Non-emptyness}: $\settone \neq \emptyset$, and 
\item \emph{Strong normalization}: $\settone\subseteq \sn\toss$.
\end{enumerate}
A generator $\settone$ of type $\form$ (resp. co-type $\form$) is a \emph{candidate} if 
\begin{enumerate}
\item \emph{Variables}: $\vars_{\form}\subseteq \settone$ (resp. $\elvars_{\form}\subseteq \settone$), and 
\item \emph{Bi-orthogonal}: $\settone=\settone^{\bot\bot}$.
\end{enumerate}
\end{defi}

 \begin{lem}[Duality and SN]
\label{l:sn-to-subterms} % \reflemmap{sn-to-subterms}{}
Let $\tm \in \lltermsform$ and $\elctx \in \elctxsform$. 
\begin{enumerate}
\item \label{p:sn-to-subterms-one} 
If $\elctxfp\tm\in\sn\toss$ then $\tm \in \sn\toss$ and $\elctx\in\sn\toss$.

\item \label{p:sn-to-subterms-two} 
If $\tm \in \sn\toss$ and $\elctx \in \elvars_{\form}$ then $\elctxfp\tm\in \sn\toss$.

\item \label{p:sn-to-subterms-three} 
If $\elctx \in \sn\toss$ and $\vartwo \in \vars_{\form}$ then $\elctxfp\vartwo\in \sn\toss$.
\end{enumerate}
\end{lem}

\begin{proof}
\hfill
\begin{enumerate}
\item By induction on $\elctxfp\tm\in\sn\toss$. Let $\tm = \lctxp\val$ and $\elctx = 
\lctxtwop{\cuta\ctxhole\var\tmtwo}$, so that $\elctxfp\tm = \lctxtwop{\lctxp{\cuta\val\var\tmtwo}}$. The proof of the 
statement is based on the obvious fact that every step from $\tm$ or $\lctxtwop\tmtwo$ (we rely on \refrem{sn-renaming}) 
can be mimicked on $\lctxtwop{\lctxp{\cuta\val\var\tmtwo}}$, so that one can then apply the \ih
%%%
\item If $\elctx \in \elvars_{\form}$ then $\elctx = \cuta\ctxhole\var\var$. Let $\tm = \lctxp\val$. We  have to 
prove that $\elctxfp\tm= \lctxp{\cuta\val\var\var} \in \sn\toss$. By induction on $\tm \in \sn\toss$ we show that all the reducts of $\tmtwo$ are in $\sn\toss$. If $\elctxfp\tm \toss \tm$ by reducing the cut on $\var$, 
then the reduct is $\tm$, which is in $\sn\toss$ by hypothesis. Otherwise, $\elctxfp\tm$ makes a step in $\val$, or in 
$\lctx$, or involving both $\lctx$ and $\val$. The same step can be done on $\tm$, and thus by \ih the reduct is in 
$\sn\toss$.
\item Let $\lctxp{\cuta\ctxhole\var\tm}\in \elctx$, so that $\elctxfp\vartwo = \lctxp{\cuta\vartwo\var\tm}$. If 
$\elctxfp\vartwo \toss \lctxp{\cutsub\vartwo\var\tm}$ then the reduct is in $\sn\toss$ by the hypothesis on $\elctx$. 
Otherwise, $\elctxfp\vartwo$ makes a step in $\tm$, or in $\lctx$, or involving both $\lctx$ and $\tm$. The same step 
can be done on $\lctxp\tm$ (we rely on \refrem{sn-renaming}), and thus by \ih the reduct is in $\sn\toss$.\qedhere
\end{enumerate}
\end{proof}

\begin{toappendix}
 \begin{prop}
 \label{prop:generators-give-red-cands}
If $\settone$ is a generator then $\settone^\bot$ is a candidate. 
 \end{prop}
\end{toappendix}

 \begin{proof}
We first consider the case $\settone \subseteq \lltermsform$. Properties:
\begin{itemize}
\item \emph{Strong normalization}: let  $\elctx = \lctxtwop{\cuta\ctxhole\var\tmtwo} \in 
\settone^\bot$ and $\tm= \lctxp\val\in\settone \neq\emptyset$. By duality, $\elctxfp\tm \in\sn\toss$. By 
\reflemmap{sn-to-subterms}{one}, $\elctx\in\sn\toss$. 

\item \emph{Variables} (which subsumes non-emptyness):  let $\tm = \lctxp\val\in \settone$ and $\cuta\ctxhole\var\var\in 
\elvars_{\form}$. By $\settone \subseteq \sn\toss$ and \reflemmap{sn-to-subterms}{two} we obtain 
$\lctxp{\cuta\val\var\var}\in\sn\toss$, that is, $\cuta\ctxhole\var\var \in \settone^{\bot}$.

\item \emph{Bi-orthogonal}: by the by-orthogonal property of duality (\reflemma{basic-prop-duality}).
\end{itemize}
Now, the case $\settone\subseteq\elctxsform$. 
Properties:
\begin{itemize}
\item \emph{Strong normalization}: let  $\tm= \lctxp\val \in 
\settone^\bot$ and $\elctx = \lctxtwop{\cuta\ctxhole\var\tmtwo} \in\settone\neq\emptyset$. By duality, $\elctxfp\tm 
\in\sn\toss$. By \reflemmap{sn-to-subterms}{one}, $\tm\in\sn\toss$. 

\item \emph{Variables} (which subsumes non-emptyness):  let $\elctx \in\settone$ and $\var\in \vars_{\form}$. By 
$\settone \subseteq \sn\toss$ and \reflemmap{sn-to-subterms}{three} we obtain $\elctxfp\var\in\sn\toss$, that is, $\var 
\in \settone^{\bot}$.

\item \emph{Bi-orthogonal}: by the by-orthogonal property of duality (\reflemma{basic-prop-duality}).\qedhere
\end{itemize}
\end{proof}

We now associate to every formula $\form$ a set $\settone_{\redc\form}$ that we then prove to be a generator, so that 
its bi-orthogonal is a candidate by \refprop{generators-give-red-cands}. A minor unusual point is that we define the set 
$\settone_{\redc\aform}$ for the atomic formula as the set of multiplicative variables. The literature rather defines it 
as $\lltermsp\aform\cap\sn{\toss}$, which is wacky, as this is actually what the method is meant to prove!

\begin{defi}[Formulas candidates]
Set $\redc{\form} \defeq \settone_{\redc\form}^{\bot\bot}$, where 
$\settone_{\redc\form}$ is defined by induction on $\form$:
 \begin{itemize}
  \item  $\settone_{\redc\aform} \defeq \mvars$;
  \item  $\settone_{\redc{\form\tens\formtwo}} \defeq \set{\pair\tm\tmtwo\ |\ 
\tm\in\redc{\form}, \tmtwo\in\redc{\formtwo}, \pair\tm\tmtwo \mbox{ is \proper}}$;

  \item  $\settone_{\redc{\form\lolli\formtwo}} \defeq \left\{\la\var\tm\ 
\Bigg|\ \begin{array}{l}
\multiform, \var\hastype\form\vdash\tm\hastype\formtwo, \\
\lctxp{\cuta\val\var\tm} \in\redc{\formtwo}\, \forall\tmtwo =\lctxp\val\in\redc\form,
\mbox{ and }\\
\la\var\tm\mbox{ is \proper}
\end{array}\right\}$;

  \item  $\settone_{\redc{\bang\form}} \defeq \set{\bang\tm\ |\ \tm\in\redc{\form},  \bang\tm\mbox{ is 
\proper}}$.
 \end{itemize}
\end{defi}

The proof that $\redc{\form}$ is a candidate for every $\form$ requires a lemma for the lolli case.

\begin{lem}
\label{l:formulas-give-redc-aux}
If $\redc\form\subseteq\sn\toss$ and $\vars_\formtwo \subseteq \redc\formtwo$ then $\settone_{\redc{\form\lolli\formtwo}} \neq \emptyset$.
\end{lem}
\begin{proof}
We show that $\la\var\suba\mvar\var\vartwo\vartwo \in \settone_{\redc{\form\lolli\formtwo}}$. The idea is that the body $\suba\mvar\var\vartwo\vartwo$ of that abstraction is the smallest term $\tm$ verifying the typing requirement $\multiform, \var\hastype\form\vdash\tm\hastype\formtwo$ when $\form\neq\formtwo$. By definition, $\la\var\suba\mvar\var\vartwo\vartwo \in \settone_{\redc{\form\lolli\formtwo}}$ holds if $\lctxp{\cuta\val\var \suba\mvar\var\vartwo\vartwo} \in \redc{\formtwo}$ for every $\lctxp\val\in\redc\form$, that is,  $\elctxfp{\lctxp{\cuta\val\var \suba\mvar\var\vartwo\vartwo}} \in \sn\toss$ for every $\elctx\in\redc\formtwo^{\bot}$. 
By hypothesis, $\vars_{\formtwo} \subseteq \redc\formtwo$, and duality gives $\elctxfp\vartwo \in \sn\toss$. By hypothesis, we also have $\lctxp\val \in \sn\toss$. The proof is in two steps:\begin{itemize}
\item Proving $\tm \defeq \elctxfp{\lctxp{\suba\mvar\val\vartwo\vartwo}}\in \sn\toss$: note that the two components $\lctxp\val$ and $\elctxfp{\vartwo}$ cannot interact in any way, so that every step of $\tm$ comes from a step of either $\lctxp\val$ or $\elctxfp{\vartwo}$, as in the proof of the extension property. Formally, the proof is by induction on $(\lctxp\val \in \sn\toss, \elctxfp{\vartwo} \in \sn\toss)$.
\item Proving $\tmtwo \defeq \elctxfp{\lctxp{\cuta\val\var\suba\mvar\var\vartwo\vartwo} } \in \sn\toss$: note that  $\tm = \elctxfp{\lctxp{\suba\mvar\val\vartwo\vartwo}} = \cutsub\val\var\elctxfp{\lctxp{\suba\mvar\var\vartwo\vartwo}}$, thus by root cut expansion we obtain $\cuta\val\var\elctxfp{\lctxp{\suba\mvar\var\vartwo\vartwo}} \in \sn\toss$. By structural stability, $\cuta\val\var\elctxfp{\lctxp{\suba\mvar\var\vartwo\vartwo}} \cuteq \tmtwo \in \sn\toss$.
\end{itemize}
Therefore, $\la\var\suba\mvar\var\vartwo\vartwo \in \settone_{\redc{\form\lolli\formtwo}} \neq \emptyset$.
\end{proof}
The proof of the next proposition is simple and yet tricky: by induction on $\form$, it uses 
\refprop{generators-give-red-cands} to prove 2 from 1, and 3 from 2, but it also needs 3 (on sub-formulas) to prove 1. 

\begin{prop}[Formulas induce reducibility candidates]
\label{prop:formulas-give-redc}
Let $\form$ be a \imell\ formula. 
\begin{enumerate}
\item \emph{Generators}: $\settone_{\redc\form}$ is a generator.
\item \emph{Dual candidates}: $\redc{\form}^{\bot}$ is a  candidate.

\item \emph{Candidates}: $\redc{\form}$ is a candidate.
\end{enumerate}
\end{prop}

\input{\proofspath/sn-explicit/formulas-induce-candidates}

\paragraph{Reducibility and Adequacy.} In the following notion of reducible derivations, the rigid structure of terms 
forces an order between the assignments in the typing context $\multiform$, that is, it treats $\multiform$ as a 
\emph{list} rather than as a multi-set. The proof of adequacy then considers that the sequent calculus comes with an 
exchange rule, treated by one of the cases. 
\begin{defi}[Reducible derivations]
Let $\tderiv\pof\multiForm \vdash \tm\hastype \form$ be a typing derivation, and  $\multiForm = 
\var_{1}\hastype\formtwo_{1},\mydots, 
\var_{k}\hastype\formtwo_{k} $. Then $\tderiv$ is \emph{reducible} if 
$\cuta{\val_{1}}{\var_{1}}\ldots \cuta{\val_{k}}{\var_{k}} \tm \in \redc\form$
 for every value 
$\val_{i}\in\settone_{\redc{\formtwo_{i}}}$ such that the introduced cuts are independent, that is, 
$\fv{\val_{i}} \cap \dom\multiForm = \emptyset$. 
To ease notations, we shorten $\cuta{\val_{1}}{\var_{1}}\ldots \cuta{\val_{k}}{\var_{k}} \tm$ to 
$\cuta{\val_{i}}{\var_{i}}_{\multiForm}\tm$.
\end{defi}

Before proving adequacy, we provide an equivalent but slightly extended reformulation of the reducibility clause, in 
which the variables $\var_i$ are cut with arbitrary terms, not necessarily with values. The variant is used in the  
proof of adequacy below when dealing with: 
\begin{itemize}
\item the left rules, as they either have $\lctxp\val$ as hypothesis (cut and subtraction) or their associated rewriting rules (which are used in the proof) split terms on-the-fly (par and dereliction);
\item the right rule for $\lolli$, because split terms are used in the definition of $\settone_{\redc{\form\lolli\formtwo}}$.
\end{itemize}
\emph{Notations}: if $\tmtwo = \lctxp\val$, we use $\cuta{\ctxholep\tmtwo}\var\tm$ for $\lctxp{\cuta\val\var\tm}$, and given $\multiForm = \var_{1}\hastype\formtwo_{1},\mydots, 
\var_{k}\hastype\formtwo_{k} $ and $\tmtwo_i$ of type $\formtwo_{i}$ for $i=1,\ldots,k$ we use
$\cuta{\ctxholep{\tmtwo_{i}}}{\var_{i}}_{\multiForm}\tm$ for $\cuta{\ctxholep{\tmtwo_{1}}}{\var_{1}}\ldots 
\cuta{\ctxholep{\tmtwo_{k}}}{\var_{k}} \tm$.

\begin{lem}[Extended reducibility clause]
\label{l:reduc-hyp-simpl}
Let $\tderiv\pof \multiForm \vdash \tm\hastype \form$ with $\multiForm = \var_{1}\hastype\formtwo_{1},\mydots, 
\var_{k}\hastype\formtwo_{k} $  be a typed derivation. Then $\tderiv$ is \emph{reducible} if and 
only if $\cuta{\ctxholep{\tmtwo}}{\var_{i}}_{\multiForm}\tm \in \redc\form$ for every $\tmtwo\in \redc{\formtwo_{i}}$ 
such that the introduced cuts are independent.
\end{lem}

\input{\proofspath/sn-explicit/simpler-reducibility-clause}
\begin{toappendix}
\begin{thm}[Adequacy]
\label{tm:adequacy}
%Let $\to$ be a substitutive rewriting relation. Every \pn\ is $\to$-reducible.
Let $\tderiv \pof \multiForm \vdash \tm\hastype \form$ a type derivation. Then $\tderiv$ is reducible.
\end{thm}
\end{toappendix}

\input{\proofspath/sn-explicit/adequacy-body}

\begin{toappendix}
\begin{cor}[Typable terms are SN]
\label{coro:sn}
Let $\tm$ be a typable term. Then $\tm\in\sn\toss$ and $\tm\in\sn\toms$.
\end{cor}
\end{toappendix}

\begin{proof}
Since $\tm$ is typable, we have $\tderiv\pof\var_{1}\hastype\formtwo_{1},\ldots, \var_{k}\hastype\formtwo_{k} \vdash 
\tm\hastype \form$, for some derivation $\tderiv$.  By adequacy (\reftm{adequacy}), $\tderiv$ is reducible. By 
\refprop{formulas-give-redc}, $\var_{i} \in \redc{\formtwo_{i}}$ 
for $i\in\set{1,\ldots,k}$ and 
$\cuta\ctxhole\vartwo\vartwo \in \redc\form^{\bot}$. Let $\tm = \lctxp\val$. By reducibility of $\tderiv$, 
$\tmtwo\defeq\cuta{\var_{1}}{\var_{1}}\ldots\cuta{\var_{k}}{\var_{k}}\lctxp{\cuta{\val}\vartwo\vartwo} \in \sn\toss$. 
Note that $\tmtwo \toss^* \tm$, thus $\tm\in\sn\toss$. By PSN (\refthm{PSN}), $\tm\in\sn\toms$.
\end{proof}

%% file: proofs/sn-explicit/formulas-induce-candidates.tex
% !TEX root = ../../main.tex
\begin{proof}
We prove the first point, the second follows from the first and \refprop{generators-give-red-cands}, the third one follows from the second and \refprop{generators-give-red-cands}. By induction on $\form$. All the inductive cases of the proof use the extension property (\reflemma{extension}). Cases:
\begin{itemize}
\item \emph{Base}, \ie $\form = \aform$. Note that $\settone_{\redc\aform}$ is non-empty by definition and that all its elements are normal. 

\item \emph{Tensor}, \ie $\form = \formtwo\tens\formthree$. 
By point 3 of the \ih, $\redc\formtwo,\redc\formthree\subseteq \sn\toss$ and they are non-empty. Then $\settone_{\redc{\formtwo\tens\formthree}}$ is non-empty. By extension, $\settone_{\redc{\formtwo\tens\formthree}} \subseteq \sn\toss$.

\item \emph{Implication}, \ie $\form = \formtwo\lolli\formthree$. 
Proving that $\settone_{\redc{\formtwo\lolli\formthree}}$ is non-empty requires the previous lemma, because by definition it does not contains the variables of type $\formtwo\lolli\formthree$, nor it is defined by abstracting terms in $\redc\formthree$. 
By \ih (point 3) applied to $\formthree$, $\vars_{\formthree} \subseteq \redc\formthree$, and By \ih (point 3) applied to $\formtwo$, $\lctxp\val \in \sn\toss$. Then by \reflemma{formulas-give-redc-aux}, $\settone_{\redc{\formtwo\lolli\formthree}}\neq\emptyset$.

 Since variables are in $\redc\formtwo$ (by Point 3 of the \ih), if $\la\var\tmthree\in \settone_{\redc{\formtwo\lolli\formthree}}$ then $\cuta\var\var\tmthree \in \redc\formthree$. By \ih (Point 3), $\cuta\var\var\tmthree \in \sn\toss$ and so does $\tmthree$. By extension, $\la\var\tmthree \in \sn\toss$. Therefore $\settone_{\redc{\formtwo\lolli\formthree}} \subseteq \sn\toss$.

\item \emph{Bang}, \ie $\form = \bang\formtwo$. 
By point 3 of the \ih, $\redc\formtwo\subseteq \sn\toss$ and it is non-empty. Then $\settone_{\redc{\bang\formtwo}}$ is non-empty. By extension, $\settone_{\redc{\bang\formtwo}} \subseteq \sn\toss$. \qedhere
\end{itemize}
\end{proof}

%% file: proofs/sn-explicit/simpler-reducibility-clause.tex
% !TEX root = ../../main.tex
\begin{proof}
Direction $\Leftarrow$ is obvious because $\settone_{\redc{\formtwo_{i}}} \subseteq \redc{\formtwo_{i}}$, we prove direction $\Rightarrow$. Let $\multiform = \multiFormtwo,\var\hastype\formtwo,\multiFormthree$.  The hypothesis is 
\begin{equation}
\begin{array}{llllllll}
\elctxfp{\cuta{\val_{i}}{\vartwo_{i}}_{\multiFormthree}\cuta{\val}\var\cuta{\val_{j}}{\varthree_{j}}_{\multiFormtwo}\tm} & \in &\sn\toss.
\end{array}
\label{eq:red-hyp}
\end{equation} 
for every $\val\in\settone_{\redc\formtwo}$ and appropriate $\val_{i}$ and $\val_{j}$. We show that we can replace $\val$ with $\tmtwo \in \redc\formtwo$, that is, that the following holds:
\begin{equation}
\begin{array}{llllllll}
\elctxfp{\cuta{\val_{i}}{\vartwo_{i}}_{\multiFormthree}\cuta{\ctxholep\tmtwo}\var\cuta{\val_{j}}{\varthree_{j}}_{\multiFormtwo}\tm} & \in &\sn\toss.
\end{array}
\label{eq:red-concl}
\end{equation}
By iterating the reasoning on all other values $\val_{i}$ and $\val_{j}$ one obtains the statement. Since \refeq{red-hyp} holds for all $\val\in\settone_{\redc\formtwo}$ we have that 
$
\elctxtwo \defeq \elctxfp{\cuta{\val_{i}}{\vartwo_{i}}_{\multiFormthree}\cuta{\ctxhole}\var\cuta{\val_{j}}{\varthree_{j}}_{\multiFormtwo}\tm}\in \settone_{\redc\formtwo}^{\bot} = \redc\formtwo^{\bot}
$.
Note that the cuts $\cuta{\val_{i}}{\vartwo_{i}}_{\multiFormthree}$ and $\cuta{\val_{j}}{\varthree_{j}}_{\multiFormtwo}$ can be added to the elimination context exactly because they are independent. Now, by duality we obtain $\elctxtwofp\tmtwo \in \sn\toss$ for every $\tmtwo \in \redc\formtwo$, which is exactly \refeq{red-concl}.
\end{proof}

%% file: proofs/sn-explicit/adequacy-body.tex
% !TEX root = ../../main.tex
\begin{proof}
By induction on $\tderiv$. The proof rests on the three rewriting properties mentioned at the beginning of the section, namely extension, root cut expansion, and  structural stability. Let $\multiForm = \var_{1}\hastype\formtwo_{1},\mydots, 
\var_{k}\hastype\formtwo_{k} $. We show the most relevant cases, the other ones are in \refapp{app-sn}, page \pageref{sect:app-sn}. Cases of the last rule of $\tderiv$:%To ease the notation, along the proof we do not write the permutation $\sigma$ in the clause of reducible terms, and when the last rule of $\tderiv$ is a left rule, we only deal with the case in which the first cut is on it---the reasoning shall not however depend on the order of the cuts. 
\begin{itemize}
\item \emph{Exponential axiom}: 
\begin{center}
\AxiomC{}
	\RightLabel{$\ax_{\esym}$}
	\UnaryInfC{$\evar\hastype\bang\form \vdash 
\evar\hastype\bang\form$}
	\DisplayProof
\end{center}
We need to show that $\elctxfp{\cuta{\bang\tmtwo}\evar\evar} 
\in\sn\toss$ for every term $\tmtwo\in\redc{\form}\subseteq\sn\toss$ and every elimination context $\elctx\in\redc{\bang\form}^{\bot}$. Note that 
$\elctxfp{\cuta{\bang\tmtwo}\evar\evar} \tobang \elctxfp{\bang\tmtwo}$ which 
is in $\sn\toss$ by duality. Note also  that $\elctxfp{\bang\tmtwo} = \elctxfp{\cutsub{\bang\tmtwo}\evar\evar} = 
\cutsub{\bang\tmtwo}\evar\elctxfp{\evar}$. By extension, $\bang\tmtwo\in\sn\toss$. Since both $\bang\tmtwo$ and 
$\cutsub{\bang\tmtwo}\evar\elctxfp{\evar}$ are in $\sn\toss$, root cut expansion gives $\cuta{\bang\tmtwo}\evar\elctxfp{\evar} 
\in\sn\toss$. By structural stability, $\cuta{\bang\tmtwo}\evar\elctxfp{\evar} \cuteq 
\elctxfp{\cuta{\bang\tmtwo}\evar\evar}\in\sn\toss$. 

Note that the two terms $\cuta{\bang\tmtwo}\evar\elctxfp{\evar}$ and $
\elctxfp{\cuta{\bang\tmtwo}\evar\evar}$ do translate to the same proof net, which is why structural stability does not appear in the proofs of SN for proof nets. In the intuitionistic case with weakenings, however, one needs some form of jumps and jump rewiring rules, so an analogous of structural stability would be most probably needed. There are no proofs of SN for intuitionistic proof nets in the literature.

%%%%%%%%%%%%%%%%%
%%%%%%%%%%%%%
\item \emph{Exponential inductive cases}. We show the promotion and contraction cases, the weakening and dereliction cases are minor variations over the contraction one.
\begin{itemize}
 
 %%%%%%%%%%%%%%
 \item \emph{Promotion}:
\begin{center}
	\AxiomC{$ \tderivtwo\pof \bang\multiFormtwo \vdash \tmtwo\hastype\formthree$}
	\RightLabel{$ \bangRightRule$}
	\UnaryInfC{$  \bang\multiFormtwo \vdash \bang\tmtwo\hastype\bang\formthree$}
	\DisplayProof
\end{center}
With $\tm = \bang\tmtwo$ and $\multiForm= \bang\multiFormtwo$. Let  $\bang\multiFormtwo= \evar_{1}\hastype\bang\formtwo_{1},\ldots,\evar_{k}\hastype\bang\formtwo_{k}$. We need to 
prove that $\cuta{\bang\tmthree_{i}}{\evar_{i}}_{\bang\multiFormtwo}\bang \tmtwo\in \redc{\bang\formthree}$ for every term $\tmthree_{i}\in\redc{\formtwo_{i}}$, that is, $\elctxfp{\cuta{\bang\tmthree_{i}}{\evar_{i}}_{\bang\multiFormtwo}\bang \tmtwo} \in\sn\toss$   for every elimination context $\elctx\in\redc{\bang\formthree}^{\bot}$. By \ih, 
$\tderivtwo$ is reducible, therefore $\cuta{\bang\tmthree_{i}}{\evar_{i}}_{\bang\multiFormtwo}\tmtwo\in\redc\formthree$, 
and so $\bang\cuta{\bang\tmthree_{i}}{\evar_{i}}_{\bang\multiFormtwo}\tmtwo\in\redc{\bang\formthree}$, that is, 
$\elctxfp{\bang\cuta{\bang\tmthree_{i}}{\evar_{i}}_{\bang\multiFormtwo}\tmtwo} \in\sn\toss$. Note that we have
\begin{center}\arraycolsep=4pt
$\begin{array}{llllllll}
\elctxfp{\bang\cuta{\bang\tmthree_{i}}{\evar_{i}}_{\bang\multiFormtwo}\tmtwo} 
&\toess^{*} &
\elctxfp{\bang\cutsub{\bang\tmthree_{i}}{\evar_{i}}_{\bang\multiFormtwo}\tmtwo}
&=&
\elctxfp{\cutsub{\bang\tmthree_{i}}{\evar_{i}}_{\bang\multiFormtwo}\bang\tmtwo} 
&=& 
\cutsub{\bang\tmthree_{i}}{\evar_{i}}_{\bang\multiFormtwo}\elctxfp{\bang\tmtwo}
\end{array}$\end{center} 
which is then in $\sn\toss$. By hypothesis, $\tmthree_{i}\in\redc{\formtwo_{i}}$, which implies $\tmthree_{i}\in\sn\toss$ by the properties of candidates (\refprop{formulas-give-redc}), thus 
$\bang\tmthree_{i}\in\sn\toss$ by extension. By root cut expansion, 
$\cuta{\bang\tmthree_{i}}{\evar_{i}}_{\bang\multiFormtwo}\elctxfp{\bang\tmtwo} \in \sn\toss$. By structural stability,  
\begin{center}$\begin{array}{llllllll}
\cuta{\bang\tmthree_{i}}{\evar_{i}}_{\bang\multiFormtwo}\elctxfp{\bang\tmtwo} 
&\cuteq &
\elctxfp{\cuta{\bang\tmthree_{i}}{\evar_{i}}_{\bang\multiFormtwo}\bang\tmtwo}
&\in&\sn\toss.
\end{array}$\end{center} 

%%%%%%%%%%%%%%
\item \emph{Contraction}:
\begin{center}
		\AxiomC{$  \multiFormtwo,\evar\hastype\bang\formtwo, \evartwo\hastype\bang\formtwo  \vdash \tmtwo\hastype\form$}
	\RightLabel{$ \contrRule $}
	\UnaryInfC{$   \multiFormtwo,\evar\hastype\bang\formtwo \vdash \cutsub\evar\evartwo\tmtwo\hastype \form$}
	\DisplayProof
\end{center}
with $\tm = \cutsub\evar\evartwo\tmtwo$ and $\multiForm = \multiFormtwo, \evar\hastype\bang\formtwo$. We use the notation 
$\var_{i}\hastype\formtwo_{i}$ for the assignments in $\multiformtwo$.
 By \ih, 
$\tmtwo' \defeq 
\elctxfp{\cuta{\val_{i}}{\var_{i}}_{\multiformtwo}\cuta{\bang\tmthree}\evar\cuta{\bang\tmthree}\evartwo\tmtwo} 
\in\sn\toss$ for every $\tmthree\in\redc\formtwo$, every $\val_{i}\in\settone_{\redc{\formtwo_{i}}}$ for 
$i\in\set{1,\ldots,k}$, and every $\elctx\in\redc\form^{\bot}$. We have: 
\begin{center}$\begin{array}{llllllll}
\tmtwo' &\toess^{*} &
\elctxfp{\cuta{\val_{i}}{\var_{i}}_{\multiformtwo}\cutsub{\bang\tmthree}\evar\cutsub{\bang\tmthree}\evartwo\tmtwo}& \in &
\sn\toss.\end{array}$\end{center} 
By the properties of meta-level substitution (namely \reflemmap{bang-subs-prop}{one}), we have 
\begin{center}$\begin{array}{llllllll}
\elctxfp{ 
\cuta{\val_{i}}{\var_{i}}_{\multiformtwo}\cutsub{\bang\tmthree}\evar\cutsub{\bang\tmthree}\evartwo\tmtwo} & =& 
\elctxfp{\cuta{\val_{i}}{\var_{i}}_{\multiformtwo}\cutsub{\bang\tmthree}\evar\cutsub\evar\evartwo\tmtwo}\end{array}$\end{center} 
which is 
equal to $\elctxfp{\cutsub{\bang\tmthree}\evar\cuta{\val_{i}}{\var_{i}}_{\multiformtwo}\cutsub\evar\evartwo\tmtwo}$
 by 
the independence of cuts in the definition of reducibility, in turn equal to 
$\cutsub{\bang\tmthree}\evar\elctxfp{\cuta{\val_{i}}{\var_{i}}_{\multiformtwo}\cutsub\evar\evartwo\tmtwo}$. Since 
$\tmthree\in\redc\formtwo$, we have $\tmthree\in \sn\toss$  by the properties of candidates (\refprop{formulas-give-redc}), thus $\bang\tmthree\in \sn\toss$ by extension. By root cut expansion, 
$\cuta{\bang\tmthree}\evar\elctxfp{\cuta{\val_{i}}{\var_{i}}_{\multiformtwo}\cutsub\evar\evartwo\tmtwo} \in\sn\toss$. 
By structural stability, 
\begin{center}$\begin{array}{llllllll}
\cuta{\bang\tmthree}\evar\elctxfp{\cuta{\val_{i}}{\var_{i}}_{\multiformtwo}\cutsub\evar\evartwo\tmtwo} 
&\cuteq& 
\elctxfp{\cuta{\val_{i}}{\var_{i}}_{\multiformtwo}\cuta{\bang\tmthree}\evar\cutsub\evar\evartwo\tmtwo} &\in&\sn\toss.
\end{array}$\end{center} 
which is exactly what is required.
\end{itemize}

%%%%%%%%%%%%%%
%%%%%%%%%%%%%%
\item \emph{Multiplicative inductive cases}. We show the lolli and subtraction cases. The subtraction case, in particular, combines the reasoning at work in both the simpler cut and par cases, which are in the appendix. 
\begin{itemize}	
 \item \emph{Lolli}:
\begin{center}
\AxiomC{$\tderivtwo\pof  \var\hastype\formtwo, \multiForm \vdash \tmtwo\hastype\formthree$}
	\RightLabel{$ \lolliRightRule $}
	\UnaryInfC{$  \multiForm \vdash \la\var\tmtwo\hastype\formtwo \lolli \formthree$}
	\DisplayProof
		\end{center}
	with $\tm =\la\var\tmtwo$ and $\form = \formtwo \lolli \formthree$.
	We have to show that $\tm' \defeq \elctxfp{\cuta{\val_{i}}{\var_{i}}_{\multiform}\la\var\tmtwo} 
\in\sn\toss$ for every $\val_{i} \in \settone_{\redc{\formtwo_{i}}}$ for $i\in\set{1,\ldots,k}$ and every 
$\elctx\in\redc{\formtwo \lolli \formthree}^{\bot}$.  By \ih on $\tderivtwo$ and the extended reducibility clause (\reflemma{reduc-hyp-simpl}), we have  
$\lctxtwop{\cuta\val\var\cuta{\val_{i}}{\var_{i}}_{\multiform}\tmtwo} \in\redc\formthree$ for every 
$\lctxtwop\val\in\redc\formtwo$ and with $\var\notin \fv{\val_i}$ for all $i$ because of the independence of cuts in the definition of reducibility. Then $\la\var\cuta{\val_{i}}{\var_{i}}_{\multiform}\tmtwo 
\in\redc{\formtwo\lolli\formthree}$. By duality, 
$\elctxfp{\la\var\cuta{\val_{i}}{\var_{i}}_{\multiform}\tmtwo} \in \sn\toss$. By structural stability and $\var\notin \fv{\val_i}$,
\begin{center}$\begin{array}{llllllll}
\elctxfp{\la\var\cuta{\val_{i}}{\var_{i}}_{\multiform}\tmtwo} &\cuteq &
\elctxfp{\cuta{\val_{i}}{\var_{i}}_{\multiform}\la\var\tmtwo} &= &\tm' &\in &\sn\toss.
\end{array}$\end{center} 
%%%%%%%%%%%%%%%%%
 \item \emph{Subtraction}:
\begin{center}
\AxiomC{$\tderiv_{l}\pof  \multiForm_{l} \vdash \lctxp\val\hastype\formthree$}
	\AxiomC{$\tderiv_{r}\pof  \multiForm_{r}, \var\hastype\formtwo \vdash \tmtwo\hastype\form$}
		\AxiomC{$\multiForm\#(\multiForm_{r}, \var\hastype\formtwo)$, $\mvar$ fresh}
	\RightLabel{$ \lolliLeftRule $}
	\TrinaryInfC{$   \multiForm_{l}, \multiForm_{r},\mvar\hastype \formthree \lolli \formtwo \vdash 
\lctxp{\suba\mvar\val\var  \tmtwo} \hastype\form$}
	\DisplayProof
	\end{center} 
	with $\tm =\lctxp{\suba\mvar\val\var  \tmtwo}$ and $\multiForm = \multiForm_{l}, \multiForm_{r},\mvar\hastype \formthree \lolli \formtwo$. This 
case combines the reasoning for the unary left rules with the one for cut. We have to show that 
\begin{center}$\begin{array}{llllllll}
\tm' &\defeq &
\elctxfp{\cuta{\val_{i}}{\var_{i}}_{\multiform}\cuta{\la\vartwo\tmfour}\mvar\lctxp{\suba\mvar\val\var  \tmtwo} }
&\in&\sn\toss.
\end{array}$\end{center}
 for every $\val_{i}\in\settone_{\redc{\formtwo_{i}}}$ for $i\in\set{1,\ldots,k}$, every 
$\la\vartwo\tmfour\in\settone_{\redc{\formthree\lolli\formtwo}}$, and every $\elctx\in\redc\form^{\bot}$. 
	By \ih on $\tderiv_{l}$, we have $\cuta{\val_{i}}{\var_{i}}_{\multiform_{l}}\lctxp{\val} 
\in\redc\formthree$. By definition of $\la\vartwo\tmfour\in\settone_{\redc{\formthree\lolli\formtwo}}$, we obtain 
$\cuta{\val_{i}}{\var_{i}}_{\multiform_{l}}\lctxp{\cuta\val\vartwo \tmfour}  \in\redc\formtwo$. Let $\tmfour =\lctxtwop\valtwo$. Then $\cuta{\val_{i}}{\var_{i}}_{\multiform_{l}}\lctxp{\cuta\val\vartwo \lctxtwop\valtwo}  \in\redc\formtwo$.
	By \ih $\tderiv_{r}$ is reducible. By the extended reducibility clause (\reflemma{reduc-hyp-simpl}), $\cuta{\val_{i}}{\var_{i}}_{\multiform_{r}} 
\cuta{\val_{i}}{\var_{i}}_{\multiform_{l}}\lctxp{\cuta\val\vartwo \lctxtwop{\cuta\valtwo\var
\tmtwo} } \in \redc\form$, that is,
	\begin{center}$\begin{array}{lll}
	\elctxfp{\cuta{\val_{i}}{\var_{i}}_{\multiform_{r}} 
\cuta{\val_{i}}{\var_{i}}_{\multiform_{l}}\lctxp{\cuta\val\vartwo \lctxtwop{\cuta\valtwo\var
\tmtwo} }}
	& =
		\\
	\elctxfp{\cuta{\val_{i}}{\var_{i}}_{\multiform} 
\lctxp{\cuta\val\vartwo \lctxtwop{\cuta\valtwo\var
\tmtwo} }} & \in&\sn\toss.
	\end{array}$\end{center}
%	By \ben{STRUCTURAL STABILITY},
%	\begin{center}$\begin{array}{lll}
%\elctxfp{\cuta{\val_{i}}{\var_{i}}_{\multiform} 
%\lctxp{\cuta\val\vartwo \lctxtwop{\cuta\valtwo\var
%\tmtwo} }}	& \streq 
%	\\
%	
%\lctxp{\cuta\val\vartwo \lctxtwop{\cuta\valtwo\var\elctxfp{\cuta{\val_{i}}{\var_{i}}_{\multiform} 
%
%\tmtwo} }} & \in\sn\toss
%	\end{array}$\end{center}
	By root cut expansion, $\cuta{\la\vartwo\lctxtwop\valtwo}{\mvar}  \elctxfp{\cuta{\val_{i}}{\var_{i}}_{\multiform}  \lctxp{\suba\mvar\val\var \tmtwo}}\in\sn\toss$. By structural stability, 
	$\cuta{\la\vartwo\lctxtwop\valtwo}{\mvar}  \elctxfp{\cuta{\val_{i}}{\var_{i}}_{\multiform}  \lctxp{\suba\mvar\val\var \tmtwo}} \cuteq \tm' \in \sn\toss$.\qedhere
	
	\end{itemize}
\end{itemize}
\end{proof}

%% file: 14-Conclusions.tex
% !TEX root = main.tex
\section{Conclusions}
\label{sect:conclusions}
This work lifts Accattoli and Kesner's \emph{linear substitution calculus} to the sequent calculus for IMELL via three generalizations: adding \emph{explicit (non-)linearity}, a simple form of \emph{pattern matching} ($\tens$), and replacing applications with \emph{subtractions}. 

We provide three main contributions for the new \emph{exponential substitution calculus}. Firstly, we show that IMELL untyped exponentials, in contrast to the classical case, are strongly normalizing. Secondly, we define the \emph{good evaluation strategy} and show that it has the sub-term property, obtaining the \emph{first} polynomial cost model for IMELL. Thirdly, we provide elegant proofs of confluence,  strong normalization, and preservation of strong normalization.

Methodologically, we provide an extensive study of linear logic cut elimination that is \emph{not} based on proof nets. 

\paragraph{The Next Step.} In the $\l$-calculus, the number of (leftmost) $\beta$-steps is a polynomial cost model \cite{DBLP:journals/corr/AccattoliL16}. Since \emph{one} $\beta$-steps is simulated in IMELL by \emph{one} multilicative step (actually a $\tololli$ step) followed by possibly \emph{many} exponential steps, it means that one can count only the number of multiplicative (or even $\tololli$) steps, that is, \emph{one can count zero for exponential steps}. Such a surprising fact is enabled by a sophisticated technique called \emph{useful sharing} \cite{DBLP:journals/corr/AccattoliL16,DBLP:conf/lics/AccattoliCC21,DBLP:conf/csl/AccattoliL22}. Can useful sharing be generalized to the ESC as to turn the number of multiplicative/$\tololli$ good steps into a polynomial time cost model? For this, termination of untyped exponentials is mandatory, thus, the results of this paper suggest that it might be possible. Still, the question is far from obvious, as in the $\l$-calculus there is a strong, hardcoded correlation between multiplicatives and exponentials, not present in IMELL. In the standard call-by-name/value encodings of $\l$-calculus in IMELL, indeed, multiplicatives and exponentials connectives rigidly alternate, while IMELL enables consecutive exponentials, as in $\bang\bang\form$, enabling wilder exponential behaviour.

\paragraph{Further Future Work} Beyond refining the cost model, there are many possible directions for future work. Is the good strategy normalizing on \emph{untyped} (clash-free) terms (that is, does it reach a cut-free reduct whenever there is one)? Good redexes are not stable by cut/left-equivalence: is there a generalization that is stable? Is it possible to measure good steps via relational semantics/multi types, refining the work of de Carvalho et al. \cite{DBLP:journals/tcs/CarvalhoPF11}? Being a generalization of linear head reduction, does the good strategy have a similar connection to game semantics (see \cite{DBLP:conf/lics/DanosHR96,DBLP:journals/corr/Clairambault15})? Can the good strategy be exploited for implicit computational complexity? Can the \emph{left context/value} splitting be given a logical status? What about a LSC/cost-aware approach for the additives?

%% file: APP-Typed_Strong_Normalization.tex
% !TEX root = main.tex
\section{Appendix of \refsect{SN} (Typed Strong Normalization)}
\label{sect:app-sn}

\gettoappendix{tm:adequacy}
\input{\proofspath/sn-explicit/adequacy-appendix}

%% file: proofs/sn-explicit/adequacy-appendix.tex
% !TEX root = ../../main.tex
\begin{proof}
Here we give the cases that are omitted from the body of the paper. Let $\multiForm = \var_{1}\hastype\formtwo_{1},\mydots, 
\var_{k}\hastype\formtwo_{k} $. Cases of the last rule of $\tderiv$:%To ease the notation, along the proof we do not write the permutation $\sigma$ in the clause of reducible terms, and when the last rule of $\tderiv$ is a left rule, we only deal with the case in which the first cut is on it---the reasoning shall not however depend on the order of the cuts. 
\begin{itemize}
\item \emph{Multiplicative axiom}:
\begin{center}
\AxiomC{$\form \neq \bang \formtwo$}
	\RightLabel{$\ax_{\msym}$}
	\UnaryInfC{$\mvar\hastype\form \vdash \mvar\hastype\form$}
	\DisplayProof
\end{center}
We need to show that $\elctxfp{\cuta\mval\mvar\mvar}\in\sn\toss$ for 
$\mval \in\settone_{\redc{\form}}\subseteq\sn\toss$ and $\elctx\in\redc\form^{\bot}$.  Note that 
$\elctxfp{\cuta{\mval}\mvar\mvar} \toaxmone \elctxfp{\cutsub{\mval}\mvar\mvar} = \elctxfp{\mval}$ which is in $\sn\toss$ 
by duality, and that $\elctxfp\mval =  \elctxfp{\cutsub{\mval}\mvar\mvar} = \cutsub{\mval}\mvar\elctxfp{\mvar}$. By 
root cut expansion, $\cuta{\mval}\mvar\elctxfp{\mvar} \in\sn\toss$. By structural stability, $\cuta{\mval}\mvar\elctxfp{\mvar} 
\cuteq \elctxfp{\cuta{\mval}\mvar\mvar } \in\sn\toss$.

%%%%%%%%%%%%%%%%%
%%%%%%%%%%%%%%%%%
\item \emph{Exchange}:
\begin{center}
\AxiomC{$ \tderivtwo\pof \multiformtwo,\var\hastype\formtwo,\vartwo\hastype\formthree,\multiformthree \vdash \tm\hastype \form$}
	\RightLabel{$\exch$}
	\UnaryInfC{$  \multiformtwo,\vartwo\hastype\formthree,\var\hastype\formtwo,\multiformthree \vdash \tm\hastype \form$}	
	\DisplayProof
	\end{center}
	with $\multiform = \multiformtwo,\vartwo\hastype\formthree,\var\hastype\formtwo,\multiformthree$. Let $\multiFormtwo = \var_{1}\hastype\formtwo_{1},\mydots, 
\var_{k}\hastype\formtwo_{k} $ and $\multiFormthree = \vartwo_{1}\hastype\formthree_{1},\mydots, 
\vartwo_{h}\hastype\formthree_{h} $. We have to 
show that
\begin{center}$\begin{array}{llllllll}
\tm' &\defeq& \elctxfp{\cuta{\val_{i}}{\var_{i}}_{\multiFormtwo}\cuta\val\vartwo\cuta\valtwo\var \cuta{\valtwo_{j}}{\vartwo_{j}}_{\multiFormthree} \tm} &\in& \sn\toss.
\end{array}$\end{center} 
for 
every $\val\in \settone_{\redc{\formthree}}$, $\valtwo\in \settone_{\redc{\formtwo}}$, $\val_{i} \in\settone_{\redc{\formtwo_{i}}}$ for $i\in\set{1,\ldots,k}$, every $\valtwo_{j} \in\settone_{\redc{\formthree_{j}}}$ for $i\in\set{1,\ldots,h}$, and every 
$\elctx\in\redc\form^{\bot}$.   By \ih 
on 
$\tderivtwo$, we have 
\begin{center}$\begin{array}{llllllll}
\tm'' &\defeq& \elctxfp{\cuta{\val_{i}}{\var_{i}}_{\multiFormtwo}\cuta\valtwo\var \cuta\val\vartwo\cuta{\valtwo_{j}}{\vartwo_{j}}_{\multiFormthree} \tm} &\in&\sn\toss.
\end{array}$\end{center} 
By structural stability and the independence of cuts in the definition of reducibility, $\tm'' \cuteq \tm' \in \sn\toss$.

%%%%%%%%%%%%%
%%%%%%%%%%%%%
\item \emph{Cut}:
\begin{center}
\AxiomC{$\tderiv_{l}\pof  \multiForm_{l} \vdash \lctxp\val\hastype\formtwo$}
	\AxiomC{$\tderiv_{\tmtwo}\pof  \multiForm_{\tmtwo}, \var\hastype\formtwo \vdash \tmtwo\hastype\form$}
		\AxiomC{$\multiForm_{l}\#(\multiForm_{\tmtwo}, \var\hastype\formtwo)$}
	\RightLabel{$ \cut $}
	\TrinaryInfC{$  \multiForm_{l}, \multiForm_{\tmtwo} \vdash \lctxp{\cuta{\val}\var \tmtwo} \hastype\form$}
	\DisplayProof
	\end{center} 
		with $\tm =\lctxp{\cuta{\val}\var \tmtwo}$ and $\multiForm = \multiForm_{l}, \multiForm_{\tmtwo}$.  We have to 
show that 
\begin{center}$\begin{array}{llllllll}
\tm' &\defeq& \elctxfp{\cuta{\val_{i}}{\var_{i}}_{\multiForm}\lctxp{\cuta\val\var \tmtwo}} &\in&\sn\toss.
\end{array}$\end{center} 
for 
every $\val_{i} \in\settone_{\redc{\formtwo_{i}}}$ for $i\in\set{1,\ldots,k}$ and every 
$\elctx\in\redc\form^{\bot}$.   By \ih 
on 
$\tderiv_{l}$, we have $\cuta{\val_{i}}{\var_{i}}_{\multiform_{l}}\lctxp\val \in\redc\formtwo$. By \ih,  
$\tderiv_{\tmtwo}$ is reducible, and by the extended reducibility clause (\reflemma{reduc-hyp-simpl}) we obtain $\cuta{\val_{i}}{\var_{i}}_{\multiform_{\tmtwo}} \cuta{\val_{i}}{\var_{i}}_{\multiform_{l}} 
\lctxp{\cuta{
\val}\var\tmtwo}\in \redc\form$, that is,
\begin{center}$\begin{array}{llllllll}
\elctxfp{\cuta{\val_{i}}{\var_{i}}_{\multiform_{\tmtwo}} \cuta{\val_{i}}{\var_{i}}_{\multiform_{l}} 
\lctxp{\cuta{
\val}\var\tmtwo}} &\in&\sn\toss.
\end{array}$\end{center} 
By structural stability and the independence of cuts in the definition of reducibility, 
\begin{center}$\begin{array}{llllllll}
\elctxfp{\cuta{\val_{i}}{\var_{i}}_{\multiform_{\tmtwo}} \cuta{\val_{i}}{\var_{i}}_{\multiform_{l}} 
\lctxp{\cuta{
\val}\var\tmtwo}} &\cuteq &\tm' &\in&\sn\toss.
\end{array}$\end{center}

%%%%%%%%%%%%%
%%%%%%%%%%%%%
\item \emph{Exponential inductive cases}. 
\begin{itemize}

%%%%%%%%%%%
 \item \emph{Weakening}:
\begin{center}
	 		\AxiomC{$\tderivtwo\pof  \multiFormtwo \vdash \tm\hastype\form$}
				\AxiomC{$\evar$ fresh}
	\RightLabel{$ \weakRule $}
	\BinaryInfC{$ \multiFormtwo,\evar\hastype\bang\formtwo \vdash \tm\hastype \form$}
	\DisplayProof
\end{center}
with $\multiForm= \multiFormtwo,\evar\hastype\bang\formtwo$. We use the notation $\var_{i}\hastype\formtwo_{i}$ for the 
assignments in $\multiformtwo$.  We have to show that 
$\elctxfp{\cuta{\val_{i}}{\var_{i}}_{\multiformtwo}\cuta{\bang\tmthree}\evar\tm} \in\sn\toss$ for every 
$\tmthree\in\redc\formtwo$, every $\val_{i}\in\settone_{\redc{\formtwo_{i}}}$ for $i\in\set{1,\ldots,k}$, and every 
$\elctx\in\redc\form^{\bot}$. By \ih on $\tderivtwo$, $\elctxfp{\cuta{\val_{i}}{\var_{i}}_{\multiformtwo}\tm} \in \sn\toss$. By freshness of $\evar$, 
\begin{center}$\begin{array}{llllllll}
\elctxfp{\cuta{\val_{i}}{\var_{i}}_{\multiformtwo}\tm} &=&
\elctxfp{\cutsub{\bang\tmthree}\evar\cuta{\val_{i}}{\var_{i}}_{\multiformtwo}\tm} &= &
\cutsub{\bang\tmthree}\evar\elctxfp{\cuta{\val_{i}}{\var_{i}}_{\multiformtwo}\tm}.
\end{array}$\end{center}
Since 
$\tmthree\in\redc\formtwo$, we have $\tmthree\in \sn\toss$ by the properties of candidates (\refprop{formulas-give-redc}), thus $\bang\tmthree\in \sn\toss$ by extension. By root cut expansion, 
$\cuta{\bang\tmthree}\evar\elctxfp{\cuta{\val_{i}}{\var_{i}}_{\multiformtwo}\tm} \in \sn\toss$. By structural 
stability and the independence of cuts in the definition of reducibility, $\cuta{\bang\tmthree}\evar\elctxfp{\cuta{\val_{i}}{\var_{i}}_{\multiformtwo}\tm} \cuteq 
\elctxfp{\cuta{\val_{i}}{\var_{i}}_{\multiformtwo}\cuta{\bang\tmthree}\evar\tm}\in\sn\toss$.
 
 %%%%%%%%%%%
 \item \emph{Dereliction}:
\begin{center}
	 			\AxiomC{$ \tderivtwo\pof \multiFormtwo, \var\hastype\formtwo\vdash \tmtwo\hastype\form$}
	\AxiomC{$\evar$ fresh}
	\RightLabel{$ \bangLeftRule$}
	\BinaryInfC{$  \multiFormtwo, \evar\hastype\bang\formtwo \vdash \dera\evar\var\tmtwo \hastype\form$}
	\DisplayProof
\end{center}
with $\tm = \dera\evar\var\tmtwo$ and $\multiForm = \multiFormtwo, \evar\hastype\bang\formtwo$. We use the notation 
$\var_{i}\hastype\formtwo_{i}$ for the assignments in $\multiformtwo$. We have to 
show that 
\begin{center}$\begin{array}{llllllll}
\elctxfp{\cuta{\val_{i}}{\var_{i}}_{\multiformtwo}\cuta{\bang\tmthree}\evar \dera\evar\var\tmtwo} 
&\in&\sn\toss.
\end{array}$\end{center}
 for every $\tmthree\in\redc\formtwo$, every $\val_{i}\in\settone_{\redc{\formtwo_{i}}}$ for 
$i\in\set{1,\ldots,k}$, and every $\elctx\in\redc\form^{\bot}$. Let $\tmthree = \lctxp\val$. By \ih, $\tderivtwo$ is reducible. by the extended reducibility clause (\reflemma{reduc-hyp-simpl}), $\cuta{\val_{i}}{\var_{i}}_{\multiformtwo}\lctxp{\cuta\val\var\tmtwo} \in \redc\form$, that is,
\begin{center}$\begin{array}{llllllll}
\elctxfp{\cuta{\val_{i}}{\var_{i}}_{\multiformtwo}\lctxp{\cuta\val\var\tmtwo}}
&\in&\sn\toss.
\end{array}$\end{center}
Note that by freshness of $\evar$ and the independence of cuts in the definition of reducibility, we have 
\begin{center}
\arraycolsep=3pt
$\begin{array}{llllllll}
\elctxfp{\cuta{\val_{i}}{\var_{i}}_{\multiformtwo}\lctxp{\cuta\val\var\tmtwo}}  
& = &
\elctxfp{\cuta{\val_{i}}{\var_{i}}_{\multiformtwo}\cutsub{\bang\tmthree}\evar\dera\evar\var\tmtwo} 
& = &
\cutsub{\bang\tmthree}\evar\elctxfp{\cuta{\val_{i}}{\var_{i}}_{\multiformtwo}\dera\evar\var\tmtwo}.
\end{array}$\end{center}
Since 
$\tmthree\in\redc\formtwo$, we have $\tmthree\in \sn\toss$  by the properties of candidates (\refprop{formulas-give-redc}), thus $\bang\tmthree\in \sn\toss$ by extension. By root cut expansion, 
$\cuta{\bang\tmthree}\evar\elctxfp{\cuta{\val_{i}}{\var_{i}}_{\multiformtwo}\dera\evar\var\tmtwo} \in \sn\toss$. 
By structural stability,  
\begin{center}$\begin{array}{llllllll}
\cuta{\bang\tmthree}\evar\elctxfp{\cuta{\val_{i}}{\var_{i}}_{\multiformtwo}\dera\evar\var\tmtwo} 
&\cuteq &
\elctxfp{ \cuta{\val_{i}}{\var_{i}}_{\multiformtwo}\cuta{\bang\tmthree}\evar\dera\evar\var\tmtwo}
&\in&\sn\toss.
\end{array}$\end{center}
 \end{itemize}
 
%%%%%%%%%%%%%%
%%%%%%%%%%%%%%
\item \emph{Multiplicative inductive cases}.\begin{itemize}
%%%%%%%%%%%%%%%%
\item \emph{Tensor}:
\begin{center}
\AxiomC{$ \tderiv_{\tmtwo} \pof \multiForm_{\tmtwo}\vdash \tmtwo\hastype\formtwo$}
 	\AxiomC{$ \tderiv_{\tmthree} \pof \multiForm_{\tmthree}\vdash \tmthree\hastype\formthree$}
	\AxiomC{$\multiForm_{\tmtwo}\#\multiForm_{\tmthree}$}
 	\RightLabel{$ \tensRightRule $}
 	\TrinaryInfC{$  \multiForm_{\tmtwo}, \multiForm_{\tmthree}\vdash \pair\tmtwo\tmthree \hastype\formtwo\tens \formthree$}
	\DisplayProof
	\end{center}
	with $\tm =\pair\tmtwo\tmthree$, $\form = \formtwo \tens \formthree$, and $\multiForm = \multiForm_{\tmtwo}, \multiForm_{\tmthree}$.	We have to show that $\tm' \defeq \elctxfp{ \cuta{\val_{i}}{\var_{i}}_{\multiForm} \pair\tmtwo\tmthree}$ is $\sn\toss$ for every $\val_{i}\in \settone_{\redc{\formtwo_{i}}}$ for $i\in\set{1,\ldots,k}$ and every $\elctx\in\redc{\formtwo \tens \formthree}^{\bot}$.  By \ih on $\tderiv_{\tmtwo}$ and $\tderiv_{\tmthree}$, $\tmtwo' \defeq \cuta{\val_{i}}{\var_{i}}_{\multiForm_{\tmtwo}}\tmtwo \in\redc\formtwo$ and $\tmthree' \defeq \cuta{\val_{i}}{\var_{i}}_{\multiForm_{\tmthree}}\tmthree \in\redc\formthree$, and so $\pair{\tmtwo'}{\tmthree'}\in \settone_{\redc{\formtwo\tens\formthree}} \subseteq \redc{\formtwo\tens\formthree}$, that is, $\elctxfp{\pair{\tmtwo'}{\tmthree'}}\in\sn\toss$. By structural stability, 
	\[\begin{array}{lllll}
	\elctxfp{\pair{\tmtwo'}{\tmthree'}} &=
	\\ 
	\elctxfp{\pair{\cuta{\val_{i}}{\var_{i}}_{\multiForm_{\tmtwo}} \tmtwo}{\cuta{\val_{i}}{\var_{i}}_{\multiForm_{\tmthree}}\tmthree}} & \cuteq 
	\\
	\elctxfp{\cuta{\val_{i}}{\var_{i}}_{\multiForm_{\tmtwo}}\cuta{\val_{i}}{\var_{i}}_{\multiForm_{\tmthree}} \pair{\tmtwo}{\tmthree}} & =
	\\
	\tm' \in\sn\toss.\end{array}\]%. % Was this dot intended? If yes, let us know

%%%%%%%%%%%%%%%%%%%%%%
 \item \emph{Par}:
\begin{center}
 \AxiomC{$\tderivtwo \pof \multiFormtwo, \var\hastype\formtwo, \vartwo\hastype\formthree \vdash \tmtwo\hastype\form$}
	\AxiomC{$\mvar$ fresh}
	\RightLabel{$ \tensLeftRule $}
	\BinaryInfC{$  \multiFormtwo, \mvar\hastype\formtwo \tens \formthree \vdash \para\mvar\var\vartwo  \tmtwo \hastype 
\form$}	
	\DisplayProof
	\end{center} 
	with $\tm =\para\mvar\var\vartwo  \tmtwo$  and $\multiForm = \multiFormtwo, \mvar\hastype\formtwo \tens 
\formthree$. We use the notation $\var_{i}\hastype\formtwo_{i}$ for the assignments in $\multiformtwo$.
The reducibility clause to prove is 
\begin{center}$\begin{array}{llllllll}
\tm' &\defeq &
\elctxfp{ \cuta{\val_{i}}{\var_{i}}_{\multiformtwo}\cuta{\pair{\tmthree}{\tmfour}}\mvar\para\mvar\var\vartwo  \tmtwo}&\in&\sn\toss
\end{array}$\end{center}

for every $\val_{i}\in\settone_{\redc{\formtwo_{i}}}$ for $i\in\set{1,\ldots,k}$, every 
$\pair{\tmthree}{\tmfour}\in\settone_{\redc{\formtwo\tens\formthree}}$, and every 
$\elctx\in\redc\form^{\bot}$. Let $\tmthree=\lctxp\val$ and $\tmfour=\lctxtwop\valtwo$. By \ih, $\tderivtwo$ is reducible, and by the extended reducibility clause (\reflemma{reduc-hyp-simpl}), $\cuta{\val_{i}}{\var_{i}}_{\multiformtwo} \lctxtwop{\cuta{\val}\var\lctxthreep{\cuta{\valtwo}\vartwo
\tmtwo}} \in \redc\form$, that is,
\begin{center}$\begin{array}{llllllll}
\tm'' &\defeq &
\elctxfp{ \cuta{\val_{i}}{\var_{i}}_{\multiformtwo} \lctxtwop{\cuta{\val}\var\lctxthreep{\cuta{\valtwo}\vartwo
\tmtwo}}}&\in&\sn\toss.
\end{array}$\end{center}
By root cut expansion, we have $\cuta{\pair{\tmthree}{\tmfour}}\mvar\elctxfp{ \cuta{\val_{i}}{\var_{i}}_{\multiformtwo}\para\mvar\var\vartwo  \tmtwo} \in\sn\toss$. By structural stability and the independence of cuts in the definition of reducibility, 
\begin{center}$\begin{array}{llllllll}
[\elctxfp{ \cuta{\val_{i}}{\var_{i}}_{\multiformtwo}\cuta{\pair{\tmthree}{\tmfour}}\mvar \para\mvar\var\vartwo  \tmtwo}&\in&\sn\toss.
\end{array}$\end{center}\vspace*{-\baselineskip}\qedhere

	\end{itemize}
\end{itemize}
\end{proof}